\documentclass[a4paper,onecolumn]{report}
\pdfoutput=1
\usepackage{preamble}
\usepackage[OT2,T1]{fontenc}
\DeclareSymbolFont{cyrletters}{OT2}{wncyr}{m}{n}
\DeclareMathSymbol{\Sha}{\mathalpha}{cyrletters}{"58}
\usepackage{hyperref}
\usepackage{tabularx}
\usepackage[numbers,sort]{natbib}

\title{The fabulous world of GKP codes}
\author{\names}
\date{\today}

\graphicspath{{figures/}}

\begin{document}

\begin{titlepage}

    \begin{center}
        \vspace*{1cm}
        
        \begin{Large}\textbf{The fabulous world of GKP codes} \end{Large}\\
       

        \vspace{3cm}
        
        { \textbf{Dissertation}}\\
        \text{zur Erlangung des Grades eines Doktors der Naturwissenschaften}\\
        \textit{(Dr.\ rer.\ nat)}
        
		\vspace{1cm}
		\text{am Fachbereich Physik}\\
		\text{der Freien Universität Berlin     }\\

		  \vspace{3cm}
		  \text{vorgelegt von}\\
		  \text{\large Jonathan Conrad}

		    \vfill
		   
        \text{Berlin,  Juli 2024}

        \vfill

	\end{center}     

\newpage
\thispagestyle{empty}
{ 

\begin{align*}
\text{Erstgutachter/in:}\hspace{1cm}   &\text{Prof.\  Dr.\  Jens Eisert} \\
\text{Zweitgutachter/in:}\hspace{1cm} &\text{Prof.\  Dr.\  Christiane Koch} \vspace{3cm} \\  \\
\text{Tag der Disputation:}\hspace{1cm} &\text{29.11.2024 }
\end{align*}
}
        
\end{titlepage}

\newpage
\pagenumbering{Roman}
\include{abstract_en}
\include{abstract_ger}
\section*{List of publications}

This dissertation builds on the following \textbf{publications}, prepared and published during the course of my PhD,

\begin{itemize}
    \item Terhal, B. M., Conrad, J., Vuillot, C., ``Towards Scalable Bosonic Quantum Error Correction'', \href{https://doi.org/10.1088/2058-9565/ab98a5}{Quantum Science and Technology, vol. 5, no. 4 (2020)}.
    \item Conrad, J., ``Twirling and Hamiltonian engineering via dynamical decoupling for Gottesman-Kitaev-Preskill quantum computing'', \href{https://doi.org/10.1103/PhysRevA.103.022404}{Phys. Rev. A 103, 022404 (2021)}.
    \item Conrad, J., Eisert, J., Arzani, F., ``Gottesman-Kitaev-Preskill Codes: A Lattice Perspective'', \href{https://doi.org/10.22331/q-2022-02-10-648}{Quantum, vol. 6 (2022)}.
    \item Conrad, J., Eisert, J., Seifert, J.P., ``Good Gottesman-Kitaev-Preskill codes from the NTRU cryptosystem'', \href{http://dx.doi.org/10.22331/q-2024-07-04-1398}{Quantum 8, 1398 (2024)}.
\end{itemize}%
Also included are the contents of the following \textbf{preprint},

\begin{itemize}

    \item Conrad, J., Burchards, A.G., Flammia, S.T., ``Lattices, Gates, and Curves: GKP codes as a Rosetta stone'', \href{https://arxiv.org/abs/2407.03270}{arXiv:2407.03270 (2024)}.

\end{itemize}

\section*{Acknowledgments}
There are many people I have to thank that have influenced and guided me to this point.
 I am grateful to all my teachers, my colleagues, my dear friends, and my family for continuously supporting and inspiring me; and for helping me to create space to foster and follow my curiosities. 

First and foremost I would like to thank Jens Eisert for giving me the opportunity to pursue research in his group and his constant dedication to building a welcoming, supportive and inspiring environment. None of the research presented in this thesis would have been possible without the positive space that Jens has created for me to grow in or without his support which I could always rely on. I am inspired by your pursuit of  building bridges between academic fields of research and your consistency in pursuing method- and insight oriented research. Showing that this is possible in a technology-oriented and increasingly competitive field is valuable inspiration that will continue to drive me.

I would like to thank Barbara Terhal for having guided me on my first steps in research and for being a highly influential person in my development. She has not only provided me with many life lessons, shaped my interest and attitude in research and showed overwhelming support for me in my steps into this world, but also guided me into research on Gottesman-Kitaev-Preskill codes, which are the topic of this thesis. 

I thank Arne Grimsmo for hosting me at the AWS center of quantum computing in Pasadena during the winter of 2022/2023, which ended up becoming a very valuable experience for me and  I thank Victor Albert for inviting me to QuICS for an enriching visit and our continuing exchanges. 

Pursuing research in this young and dynamic field has also filled my life with friends and inspiring people I am lucky to have gotten the chance to meet. I thank Cica Guistiani for having been a consistent companion to me in the past years and for the chance to have grown up in this field together. I thank Ryan Sweke for being a good friend, a cool- and inspiring guy, and for taking me by the hand into the world of climbing after I arrived in Berlin, which has gifted my much joy in the past years, and thanks to whom I have met my dear friends Basti and Yan.
I am eternally grateful to have become close friends, climbing partners, house- and office mates with Julio Carlos Magdalena de la Fuente, whose energy and endless curiosity has continued to help foster mine and I am very proud to have been able to see him become a leading figure in the group and his field of research. 

There are many people that I have met along this path, whose presence continues to enrich  and influence my life in- and outside of research, which are many more people than I can list here and I aplogize for any omission.
I'd like to specially thank Ben Baragiola for his continuous commitment to organize cool conferences and connect people in the field as well as all the interesting discussions on GKP. I thank Rafael Alexander for many interesting discussions on continuous variable error correction, for working on making the dream of GKP a reality and his invitations to visit Xanadu. I thank Franceosco Arzani for our collaborations and the jams, and in particular his initial motivation to study lattice theory in the context of GKP.

I thank Jonas Haferkamp, Yihui Quek,  Niko Breuckmann, Armanda Quintavalle and Alessandro Ciani for all the interesting discussions and collaborations and for always helping me with life advice.
I would like to thank Steve Flammia for consistently being up for exciting discussions about science and math and in particular for making me feel validated in my random mathematical interests. Thanks also to Jean-Pierre Seifert for the interesting collaborations and for sharing his expertise and curiosity about lattice theory.

I am grateful to all past and present members of our group and their contributions to creating an inspiring and fun research environment. Thanks in particular to Felix Witte and Claudia Thomas for making everything possible and I would also like to thank Alexander Townsend-Teague, Ansgar Burchards and Peter-Jan Derks for carrying our quantum error correction subgroup into an exciting future. Thanks to my office mates Frederik Wilde and Christian Bertoni for being a fun part of my journey. I thank Julio Magdalena and Armanda Quintavalle for valuable feedback on this thesis and in particular Julio for our daily random discussions on everything within and outside of research. I thank Lennart Bittel for teaching me about computational complexity theory, Steve Simons and Julio for many discussions on connections between GKP codes and quantum field theories and Daniel Weigand for dicussions on the implementation of GKP codes as well as on how to control quantum systems. 

\vspace{.5cm}

Being able to fill my life with so many inspiring and supportive people through our shared curiosities has been the greatest privilege on this journey.
\newpage
\begin{center}
\large
\textit{
To all my teachers.
}
\end{center}

\thispagestyle{empty}
\newpage
\tableofcontents
\newpage
\listoffigures
\leavevmode\thispagestyle{empty}\newpage
\pagenumbering{arabic}
\newpage
\chapter{The world of GKP codes}\label{chap:intro}

Gottesman-Kitaev-Preskill (GKP) codes  were conceived in the year 2000 by their namesakes in \textit{``Encoding a qubit in an oscillator''} \cite{GKP} as a quantum error correcting code.  That is, as a specific way to associate logical quantum information represented by a so-called code space $\CH_C\subseteq \CH$ to a physical quantum system. The hope in the design of quantum error correcting codes is typically to choose the subspace such that
\begin{enumerate}
\item It is -- or can be made --  robust to physically well-motivated noise in some sense and
\item Logical gates, that is unitary operations applied to the code space are physically easy to implement and can also be made robust in some sense.
\end{enumerate}

We will see a more explicit formulation of these desiderata in the course of this thesis. 
Gottesman, Kitaev and Preskill proposed to encode discrete quantum information into a phase-space translation invariant subspace of the infinite Hilbert space attributed to a collection of quantum harmonic oscillator modes. While these quantum error correcting codes appeared to have nice properties, in particular in terms of robustness against errors from a natural error basis in such physical systems and they come with a simple set of robust logical gates, the fine-grained experimental control seemed daunting at the time and, aside from some interesting theoretical observations made early on by Harrington and Preskill \cite{HarringtonPreskill, Harrington_Thesis}, only little work had cumulated. Two notable proposals that emerged in a long period of silence are the proposals to implement photonic measurement-based quantum computation using the GKP code by Meniccuci \cite{Menicucci_2011}  and a proposal by Terhal and Weigand to implement the GKP code in superconducting circuits \cite{Terhal_2016}. Finally, as experimental technology began to catch up, silence broke, and research on the GKP code began to experience a renaissance with the first demonstration of code-state preparation in a trapped-ion system by Fluehmann et al. in 2018 \cite{Fluehmann_2019} followed by its implementation in superconducting circuits by Campagne-Ibarcq \cite{Campagne_Ibarcq_2020} and Sivak \cite{Sivak_2023} that showed dramatically improved performances. In refs.~\cite{Albert_2018,Noh_Capacity} Albert and Noh showed that, despite not being explicitly designed to deal with such noise, the GKP code demonstrates superior performance in its protection against photon loss, a natural and physically relevant noise source for the physical systems considered.
This cumulation of events has lifted quantum error correction with the GKP code to a topic of broad interest in recent years, such that even commercial start-ups like \textit{Xanadu} \cite{Bourassa_Blue_2021} and \textit{Nord Quantique} \cite{lachancequirion2023autonomous} orient their efforts towards realizing quantum computation using GKP codes.

Despite their relevance for some promising technological developments, the primary focus of this thesis will not be on developing the technology and use of GKP codes. Rather, the goal will be to better our base understanding of GKP codes, its coding theory and its connections to other areas in mathematics, computer science and physics. My hope for this work is to convince the reader that the \textit{ looking glass of GKP codes} provides a unique perspective on quantum computation and quantum error correction and may help to form meaningful connections far beyond its own realm. Of course, we point to technological contributions that naturally emerge along this quest.

\begin{figure}
\center
\includegraphics[width=\textwidth]{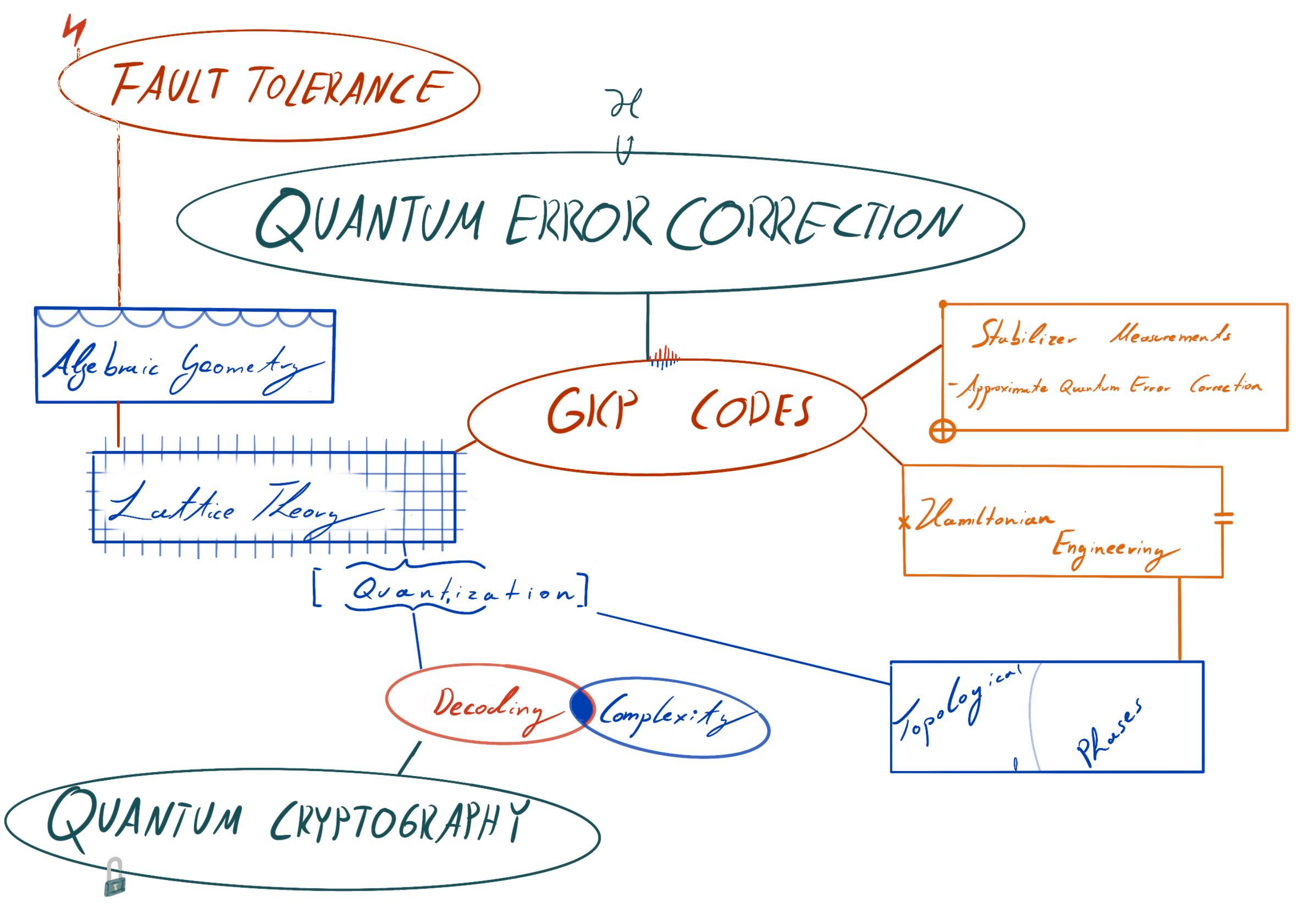}
\caption{High level overview of research areas highlighted in this thesis that are connected to GKP codes.}
\end{figure}

\section{From classical to quantum displacements}
Before we define GKP codes from a more conventional perspective, we briefly sketch out how their structure emerges quite naturally as \textit{``the thing that stays classical in quantization"}. We remain deliberately vague here to keep things manageable, but I believe that extensions of the following idea could be interesting to further formalize and generalize.
There are many different approaches to quantization, which is to provide a clear mathematical structure to pass from classical- to quantum physics, see refs.~\cite{dirac1958principles, Blau_GQ}. Here we softly touch on the framework referred to as \textit{geometric quantization} \cite{Blau_GQ} for its natural connection to symplectic geometry.

This picture begins with classical Hamiltonian mechanics of a particle on the line \cite{Arnold1988}. The state of the particle is described by a position index $q\in \R$ and a canonical momentum index $p\in \R$. The associated configuration space, $\R^2 \ni (q,p)$ that hosts the state of the particle is called \textit{phase space}. A Hamiltonian function $H(q,p)$ determines the energy of the particle in the respective state and dictates how the state of the particle changes in time. We denote the phase space index by
$\bs{x}= \lr{q,p}$ and the gradient by $\nabla=\lr{\partial/\partial q \;  \partial/\partial p}^T$. The Hamiltonian evolution of the particle is given by the Hamiltonian equations of motion
 \begin{equation}
 \frac{d}{dt}\bs{x}=J_2\nabla H, \label{eq:EOM}
 \end{equation}
 where we have also defined the symplectic form \begin{equation}
 J_{2n}=\begin{pmatrix}
 0 & I_n \\ -I_n & 0
 \end{pmatrix}.
 \end{equation}
 We will encounter the symplectic form often throughout this thesis and will generally omit the index $2n$ when the dimensions of the matrix are clear from context.
 
 The Hamilton function is a classical observable, whose evaluated value $H(\bs{x})$ outputs information about the state of the system. The time evolution of any other observable $f\lr{\bs{x}}$ along the trajectories dictated by the Hamiltonian evolution is given by
 \begin{align}
 \frac{d}{dt}f&=-\lr{\nabla H}^T J \nabla f \label{eq:X_H} \\
 &=\frac{\partial f}{\partial q}\frac{\partial H}{\partial p}-\frac{\partial f}{\partial p}\frac{\partial H}{\partial q}=:\lrc{f, H},
 \end{align}   
where in the last line $\lrc{\cdot, \cdot}$ denotes the \textit{Poisson bracket}. The differential operator $X_H:= \lr{J\nabla H}^T  \nabla $ appearing in eq.~\eqref{eq:X_H} is the so-called \textit{Hamiltonian vector field} and a formal solution to the corresponding differential equation for the evolution of the observable $f$ is
\begin{equation}
f\lr{\bs{x}\lr{t}}=e^{t X_H} f\lr{\bs{x}\lr{t=0}}. \label{eq:evolv_obs}
\end{equation}

We have $\lrq{X_f, X_g}=X_fX_g-X_gX_f=-X_{\lrc{f,g}}$, such that the evolutions generated by two Hamiltonian vector fields with Hamiltonians $f$ and $g$ commute when the Poisson bracket of the respective Hamiltonians is constant.
Consider the simplest non-trivial Hamiltonians: $H_q\lr{\bs{x}}=-v q$ and $H_p\lr{\bs{x}}=u p$. Per eq.~\eqref{eq:EOM} these Hamiltonians generate the evolutions

\begin{equation}
H_q:\; \frac{d}{dt}\bs{x}=\begin{pmatrix}
0 \\ v
\end{pmatrix},\hspace{1cm}
H_p:\; \frac{d}{dt}\bs{x}=\begin{pmatrix}
u \\ 0
\end{pmatrix},
\end{equation}
that is, $H_q$ induces a change of momentum with constant rate $q$ and $H_p$ induces a change of position with constant rate $v$. The corresponding evolution operators that generate the effective time evolution of observables $e^{tX_H}$ as in eq.~\eqref{eq:evolv_obs} are given by
\begin{equation}
D_q\lr{ut}=e^{ut \frac{\partial}{\partial p}},\hspace{1cm} D_p\lr{vt}=e^{vt \frac{\partial}{\partial q}}.
\end{equation}
These \textit{classical displacement operators} displace phase-space points on classical observables
\begin{equation}
D_q\lr{v}f\lr{q,p}=f(q, p+v),\hspace{1cm} D_p\lr{u} f(q+u,p),
\end{equation}
i.e. they effectively implement shifts $(q,p)\mapsto (q+u, p+v)$ in phase space. Since the Poisson bracket $\lrc{q,p}=1$ is constant, these classical displacement operators also do, 
\begin{equation}
D_q(v)D_p(u)=D_p(u)D_q(v),\label{eq:class_commutation}
\end{equation}
and as linear operators that act on observables they naturally preserve each others eigenspaces. In particular all observables invariant under shifts $D_p(u): q\mapsto q+u$ will maintain that invariance relative to any offset in $p$ via $D_q(v): p\mapsto p+v$.

This situation changes distinctively when moving from classical- to quantum mechanics. Instead of describing the state of a particle by deterministic phase space indices $q$ and $p$, position and momentum are lifted to infinite dimensional operators $\hat{q}, \hat{p}$ acting on a separable Hilbert space $\CH=L^2\lr{\R}$ of square integrable functions and linearity over $\C$, in a manner that the classical Poisson bracket gets replaced by the commutator of these operators $\lrc{\cdot , \cdot}\mapsto -i/\hbar\, \lrq{\cdot, \cdot}$, and we obtain the canonical commutation relations\footnotemark
\begin{equation}
\lrq{\hat{q}, \hat{p}}=i\hbar \hat{I}. \label{eq:CCR}
\end{equation}
The failure of position- and momentum operators to commute is reflected in Heisenberg's uncertainty-relation, ${\rm Var}\lr{q}{\rm Var}\lr{p}\geq \frac{\hbar^2}{4}$, so that one cannot be measured without perturbing the potential measurement outcome of the other.

\footnotetext{Note that this equation is easily misinterpreted to yield nonsense statements. If, e.g., the operators in this equation were treated as finite dimensional objects, taking the trace of this equation would imply the statement ``$1=0$''. The rigorous way to define this is by instead considering the Weyl form of the commutation relations. \label{foot:one}}

The displacement operators that we have encountered before now have become
\begin{equation}
\hat{D}_q\lr{v}=e^{iv\hat{q}},\hspace{1cm} \hat{D}_p\lr{u}=e^{-iu\hat{p}},
\end{equation}
which form one-parameter unitary groups and fail to commute as
\begin{equation}
\hat{D}_q(v)\hat{D}_p(u)=e^{iuv}\hat{D}_p(u)\hat{D}_q(v). \label{eq:WeylForm}
\end{equation}
This is the so-called Weyl form of the canonical commutation relations. The Stone-von Neumann theorem \cite{Neumann1932, Rosenberg} guarantees the uniqueness of any pair of such one-parameter groups of unitary operators, which is a cornerstone of quantum mechanics. They form a so-called \textit{Heisenberg(-Weyl)} group $H$. This group is isomorphic to $U(1)\times \R^2$, where the $ \R^2\supset (u, v) $ component refers to the indices we have been using and $U(1)$ attaches a phase to each element in the group. 

This group is said to fit into the exact sequence  \begin{equation}
1\rightarrow U(1)\rightarrow H \rightarrow \R^2 \rightarrow 0,\label{eq:exact_sequence}
\end{equation}
meaning that each arrow indicates a group homomorphism and the image of one arrow is exactly the kernel of the next. In this example this essentially refers to the fact that the identity element in $H$ corresponds to  $(0,0,0)\in U(1)\times \R^2$. 

We take away that displacement operators are  clearly important in quantum mechanics -- we will see many more properties of them in the coming chapters --  and comprise in a certain sense the core of this thesis. 

The group of displacement operators contains a subset that mimics their classical ancestors: that is the algebra of displacement operators spanned by
\begin{equation}
S_0=\lrc{\hat{D}_q(v), \hat{D}_p(u): \, uv \in 2\pi\Z}
\end{equation}
maintains the commutation relations of the classical displacement operators in eq. \eqref{eq:class_commutation} and leave the mutual eigenspaces invariant. To phrase it differently, notice that one realization of  the operators in $S_0$ is spanned by the displacement operators generated by $  D_q\lr{\sqrt{2\pi}\lambda^{-1}},  D_p\lr{\sqrt{2\pi}\lambda}$ for some $\lambda >0$.  These operators commute and their measurements can be interpreted to correspond to measurements of the modular quadratures 
\begin{equation}
\hat{q} \mod \sqrt{2\pi}\lambda \hspace{1cm} \text{and} \hspace{1cm}  \hat{p} \mod \sqrt{2\pi}\lambda^{-1},
\end{equation}
which behave as if they were classical observables. 
This spectacular magic trick is one operational perspective on the essence of the GKP code;  it has found a direct implementation in a quantum displacement sensor scheme by Duivenvoorden et al. \cite{Duivenvoorden_Sensor} and GKP stabilizer measurements that we will discuss later.

Before we continue, note that this section has also illustrated the role of the symplectic matrices 
\begin{equation}
\Sp_{2n}\lr{\R}:=\lrc{S\in \GL_{2n}\lr{\R}:\; S^TJS=J}
\end{equation}
in classical Hamiltonian dynamics. This is the group of matrices that preserve the symplectic form $J$ and is closed under transposition. From eq.~\eqref{eq:X_H} we can see that a basis transformation $\bs{x}\mapsto S\bs{x}=\bs{x'} \Leftrightarrow \nabla=S^T\nabla ', S\in\Sp_{2}\lr{\R} $ preserves the time evolution of observables and here corresponds to our freedom of implementing coordinate transforms that yield the same observable dynamics. Symplectic matrices will accompany us throughout the coming chapters and take a special role in the description of GKP codes.

For the rest of this work, when clear from context and not specifically relevant, we will omit the operator hats $\hat{\cdot}$ on operators and adhere to the convention $\hbar=1$.

\section{Translation invariant functions on phase space}

\begin{mybox}
\subsubsection*{What is ... a holomorphic function?}

A holomorphic function $f: \C \to \C$ is a complex function that is complex differentiable on every point of an open subset $U\subseteq \C$ such that the limit
\begin{equation}
f'(z)=\lim_{h\rightarrow 0} \frac{f(z+h)-f(z)}{h}\label{eq:holo_lim}
\end{equation}
exists $\forall z \in U$. 

Every complex number $z=\Re\lr{z}+i\Im\lr{z} \in \C\sim \R^2$ can be labeled by two real numbers $z\mapsto \bs{h}\lr{z}=Re\lr{z}\oplus \Im\lr{z}$, and there is a homomorphism 
\begin{equation}
h(z)=\begin{pmatrix}
\Re\lr{z} & -\Im \lr{z} \\ \Im\lr{z} &\Re\lr{z}
\end{pmatrix},
\end{equation}
that satisfies $h(w+z)=h(w)+h(z)$ and $h(wz)=h(w)h(z)$ for all $w, z \in \C$. The function $ \bs{h}\lr{z}$ is a homomorphism for the additive structure of $\C$ while $h(z)$ also carries a homomorphism for the complex multiplication to the vectors via
\begin{equation}
h(w) \bs{h}\lr{z}=\bs{h}(wz).
\end{equation}
Notice that $h(1)=I_2$ and $h\lr{i}=-J_2$, that is, the symplectic form $J_2$ forms a linear representation of the multiplication with the complex unit $i$ and $h\lr{e^{i\phi}} =: R_{\phi} =\cos\lr{\phi}  I_2 - \sin\lr{\phi} J_2 $ is a rotation of the plane.

By identifying the complex function $f(z)=f\lr{x+iy}=u(x,y)+iv(x,y)$ with a function with domain $\R^2$, we can compute
\begin{equation}
\bs{h}\lr{df\lr{x+iy} }=\begin{pmatrix}
\partial_x u & \partial_y u \\ \partial_x v & \partial_y v 
\end{pmatrix} 
\begin{pmatrix}
dx \\ dy
\end{pmatrix}.\label{eq:holo_homo}
\end{equation}

In the definition of the complex derivative in eq.~\eqref{eq:holo_lim} it is important that the derivative should be well-defined and independent of the direction $\phi$ the parameter $h=|h|e^{i\phi}$ approaches $0$. Translated into the representation in eq.~\eqref{eq:holo_homo} this in particular means that the outcome of eq.~\eqref{eq:holo_homo} should be linear in rotations $dx \oplus dy \mapsto R_{\phi}\lr{dx\oplus dy}$ and the complex derivative should satisfy $\bs{h}\lr{f'\lr{z}dz}=h\lr{f'\lr{z}} \bs{h}\lr{dx+idy}$ with $df=f'\lr{z}dz$. Altogether, this requirement that the derivative should behave complex linear implies the \textit{Cauchy-Riemann equations}
\begin{equation}
\partial_x u=\partial_y v , \hspace{1cm} \partial_x v= -\partial_y u.
\end{equation}
This behavior endows holomorphic functions with a particular high degree of structure. In particular, holomorphic functions $f:\, \C \rightarrow \C$ are \textit{analytic}, that is, they can be expanded into a power series
\begin{equation}
f\lr{z-z_0}=\sum_{k\geq 0} a_k (z-z_0)^k
\end{equation}
in the open neighborhood $U\subseteq \C$ of every point $z_0\in \C$. Another important property is that, via Liouvilles theorem, every bounded holomorphic function $|f(z)|<M,\, \forall z\in \C$ is constant \cite{Nelson1961}; which also implies that there are no non-trivial doubly translation invariant holomorphic functions.
\end{mybox}

In quantum mechanics, the classical phase space we have encountered earlier finds a new meaning: it becomes the domain for wave-functions and (quasi-)  probability distributions that determine the statistics of measuring a certain position and momentum. For a single degree of freedom, it is given by indices  $(q,p)\in\R^2$ which we can complexify to obtain the more compact labelling $z=q+ip \in \C$. Now the functions that determine the physics of our system are complex functions and since they are supposed to represent physical quantities, it makes sense to require them to be very well-behaved. Concretely, we require them to be \textit{holomorphic} functions and to be normalizable under the scalar-product induced norm of a Hilbert space they reside in.
The relevant scalar product of their Hilbert space is given by 
\begin{equation}
\braket{f|g}=\frac{1}{\pi}\int_{\C}dz\;  e^{-|z|^2} \overline{f(z)}g(z),\label{eq:SB_prod}
\end{equation}
which constructs the \textit{Segal-Bargmann representation} of quantum mechanics \cite{Chabaud_2022, Bargmann1961, segal1967mathematical}.  

In the preceding section we have motivated the relevance of phase space translation-symmetric states. Since phase space has two (real) dimensions we also need to specify to translation axes, one of which we fix with the basis ``vector'' $1\in \C$ and the other by $\tau\in \C, \; \Im\lr{\tau} \neq 0 $. W.l.o.g. we require $\tau \in \hh=\lrc{z\in C,\, \Im\lr{z}>0}$ to be a vector in the complex upper half plane such that the complex lattice of translational symmetries of our desired functions is given by $\Lambda=\Z+\tau \Z$. 

Now we have hit a roadblock. By Liouville's theorem there are no non-constant holomorphic functions that have a doubly translational symmetry $f(z)=f(z+\lambda)\; \forall \lambda \in \Lambda$. To obtain non-trivial functions with doubly translational symmetries has to either allow for poles, i.e. relax to use meromorphic functions, or relax the requirement on periodicity. The former leads to the theory of so-called \textit{elliptic functions}, which we will encounter later in chapter~\ref{chap:Theory}. Here we decide to relax the requirement of periodicity which brings us naturally to the theory of \textit{theta functions}. Jacobi's theta function is defined as \cite{Tata_1} 
\begin{equation}
\vartheta(\tau, z)=\sum_{n\in \Z} e^{i\pi n^2\tau+i2\pi n z}.
\end{equation}
 The infinite series converges compactly on $\hh \times \C$ and yields a holomorphic function in $z\in \C$  It is doubly periodic as
 \begin{equation}
\vartheta(\tau, z+1)=\theta(\tau, z), \hspace{1cm} \vartheta(\tau, z+\tau)=e^{-i\pi \tau - i2\pi z}\vartheta(\tau, z). 
 \end{equation}

The non-trivial factor on the r.h.s. of this equation establishes its \textit{quasi periodic behaviour}. 

Following ref.~\cite{Tata_1}, define the ``holomorphic displacement operators" $S_a, T_b$ with $(a,b)\in \R^2$ by
\begin{equation}
S_a f(z)=f(z+a), \hspace{1cm} T_b f(z)=e^{i \pi b^2\tau +i2\pi bz} f(z+b\tau).
\end{equation}

These operators naturally satisfy $S_{a_1}S_{a_2}=S_{a_1+a_2}, T_{b_1}T_{b_2}=T_{b_1+b_2}$ and the commutation relation 
\begin{equation}
S_a T_b= e^{i2\pi ab}T_bS_a,
\end{equation}
equivalent to the Weyl form in eq.~\eqref{eq:WeylForm} such that, by the Stone-von Neumann theorem, these displacement operators are in fact equivalent to those we have found before, and we again obtain a representation of the Heisenberg group $H= U(1)\rtimes \R^2$ that acts on holomorphic functions as
\begin{equation}
U_{\lr{\lambda, a, b}}f(z)=\lambda e^{i2\pi ab}e^{i\pi b^2 \tau +i2\pi bz} f(z+a+b\tau).
\end{equation}

The theta functions $\vartheta(\tau, z)$ are (up to scalars) the unique functions invariant under the action of the subgroup of the Heisenberg group
\begin{equation}
H_1=\lrc{\lr{1, a, b} ,\; (a,b)\in \Z^2} \subset H,
\end{equation}
for which the commutation factor $e^{i2\pi ab}=1$ is always trivial.

We can obtain even more functions with a similar behavior by scaling the lattice of translational symmetries by an integer factor $d\in \N: \Lambda \rightarrow d\Lambda$. The space $V_d$ of holomorphic functions invariant under the group
\begin{equation}
H_d=\lrc{\lr{1, a, b} ,\; (a,b)\in d\Z^2} \subseteq H_1
\end{equation}
is in fact $\dim\lr{V_d}=d^2$ dimensional. To see this, note that $V_d$ is closed under the operators $S_{1/d}$ and $T_{1/d}$ as these operators commute with the displacements representing $H_d$. Denote the group generated by these operators with $H_d^{\perp}=r_{d}\times \lr{\frac{1}{d}\Z_d}^2 \subset H$, where $r_d$ denotes the set of $d$-th roots of unity. This group is called the finite Heisenberg-Weyl group, and similar to the continuous Heisenberg-Weyl group there is an exact sequence
\begin{equation}
1\rightarrow r_{d^2} \rightarrow H_d^{\perp} \rightarrow \lr{\frac{1}{d}\Z_d}^2\rightarrow 0, \label{eq:HW_discr}
\end{equation}
where now the action of $H_d^{\perp}$ is trivial under the preimage $H_d$ of $\lr{d\Z}^2$. \footnote{This discussion has been motivated by the presentation of D. Arapura provided in ref.~\cite{Arapura_Shimura}.} 
This construction of a discrete Heisenberg-Weyl group as a subgroup of the continuous one -- i.e. the fact that each element in the sequence in eq.~\eqref{eq:HW_discr} embeds into the corresponding element in the sequence in \eqref{eq:exact_sequence} -- is the first example of a GKP-code, which we will discuss more in-depth in the following. The unitary representation of discrete Heisenberg-Weyl algebra is also sometimes referred to as the (generalized) qudit Pauli-group, which plays an important role in quantum computing and quantum error correction.

The discrete Heisenberg-Weyl group has a irreducible action on $V_d$ \cite{Tata_1} and defines a basis for $V_d$ given by the theta functions with characteristic $(a,b)\in d^{-1}\Z / \Z$
\begin{equation}
\vartheta_{a,b}(z)=S_bT_a \vartheta(z).
\end{equation}
It can be shown that for $p,q  \in \Z$ 
\begin{equation}
\vartheta_{a+p,b+q}(z)=e^{i2\pi a q} \vartheta_{a,b}(z),
\end{equation}
which shows that $\vartheta_{a,b}(z)$, up to a constant, only depends on characteristics  $(a,b)\in \frac{1}{d}\Z / \Z$ and these $d^2$ theta functions form a basis for $V_d$.

Theta functions are not going to be in the focus of the upcoming presentations, but nevertheless form important objects that underlie many of the topics and ideas we are about to discuss. 
Equipped with the following understanding, the relevance  of the role of theta functions will hopefully become more clear:
\begin{enumerate}
\item Theta functions with characteristic form a basis for logical GKP states, and
\item Theta functions yield a projective embedding of the complex torus $E_{\tau}=\C/\lr{\Z+\tau \Z} \rightarrow \CPP^{d^2-1}$ into complex projective space. 
\end{enumerate}
While the first point is going to become clearer very soon, the second point will only be briefly discussed in sec.~\ref{sec:rosetta}.   By Chow's theorem \cite{Chowsthm}, it is this embedding that implies that complex tori, which we will learn to interpret as GKP codes, are in fact algebraic curves. I encourage the reader to return to this point after reading the next chapter and hope that at that point it will become less of a mystery and more of a wonder why such mathematical idiosyncrasies appear in a thesis about quantum error correction.

\section{Outline}
This thesis will be arranged in four core chapters. Chapter~\ref{chap:Theory} will focus on the abstract coding theory of GKP codes and the mathematical structure behind fault-tolerant quantum computation with GKP codes. In the first part of that chapter we will open the toolbox of lattice theory to examine GKP codes, derive coding theoretic properties and tradeoffs. By examining the structure of logical Clifford gates for the GKP code, we will discover a route to zoom out further and develop an algebraic geometric perspective on GKP codes, where we classify the structure of ``moduli spaces of GKP codes" and build a close link between the theory of moduli spaces of elliptic curves and the theory of fault tolerance for the GKP code. This work on the coding theory of GKP codes will form the centerpiece of this thesis. 

In chapter~\ref{chap:constructions}, we apply the developed structure to discuss concrete GKP codes with their lattice-theoretic and coding-theoretic properties. This chapter discusses the structure of some GKP codes already found in the literature and also features  a class exotic GKP codes with good properties that can be derived from a lattice-based post-quantum cryptosystem called \textit{NTRU}. 

For practical implementations, merely knowing the structure and properties of codes is not enough. Concretely, to use a quantum error correcting code, it is paramount to be able to find classical strategies to process how best to correct the errors on the system given partial information extracted from measuring its stabilizers. We discuss the \textit{decoding problem} associated to GKP codes through a complexity theoretic lens in chapter~\ref{chap:complexity} and show how the decoding problem for the NTRU-GKP codes presented in chapter~\ref{chap:constructions} implies a quantum cryptographic scheme. This is done by proposing a private quantum channel that builds on the post-quantum cryptographic properties of the NTRU-GKP codes.

As a strong motivation to examine GKP codes are their technological realization, the fourth chapter~\ref{chap:implementation} will discuss how GKP codes can be implemented. We discuss various ways to implement the GKP code through active error correction that are relevant for photonic-  and superconducting circuit architectures and discuss the physics behind passive quantum error correction with the GKP code.
A new result presented in this chapter is a  ``passive'' error correction scheme that engineers a Hamiltonian that hosts the GKP code space in its ground state utilizing time-dependent control. 

Inspired by the AMS ``What is...?'' column series \cite{What_is} I include topical ``What is...?''  boxes to provide very brief introductions to specialized ideas that appear throughout this work. Each chapter closes with a \textit{Dream}, in which I highlight a perspective on interesting follow-up work that extends the work presented in the respective chapter. I hope that the ideas and dreams presented in this thesis will encourage the reader to take interest in the GKP code and its various applications in- and outside quantum computation and stimulate the reader to dive into the dreams presented here.
\chapter{Continuous variable basics}\label{chap:CVBasics}

A continuous variable (CV) quantum system is one, where states in the corresponding Hilbert space $\CH = L^2\lr{\R^{2n}}$ are naturally labelled by continuous degrees of freedom, such as the position-  and momentum variables of a particle we have encountered before. In the literature one often also encounters the association to the Hilbert space of a quantum harmonic oscillator (QHO) or the reference to a bosonic system \cite{GKP, Terhal_2020}. These nomenclatures refer to the fact that the state space is naturally given by that of a physical collection of $n$ quantum harmonic oscillators with Hamiltonian governing the time evolution 
\begin{equation}
H=\sum_{i=1}^n \frac{\hat{q}_i^2+\hat{p}_i^2}{2}=\frac{\bs{\hat{x}}^{\dagger}\bs{\hat{x}}}{2}, \hspace{1cm} \bs{\hat{x}}=\lr{\hat{q}_1 \hdots \hat{q}_n \; \hat{p}_1\hdots \hat{p}_n}^T,
\end{equation}
and highlight fact that the relevant quadratures satisfy bosonic commutation relations\footnote{See also footnote~\ref{foot:one} in the previous chapter.} 
\begin{equation}
\lrq{\hat{x}_i, \hat{x}_j}=iJ_{ij}.
\end{equation}
Expressed in annihilation operators $\hat{a}_i=\lr{\hat{q}_i+i\hat{p}_i}/\sqrt{2},$ the Hamiltonian becomes $H=\sum_i^n \hat{n}_i+\frac{n}{2},$ where the number operators are given by $\hat{n}_i=\hat{a}_i^{\dagger}\hat{a}_i$. The number operators have a countably infinite spectrum $n_i\in \N_0$ and eigenbasis given by Fock states $\lrc{\ket{n_i}}_{n_i=0}^{\infty}$.

Quantum harmonic oscillator systems are ubiquitous in nature. Some systems, such as the electromagnetic field of a propagating photon, the fluxes and charges of a superconducting LC-circuit, a trapped ion, or a very very small mechanical spring are directly equipped with such Hamiltonians and state-spaces. Beyond these, a popular physicists' argument is that every natural potential has local minima, each of which can be well approximated by a quantum harmonic oscillator system. This argument is not truly universally applicable but shall suffice to motivate that quantum harmonic oscillators are both relevant building blocks of nature and ``easy'' to construct in a controlled environment.

We define displacement operators with amplitude $\bs{\xi}\in \R^{2n}$ 
\begin{align}
D\lr{\bs{\xi}} = \exp\left\{-i \sqrt{2\pi} \bs{\xi}^T J \bs{\hat{x}}\right\}. \label{eq:disp}
\end{align}
These displacement operators act linearly on $\CH$ to implement shifts of the quadratures
\begin{equation}
D^{\dagger}(\bs{\xi}) \hat{\bs{x}}D(\bs{\xi})=  \hat{\bs{x}}+ \sqrt{2\pi}\bs{\xi},
\end{equation}
 as can be verified using the well-known identity  $e^A B e^{-A}=e^{\lrq{A, \cdot}}B$. They are orthonormal with
\begin{equation}
\Tr\lrq{D^{\dagger}\lr{\bs{\xi}}D\lr{\bs{\eta}}}=\delta\lr{\bs{\xi}-\bs{\eta}},\label{eq:disp_orth}
\end{equation}
and commute and close as
\begin{align}
D\lr{\bs{\xi}}D\lr{\bs{\eta}}&=e^{-i\pi \bs{\xi}^TJ\bs{\eta}}D\lr{\bs{\xi}+\bs{\eta}}, \nonumber\\
&=e^{-i2\pi \bs{\xi}^TJ\bs{\eta}}D\lr{\bs{\eta}}D\lr{\bs{\xi}}.\label{eq:disp_comm}
\end{align}

In this equation we recognize the Weyl-form of the canonical commutation relations encountered before (set $n=1$). The orthonormality of the displacement operators allows to express any (trace-class) operator to be expressed as a continuous linear combination of displacements
\begin{equation}
\hat{A}=\int_{\R^{2n}} d\bs{x}\, c_A(\bs{x})D\lr{\bs{x}}, \;  c_A\lr{x}=\Tr\lrq{D^{\dagger}\lr{\bs{x}} \hat{A}},
\end{equation}
where the function $c_A\lr{x}$ is called the \textit{characteristic function} of $\hat{A}$. 
Closely related is the \textit{Wigner function} of $\hat{A}$, given by the symplectic Fourier transform

\begin{equation}
W_{A}\lr{\bs{x}}
= \int_{\mathbb{R}^{2n}} d\bs{\eta}\, e^{-i2\pi \bs{x}^T J \bs{\eta}} \Tr\lrq{D\lr{\bs{\eta}} \hat{A}}.
\end{equation}

The Wigner function of a state $W_{\rho}$ is a quasi-probability distribution for the state, such that the probability to measure a quadrature along axis $\bs{e}_i$ is given by marginalizing over the remaining coordinates
\begin{equation}
P\lr{\bs{e}_i^T\bs{x}}=\int_{\bs{e}_i^{\perp}} d\bs{x}\, W_{\rho}\lr{\bs{x}},
\end{equation}
and with
\begin{equation}
\Tr\lrq{AB}=\int_{\R^{2n}}d\bs{x}\, W_{AB}\lr{\bs{x}}=\int_{\R^{2n}}d\bs{x}\, W_{A}\lr{\bs{x}}W_{B}\lr{\bs{x}},\label{eq:WignerAB}
\end{equation}
it allows to compute expectation values of observables $\hat{O}$ as
\begin{equation}
\braket{\hat{O}}=\Tr\lrq{\rho \hat{O}}=\int_{\R^{2n}}d\bs{x} \, W_{\rho}\lr{\bs{x}}W_{\hat{O}}\lr{\bs{x}}.\label{eq:Wigner_exp}
\end{equation}
The Wigner functions of a vacuum-, squeezed-  and GKP state are shown in fig. ~\ref{fig:Wigner_GKP}.
\begin{figure}
\begin{minipage}{.33\textwidth}
\includegraphics[width=\textwidth]{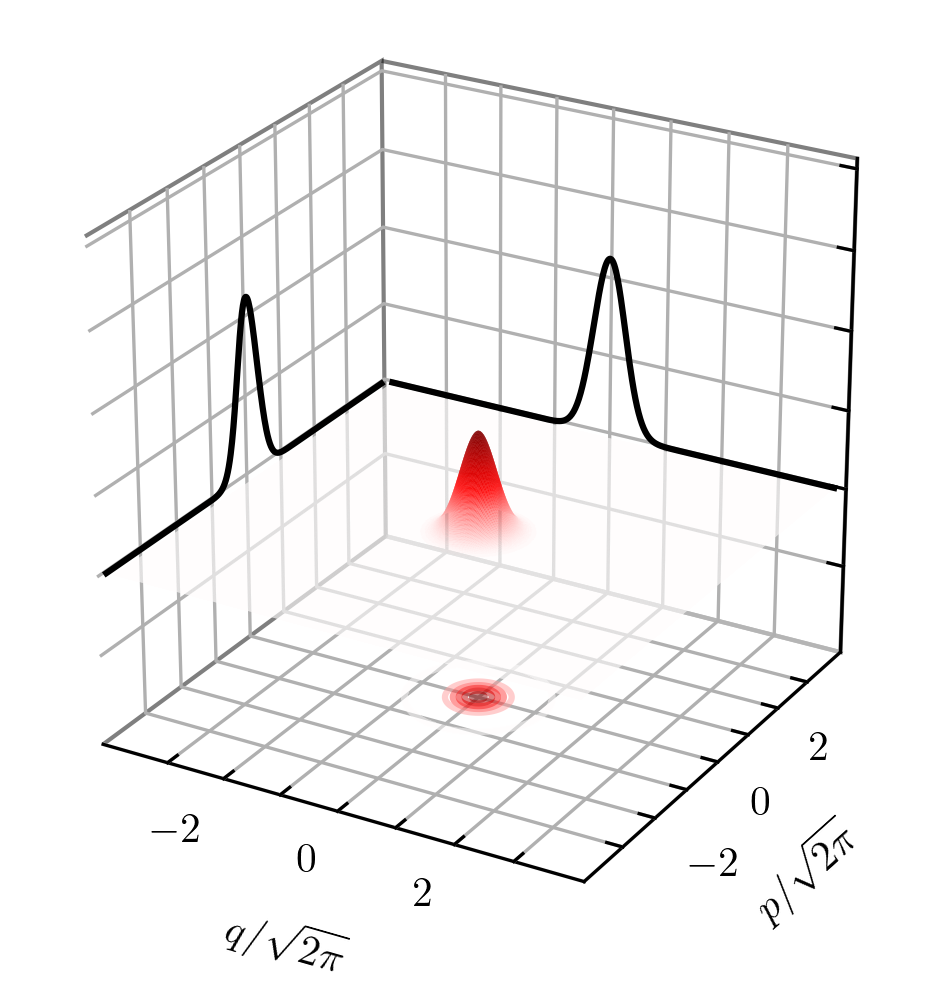}
\end{minipage}%
\begin{minipage}{.33\textwidth}
\includegraphics[width=\textwidth]{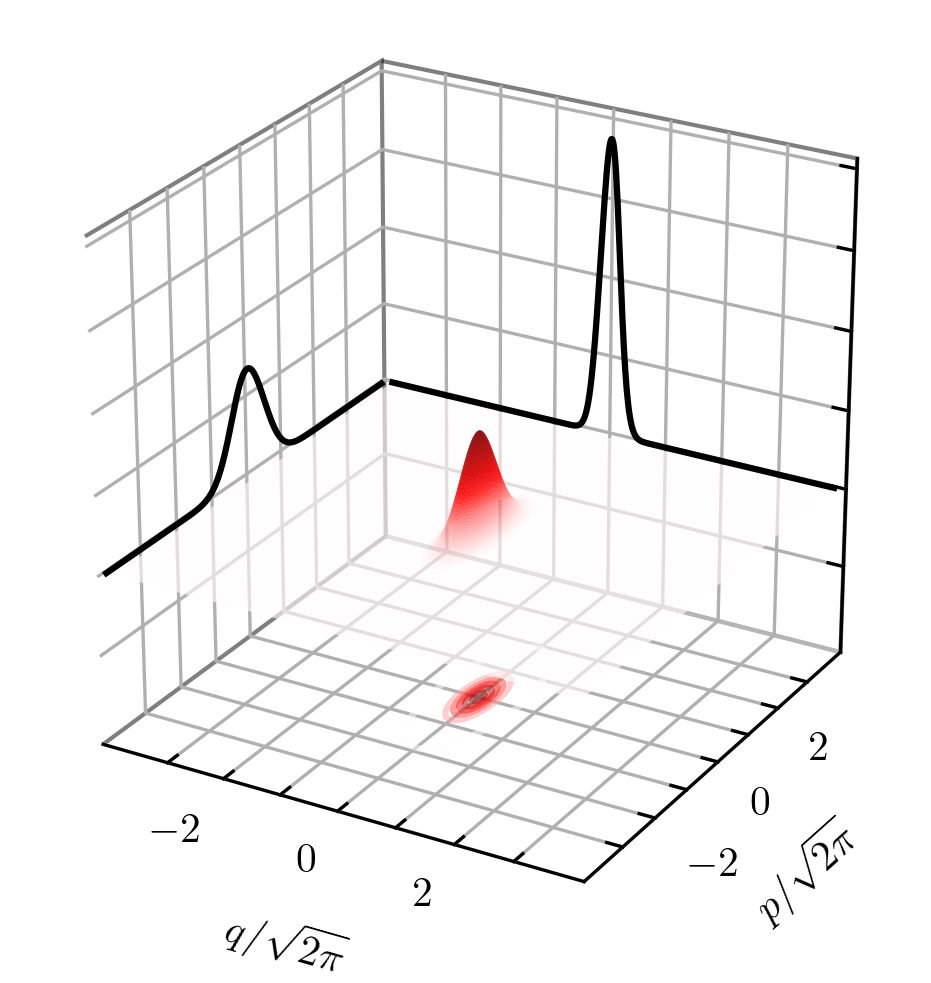}
\end{minipage}%
\begin{minipage}{.33\textwidth}
\includegraphics[width=\textwidth]{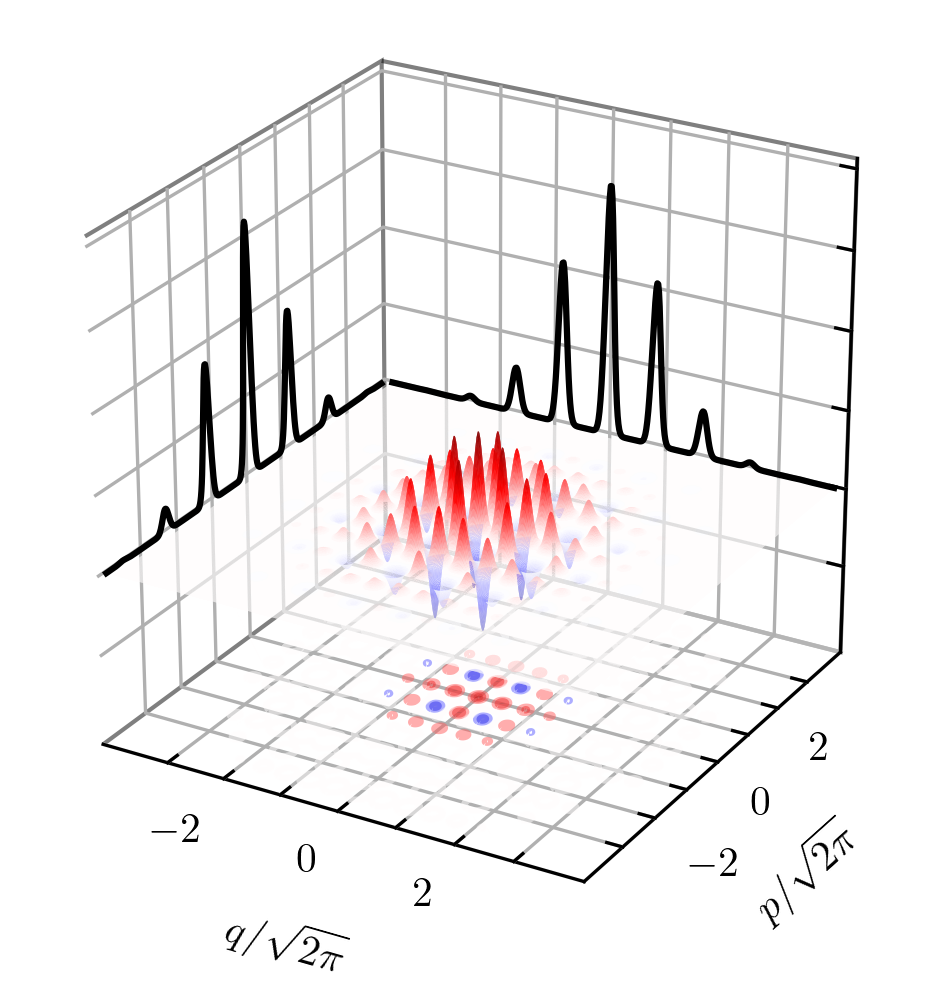}
\end{minipage}%
\caption{The Wigner function of a  vacuum-, squeezed- and approximate GKP state. These plots were produced using a code base provided in ref.~\cite{Weigand_code}.}\label{fig:Wigner_GKP}
\end{figure}
The vacuum state $\ket{\bs{0}}=\ket{0}^{\otimes n}$ is a so-called Gaussian state, whose Wigner function is given by a Gaussian distribution. Tightly related are the so-called coherent states

\begin{equation}
\ket{\bs{\alpha}}=D\lr{\bs{\alpha}}\ket{\bs{0}},
\end{equation}
which are eigenstates of the generalized annihilation operator $\bs{\hat{a}}=\lr{\hat{a}_1\hdots \hat{a}_n}^T$ with eigenvalues $\sqrt{2\pi}\bs{\alpha}$. Coherent states are simply Gaussian vacuum states centered at phase space point $\sqrt{2\pi}\bs{\alpha}$. As operator, they can be decomposed into displacements as 

\begin{align}
\ketbra{\bs{\alpha}}&=\int_{\R^{2n}} d\bs{\beta}\, \Tr\lrq{D^{\dagger}\lr{\bs{\beta}} \ketbra{\bs{\alpha}}} D\lr{\bs{\beta}} \\
&=\int_{\R^{2n}} d\bs{\beta}\,  e^{-\frac{\pi}{2} \bs{\beta}^T\bs{\beta}-i2\pi\bs{\alpha}^TJ\bs{\beta}}D\lr{\bs{\beta}} ,
\end{align}
such that a resolution of the identity is given by
\begin{equation}
\int_{\R^{2n}}d\bs{\alpha}\, \ketbra{\bs{\alpha}} = I,
\end{equation}
and we have
\begin{equation}
    \braket{\bs{\beta}|\bs{\alpha}}=e^{-\frac{\pi}{2}\lr{\|\bs{\alpha}-\bs{\beta}\|^2+i2\bs{\alpha}^TJ\bs{\beta}}},
\end{equation}
that is, coherent states form a non-orthogonal over-complete basis for operators on phase space.

There also exists a different parametrization of displacement operators labeled by complex indices $\bs{\gamma}\in \C^n$ given by
\begin{align}
D_c\lr{\bs{\gamma}}=\exp\lrc{\sqrt{\pi}\lr{\bs{\gamma}^T\bs{\hat{a}^{\dagger}}-\bs{\gamma}^{* T}\bs{\hat{a}}}} =D\lr{\bs{\xi}_{\bs{\gamma}}}, \label{eq:disp_complex}
\end{align}
where the equivalent real parameter is $\bs{\xi}_{\bs{\gamma}}=\Re\lr{\bs{\gamma}}\oplus \Im\lr{\bs{\gamma}}\in \R^{2n}$.

In this parametrization the displacement operator acts as
\begin{equation}
D_c\lr{\bs{\gamma}}^{\dagger}\bs{\hat{a}}D_c\lr{\bs{\gamma}}=\bs{\hat{a}}+\sqrt{2\pi}\bs{\gamma}.
\end{equation}
and the commutation relation is given by
\begin{equation}
D_c\lr{\bs{\gamma}}D_c\lr{\bs{\delta}}= e^{-i2\pi\Im\lr{ \bs{\gamma}^{\dagger}\bs{\delta}}} D_c\lr{\bs{\delta}}D_c\lr{\bs{\gamma}}.
\end{equation}
The symplectic form $\omega\lr{\bs{\gamma}, \bs{\delta}}=\Im\lr{ \bs{\gamma}^{\dagger}\bs{\delta}}$ is a skew-symmetric function inherited from the hermitian form $H\lr{\bs{\gamma}, \bs{\delta}}= \bs{\gamma}^{\dagger}\bs{\delta}$ as its imaginary part.

The complex parametrization provides the usual notation for coherent states $\ket{\bs{\gamma}}=D_c\lr{\bs{\gamma}}\ket{\bs{0}}$, which is equivalent to the previous definition. Expressed as a function of this index, the function $Q\lr{\bs{\gamma}}=\braket{\bs{\gamma}|\rho |\bs{\gamma}}$ is the so-called Husimi-Q function. It is holomorphic in the complex parameter and can also be obtained as a Gaussian-smoothed Wigner-function \cite{CahillGlauber_1969}, as can also be computed using eq.~\eqref{eq:Wigner_exp}. One can interpret the Husimi-Q function as the probability distribution for joint position- and momentum measurements (taking into account the Heisenberg uncertainty) and it is  connected to representation of states in the Segal-Bargmann (or \textit{stellar-}) representation encountered in the previous section by the correspondence~\cite{Chabaud_2022}

\begin{equation}
Q\lr{\bs{\gamma}}= e^{-\pi\bs{\gamma}^{\dagger}\bs{\gamma}}\left\vert f\lr{\overline{\bs{\gamma}}}\right\vert^2.
\end{equation}
Note the difference in convention relative to ref.~\cite{Chabaud_2022} since we have chosen to scale the domain by a factor $\sqrt{2\pi}$ in the real representation. This representation has nice properties: the stellar function, aside from being holomorphic and normalizable under the norm given by eq.~\eqref{eq:SB_prod}, is not required to immediately represent a physical observable but is allowed to have infinite support. This feature makes it a nice representation for the treatment of GKP states.

Unitary evolution via Hamiltonians strictly quadratic in the quadrature operators implement symplectic transformations, as can again be verified using $e^A B e^{-A}=e^{\lrq{A, \cdot}}B$, 
\begin{align}
U_S=e^{-\frac{i}{2}\bs{\hat{x}}^T C \bs{\hat{x}}},\; C=C^T, \\
U_S^{\dagger}\bs{\hat{x}}U_S=S\bs{\hat{x}},\; S=e^{CJ}, \label{eq:sympH}
\end{align}
where $S\in \Sp_{2n}\lr{\R}$ is a symplectic matrix which follows from unitarity of $U_S$, and we have 
\begin{equation}
U_{S}D\lr{\bs{\xi}}U_{S}^{\dagger}=D\lr{S\bs{\xi}},\label{eq:disp_S}
\end{equation}
such that also
\begin{equation}
W_{U_S\rho U_S^{\dagger}}\lr{\bs{x}}=W_{\rho}\lr{S\bs{x}}.
\end{equation}

The symplectic group contains an important subgroup of \textit{symplectic orthogonal} matrices isomorphic to the unitary group $O_{2n}\lr{\R}\cap \Sp_{2n}\lr{\R} =U_n\lr{\C}$, where for elements $U\in U_n\lr{\C} $ the isomorphism is given by
\begin{equation}
\begin{pmatrix}
\Re U & -\Im U  \\ \Im U & \Re U
\end{pmatrix} \in O_{2n}\lr{\R}\cap \Sp_{2n}\lr{\R} .
\end{equation}

These symplectic orthogonal transformations are those, that can be implemented by purely passive -- photon-number preserving -- linear optical elements: beamsplitters and phase-shifters, and do not require additional squeezing. For this reason they are considered particularly cheap to implement. This becomes explicit via the Bloch-Messiah decomposition. Every symplectic matrix $S\in \Sp_{2n}\lr{\R}$ admits a decomposition into symplectic orthogonal matrices $O_1, O_2 \in  O_{2n}\cap \Sp_{2n} $ and a diagonal matrix $D={\rm diag}\lr{\lambda_1\hdots \lambda_n \lambda_1^{-1}\hdots \lambda_n^{-1}}, \lambda_i >1$ such that
\begin{equation}
S=O_1DO_2.
\end{equation}
The diagonal matrix $D$ implements a squeezing of the quadratures $(\hat{q}_i, \hat{p}_i)\mapsto ( \lambda_i\hat{q}_i, \lambda_i^{-1}\hat{p}_i)$, of which the maximum ${\tt  sq}\lr{S}=\sqrt{\lambda_{\rm max}\lr{S^TS}}=\max_i \lambda_i$ squeezing value yields an indicator for the amount of energy that needs to be pumped into the QHOs to realize $U_S$.

Unitary operators of the form $U(S, \bs{x})=D\lr{\bs{x}}U_S$ are so-called \textit{Gaussian unitaries}, as they can be generated by Hamiltonians of maximal quadratic order in the quadratures, and they preserve Gaussianity of the Wigner-functions of states they act on. Due to the simplicity of their generating Hamiltonian, and their simple linear-order action on the quadrature vector, Gaussian unitaries are considered as especially desirable and robust in physical implementation. 
 
This quick overview concludes the basics technical physics background to understand what follows. These tools are part of a larger toolbox typically attributed to the topic of \textit{quantum optics}, and many good textbooks and review articles exist to which we refer the reader for further studies, see e.g. refs.~\cite{Gerry_Knight_2004, Weedbrook_2012} and references therein.
\chapter{Quantum computation and quantum error correction}\label{chap:QC}

Since the core technological contribution of the GKP is to facilitate quantum computation,  in this chapter we very briefly review what quantum computation is about and how quantum error correction is meant to make it possible.  The purpose of this brief review is to both provide context for this work,  and to understand the technological goalposts of the development of the GKP code.  We keep the presentation here to a minimum and refer the reader to excellent resources: See refs.~\cite{shor2001introduction,  preskill2023quantum} for historical accounts and ref.~\cite{nielsen00} for a comprehensive overview.

The idea of using a quantum mechanical system to perform computation or to simulate other physical quantum systems traces its roots to Manin,  Beninoff and Feynman~\cite{manin1980computable,  benioff1980computer,  feynman21simulating} in the early 1980s.  Motivated by classical computing,  the typical setup is to consider a computation based on performing unitary operations and measurements on a system of $n$-qubits,  each of which are described by a two-dimensional Hilbert space $\CH_1={\rm span}_{\C}\lrc{\ket{0},  \ket{1}}$.  The evolution of a state thus takes place in $\dim\lr{\CH_1^{\otimes n}}=2^n$ dimensional state space. This exponentially large size yields plenty of opportunity to design quantum algorithms to take advantage of the ability to create superpositions of basis states,   engineer interferences to amplify correct answers to input problems and to take advantage of state collapse under projective quantum measurements in order to single out individual state evolutions.  In the 1990s this toolbox was then exploited by Deutsch,  Joza,  Shor and Grover~\cite{DeutschJoza,  Shor,  Grover} to show that quantum algorithms can be designed that may outperform their classical counterparts.  Concurrent research into possible physical realizations of quantum computation,  see e.g.  refs.~\cite{Lloyd1993,  Loss_1998,  CiracZoller},  however, also made it clear that physical realizations of qubits as effective degrees of freedom in a real quantum system will never be perfectly shielded from environmental influences or imperfect separation from other degrees of freedom of its embedding system.  In order for the qubit system to be able to carry out the desired computations,  it is necessary for the quantum states to maintain superposition for a long time and logical operations on those qubits would need to be packaged into a form to not incur dramatic errors on the state of the computation.  This necessity gave rise to the theory of quantum error correction and fault tolerance~\cite{Shorcode,  Terhal_2015,  gottesman1997stabilizer}.

Fault-tolerance refers to the rather qualitative idea of physically implementing effective logical channels on encoded quantum information in a manner that is robust (tolerant) towards imperfections (faults) on the physical realization of the channel.  There is a large variety in interpretations of this property which is claimed under various different assumptions in the literature on a case-by-case basis tailored to the individual engineering problem.  A general quantitative framework and definition for fault tolerance,  that is hoped to encapsulate existing ideas,  was proposed by Gottesman and Zhang in ref.~\cite{gottesman2017fiber}. This framework will be important for us, and we will spend some time discussing it in chapter~\ref{chap:Theory}.  

\section{How to quantum compute}

The basic tool for quantum computation are the Pauli operators.  Consider a qubit with Hilbert space $\CH_2={\rm span}_{\C}\lrc{\ket{0},  \ket{1}}$.  The distinguished basis is given by the states $\ket{0},  \ket{1}$ is called the \textit{computational basis}. It is defined as the eigenbasis of the Pauli $Z$ operator and permuted by the Pauli-$X$ operator
\begin{equation}
Z=\ketbra{0}-\ketbra{1},\hspace{1cm} X=\ket{0}\!\bra{1}+\ket{1}\!\bra{0},  \; ZX=-XZ,
\end{equation}
which together generate the single qubit Pauli group $\CP_1=\langle iI,  X,  Z \rangle$,  which is the unitary representation of the discrete Heisenberg group $H_2'$ we encountered earlier.  Note that everything discussed here can be extended to the case of qudits of dimension $d$,  where $\CH_1={\rm span}_{\C}\lrc{\ket{0},  \hdots \ket{d-1}}$ and the generalized Pauli operators act as
\begin{equation}
Z_d\ket{j}=e^{i\frac{2\pi}{d}}\ket{j},\hspace{1cm} X_d\ket{j}=\ket{j+1\,{\rm mod}\, d},  \; Z_dX_d=e^{i\frac{2\pi}{d}}X_dZ_d. 
\end{equation}
We denote the $n$-qudit pauli group by $\CP_n\lr{d}=\langle e^{i\frac{2\pi}{d}}I, X_d, Z_d\rangle$.
The $n$-qudit Pauli group $\CP_n\lr{d}=\CP_1^{\otimes n}\lr{d}$ is normalized by the Clifford group
\begin{equation}
\Cl_n\lr{d}= \lrc{U\in \CU_n(d):\; U PU^{\dagger}\in \CP_n\,  \forall\, P\in \CP_n},
\end{equation}
which is the subgroup of the $n$-qudit unitary group $\CU_n(d)$ that normalizes the Pauli group. We  will mostly stick to $d=2$ for convenience and since this is the most important special case. The qubit Clifford group together with a unitary operator outside the Clifford group,  such as the magic gate $T=\ketbra{0}+e^{i\pi/4}\ketbra{1}$ allows to efficiently approximate any unitary operation on the Hilbert space. Together with the of operations ``Initialization of computational basis states''  and ``measurement of Pauli-$Z$ on qubits in an $n$ qubit register'' these resources hence constitute a complete set of building blocks to realize any quantum computation.  
Clifford gates are typically implemented through a sequence of gates from a generating set given by the phase $S$ and  Hadamard $H$ gates satisfying
\begin{equation}
SXS^{\dagger}=iXZ,\, SZS^{\dagger} =Z,\; HXH^{\dagger}=Z,
\end{equation}
as well as the controlled NOT gate
\begin{align}
{\rm CNOT}_{ab}X_a{\rm CNOT}_{ab}^{\dagger}&=X_aX_b,  \hspace{1cm}{\rm CNOT}_{ab}Z_b{\rm CNOT}_{ab}^{\dagger}&&=Z_aZ_b,\\
{\rm CNOT}_{ab}Z_a{\rm CNOT}_{ab}^{\dagger}&=Z_a,  \hspace{1.5cm}  {\rm CNOT}_{ab}X_b{\rm CNOT}_{ab}^{\dagger}&&=X_b.
\end{align}

The ``magic'' $T$ gates are typically implemented using a magic resource state $\ket{T}=T\ket{+}$ through a protocol called magic gate injection \cite{GKP,  Bravyi_2005} pictured in fig. ~\ref{fig:Tgate}, since this is one of the simplest fault-tolerant ways to implement this gate building on the fault-tolerant implementations of the CNOT gate, the Pauli measurement as well as the faithful preparation of the magic state.

\begin{figure}
\center
\includegraphics[width=.4\textwidth]{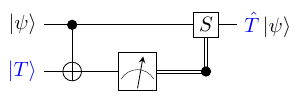}
\caption{The magic gate injection protocol.  The auxiliary magic state is coupled to the data qubit using a CNOT gate and measured in the Pauli-$Z$ basis.  If the outcome is $1$ a classical correction with $S$ is applied to the data qubit.}\label{fig:Tgate}
\end{figure}

The important feature to notice in this section is the close similarity between the behavior of the (generalized) Pauli operators and the displacement operators encountered before.  This is not a coincidence as they both form a realization of the Heisenberg-Weyl group in either the discrete or continuous setting,  and in the previous section we have even seen how the discrete Heisenberg-Weyl operators arise as discrete subgroup of the continuous one when acting on the subspace $V_d$.  Similarly, the Clifford group,  as the unitary normalizer of the Pauli group,  is strictly analogous to the group of Gaussian unitary operators,  which form the normalizer of the displacement operators (see eq. ~\eqref{eq:disp_S}).  This close analogy strongly suggest that the embedding of the discrete Heisenberg-Weyl groups into displacement operators would yield a natural way to implement quantum computation using discrete degrees of freedom on a continuous variable system.  Such an embedding would be realized through an irreducible representation of the discrete Heisenberg-Weyl group,  which is provided by a subspace of the form $V_d$ defined in chap.~\ref{chap:intro}. 

\section{How to quantum error correct}

The core idea of quantum error correction (QEC) is to encode logical information into a subspace $\CH_C$ of a physical Hilbert space $\CH$,  which we refer to as code space.  To keep things simple,  let's assume that the physical Hilbert space $\CH=\CH_2^{\otimes n}$ is made up of a collection of $n$ qubits (generalization to qudits is straightforward) and acted upon by the Pauli group $\CP_n$.  Like displacement operators,  Pauli operators form a complete basis for operators acting on Hilbert space and every operator can be decomposed into a linear combination of Pauli operators $P_{\alpha},  \alpha=0\hdots 4^n-1$
\begin{equation}
E=\sum_{\alpha=0}^{4^n-1} c_{\alpha} P_{\alpha}
\end{equation}
and  two Pauli operators either commute or anticommute.  The most prominent way to distinguish this code space is via the \textit{stabilizer formalism} \cite{gottesman_thesis},  where a subspace of the physical Hilbert space is distinguished by a so-called \textit{stabilizer group},  generated by $r$ independent commuting Pauli operators $\CG=\lrc{g_1,\hdots, g_r}\subset \CP_n$.  The stabilizer group 
\begin{equation}
\CS=\langle \CG \rangle - \lrc{-I, \pm iI}
\end{equation}
is an \textit{Abelian} group and as a subgroup of the Pauli group,  each element has order $2$ and Eigenvalues $\pm 1$.  The stabilizer group carries this name as it acts as the stabilizer of the code space,  which has dimension $\dim\lr{\CH_C}=2^{n-r}=:2^k$,  such that it encodes $k$ logical qubits.  A distinguished subgroup with irreducible action on code space is given by the logical Pauli operators,  which  need to commute with the stabilizer group and hence are contained in its centeralizer $\CC\lr{\CS}\subseteq \CP_n$ within the Pauli operators.  The stabilizer group $\CS\subseteq \CC\lr{\CS}$ is Abelian if and only if it is  a subgroup of its own centralizer, and we obtain representatives for logical Pauli operators as elements in the quotient group $\CC\lr{\CS}/ \CS$,  which has order $4^k$ up to phases.
The logical Pauli group forms the basis to analyse the implementation of logical operators for a stabilizer code. The size of the smallest non-trivial logical operator 
\begin{equation}
d\lr{\CS}=\min_{x \in \CC\lr{\CS} -  \CS } \|x\|_H
\end{equation}
is called the \textit{distance} of the code,  where $\|\cdot\|_H$ denotes the Hamming norm,  and quantifies the minimum number of single-qubit error that need to happen to make up a logical error.  A different way to capture this definition is through the quantum error correction conditions.  Let $\Pi_C=2^{-r}\sum_{s\in \CS}s$ denote the projector on $\CH_C$ and let $\lrc{E_k}$ constitute a set of relevant Kraus operators that capture physical error processes on $\CH$.  The QEC conditions read
\begin{equation}
\Pi_C E_k^{\dagger}E_{k'}\Pi_C = c_{k k'} \Pi_C, \; 
\end{equation}
with Hermitian matrix $c=c^{\dagger}$. This condition captures the requirement that physical errors maintain orthogonality of code states in $\CH_C$ and do not leak information about the code state to the environment. Fulfilling the condition for an error channel relative to a code guarantees the possibility to recover the logical information from its corruption \cite{nielsen00}.  In this language, the distance becomes the Hamming weight of the smallest operator $P=E_k^{\dagger}E_{k'}\in \CP_n$ that violates the quantum error correction conditions.  A quantum error correcting code that encodes $k$ logical qubits into $n$ physical qubits with a distance $d$ is also denoted by $[\![ n, k, d ]\!]$.

Operationally,  quantum error correction proceeds as follows.  A code state $\ket{\psi}\in \CH_C$ undergoes a physical noise channel $\CN$,  which probabilistically applies an error
\begin{equation}
E=\sum_{\alpha=0}^{4^n-1} c_{\alpha} P_{\alpha}
\end{equation}
to the input state with probability $p(E)$.  Upon measurement of the generators of the stabilizer group $\CG=\lrc{g_1,\hdots, g_r}\subset \CP_n$,  the state collapses onto a state with a definite syndrome $\bs{s}$,  which is stabilized by 
\begin{equation}
\CS\lr{\bs{s}}=\langle (-1)^{s_1}g_1,  (-1)^{s_2}g_2,  \hdots,  (-1)^{s_r}g_r \rangle - \lrc{-I, \pm iI}.
\end{equation}
The vector $\bs{s}$ of stabilizer generators eigenvalues measured reflects the subspace that the state has collapsed onto.  
This process -- the measurement collapse -- is the central ingredient to facilitate active quantum error correction.  By collapsing the state onto a fixed syndrome,  the structure of the effective error that has been applied to the state is significantly simplified and behaves almost like a pure Pauli error.  Now to return to code space,  one applies a corrective Pauli operator 
\begin{equation}
P_{r\lr{\bs{s}}}: P_{r\lr{\bs{s}}} \CS\lr{\bs{0}} P_{r\lr{\bs{s}}}^{\dagger} = \CS\lr{\bs{s}}, \label{eq:destabilizer}
\end{equation}  
which produces the same syndrome as we have measured. Since the syndrome,  by design,  contains no information about the logical state,  the correction leaves open the possibility for a logical error to have happened either through the noise channel or via a bad choice of correction $r\lr{\bs{s}}$.  In order to minimize this effect one hence needs to choose $r\lr{\bs{s}}$ such that the probability of a remaining logical error is minimized.  For a Pauli error model (i.e. one for which each Kraus operator is of Pauli-type)  the optimal strategy is to amend a generic choice of correction $r\lr{\bs{s}}$ through a classical optimization called \textit{maximum likelihood decoding}(MLD),  where a logical post-correction is found by solving
\begin{equation}
{\rm MLD}\lr{P_{r\lr{\bs{s}}}}=\max_{L \in \CC\lr{\CS}/\CS} \sum_{S\in \CS} p\lr{P_{r\lr{\bs{s}}}LS}.\label{eq:MLD}
\end{equation}
To see why this works,  note that every Pauli operator $E=P_{r\lr{\bs{s}}}LS$ can be decomposed as a product of a a \textit{pure error}/\textit{destabilizer},  which is a generic Pauli operator that reproduces the correct syndrome according to eq.  \eqref{eq:destabilizer},  a logical representative $L \in \CC\lr{\CS}/\CS$ and a stabilizer $s\in \CS$.  
Solving this optimization maximizes the probability that the error correction process returns the state to the correct code state,  taking into account the degeneracy provided by stabilizer operations.  Since this optimal process requires a rather costly computation of the cost function (note e.g. that $|\CS|=2^{n-k}$ grows exponentially with the number physical qubits),  in practice one often resorts to more computationally efficient methods,  such as simply optimizing 
\begin{equation}
{\rm MLE}\lr{\bs{s}}=\max_{P: P\CS\lr{\bs{0}}P^{\dagger}=\CS\lr{\bs{s}}} p(P),
\end{equation}
dubbed \textit{maximum energy decoding} \cite{Vuillot_2019}. Alternatively, one may also use tensor network methods to approximate the sum in eq.~\eqref{eq:MLD} \cite{Bravyi_2014,  chubb2021general} or one of many other approximations found throughout the literature.  An interesting viewpoint, proposed by Dennis et al. in ref.~\cite{Dennis_2002}, is that the structure of the MLD decoder in eq.~\eqref{eq:MLD} may also be interpreted as the partition function of a classical statistical mechanical model,  providing an interesting connection between quantum error correction and condensed matter theory (see also refs.~\cite{Vuillot_2019,  Chubb2021}).  In this identification \textit{thresholds} of quantum error correcting codes,  i.e.  error parameters below which an increase in the size of the code $n$ leads to an exponential suppression of the remainder logical error probability, are understood as phase-transitions of the corresponding statistical mechanical model. This understanding provides an argument for the existence of thresholds \cite{Dennis_2002,  Vuillot_2019},  and,  reversely,  suggests the possibility to simulate physical phenomena in condensed matter theory through the implementation of quantum error correction.
This  fascinating correspondence is only one of the many interesting connections  quantum error correction offers to foundational topics in physics and provides example for how the theory of quantum error correction,  albeit motivated as a means to facilitate technology,   has the potential to become an integral tool to the basic sciences.
\chapter{GKP coding theory}\label{chap:Theory}
\blfootnote{The content of this chapter is oriented along the publications ref.~\cite{Conrad_2022} and ref.~\cite{Conrad_2024}. In particular the content of secs.~\ref{sec:lattice_perspective}, ~\ref{sec:bases} and~\ref{sec:distance} are adapted from ref.~\cite{Conrad_2022} and the content of secs. ~\ref{sec:Cliff} and~\ref{sec:rosetta} are adapted from ref.~\cite{Conrad_2024}.} 

In this chapter, we will build a abstract foundation for the theory of GKP codes. We begin by discussing a lattice theoretic perspective theory on GKP codes, which extends the lattice theoretic formulation already present in the original work \cite{GKP} and refs.~\cite{HarringtonPreskill, Harrington_Thesis}. The presentation here will mostly be guided by my work in ref.~\cite{Conrad_2022}, but it is worth mentioning that around the time ref.~\cite{Conrad_2022} appeared, two independent lattice theoretic investigations into the GKP code were also published in refs.~\cite{Royer_2022, Schmidt_2022}, which are valuable complementary resources. 
In sec.~\ref{sec:Cliff} we investigate the structure of Clifford gates for the GKP code. Their understanding will naturally lead to a algebraic geometric perspective on GKP codes, which yields a classification of the \textit{space of GKP codes}. We will show how this perspective naturally yields a description of fault-tolerant quantum computation with the GKP code in Gottesman and Zhangs framework for fiber-bundle fault tolerance.

\section{GKP codes: A lattice perspective}\label{sec:lattice_perspective}

A \emph{GKP code} \cite{GKP} is a stabilizer code acting on the Hilbert space of $n$ bosonic modes, where stabilizers are given by displacement operators.

\begin{mydef}[GKP stabilizer group \cite{GKP}]\label{def:GKP}
The stabilizer group of a GKP code is given by a set of displacements
\begin{equation}
\mathcal{S}:=\langle D\lr{\bs{\xi}_1},\dots, D\lr{\bs{\xi}_{2n}} \rangle,
\end{equation}
where $\lrc{\bs{\xi}_i}_{i=1}^{2n}$ are linearly independent and we have  $\bs{\xi}_i^TJ\bs{\xi}_j \in \mathbb{Z} \forall\,
 i,\, j$.
\end{mydef}

\begin{mybox}
\subsubsection*{What is  ... a lattice?}
There are different definitions of lattices in the literature. A very general definition is to define a lattice $L \subset G$ as a discrete subgroup of \textit{topological group}, i.e. a topological space with a product operation between its elements that satisfy the group axioms \cite{morris2015introduction}, such that $\Gamma \backslash G$ \footnotemark has a finite volume relative to the constant (Haar) measure on $G$.

The most common type of lattices we will consider here are discrete additive subgroups of $\R^n$, which are always isomorphic to $\Z^m$ for some $m\leq n$, such that $L= M^T\Z^n \subset R^n$ for some $m\times n$ matrix $M^T$. If $m=n$ the lattice is called \textit{full-rank} and else \textit{degenerate}.

The definition is such that other topological groups can also give rise to lattices. Important examples are $G=\SL_{n}\lr{\R}$ and $G=\Sp_{2n}\lr{\R}$, where the discrete subgroups $SL_n\lr{\Z}$, respectively $\Sp_{2n}\lr{\Z}$ form lattices in $\SL_{n}\lr{\R}$ and  $\Sp_{2n}\lr{\R}$ \cite{morris2015introduction, omeara_symplectic_1978}.
\end{mybox}
\footnotetext{Throughout this work, we will use the backslash $\cdot \backslash \cdot$ for left-quotients, while the usual set exclusion will be denoted by a minus $-$ sign.}

Exploiting the structure of the displacement operators, the GKP construction defines a 
stabilizer group isomorphic to a \emph{lattice} with generator matrix
\begin{equation}
M=\begin{pmatrix}
\bs{\xi}_1^T\\ \vdots \\ \bs{\xi}_{2n}^T
\end{pmatrix},
\end{equation}%
which is simply the set of integer linear combinations of basis elements%
\begin{equation}
\mathcal{L}=\mathcal{L}\lr{M}=\lrc{\bs{\xi} \in \mathbb{R}^{2n}\big\vert \;\bs{\xi}^T=\bs{a}^TM, \, \bs{a}\in \mathbb{Z}^{2n \times 2n}} .
\end{equation}%
We follow the convention of ref.~\cite{GKP}, where the basis vectors of the lattice $\mathcal{L}$ constitute the \textit{rows} of the generator matrix, which is the convention common in coding theory; lattice theory more conventionally uses a column convention, but we will mostly stick to the row convention and specially indicate when we do deviate from it. A lattice is a $\Z$ module: It behaves almost like a vector space in terms of addition of elements and closure under multiplication by elements in $\Z$, but since the set of integers $\Z$ does not contain a multiplicative inverse for every element\footnote{It is a \textit{Ring}.}, the lattice formally is a module (a vector space requires to be built on top of a \textit{field}; see also the ``what is ...'' box~\ref{whatis:rings}).

For $\mathcal{S}$ to constitute a stabilizer group, it needs to be (1.)  a group and (2. ) Abelian. By construction in definition~\ref{def:GKP}, $\mathcal{S}$ is a group and by means of 
 eq.~\eqref{eq:disp_comm} it can be observed that $\mathcal{S}$ is Abelian if and only if the \textit{symplectic Gram matrix} associated with \textit{any} generator $M$ of the lattice%
\begin{equation}
A:=MJM^T
\end{equation}
has only \textit{integer} entries. We then say that the lattice $\mathcal{L}(M)$ is \textit{symplectically integral}. 
 
\begin{lem}
Each element of the stabilizer group $\mathcal{S}$ can be written as
\begin{align}
\prod_{i=1}^{2n} D\lr{\bs{\xi}_i}^{a_i} =e^{i\pi\bs{a}^T  A_{\lowertriangle}  \bs{a}} D\lr{(\bs{a}^TM)^T},
\end{align}
where $A_{\lowertriangle}:\; (A_{\lowertriangle})_{i,j}=A_{i,j}\, \forall\, i>j$  is the lower triangular matrix of $A$.
\end{lem}
\proof
This is verified as follows. Let $\bs{a}_{[i]}$ denote the vector $\bs{a}=\bs{a}_{[2n]}$ with all entries $a_j=0\, \forall j>i$ and let $\phi_{i}$ denote the cumulative global phase when simplifying the product (row vector convention in $D\lr{}$ here)
\begin{align}
\prod_{i=1}^{2n} D\lr{\bs{\xi}_i}^{a_i} 
&= \prod_{i=1}^{2n} D\lr{a_i \bs{\xi}_i} \\
&=\lrq{\prod_{i=1}^{2n-1} D\lr{a_i \bs{\xi}_i} } D\lr{a_{2n} \bs{\xi}_{2n}} \\
&= e^{i\phi_{2n-1}} D\lr{\bs{a}_{[2n-1]}^TM} D\lr{a_{2n} \bs{\xi}_{2n}} \\
&= e^{i\phi_{2n-1}}  e^{-i\pi \bs{a}^T_{[2n-1]}MJ(a_{2n}\bs{\xi}_{2n})} D\lr{\bs{a}_{[2n]}^TM}.
\end{align}
We combined the displacements using eq.~\eqref{eq:disp_comm}, from where we can already see that $e^{i \phi_j}=\pm 1 \forall j$ when $A$ is integer. The expression above allows us to write down the recursion
\begin{align}
\phi_1&=0 \\
\phi_i 
&=\phi_{i-1}-\pi \bs{a}^T_{[i-1]}MJ(a_{i}\bs{\xi}_{i}) \\
&=\phi_{i-1}+\pi \sum_{j=1}^{i-1} a_i A_{i, j} a_j.
\end{align} 
Computing the recursion, we observe that 
\begin{equation}
\phi_{2n}=\pi \sum_{i>j} a_i A_{i,j} a_j = \pi \bs{a}^T A_{\lowertriangle} \bs{a},
\end{equation}
where $A_{\lowertriangle}:\; (A_{\lowertriangle})_{i,j}=A_{i,j}\, \forall\, i>j$  is the lower triangular matrix of $A$.
\endproof

The stabilizer group $\mathcal{S}$ is hence given by
\begin{equation}
\mathcal{S}=\lrc{e^{i\phi_M\lr{\bs{\xi}}} D\lr{\bs{\xi}} \vert \bs{\xi} \in \CL},\label{eq:GKP_2}
\end{equation}
where
\begin{equation}
\phi_M\lr{\bs{\xi}}=\pi \bs{a}^T A_{\lowertriangle} \bs{a}, \; \bs{a}^T=\bs{\xi}^TM^{-1}
\end{equation}
determines the phases attached to each lattice displacement relative to the \emph{pivot basis} $M$, i.e. the set of basis vectors for which  each associated displacement operator is fixed to eigenvalue $+1$ by definition~\ref{def:GKP}. We will also use the notation $(M, \phi_M)$  to specify GKP stabilizer groups from its generator $M$ when the phase-sector is relevant. The pair $(M, \phi_N)$  specifies that the generator for the stabilizer group $D\lr{M_i^T}$ is fixed by eigenvalue $e^{i\phi_N\lr{M_i^T}} $ in code space. Consistently, we have $(M, \phi_M)=(M, 0)$.

 When $\CL$ is symplectically integral, we have $e^{i\phi_M\lr{\bs{\xi}}}=\pm 1\, \forall \, \bs{\xi} \in \CL$, and $e^{i\phi_M\lr{\bs{\xi}}}=1$ holds when  $A \in \lr{2\mathbb{Z}}^{2n \times 2n}$ is \textit{even}. While these additional phases have no effect on the projective Hilbert space, they determine the eigenvalue of the corresponding displacement operator in the codespace determined by $(M, \phi_M)=(M, 0)$
These phases are e.g. non-trivial for the sensor state \cite{displacement_sensor} with stabilizer group
\begin{equation}
\mathcal{S}_{\text{sensor}}=  \langle D\lr{\bs{e}_1},\, D\lr{\bs{e}_2}  \rangle= \lrc{(-1)^{mn }D\lr{m\bs{e}_1+n\bs{e}_2},\; (m,n)\in \mathbb{Z}^2 },\label{eq:sensor}
\end{equation} where $\bs{e}_j$ are canonical basis vectors, so that code-states should have eigenvalue $-1$ on displacement for which $mn$ is odd if they have eigenvalue $+1$ on each generator $D\lr{\bs{e}_1},\, D\lr{\bs{e}_2}$.  In an alternative definition, used in ref.~\cite{GKP}, with $\mathcal{S}' = \lrc{D\lr{\bs{\xi}}, \bs{\xi \in \CL}}$, $D\lr{\bs{e}_1}D\lr{\bs{e}_2}=(-1)D\lr{\bs{e}_1+\bs{e}_2}$ would formally not be included in the stabilizer group. 

 To encode discrete quantum information, such as a qubit, the continuous state space needs to be fully ``discretized'' by introducing suitable constraints. This is done by choosing a generating set for the stabilizer group with $2n$ linearly independent generators. Each linearly independent generator can be seen as quantizing one direction in phase space. This is the reason why  we require the lattice $\mathcal{L}$ to be \textit{full rank}, i.e., $M$ has full row-rank. It is possible to specify a lattice in $\mathbb{R}^{2n}$ using more than $2n$ basis vectors, but they can always be reduced to $2n$ linearly independent vectors, a process for which a number of (efficient) algorithms are known \cite{CryptoLectureNotes}. Note also that non-full rank (or \textit{degenerate}) lattices had recently been explored to define effective GKP codes \cite{Noh_2020, oscillator_to_oscillator}. The code space of such codes retains continuous quadratures which allows to encode and perform error correction on encoded CV states. In this work we focus on encoding qubits or qudits defined via full-rank lattices and refer to refs.~\cite{Noh_2020, Conrad_2021} for discussions on encoding continuous information using the GKP code.

While a lattice $\CL=\CL\lr{M}$ is unique as a geometric object, different generator matrices $M, \,M' $  can generate the same lattice if and only if there exists an unimodular matrix $U\in \GL_{2n}\lr{\Z},\; |\det\lr{U}|=1$ such that
\begin{equation}
M'=UM.
\end{equation}
Such transformation also transforms 
\begin{equation}
A \mapsto A'=UAU^T,
\end{equation}
which has even entries if and only if $A$ does. 
Due to the phases appearing in eq.~\eqref{eq:GKP_2}, when a different basis $M'$ is used to fix the stabilizer group as in definition~\ref{def:GKP}, the generating set for the stabilizer group needs to be chosen as 
\begin{equation}
\lr{M', \, \phi_M}
\end{equation}
to yield the same stabilizer group and the same code-space.

The symplectic Gram matrix is invariant under a symplectic transformation 
\begin{equation}
M\mapsto MS^T=\begin{pmatrix}
(S\bs{\xi}_1)^T\\ \vdots \\ (S\bs{\xi}_{2n})^T
\end{pmatrix},\label{eq:M_symp}
\end{equation}%
which generally  change the lattice $\CL\mapsto S\CL$ but leave the symplectic Gram matrix invariant. Comparing to the discussion in the introduction, see eq.~\eqref{eq:X_H}, one can also understand the symplectic transformation as a transformation on the basis or ambient space which preserves symplectic inner products.

A \textit{sublattice} $ \mathcal{L}' \subseteq \mathcal{L}$ is a subset of $\mathcal{L}$ that is itself a lattice. Any $d\leq 2n$ dimensional sublattice of $\mathcal{L}$ can be specified by the basis
\begin{equation}
S=BM,
\end{equation}
where $B\in \mathbb{Z}^{d\times 2n}$ is an integer matrix of full row rank.
The \textit{symplectic dual} (simply ``dual'' in the following) of a lattice $\CL$ is the lattice $\CL^{\perp}$ that consists of all vectors that have integer symplectic inner product with any vector from $\CL$%
\begin{equation}
\CL^{\perp} = \lrc{ \mathcal{\bs{\xi}}^\perp \in \mathbb{R}^{2n}\ |\ \lr{ \bs{\xi}^\perp }^T J \bs{\xi} \in \mathbb{Z}\ \forall \bs{\xi}\in \mathcal{L}   }.
\end{equation}%
Within the GKP code construction vectors $\mathcal{\bs{\xi}}^\perp \in \CL^\perp$  correspond to the displacement amplitudes associated to the centralizer of the stabilizers within the set of displacement operators.
A canonical choice of basis $M^{\perp}$ for the symplectic dual is specified by fixing $M^{\perp}$ to satisfy
\begin{equation}
M^{\perp}JM^T=I.
\end{equation} 
Since we will focus on full rank lattices, for which $M$ is non-singular, we obtain the canonical dual basis as \begin{equation}
M^{\perp}=(JM^T)^{-1}=M^{-T}J^T.\label{eq:canonical_perp}
\end{equation}
Together with the definition of $A$ it can be shown that 
\begin{equation}
M=AM^{\perp},
\end{equation}
that is, the symplectic Gram-matrix $A$ describes how the sublattice $\CL=\CL\lr{M}\subseteq \CL^{\perp}=\CL\lr{M^{\perp}}$ associated to stabilizer operators embeds into the lattice associated to its centralizer. 
The \textit{dual quotient}  of $\CL$, 
\begin{equation}
\CL^{\perp}/\CL \sim \CC\lr{\CS}/\CS,
\end{equation}
thus lists the logically distinct displacements admitted by the GKP code $\mathcal{S}$. These displacements form the effective logical Pauli group. 

Finally, the number of logically distinct centralizer elements associated to $\mathcal{S}$ is 
\begin{equation}
\dim\lr{\CH_C}^2=|\CL^{\perp}/\CL|=|\det\lr{M}|/|\det\lr{M^{\perp}}|=|\det{A}|=|\det\lr{M}|^2.
\end{equation}
One can verify this formula geometrically by imagining the partition of unit cells of $\CL$ with patches unit cells of $\CL^{\perp}$, see fig.~\ref{fig:Z2_A2}.
Under basis transformation of the direct lattice $M\mapsto UM$ the canonical dual basis transforms as
\begin{equation}
M^{\perp}\mapsto U^{-T}M^{\perp},
\end{equation}
and under symplectic transformations $M\mapsto MS^T $ we obtain $M^{\perp}\mapsto M^{\perp} S^T$
The symplectic Gram matrix of the dual lattice can be shown to satisfy 
\begin{equation}
A^{\perp}:=M^{\perp}J\lr{M^{\perp}}^T=-A^{-1}. \end{equation}
It is similarly common  to define the \textit{euclidean} dual of a lattice $\CL$ which we shall denote by $\CL^*$. This is the lattice in $\mathbb{R}^{2n}$ consisting of all vectors with integer \textit{euclidean} inner product with every vector in $\CL$. 

The euclidean dual is more common than the symplectic one in the lattice theory literature, and hence often simply called dual. Similar to the symplectic case, a canonical basis $M^*$ for the dual lattice can be fixed to satisfy $M^* M^T=I $.
Since $J\in O(2n)$ is an orthogonal matrix, we can observe that the symplectic dual is equivalent to the euclidean dual up to an orthogonal rotation. In particular, the distribution of lengths of vectors in $\CL^{\perp}$ is equal to  the distribution of lengths of vectors in $\CL^*$. We will take advantage of this fact when analyzing the code distance of GKP codes in section~\ref{sec:distance}.

We close this section by introducing the two most studied classes of GKP codes, which we will use for illustrations throughout this manuscript. The first class, which we will refer to as \textit{scaled codes}, have been thoroughly examined in ref.~\cite{Harrington_Thesis} and build on symplectic \textit{self-dual} lattices. A symplectic self-dual lattice $\CL_0$ is a symplectically integral lattice for which $|\det\lr{M_0}|=|\det\lr{A_0}|=1$ and, consequently, $\mathcal{L}_0=\mathcal{L}_0^{\perp}$ and $A_0=J$: The generator matrix $M_0\in \Sp_{2n}\lr{\R}$ is (up to a basis transformation equivalent to a) symplectic matrix and the associated code-space is one-dimensional. GKP codes associated to a symplectically self-dual lattice are analogous to so-called \textit{stabilizer states} known in quantum information and the sensor state in eq.~\eqref{eq:sensor} already presented one such example.

\subsubsection{Scaled GKP codes.} 

\begin{mydef}
A \textit{scaled GKP code}  is obtained by rescaling a symplectically self-dual lattice $\CL_0=\CL^{\perp}$ to $\CL=\sqrt{d} \CL_0$ with an integer $d\in \N$.
\end{mydef}

Let $M_0$ be a generator matrix for $\CL_0$. We will see shortly that one can always choose a basis such that $M_0JM_0^T=J$ is itself a symplectic matrix.
 
The symplectic Gram matrix associated to the scaled GKP code becomes $A=d J$ and the dimension of the code-space is $\dim\lr{\CH_C}=|\sqrt{\det\lr{ d J}}|=d^n$. By choosing $d=2$ this yields a code with $k=n\log_2\lr{\dim\lr{\CH_C}}=n$ encoded qubits. 
Let 
\begin{equation}
\lambda_1\lr{\CL_0}:=\min\lrc{\|\bs{x}\|, \, \bs{x}\in \CL_0-\lrc{0}}
\end{equation}
be the length of the shortest vector in $\CL_0=\CL_0^{\perp}$. Note that throughout this manuscript $\|\cdot\|=\|\cdot\|_2$ denotes the euclidean \textit{2-norm}.
The rescaling implies  $\lambda_1\lr{\CL}= d^{\frac{1}{2}}\lambda_1\lr{\CL_0}$ and $\lambda_1\lr{\CL^{\perp}}=d^{-\frac{1}{2}}\lambda_1\lr{\CL_0}$. For $d>1$ the symplectic dual vector corresponding to $\lambda_1\lr{\CL^{\perp}}$ cannot be in $\CL$ and hence constitutes the shortest logically non-trivial displacement, the length of which decreases with the number of logical dimensions that are squeezed into the code.
Prominent examples of scaled GKP codes that have already been discussed in ref.~\cite{GKP} are generated by
\begin{equation}
    \sqrt{2}M_{\Z^2}=\sqrt{2}I_2,
\end{equation}
also known as the \textit{square GKP code}, that encodes $k=1$ qubit into $n=1$ oscillator and
\begin{equation}
    \sqrt{2}M_{A_2}=3^{-\frac{1}{4}}\begin{pmatrix}
2 & 0 \\ 1 & \sqrt{3}
\end{pmatrix},\label{eq:GKP_hex}
\end{equation}
which is known as the \textit{hexagonal GKP code}, similarly with parameter $k=n=1$ but with a different lattice geometry that requires a slightly larger displacement amplitude to implement a logical operator. 
The hexagonal GKP code is labeled by the symbol $A_2$, which we will explain in chapter~\ref{chap:constructions}.
Here we have  $\lambda_1\lr{\CL^{\perp}}=3^{-\frac{1}{4}}$ compared to $\lambda_1\lr{\CL^{\perp}}=2^{-\frac{1}{2}}$ for the square GKP code. The simplest $\Z^2$ and $A_2$ scaled GKP codes with $d=2$ are illustrated in fig.~\ref{fig:Z2_A2}.
Further symplectic self-dual lattices with larger $\lambda_1\lr{\CL_0}$ and a numerical procedure to find symplectic self-dual lattices are detailed in ref.~\cite{Harrington_Thesis} and discussed in chapter~\ref{chap:constructions}.

Scaled GKP codes are important as every GKP code can be understood as an extended version of a scaled GKP code. We discuss this relationship in the next section.

\begin{figure}
\center
\includegraphics[width=.8\textwidth]{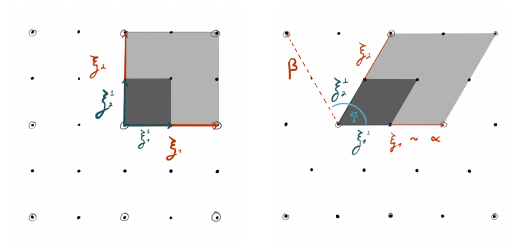}
\caption{
The symplectic lattices $\Z^2$ (left) and $A_2$ (right) scaled by $d=2$ and their respective (dual) unit cells. 
The logical displacement amplitudes are marked in turquoise and stabilizer displacements are marked in red.
}
\label{fig:Z2_A2}
\end{figure}

\subsubsection{Concatenated square GKP codes.} 
The second relevant class of GKP codes that of concatenated codes built from the square GKP code and a qubit quantum-error correcting code. Let $Q\subset \mathbb{Z}_2^{2n}$ be a set of binary symplectic vectors\footnote{That is $q^T J q' = 0\mod 2$ for all $q,\ q' \in Q$ (considering addition over reals or in $\mathbb{Z}_2$ does not make any difference).} such that 
\begin{equation}
\mathcal{S}_Q=\lrc{Z^{q_1}_1Z^{q_2}_2\hdots X_{n+1}^{q_{n+1}}\hdots X_{2n}^{q_{2n}} \big\vert \forall \bs{q} \in Q}
\end{equation} 
describes the stabilizer group 
of an $[\![n,k,d]\!]$ qubit quantum error correcting code. 
We can embed binary vectors in $\mathbb{R}^{2n}$ in the trivial way. The lattice associated to the concatenated GKP code will be given by 
\begin{equation}
\CL = \Lambda\lr{Q}:=\lrc{\bs{x}\in \mathbb{R}^{2n}\big\vert \sqrt{2}\bs{x}\mod 2 \in Q }
\end{equation}
and equally describes an encoding of $k$ logical qubits into $n$ modes.
Such lattices are known in the literature as \textit{Construction A} lattices 
\cite{ConwaySloane}.
%
Given $r=n-k$ symplectic vectors associated to a set of generators for $\mathcal{S}_Q$ that we stack row-wise into a generator $B_Q\in\mathbb{Z}_2^{r\times 2n}$, we can write down the  generator for the concatenated code as
\begin{equation}
M_\mathrm{conc}=\frac{1}{\sqrt{2}}\begin{pmatrix}
 B_Q  \\ \hline  
2I_{2n}
\end{pmatrix}.
\end{equation}
This generator defines an overcomplete basis and can be row-reduced to only consist of $2n$ generators. In slightly greater generality, a concatenated GKP code can be understood as a GKP code where the lattice $\CL$ contains a  sublattice corresponding to the union of the single mode GKP codes. We summarize this as definition

\begin{mydef}
A GKP code described by lattice $\CL \subseteq \CL^{\perp} \subset \R^{2n}$ is \textit{concatenated}, if it contains a sublattice

\begin{equation}
\sqrt{d}\CL_0^{\oplus n} \subset \CL,
\end{equation}
where $d\in \N$ and  each local sublattice, $\CL_{\rm loc}=\sqrt{d}\CL_0 \subset \R^2,\;  \CL_0=\CL_0^{\perp}$, corresponds to a single-mode GKP code. When $\CL_0=\Z^2$, we call the GKP code \textit{square concatenated}.
\end{mydef}

This definition is flexible enough to allow for arbitrary choices of single-mode GKP codes and dimensionalities of the effective local logical qudits. It however constrains each local GKP code to be the same. An even more flexible  structure is defined by \textit{glueing}, a construction further discussed in chapter~\ref{chap:constructions}.

\section{Bases of GKP codes} \label{sec:bases}

\subsection{The canonical basis and symplectic equivalence}

A core feature of symplectically integral lattices $\CL \subseteq \CL^{\perp}$ is that there always exists a choice of basis that partitions the vectors in to symplectically conjugate pairs. This is the content of the following lemma due to Frobenius \cite{Bourbaki9, GKP, beauville_theta}, which we adapt from the presentation in refs.~\cite{beauville_theta} and \cite{Bourbaki9}.

\begin{lem}[Frobenius Lemma \cite{Bourbaki9, beauville_theta}]\label{lem:Frobenius}
Let $L$ be a free finitely generated $\Z$-module and $E:\, L \times L \to \Z$ a skew-symmetric, non-degenerate form. There exists positive integers $d_1 | d_2| \hdots |d_n$ and a basis $\lr{\bs{e}_1,\, \hdots \bs{e}_n, \bs{f}_1, \hdots, \bs{f}_n}^T$ of $L$, such that the matrix $A$ of $E$ becomes $A=J_2\otimes D$, with $D={\rm diag}\lr{d_1,\hdots, d_n}$ and the ideals $ d_i \Z,  i=1\hdots n$ are uniquely determined.
\end{lem}

\proof (from \cite{beauville_theta})
Let $d_1$ be the minimum of the numbers $E\lr{\bs{a}, \bs{b}}$ for $\bs{a}, \bs{b} \in L $, $E\lr{\bs{a}, \bs{b}}>0$; choose $\bs{e}_1, \bs{f}_1 \in L$ such that $E\lr{\bs{e}_1, \bs{f}_1}=d_1$. For any $\bs{z}\in L$, $E\lr{\bs{e}_1,  \bs{z}}$ is divisible by $d_1$ -- otherwise, using Euclidean division, we would find $\bs{z}$ with $0<E\lr{\bs{e}_1, \bs{z}}<d_1$. Likewise, $E\lr{ \bs{f}_1 \bs{z}}$ is divisible by $d_1$. Set $U=\Z\bs{e}_1+ \Z\bs{f}_1$ such that we have $L=U\oplus U^{\perp}$. This holds as for any $\bs{z}\in L$ we have

\begin{equation}
\bs{z}=\frac{E\lr{\bs{z}, \bs{f}_1}}{d_1}\bs{e}_1+\frac{E\lr{\bs{e}_1, \bs{z}}}{d_1}\bs{f}_1+\lr{\bs{z}- \frac{E\lr{\bs{z}, \bs{f}_1}}{d_1}\bs{e}_1-\frac{E\lr{\bs{e}_1, \bs{z}}}{d_1}\bs{f}_1}.
\end{equation}

Reasoning by induction on the rank of $L$, we find that $L=U_1\oplus U_2 \oplus\hdots U_n$ decomposes in $n$ $2$- dimensional subspaces $U_i=\Z \bs{e}_i +\Z\bs{f}_i, \; E\lr{\bs{e}_i, \bs{f}_i}=d_i$ and $d_2|\hdots |d_n$. Using Euclidean division it can again be shown that if $d_1$ does not divide $d_2$ one can find $k \in \Z$ such that $E\lr{\bs{e}_1+\bs{e}_2, k\bs{f}_1+\bs{f}_2}<d_1$, which is a contradiction. Finally, we show that the ideals $d_i\Z$ are uniquely determined. To show this, note that there exist a canonical isomorphism $L \rightarrow L^{\lor}: \bs{x} \mapsto \bs{x}(\cdot)=\bs{x}^TJ\lr{\cdot} $. Applying this to the basis elements $\bs{e}_i, \bs{f}_i$ determines functionals into the ideals  $L\rightarrow d_i\Z$. 
 \endproof

The Frobenius lemma immediately implies a standard form for symplectic Gram matrices of GKP codes.

\begin{cor}[Normal form for GKP Codes]
Let $\CL\subseteq\CL^{\perp}$ be a GKP code with symplectic gram matrix $A\in \Z^{2n \times 2n}$. There exists a unimodular transformation $U\in \GL_{2n}\lr{\Z}$, such that
\begin{equation}
UAU^T = J_2 \otimes D
\end{equation}
with positive integer diagonal matrix  $D={\rm diag}\lr{d_1,\hdots, d_n}$ and unique integers $d_1 | d_2 | \hdots | d_n$.
\end{cor}
\proof
We use the result of lemma~\ref{lem:Frobenius}. The Lattice $\CL=M^T\Z^{2n}$ is a free finitely generated $\Z$-module and $E: \CL\times \CL \rightarrow \Z,\;  E\lr{\bs{x}, \bs{y}}=\bs{x}^TJ\bs{y}$ is the trivial skew-symmetric non-degenerate form. By lemma~\ref{lem:Frobenius} there exists a basis $\lr{\bs{e}_1,\, \hdots ,\bs{e}_n, \bs{f}_1, \hdots, \bs{f}_n}^T$ of $\CL$, such that the symplectic Gram matrix $A$ takes the form $A=J_2 \otimes D$ with $D={\rm diag}\lr{d_1,\hdots, d_n}$. Since all bases of a lattice are unimodularily equivalent, there is a unimodular matrix $U$ with $UAU^T=J_2 \otimes D$.  Per lemma~\ref{lem:Frobenius} the principal ideals $d_i\Z$ are unique.. Since two principal ideals $n\Z = m\Z$ are equal exactly when $n=\pm m$, and we have chosen $d_i > 0$, the matrix $D$ is unique. 
\endproof

We call $D$ -- either as matrix or as vector $(d_1,\hdots, d_n)^T$ depending on the context --  the \textit{type} of the GKP code. The corresponding basis $M=(\bs{\xi}_1\,\hdots\, \bs{\xi}_{2n} )^T$ is called the canonical basis, and we refer to the corresponding symplectic Gram matrix $A=J_2\otimes D$ as the \textit{standard form}.
We have shown that the canonical basis is given by symplectically conjugate pairs of vectors
\begin{align}
    \lr{\bs{\xi}_{i}, \bs{\xi}_{i+n}}:\, &\bs{\xi}_{i}^T J\bs{\xi}_{i+n}=d_i \;\;\forall i\in \lrc{1,\hdots, n}, \nonumber \\
    &\bs{\xi}_{i}^T J\bs{\xi}_{j}=0 \;\; \forall j\neq i+n\;\;\forall   i\in \lrc{1,\hdots, n},
\end{align}
and similarly, the dual basis $M^{\perp}=(\bs{\xi}_1^{\perp}\,\hdots\, \bs{\xi}_{2n}^{\perp} )^T$ can be arranged into pairs of vectors $(\bs{e}_i, \bs{f}_i)=(\bs{\xi}^{\perp}_{i}, \bs{\xi}^{\perp}_{i+n})$ with
\begin{align}
\bs{e}_i^T J\bs{f}_{i}&=\frac{1}{d_i} \;\; \forall i\in \lrc{1,\hdots, n} , \nonumber \\
 \bs{e}_i^T J\bs{f}_{j}&=0\;\; \forall  j\neq i\, \in \lrc{1,\hdots, n}.
\end{align}

These vectors constitute precisely the canonical representatives for the logical generalized Pauli operators of the GKP code and satisfy the rules of the desired Heisenberg-Weyl algebra
\begin{align}
\label{eq:HW_operators}
    X_i :=  D\lr{\bs{e}_i}, \, Z_i :=D\lr{\bs{f}_i}: \qquad
    Z_iX_i&=e^{i\frac{2\pi}{d_i}}X_iZ_i, \quad \\
    X_i^{d_i}, Z_i^{d_i} &\in \CS, \quad \\
    [X_i, Z_j] &=0\;\forall i\neq j\,. 
\end{align}

\begin{remark} The structure of this Heisenberg-Weyl algebra also lends itself to an alternative  interpretation for GKP codes:  Given a physical Hilbert space, we can imagine to first pick Heisenberg-Weyl operators $X_i, Z_i: e^{i\frac{2\pi}{d_i}}X_iZ_i$ and then to simply \textit{define} the codespace to be the subspace $\CH_C\subset \CH$ in which $X_i^{d_i}$ and $Z_i^{d_i}$ act as identity.
\end{remark}

As already pointed to in the construction of scaled GKP codes, there is a close connection between general GKP codes described by a symplectically integral (or \textit{weakly symplectically self-dual}) lattice $\CL \subseteq \CL^{\perp}$ with integral symplectic Gram matrix $A=J_2 \otimes D$ and symplectic self dual lattices with $\CL_0=\CL_0^{\perp}$ and $A=J$. Note that for scaled GKP codes $\CL=\sqrt{d}\CL_0$, the \textit{type} is $D=(d)^{\oplus n}$ and we additionally have $d=1$ for symplectic self-dual lattices.

We can reinterpret a general GKP code as follows. Take the symplectic conjugate pairs of vectors $\lr{\bs{\xi}_i, \bs{\xi}_{i+n}}: \bs{\xi}_{i}^T J\bs{\xi}_{i+n}=d_i $ that make up basis $M$ and define  rescaled pairs
\begin{equation}
\lr{\bs{\xi}_i,'  \bs{\xi}_{i+n}'} =\lr{ d_i^{-1}\bs{\xi}_i, \bs{\xi}_{i+n}}: \bs{\xi}_{i}'^T J\bs{\xi}_{i+n}'=1.
\end{equation}
The rescaled pairs of vectors now define a symplectically self-dual lattice with basis $M_0=(\bs{\xi}'_1\,\hdots\, \bs{\xi}'_{2n} )^T$, such that
\begin{equation}
    M:=(D\oplus I_n) M_0. \label{eq:scaledGKP}
\end{equation}

This decomposition implies two ways to interpret the type $D$ GKP code given by basis $M$. First of all, it implies that the lattice associated to the GKP code $\CL$ can be understood as a \textit{sublattice} of the symplectic lattice $\CL_0$  where $\CL$ is embedded into $\CL_0$ by a $ D\oplus I_n$ linear combination of the basis vectors in $M_0$. Secondly, by comparing to eq.~\eqref{eq:M_symp}, one may also interpret it as a GKP code with lattice basis $D\oplus I_n$ that underwent a symplectic transformation via the symplectic matrix $S=M_0^T$. 

In particular, this duality also immediately implies

\begin{them}[Symplectic equivalence]\label{thm:symplectic_equiv}
Let $M, N\, \in \mathbb{R}^{2n\times2n}$ be generators of full rank lattices that each satisfy $\CL \subseteq \CL^{\perp}$ . Without loss of generality, assume the bases $M, N$ are chosen such that $A_M=MJM^T=J_2\otimes D_M$ and $A_N=NJN^T=J_2\otimes D_N$ are in standard form. Then $M=NS^T$ for some symplectic matrix $S$ if and only if $D_N=D_M=:D$.
\end{them}
\proof
We have already shown that there exist symplectic matrices $M_0, N_0$ such that
\begin{align}
M=(D\oplus I_n) M_0,\\
N=(D\oplus I_n) N_0.
\end{align}
Choosing $S=N_0^{-1}M_0$ yields the desired relationship $M=NS^T$. Reversely, since symplectic transformations $S\in \Sp_{2n}\lr{\R}$ satisfy $S^TJS=J$, the sympelectic relation $M=NS^T$ immediately implies $A_M=A_N$.
\endproof

Further applying a squeezing transformation $S=D^{-\frac{1}{2}}\oplus D^{\frac{1}{2}} \in \Sp_{2n}\lr{\R}$ to the GKP code with lattice basis $(D\oplus I_n)$ shows that these codes are in fact simply squeezed square GKP codes. This observation implies that every GKP code can be obtained by symplectically transforming a square GKP code with the correct dimensionalities.

\begin{cor}[Normal form of generators]
Let $M\in \mathbb{R}^{2n\times2n}$ describe the basis of a full rank lattice $\CL\subseteq\CL^{\perp}$. W.l.o.g. assume $M$ is chosen such that $A=J_2\otimes D$ is in standard form. The code specified by the generator
\begin{equation}
N_{\square}=\oplus_{j=1}^n \sqrt{d_i} I_2
\end{equation}
is symplectically equivalent to the one specified by $M$. 
\end{cor}
\proof
We have already seen earlier that we can write $M=\lr{D\oplus I_n}M_0$ with $M_0$ symplectic. Applying the symplectic transformation $D^{-\frac{1}{2}}\oplus D^{\frac{1}{2}} M_0^{-T}$ yields
\begin{equation}
M\mapsto \sqrt{D}\oplus \sqrt{D},
\end{equation}
which is a trivial collection of $n$ single mode square with local scaling dimension $d_i$. 
\endproof
Note that, to highlight the local structure of the code we have used the indexing of quadratures $\lr{\hat{q}_1,\hat{p}_1,\hat{q}_2,\hat{p}_2,\ldots}$ where the symplectic form takes the form $J_{2n}'=I_n\otimes J_2$ which looks different from the usual representation with $J_{2n}=J_2\otimes I_n$ but constitutes an equivalent choice. In the statement of this corollary, the matrix $N_{\square}$ is diagonal, such that its stabilizers each act on a single quadrature at a time. 
Since $N_{\square}$ decomposes into a direct sum over each mode, we can always prepare a code state of $M$ by locally preparing a code state of $N_{\square}$ and apply the corresponding symplectic transformation. 

The same symplectically equivalent local decomposition can also be stated using the hexagonal GKP code on each local mode and since $A_{\square}=A_{\mhexagon}=d J_2$. Similarly, any other scaled GKP code can also be used to provide a similar decomposition. This is an example of \textit{ Gaussian code switching} which is always possible between two GKP codes $\CL\lr{M}$ and $\CL\lr{N}$ on the same number of modes whenever $D_N=D_M$ in their respective standard bases.

Based on the preceding observations, we can also deduce the following corollary.

\begin{cor}[Normal form in prime dimensions]\label{thm:symplectic_equiv_cor}
For $d=|\det(M)|$ prime, the lattice $\CL \subset \mathbb{R}^{2n}$ is symplectically equivalent to a code specified by
\begin{equation}
N'=\sqrt{d}I_2 \oplus I_2^{\oplus (n-1)}, \label{eq:symp_single}
\end{equation} 
\end{cor}
\begin{proof}
For a basis of $\CL$ in standard form, we have %
\begin{equation}
d=\prod_{j=1}^{n}d_j.
\end{equation}%
Since $D$ is a positive integer diagonal matrix, if $d$ is prime, the eigenvalues of $D$ are uniquely specified by $D=\diag(d, 1,1, \hdots)$ up to permutations. Hence by theorem~\ref{thm:symplectic_equiv} we have that $\CL$ is symplectically equivalent to the code specified by eq.~\eqref{eq:symp_single}.
\end{proof}
From the proof above observe that the number of symplectically inequivalent classes of codes with given logical dimension $d$ corresponds to the number of different factorizations of $d$. For example, for two modes encoding two qubits there are two inequivalent choices: $D = \mathrm{diag}\lr{4,1}$ or $D = \mathrm{diag}\lr{2,2}$, corresponding to two symplectically inequivalent classes of codes.

\subsection{Complex parametrization} 

Before we proceed, it is interesting to also write down a complex representation of symplectic lattices in $\C^n \sim \R^{2n}$, which represents the displacement amplitudes of GKP codes when the complex representation of displacement operators from eq.~\eqref{eq:disp_complex} is used.
To this end we can map $C: \, \lr{\bs{x}, \bs{y}}^T\in \R^{2n} \mapsto \bs{x}+i\bs{y} \in \C^n$, with reverse map $\bs{\xi}_{\bs{\gamma}}=\Re\lr{\bs{\gamma}}\oplus \Im\lr{\bs{\gamma}}$.
Under this map the linear operation $J^T$ on $\R^{2n}$ is one-to-one with the scalar multiplication by the complex unit $i$ on $\C^n$ such that euclidean and symplectic inner products become encoded in the Hermitian inner product

\begin{equation}
\bs{\gamma}^{\dagger}\bs{\delta}=\bs{\xi}_{\bs{\gamma}}^T\bs{\xi}_{\bs{\delta}}+i\bs{\xi}_{\bs{\gamma}}^TJ \bs{\xi}_{\bs{\delta}}. \label{eq:comp_inner}
\end{equation}
We have already seen the appearance of the symplectic inner product in terms of the imaginary part of a Hermitian inner product before when discussing the complex parametrization of displacement operators in eq.~\eqref{eq:disp_complex}.

Denote the complexification of the lattice $\CL$ by $\Lambda$. 
 We have seen earlier that in its real parametrization, every lattice $\CL\subseteq\CL^{\perp}$ can be written in a basis such that their generators are related by $M=AM^{\perp}$, where $A=MJM^T$ was the symplectic Gram matrix. Similarly, we can construct a complex matrix $\Pi \in \C^{n\times 2n}$, such that its \textit{columns} represent a basis for $\Lambda$, this matrix is sometimes also called the matrix of \textit{periods}, and we will only later discuss why. For now, let $\Pi, \Pi^{\perp}$ be the complexifications of $M, M^{\perp}$. The above relation simply translates into

\begin{equation}
\Pi=\Pi^{\perp}A^T.
\end{equation}

In this section we have established that every symplectically integral lattice $\CL$ can be written as a sublattice of a symplectic lattice provided by the decomposition of bases
$M=\lr{D\oplus I_n} M_0$.
Lets denote with
\begin{equation}
M_0^T=\begin{pmatrix}
\alpha & \gamma\\ \beta & \delta
\end{pmatrix}
\end{equation}
the  generator matrix for the symplectic lattice $\CL_0=\CL_0^{\perp}$. A complex basis for $\CL_0$ is then given by

\begin{equation}
\Pi_0=\begin{pmatrix}
\alpha+i\beta & \gamma+i\delta
\end{pmatrix}= \begin{pmatrix}
I & \Omega
\end{pmatrix} (\alpha+i\beta),
\end{equation}
where we have set $\Omega= \lr{\gamma + i\delta}(\alpha+i\beta)^{-1}$, which is an element of \textit{Siegel upper half space} \cite{Freitag1991}
\begin{equation}
\HH_n=\lrc{Z=Z^T\in \C^{n \times n},\; \Im\lr{Z}>0}.
\end{equation}
Generally, we can always choose a basis such that the period matrix takes the form (by aligning the axes of our coordinate system with the first $n$ columns of $\Pi_0$)
\begin{equation}
\Pi_0=\begin{pmatrix} I & \Omega\end{pmatrix},\; \Omega \in \HH_n.
\end{equation}
We have from eq.~\eqref{eq:comp_inner}
\begin{equation}
\Pi_0^{\dagger} \Pi_0 = M_0M_0^T + i M_0 J M_0^T=M_0M_0^T + i J,
\end{equation}
such that $\Im\lr{\Pi^{\dagger} \Pi}=A_0 $ reproduces the symplectic Gram matrix of $\CL_0$. When the lattice is equipped with a non-trivial type $D$, one can instead choose the basis to take the form
\begin{equation}
\Pi=\begin{pmatrix}
D & \Omega
\end{pmatrix}, \;  \Omega \in \HH_n,
\end{equation}
which satisfies
\begin{equation}
\Im \lr{\Pi^{\dagger} \Im\lr{\Omega}^{-1}\Pi}=J_2 \otimes D.
\end{equation}
In this form the Hermitian inner product with Kernel $\Im\lr{\Omega}^{-1}$ yields a positive definite form whose imaginary part reproduces the symplectic Gram matrix. 
In fact, one can see that every symplectic matrix can be associated to an element $\Omega \in \HH_n$ by defining the \textit{M{\"o}bius} or \textit{modular} action for a symplectic matrix $S=\begin{pmatrix}
A & B \\ C & D
\end{pmatrix}$ by $S.Z=(AZ+B)(CZ+B)^{-1}$ which can be shown to preserve the Siegel upper half space $\HH_n$ and acts transitively on it \cite{Freitag1991}.

These fun facts are not very relevant for now, but we will see later how they come into play when we establish a geometric picture to classify the spaces of GKP codes. What the above relation tells us is that when $D$ (and implicitly also $n$) is fixed, the degrees of freedom in choosing a GKP code are fully characterized by the choice of the matrix $\Omega \in \HH_n$ as every symplectic lattice can be described by such an $\Omega \in \HH_n$. This summarizes the key takeaway of this section: After fixing the type $D$, every GKP code is equivalent up to symplectic transformations.

\section{Distance of a GKP code}\label{sec:distance}

In reality, we need a benchmark to assess how well the logical content of a GKP state is protected from noise. 
In qubit codes, a simple standard assumption is that noise is stochastic i.i.d. for each physical qubit, such that the likelihood of an error decreases exponentially with its support or weight. This assumption makes the definition of a code distance as weight of the shortest non-trivial logical operator meaningful. But this model no longer makes sense in the bosonic setting. Here, a more reasonable assumption about the underlying error model is that weak coupling to the environment results in effectively small displacements. 
Concrete examples for realistic models of relevant noise are loss, thermal noise \cite{Noh_Capacity}, and to a limited  extent, finite squeezing errors \cite{Glancy_2006, Tzitrin_2020, Terhal_2020}. 
For simplicity, we assume a stochastic displacement noise model where small displacements are more likely than large displacements. One such noise channel is provided by the Gaussian displacement noise model
\begin{equation}
\mathcal{N}\lr{\rho}= \int d\bs{x}\; P_{\tilde{\sigma}}\lr{\bs{x}} D\lr{\bs{x}}\rho D^{\dagger}\lr{\bs{x}},\; \label{eq:noise_channel}
\end{equation}
 which will be analyzed more in-depth later, where the probability to displace the state by an amount $\bs{x}$ is determined by a centered Gaussian distribution with variance $\tilde{\sigma}^2$,
\begin{equation}
P_{\tilde{\sigma}}\lr{\bs{x}}=G\lrq{0, \tilde{\sigma}^2 I_{2n}}\lr{\bs{x}}\propto e^{-\frac{\|\bs{x}\|^2}{2\tilde{\sigma}^2}}.
\end{equation}
 Although this error model is widely used in the analysis of GKP codes \cite{toricGKP, surfGKP, Haenggli_2020} for its simplicity, one needs to be careful to note that in real implementations of GKP error correction, this is not a physically accurate model \cite{Terhal_2020, Conrad_2021} in general, but only reproduces the correct measurement statistics in specific cases, such as when the finite squeezing error is applied to a perfect GKP state and interpreted as channel or when GKP states undergo a $0-$photon loss event. The concrete details of the noise models are not very relevant for now and an in-depth discussion can be found in refs.~\cite{Terhal_2020, GKP}. The distance measure that we introduce merely assumes that larger phase-space displacements are more likely than smaller ones, which is a natural choice and is sufficiently meaningful to indicate the robustness of the code with respect to realistic noise sources.
  
\begin{mydef}[Euclidean distance of a GKP code]
The (euclidean) distance $\Delta$ of a GKP code given by lattice $\CL$  is the euclidean length of the shortest non-trivial logical operator, i.e.,
\begin{equation}
\Delta=\Delta \lr{\CL} := \min_{0\neq \bs{x}\in \CL^{\perp}- \CL} \|\bs{x}\|. \label{eq:def_distance}
\end{equation}
\end{mydef}

For GKP codes with the special decomposition
$\CL=\CL_q \oplus \CL_p$, which one may call \textit{CSS} type in reference to the analogous situation for qubit-based quantum error correcting codes, we can further distinguish 
\begin{equation}
\Delta_q := \min_{0\neq \bs{x}\in \CL_p^*-\CL_q} \|\bs{x}\|,
\end{equation}
\begin{equation}
\Delta_p := \min_{0\neq \bs{x}\in \CL_q^*-\CL_p} \|\bs{x}\|.
\end{equation}

%

\begin{lem}[Distance bound]\label{thm:distance_bound}
Let $\lambda_1\lr{\CL^{\perp}}$ denote the shortest non-zero vector in the dual lattice. We have 
\begin{equation}
\Delta \geq \lambda_1\lr{\CL^{\perp}}.
\end{equation}
\end{lem}
\begin{proof}
Because the lattice vector  $0\neq \bs{x}\in \CL^{\perp}/\CL$ for which $\|\bs{x}\|$ is minimal is also in $\CL^{\perp}$ this holds trivially.
\end{proof}%
We have already seen in the previous section that for GKP codes obtained from scaling a symplectically self-dual lattice $\CL_0$ to $\CL=\sqrt{d}\CL_0$ we have 
\begin{equation}
\Delta_\mathrm{scaled}=\lambda_1\lr{\CL^{\perp}}=d^{-\frac{1}{2}}\lambda_1\lr{\CL_0},
\end{equation}
so the distance decreases while the size of the stabilizers and number of encoded logical dimensions increases. In order to find scaled GKP codes with both growing encoded dimension \textit{distance} it hence turns out to be sensible to keep the scaling parameter $d$ fixed and consider families of lattices of growing shortest lattice vectors $\lambda_1\lr{\CL_0}$. We will see in section~\ref{sec:goodrandom} how such lattice families can arise via random constructions.

For the concatenation of a single mode GKP code into a $[\![n,k,d]\!]$ qubit quantum error correcting code, we have, due to the conversion between Hamming $\|\cdot\|_0$ norm and Euclidean norm, the following bound.

\begin{lem}
The euclidean distance of a concatenated GKP code satisfies
\begin{equation}
\Delta_{\rm conc}\geq\sqrt{d}\Delta_{\rm loc}, \label{eq:conc_dist}
\end{equation}
where $\Delta_{\rm loc}$ is the distance of the local single mode GKP code.
\end{lem}
\proof
Eq.~\eqref{eq:conc_dist} is verified by decomposing a shortest representative non-trivial logical vector $\bs{L}=\oplus_{i=1}^n \bs{l}_i $ into $n$ local sub-blocks, where $\bs{l}_i$ is a logical operator of the respective local code, such that we have $\|\bs{L}\|^2\geq d \Delta_{\rm loc}^2$.
\endproof

The local distance is of typical magnitude $\Delta_{\square}=2^{-\frac{1}{2}}$ for the square GKP code or $\Delta_{\mhexagon}=3^{-\frac{1}{4}}$ for the hexagonal GKP code. Equality in eq.~\eqref{eq:conc_dist} holds when the code we consider is CSS or the local code is the hexagonal GKP code (such that all shortest non-trivial logical vectors have the same distance). 

 The distance of a concatenated code is typically strictly larger than $\lambda_1\lr{\CL^{\perp}}$, because the shortest dual vectors would typically correspond to stabilizer displacements.
  

The length of the shortest vector in a lattice $\lambda_1\lr{\CL}$ is the first of $2n$ successive minima of the lattice. Generally, the $i$th successive minimum $\lambda_i\lr{\CL}, \, i\leq 2n$ is defined to be the smallest $r>0$, such that $\CL$ contains $i$ linearly independent vectors of length at most $r$. It holds that $\mu\lr{\CL}\geq {\lambda_{2n}\lr{\CL}}/{2}$, where $\mu\lr{\CL}$ is the \textit{covering radius} of the lattice $\CL$, i.e., the minimum radius $\mu$, such that the union of closed balls $\mathcal{B}_{\mu}\lr{\bs{x}}, \, \bs{x} \in \CL$ centred around each lattice point of $\CL$ cover the entire space $\mathbb{R}^{2n}$. A related quantity is the \textit{packing radius} of the lattice $\rho\lr{\CL}={\lambda_1\lr{\CL}}/{2}$.

\subsection{Tradeoffs and bounds for GKP codes}

Successive minima of the direct and dual lattice $\CL^*$ are related by so-called transference theorems, in particular we have \cite{Banaszczyk1993}
\begin{align}
1\leq \lambda_1\lr{\CL}\lambda_{2n}\lr{\CL^*}\leq 2n. \label{eq:transference}
\end{align}
The symplectic- and euclidean dual lattices $\CL^{\perp}$ and $\CL^*$ differ only by an orthogonal transformation, therefore it holds that 
\begin{equation}
\lambda_i\lr{\CL^{\perp}}=\lambda_i\lr{\CL^*}
\end{equation}
and we can apply eq.~\eqref{eq:transference} to relate the distance to the length of stabilizer vectors. 

\begin{them}[Distance bound]
For a GKP code with lattice $\CL$, distance $\Delta$ and maximal length $C$ of basis vectors for a fixed basis $M$ we have
\begin{equation}
\Delta\geq \lambda_1\lr{\CL^{\perp}}\geq \lambda_{2n}^{-1}\lr{\CL}\geq C^{-1},\label{eq:dist_C_bound_1}
\end{equation}
as well as
\begin{equation}
\Delta \leq \lambda_{2n}\lr{\CL^{\perp}}\leq \frac{2n}{\lambda_1\lr{\CL}}. \label{eq:dist_C_bound_2}
\end{equation}
\end{them}
\begin{proof}
The first bound follows immediately from lemma~\ref{thm:distance_bound}  and eq.~\eqref{eq:transference} by swapping the roles of $\CL^{\perp}$ and $\CL$, which is possible because $\lr{\CL^{\perp}}^{\perp}=\CL$ and further from $C\geq \lambda_{2n}\lr{\CL}$. Similar for the second bound. 
\end{proof}
These bounds indicate an intimate relation between the lengths of the stabilizers in $\CL$ and the distance of the GKP code, which we will expand on further below. 

\subsection{Symplectic transformations}

It is interesting to study how symplectic transformations change the code distance. In particular, since we have seen in theorem~\ref{thm:symplectic_equiv} that GKP codes with equal standard form are symplectically equivalent, it is possible that for specific noise models a symplectic transformation of the stabilizers $M\mapsto MS^T$ can be used to improve the code's resilience to noise. Clearly, orthogonal (symplectic) transformations $S$  satisfying $S^TS=I$ leave the distance invariant.

This is not the case for squeezing. In particular if we consider a uniform squeeze
\begin{equation}
S=\eta I_n \oplus  \eta^{-1} I_n, \quad \eta \in (0, \infty)
\end{equation}
applied to a CSS code, we have
\begin{equation}
\Delta_q \mapsto \Delta_{q,\eta}= \eta^{-1}\Delta_q,\hspace{1cm} \Delta_p \mapsto  \Delta_{p,\eta} =\eta\Delta_p.
\end{equation}
Squeezing a code allows to account for potential bias in the noise, e.g. when the stochastic displacement noise channel is governed by
\begin{equation}
P\lr{\bs{x}}=G[0, \tilde{\sigma}^2 \lr{ \eta^{-1} I_n \oplus \eta I_n } ] \lr{\bs{x}},
\end{equation}
it allows to change the effective bias experienced by the code.
For concatenated CSS GKP codes, we generally have 
\begin{equation}
\Delta_q=\sqrt{d_Z}\Delta_{q, \mathrm{loc}}, \hspace{1cm} \Delta_p=\sqrt{d_X}\Delta_{p, \mathrm{loc}},
\end{equation}
such that squeezing the local codes by $\eta$ is equivalent to increasing (decreasing) the upper level $X/Z-$ distances $d_{X/Z}$ by a factor of $\eta^2$. On the other hand, for natively unbiased noise it is  also possible to squeeze the local code and employ a qubit quantum error correcting codes tailored towards biased noise such that $\Delta_q = \Delta_p$ remains constant. Although such a setup leaves the distance invariant, it can still lead to improvements in the error correction procedure when dedicated decoders for the qubit error correcting codes are used, as was recently demonstrated in ref.~\cite{Haenggli_2020}.

Corollary~\ref{thm:symplectic_equiv_cor} can also be used to derive an upper bound on the distance of a given code from symplectic equivalence as follows. Let us first consider a code $\mathcal{L}$ with generator $M$, encoding a single qubit within $n$ modes. Suppose $M$ is in canonical form with \begin{equation}
A = MJM^T = J_2 \otimes D\; \mathrm{with }\quad D = \mathrm{diag}\lrc{2,1,\hdots,1 }.
\end{equation} 
The corollary implies that there exists a symplectic matrix $S$ with $M = N_\square S^T$ and $N_\square $ is the generator of the code with logical dimension two. Since $N_\square $ is diagonal, the corresponding lattice $\mathcal{L}_\square$ is trivially orthogonal, and so is the dual $\mathcal{L}_\square^\perp$. The shortest non-trivial logical operators are thus immediately found as $\bs{\eta}_{\square, 1}^T = \lr{1/\sqrt{2},\bs{0}_{2n-1}^T}$ and $\bs{\eta}_{\square,2}^T = \lr{0, -1/\sqrt{2}, \bs{0}_{n-2}^T,}$. We can now recall that commutation relations for the displacements are related to symplectic products of the corresponding phase space vectors and that $S$  preserves symplectic products. Hence, the transformation $S$ maps all points in $\CL^\perp_\square$ that correspond to stabilizers, $\bs{x}\in\CL_{\square} \subset \CL^\perp_\square $, to direct lattice points $S\bs{x}\in\CL$, which are stabilizers of $\CL$, and all non-trivial logical operators to non-trivial logical operators. Therefore, we can readily write down two logical operators, $\bs{\eta}_1, \bs{\eta}_2 \in \CL^{\perp}\setminus \CL$, as $\bs{\eta}_j = S\bs{\eta}_{\square, j} $ whose lengths are \begin{equation}
\norm{\bs{\eta}_j } = \norm{S\bs{\eta}_{\square,j}} \geq \Delta.
\end{equation}  
One can further relate the distance to the \emph{squeezing} contained in $S$. Applying the Bloch-Messiah ~\cite{braunsteinSqz} we can write $S = O_2KO_1$ with $O_j$ symplectic orthogonal matrices and $K = \mathrm{diag}\left\{ k_1,\ldots,k_n,1/k_1,\ldots,1/k_n \right\}$ with $k_l \in (0,\infty)$ denoting an effective squeezing operation. Similar to ref.~\cite{oscillator_to_oscillator}, we denote by $\mathtt{sq}\lr{S}=\sqrt{\lambda_{max}\lr{S^TS}}$ the root of the largest eigenvalue of $S^TS$, or equivalently, the largest squeezing factor in the  Bloch-Messiah decomposition of the associated symplectic matrix. Taking everything together, we obtain
\begin{equation}
\Delta\leq \|S\bs{\eta}_{\square,j}\| \leq \|K\|  \|\bs{\eta}_{\square,j}\| = \frac{1}{\sqrt{2}}\mathtt{sq}\lr{S}.
\end{equation} 
The bound presented above generalizes straightforwardly to codes with higher logical dimension, and we obtain the following summarized statement. 

\begin{them}[Squeezing bound to the distance]\label{thm:distance_squeezing}
Let $\CL=\CL\lr{M}$ specify a GKP code with symplectic Gram matrix $A=J_2 \otimes D$ in its canonical form. Further, let $S$ denote the symplectic matrix that transforms between $M$ and the generator $N_{\square}$ as specified in corollary~\ref{thm:symplectic_equiv_cor}. We have
\begin{equation}
\Delta \leq \sqrt{\max_j d_j}^{-1}\mathtt{sq}\lr{S}.
\end{equation}
\end{them}

This bounds a benchmark for error correcting capabilities of a code by the squeezing measure of the Gaussian unitary necessary to prepare the code state from a collection of square GKP codes.

\subsection{Good random GKP codes}\label{sec:goodrandom}

As already remarked earlier, an especially interesting family of GKP codes  is called \textit{good}.

\begin{mydef}[Good GKP codes]
A GKP code family $\CL_{n}\subset \mathbb{R}^{2n}$ parametrized by lattice dimension $2n$ with asymptotically non-vanishing rate 
\begin{equation}\log\det\lr{\CL_n}=\Omega\lr{n}
\end{equation} and distance scaling 
\begin{equation}\Delta^2=\Omega\lr{n}
\end{equation} is \textit{good}.
\end{mydef}

In particular, we obtain \textit{good scaled GKP codes} if a family of symplectic self-dual lattices $\CL_0=\CL_0^{\perp} \subset \R^{2n}$ can be found such that $\lambda_1\lr{\CL_0}\propto \sqrt{n}$, i.e. the shortest lattice vector grows as the square root with the dimension of the lattice.

The proof of existence of \textit{good} GKP codes provided by 
ref.~\cite{HarringtonPreskill} can essentially be formulated using a Haar average over the (moduli) space of all symplectic lattices \cite{Sarnak1994}. The analogous heuristic to lower bound the shortest vector in a general lattice is given by the Gaussian heuristic. 

\begin{GaussianHeuristic}\label{GaussianHeuristic}
Let $L \subset \mathbb{R}^n$ be a sufficiently random full rank lattice with large $n$, then we expect the smallest non-zero vector in the lattice will satisfy
\begin{equation}
\lambda_1\lr{L}\approx \sqrt{\frac{n }{2\pi e}} \det\lr{L}^{\frac{1}{n}}.
\end{equation}
\end{GaussianHeuristic}

\textit{Argument} \cite{Silverman_lecture, RandomLat}:
The moduli space of full rank lattices in $\R^n$ with unit covolume is given by $\CL_n=\mathrm{SL}_{n}\lr{\Z}\setminus \mathrm{SL}_{n}\lr{\R}$, where the left\footnote{We write the left quotient because of the row-convention used in the definition of lattice bases. In the literature one more commonly uses a right-quotient associated to a column-convention.} quotient $\mathrm{SL}_{n}\lr{\Z}$ indicates the equivalence up to changes of basis. There is a Haar measure $\mu_n$ over $\CL_n$, normalized to $\mu_n\lr{\CL_n}=1$, such that for Lebesque-integrable functions $f:\, \R^n \mapsto \R$, we have that \cite{macbeath_rogers_1958}

\begin{align}
    \Braket{f\lr{L-\lrc{0}}}:=\int_{L \in \CL_n}\sum_{\bs{x} \in L- \lrc{0}} f\lr{\bs{x}}\, d\mu_n = \int_{\R^n} f\lr{\bs{x}}\, d\bs{x}.\label{eq:ModuliAV}
\end{align}

Let $f\lr{\bs{x}}=\theta\lr{\|\bs{x}\|\leq R}$, where $\theta$ is the Heaviside function. Equation \eqref{eq:ModuliAV} yields
\begin{equation}
    \Big\langle \#\lrc{\bs{x}\in L:\, \|\bs{x}\|\leq R,\; \bs{x}\neq 0 }  \Big\rangle_{L\in \CL_n}=V_n(R),
\end{equation}
where 
\begin{equation}
V_n(R)=\frac{\pi^{\frac{n}{2}}}{\Gamma\lr{\frac{n}{2}+1}}R^n
\end{equation} is the volume of the centered $n-$ball $B_n(R)\subset \mathbb{R}^n$. 

We hence have that if lattices $L$ are sampled from a random distribution close to $ \mu_n$ in the moduli space of all lattices with $\det\lr{L} =1$, the average number of non-zero lattice points of length at most $R$ is given by the volume of the $n$-ball, $V_n(R)$. Similarly, it is reasonable to expect that the average number of non-zero lattice points of length at most $R$, when the lattice has $\det\lr{L} \neq 1 $ and is sampled from an approximation to the Haar measure is given by $V_n\lr{R}/\det\lr{{L}}$. 

Using Stirling's approximation, the smallest $R$ for which this number becomes non-zero is given by $R\approx \sqrt{n/2\pi e}\det\lr{L}^{\frac{1}{n}}$.
\qed

The Gaussian Heuristic is a statement accepted to be generally true in lattice theory and post-quantum cryptography. In the above argument the ``heuristic'' enters in the assumption that the design property eq.\  \eqref{eq:ModuliAV} still holds for measures that only approximate the Haar measure on the space of lattices and that it moreover also still holds when the lattices are not of $\det\lr{\CL}=1$.

The Gaussian heuristic motivates that lattices with $\lambda_1=\Omega\lr{\sqrt{n}}$ can be found amongst sufficiently random sets of lattices. Buser and Sarnak \cite{Sarnak1994} showed that there is also a Haar measure over the moduli space of \emph{symplectic} lattices, using which Harrington and Preskill identified the existence of good GKP codes by a similar calculation as presented above \cite{HarringtonPreskill}. In fact, we can make an even stronger statement here by considering a subset of symplectic lattices which still retains the goodness property. 
We construct this class using symmetric matrices $X=X^T\in \Z^{n\times n}$ matrices to define the generator
\begin{equation}
    M\lrq{X}=\begin{pmatrix}
        I & X \\ 0 & qI_n
    \end{pmatrix}.
\end{equation}
Matrices of this form are is $q$-symplectic, that is, they are such that

\begin{equation}
M\lrq{X}JM^T\lrq{X}=qJ
\end{equation}
and we can rescale $M_0\lrq{X}=q^{-\frac{1}{2}}M\lrq{X}$ to obtain a symplectic matrix.

The GKP code produced by this lattice basis is square concatenated: The lattice $L$ generated by $M\lrq{X}$ contains a sublattice $q\Z^{2n}$ and the rescaled lattice $q^{-\frac{1}{2}}L$ contains a sublattice $\sqrt{q}\Z^{2n}$.
 The top block of $M\lrq{X}$ can be interpreted as the so-called reduced row-echelon form $\lr{I \; X}$ of a classical linear $q-$ary code in $\mathbb{F}_q^{2n}$. 

Following the technique used in ref.\ \cite{Sarnak1994}, we can show the subsequent statements.

\begin{them}\label{them:randsymH}
Let
\begin{equation}
U_q := \lrc{X=X^T \in \lrc{-\frac{q}{2} ,\hdots , \frac{q}{2} }^{n \times n} }
\end{equation}
be the set of symmetric matrices in $\Z_q$ and let $f:\, \R^{2n} \rightarrow \R$ be a function with compact support. We have

\begin{equation}
  \lim_{q\rightarrow\infty}  \Big\langle  \sum_{\bs{k}\in \Z^{2n}-\lrc{0}}  f\lr{M^T_{0}\lrq{X} \bs{k}} \Big\rangle_{X \in U_q} = \int_{\R^{2n}} f\lr{\bs{x}} d\bs{x},
\end{equation}
where the expectation value on the LHS is taken uniformly over $U_q$.

\end{them}

\proof
We start from the definition
\begin{align}    
\lim_{q\rightarrow \infty}\Big\langle F(X)\Big\rangle_{X \in U_q} &= \lim_{q\rightarrow \infty} q^{-1}\sum_{X_{1,1} = -q/2}^{q/2} q^{-1}\sum_{X_{1,2} =-q/2}^{q/2} \hdots F(X)  \\
&= \int_{-1/2}^{1/2} dX_{1,1} dX_{1,2}dX_{1,3}\hdots F(qX) 
\nonumber.
\end{align}
%
We have for $\bs{k}=\bs{m} \oplus \bs{n}$
\begin{align*}
M^T_{0}\lrq{X} \bs{k} = q^{-\frac{1}{2}}\begin{pmatrix}
    \bs{m} \\ qX\bs{m} +q\bs{n} 
\end{pmatrix},   
\end{align*}
such that we can compute analogously to
the argument presented in ref.\ \cite{Sarnak1994}
\begin{align}
I(q)&=\int_{-1/2}^{1/2} dX_{1,1} dX_{1,2}dX_{1,3}\hdots \sum_{\bs{m},\bs{n} \in \Z^{n}-\lrc{0}}  f\lr{q^{-\frac{1}{2}}\begin{pmatrix}
    \bs{m} \\ qX\bs{m} +q\bs{n} 
\end{pmatrix}} \label{eq:sum_dec0}\\
&=\int_{-1/2}^{1/2} dX_{1,1} dX_{1,2}dX_{1,3}\hdots \lrc{ \sum_{\substack{ \bs{m} \in Z^n,\\ m_1\neq 0 }}+\sum_{\substack{ \bs{m} \in Z^n,\\ m_1=0 \\ m_2 \neq 0 }} + 
\sum_{\substack{ \bs{m} \in Z^n,\\ m_1=0 \\ m_2 =0 \\ m_3 \neq 0 }} + \hdots
}\\
&\hspace{1cm}\times \sum_{\bs{n} \in \Z^n} f\lr{q^{-\frac{1}{2}}
\begin{pmatrix}
    \bs{m} \\ qX\bs{m} +q\bs{n} 
\end{pmatrix}} +\sum_{\bs{n} \in \Z^n} f\lr{
\begin{pmatrix}
    \bs{0} \\ \sqrt{q}\bs{n} 
\end{pmatrix}} \label{eq:sum_dec}
\end{align}
In eq.\  \eqref{eq:sum_dec}, we consider each summation over $\bs{m}$ separately. 
In the first term with the constraint $m_1 \neq 0$
we rewrite
\begin{align}
    qX\bs{m}+q\bs{n}=\begin{pmatrix}qm_1\lr{X_{1,1}+\frac{n_1}{m_1} + m_1^{-1}\sum_{k>1}X_{1,k} m_k  }
    \\ 
    qm_1\lr{X_{2,1}+\frac{n_2}{m_1} + m_1^{-1}\sum_{k>1}X_{2,k} m_k  }
    \\ 
    qm_1\lr{X_{3,1}+\frac{n_3}{m_1} + m_1^{-1}\sum_{k>1}X_{3,k} m_k  } \\
    \vdots
    \end{pmatrix}\label{eq:sum_vec}.
\end{align}
We write for each $n_i=\lfloor \frac{n_i}{m_1}\rfloor m_1 + (n_i \mod m_1)$ and split the summation
\begin{equation}
    \sum_{n_i\in \Z} g\lr{\frac{n_i}{m_1}} = \sum_{j_i\in \Z}\sum_{n_i\in \Z_{m_1}} g\lr{j_i+ \frac{n_i}{m_1}}.
\end{equation}
This way, each summation over the integer divisors of $n_i$ with $m_1$ can be combined with the integral over $X_{i,1} \in \lrq{-1/2,1/2}$ to an integral of $X_{i,1}  \in \R$ over the real numbers. To perform this trick, start with $X_{1,1} + j_1$ in the first row of eq.~\eqref{eq:sum_vec} and realize that all subsequent rows are independent of $X_{1,1}$. After converting the integration in the first row, all remaining summand of that row can be absorbed into a shift of the $X_{1,1}$ integral.
Now the first row is also independent of $X_{2,1}=X_{1,2}$, such that we can repeat this trick, converting 
the integral over $X_{2,1}$  and summation over $j_2$ into integration of $X_{2,1}$ over $\R$ which again gets rid of the dependency on $X_{k,2}, k>1$ in this row. 
Similarly, the summations over the terms $\frac{n_i}{m_1}$ also becomes trivial and provides a factor of $m_1$. In total, after substitution $t_i=qm_1X_{i,1}$
\begin{align}
&\int_{-1/2}^{1/2} dX_{1,1} dX_{1,2}dX_{1,3}\hdots \sum_{\substack{ \bs{m} \in Z^n,\\ m_1\neq 0 }}
\sum_{\bs{n} \in \Z^n} f\lr{q^{-\frac{1}{2}}
\begin{pmatrix}
    \bs{m} \\ qX\bs{m} +q\bs{n} 
\end{pmatrix}}\\
\nonumber
&=\sum_{\substack{ \bs{m} \in Z^n,\\ m_1\neq 0 }}
\int_{-\infty}^{\infty} d\bs{t}\,
q^{-n}f\lr{q^{-\frac{1}{2}}
\begin{pmatrix}
    \bs{m} \\ \bs{t}
\end{pmatrix}}\\
&=q^{-n/2}\sum_{\substack{ \bs{m} \in Z^n,\\ m_1\neq 0 }}
\int_{-\infty}^{\infty} d\bs{t}\,
f\lr{
\begin{pmatrix}
    q^{-\frac{1}{2}}\bs{m} \\ \bs{t}
\end{pmatrix}
}.
\nonumber
\end{align}
In the second term with constraint $m_1=0,\, m_2\neq 0$ we repeat the above procedure by pulling out a factor of $qm_2$, $qX\bs{m}+q\bs{n}=qm_2( qX\bs{m}/m_2 + \bs{n}/m_2)$. Begin with the integration over $X_{2,2}$, together with the sum over $n_2$ this again extends the domain of integration of $X_{2,2}$ to $\R$. Substituting the remaining summands in the corresponding row renders the rest of $qX\bs{m}+q\bs{n}$ independent of $X_{2,i}, i>2$ such that in each other row we can combine the $X_{2,i}$ integration with the sum over $n_i$ to extend the domains of integration. 
Repeat this procedure using each $m_i\neq 0$ in eq.~\eqref{eq:sum_dec} and finally use that $f$ has compact support, such that in the limit $q\rightarrow\infty$
eq.~\eqref{eq:sum_dec0} becomes

\begin{equation}
   \lim_{q\rightarrow \infty} I(q)= \lim_{q\rightarrow \infty} q^{-n/2}\sum_{\substack{ \bs{m} \in Z^n-\lrc{0} }}
\int_{-\infty}^{\infty} d\bs{t}\,
f\lr{
\begin{pmatrix}
    q^{-\frac{1}{2}}\bs{m} \\ \bs{t}
\end{pmatrix}
}.
\end{equation}
In the limit, we again use the definition of the Riemann integral to finally obtain
\begin{equation}
   \lim_{q\rightarrow \infty} I(q)=\int_{\R^{2n}} f\lr{\bs{x}} d\bs{x}.
\end{equation}

\endproof

The proof technique used  above was adapted from a similar proof in ref.~\cite{Sarnak1994}. It is quite remarkable that it is possible to explicitly compute the expectation value of functions over the space of random lattice using standard tricks from calculus. Comparing the derived statement to the proof of the Gaussian Heuristic~\ref{GaussianHeuristic}, this immediately implies

\begin{cor}
Under the same assumptions as in theorem~\ref{them:randsymH}, Lattices $L_{q, X}$ generated by uniformly random over 
\begin{equation}
U_q := \lrc{X=X^T \in \lrc{-\frac{q}{2} ,\hdots , \frac{q}{2} }^{n \times n} }
\end{equation}
have expected shortest vector length $\overline{\lambda_1\lr{L_q, X}}\approx \sqrt{n/2\pi e}$. In particular, for large $q$, there always exists $X\in U_q$ such that 
\begin{equation}
\lambda_1\lr{L_{q, X}}\geq \sqrt{n/2\pi e}.
\end{equation}

\end{cor}

\proof

Let $L_{q,X}=L\lr{M_0\lrq{X}}$ be a lattice generated by the basis $M_0\lrq{X}$ as before. Theorem~\ref{them:randsymH} asserts that for any compact function $f:\, \R^{2n} \rightarrow \R$, it holds that

\begin{equation}
\lim_{q\rightarrow \infty}\Big\langle f\lr{L-\lrc{0}} \Big\rangle_{X\in U_q} =\int_{\R^{2n}}d\bs{x}\, f\lr{\bs{x}}.
\end{equation}
Now, same as in the proof of the Gaussian Heuristic~\ref{GaussianHeuristic}, set $f\lr{\bs{x}}=\theta\lr{\|\bs{x}\|\leq R}$ with the Heaviside theta function $\theta$. We obtain

\begin{equation}
  \lim_{q\rightarrow \infty}  \Big\langle \#\lrc{\bs{x}\in L_{q, X}:\, \|\bs{x}\|\leq R,\; \bs{x}\neq 0 }  \Big\rangle_{X\in U_q}=V_n(R),
\end{equation}
where again
\begin{equation}
V_n(R)=\frac{\pi^{\frac{n}{2}}}{\Gamma\lr{\frac{n}{2}+1}}R^n
\end{equation} is the volume of the centered $n-$ball $B_n(R)\subset \mathbb{R}^n$ which describes the average number of lattice points of $L_{q, X}$ in a Ball of radius $R$  for large $q$. Using Stirling's approximation, we see that this number is non-zero when we have at least $R\approx \sqrt{n/2\pi e}$. The average property implies the existence of instances $L_{q, X}$ with shortest vector length at least $ \sqrt{n/2\pi e}$.\endproof

Scaled GKP codes given by $\CL=\sqrt{d}L_{q, X}$ in the limit $q\rightarrow \infty$ are hence \textit{good}:  They encode logical dimension $\dim\lr{\CH_C}=d^n$ and have distance scaling $\Delta\propto \sqrt{n/d}$.  We will later find that even a further refinement of GKP codes of this form, derived from a lattice cryptosystem, will maintain this goodness property. This will be the topic of sec.~\ref{sec:NTRU}. For now, the core message of this section is, informally,

\begin{cor}
Good families of GKP codes exist.
\end{cor}

For more flexibility in choosing logical dimensionalities, one can even show the final statement of this section.
\begin{cor}\footnote{I sincerely thank Jerry Zheng for asking the question that lead to the derivation of this statement.}
For any type $D=\diag\lr{d_1,\hdots, d_n}$, there exists a GKP code with lattice $\CL \subseteq \CL^{\perp}$ and distance

\begin{equation}
\Delta\lr{\CL}\geq \frac{1}{{\rm lcm}\lr{D}} \sqrt{\frac{n}{\pi e}},
\end{equation}
where ${\rm lcm}\lr{D}$ is  the least common multiple of the factors in $D$, ${\rm lcm}\lr{D}:={\rm lcm}\lr{d_1,\hdots, d_n} \leq \det\lr{D}$.
\end{cor}
\proof
We use  the fact that every type $D$ GKP code can be understood as a $\lr{D\oplus I}$ sublattice $\CL$ of such a symplectic lattice. We have bases $M, \, M_0$ for the lattices $\CL, \, \CL_0$ respectively such that
\begin{equation}
M=\lr{D\oplus I_n}M_0.
\end{equation}
By associating the canonical dual to each of these bases demanding $M^{\perp}JM^T=I$, we find
\begin{equation}
M^{\perp}=\lr{D^{-1}\oplus I_n} M_0^{\perp}.
\end{equation}
 Note that in this equation, the lattice spanned by $M^{\perp}$ cannot be interpreted as a sublattice of the lattice spanned by $M_0^{\perp}$ unless $D^{-1}$ is integer (i.e. $D=I_n$). However, we can multiply both sides of the equation by the least common multiple of the factors in $D$, ${\rm lcm}\lr{D}:={\rm lcm}\lr{d_1,\hdots, d_n} \leq \det\lr{D}$ such that $B={\rm lcm}\lr{D}\lr{D^{-1}\oplus I_n}$ is an integer matrix. This shows that ${\rm lcm}\lr{D}\CL^{\perp} \subseteq \CL_0$ is a sublattice of the symplectic lattice as well, such that we obtain
 
 \begin{equation}
 \Delta\lr{\CL}\geq \lambda_1\lr{\CL^{\perp}}\geq \frac{1}{{\rm lcm}\lr{D}} \lambda_1\lr{\CL_0},
 \end{equation}
where we have used that if $L'\subseteq L$ is a sublattice, it also holds that $\lambda_1\lr{L'}\geq \lambda_1\lr{L}$. Now using the fact that symplectically self dual lattice $\CL_0=\CL_0^{\perp}\subset \R^{2n}$ with $\lambda_1\lr{\CL_0}\geq \sqrt{\frac{n}{\pi e}}$ exist yields the final result.
 
\endproof

We end this section on a conjecture.

\begin{conjecture}
There exists a Haar measure over all symplectically integral lattices $\CL \subseteq \CL^{\perp}\subset \R^{2n}$ with type $D$. 
\end{conjecture}
This conjecture is motivated by the fact that every type $D$ symplectically integral lattice $\CL\subseteq \CL^{\perp}$ is simply provided by the $D\oplus I_n$ sublattice of the $2n$-dimensional symplectic lattices, for which a Haar measure $\mu_{2n}\lr{\Sp_{2n}\lr{\Z}\backslash \Sp_{2n}\lr{\R}}$ is already known to exist. 

\subsection{The euclidean distance in lattice theta functions}

The relevant information of a lattice that captures its distance is contained in its so-called lattice theta function, which is the generator function for its distance distribution -- i.e. the numbers $N_d$ of lattice vectors $x\in \CL$ of square length $\bs{x}^T\bs{x}=d$.
The theta function is simply defined by summing over all vector in a lattice $\CL$,

\begin{equation}
\Theta_{\CL}(\tau)=\sum_{\bs{x}\in \CL} q^{\bs{x}^T\bs{x}}=\sum_{\delta \in \CD } N_{\delta} q^{\delta}, \label{eq:def_theta}
\end{equation}%
where $q=e^{i2\pi \tau}, \tau \in \hh$ and  $\CD=\lrc{\|\bs{x}\|^2_2, \, \bs{x}\in \CL}$ is the set of squared distances of $\CL$. We have also introduced the number of lattice vectors of a given length %
\begin{equation}
N_{\delta}=\#\lrc{\bs{x}\in \CL:\, \bs{x}^T\bs{x}=\delta}.\label{eq:def_N}
\end{equation}%
For integral lattices, i.e., when the corresponding (euclidean) Gram matrix satisfies $G=MM^T \in \mathbb{Z}^{2n\times 2n}$, we can set $\CD=\mathbb{N}$ and  $N_m$ is given by the number of integer solutions $\bs{n}\in \mathbb{Z}^{2n}$ to the equation $\bs{n}^T G \bs{n}=m$, which is an example of a \textit{Diophantine equation}. We call the pair $(\CD, N_{\delta})$ the \textit{distance distribution} of the lattice $\CL$.
The theta function converges and is holomorphic for $\Im(\tau)>0$. The first summands are given by %
\begin{equation}
\Theta_{\CL}(\tau)=1+\kappa q^{\lambda_1}+\hdots,
\end{equation}%
where $\lambda_1=\lambda_1\lr{\CL}$ is the length of the shortest vector of the lattice and $\kappa$ is known as the \textit{kissing number}, the number of minimal-length vectors of $\CL$.
It is known that the theta function for the (euclidean) dual of a lattice $\CL\subset \mathbb{R}^{2n}$ is given by%
\begin{equation}
\Theta_{\CL^*}(\tau)=\det\lr{\CL}\lr{\frac{i}{\tau}}^{n}\Theta_{\CL}\lr{-\frac{1}{\tau}}, \label{eq:theta_functional}
\end{equation}%
which follows from the Poisson summation formula.  Since this is a very useful formula, we briefly write down the Poisson formulas for sums over lattices relevant for GKP codes.

\begin{lem}[Dirac comb representation]\label{lem:dirac_comb}
Let $\CL \subseteq \mathbb{R}^{2n}$ be a lattice with generator $M$ with symplectic dual $\CL^{\perp}$, we have
\begin{equation}
\sum_{\bs{\xi} \in \CL} e^{i2\pi \bs{\xi}^TJ\bs{z}}= \frac{1}{\det\lr{\CL}}\sum_{\bs{\xi}^{\perp} \in \CL^{\perp}} \delta^{2n}\lr{\bs{z}-\bs{\xi}^{\perp}}.
\end{equation}
\end{lem}
\begin{proof}
We compute straightforwardly using one dimensional Poisson resummation 
\begin{align}
\sum_{\bs{\xi} \in \CL} e^{i2\pi \bs{\xi}^TJ\bs{z}}
&=\sum_{\bs{a} \in \mathbb{Z}^{2n}} e^{i2\pi \bs{a}^T MJ\bs{z}}  \\
&=\prod_{j}\sum_{a_j \in \mathbb{Z}} e^{i2\pi a_j \lr{MJ\bs{z}}_j} \\
&=\prod_{j}\sum_{b_j \in \mathbb{Z}} \delta\lr{\lr{MJ\bs{z}}_j-b_j} \\
&= \frac{1}{\det\lr{\CL}} \sum_{\bs{b} \in \mathbb{Z}^{2n}} \delta^{2n}\lr{\bs{z}-\lr{\bs{b}^TM^{\perp}}^T} \\
&= \frac{1}{\det\lr{\CL}}\sum_{\bs{\xi}^{\perp} \in \CL^{\perp}} \delta^{2n}\lr{\bs{z}-\bs{\xi}^{\perp}}.
\end{align}
\end{proof}

\begin{lem}[Poisson resummation]\label{lem:poisson}
Let $\CL \subseteq \mathbb{R}^{2n}$ be a lattice with symplectic dual $\CL^{\perp}$, we have for all Schwartz functions $f:\, \mathbb{R}^{2n} \rightarrow \mathbb{C}$,

\begin{equation}
\sum_{\bs{\xi}\in \CL} f\lr{\bs{\xi}}=\frac{1}{\det\lr{\CL}} \sum_{\bs{\xi}^{\perp} \in \CL^{\perp}} \hat{f}\lr{\bs{\xi}^{\perp}},
\end{equation}
where
\begin{equation}
\hat{f}\lr{\bs{x}}=\int_{\mathbb{R}^{2n}}d\bs{y}f\lr{\bs{y}}e^{i2\pi \bs{y}^TJ\bs{x}}
\end{equation}
is the symplectic Fourier transform of $f$.
\end{lem}

\begin{proof}
We proof this fact following the analogous proof using the usual Fourier transform provided in ref.~\cite{theta_course}. First define
\begin{equation}
F\lr{\bs{z}}=\sum_{\bs{\xi}\in \CL} f\lr{\bs{z}+\bs{\xi}}.
\end{equation}
This function is invariant under transformations $\bs{z}\mapsto \bs{z}+\CL$, such that it descends to a function on $\mathbb{R}^{2n}/\CL$ and has Fourier expansion
\begin{equation}
F\lr{\bs{z}}=\sum_{\bs{y} \in \CL^{\perp}} \hat{F}\lr{\bs{y}}e^{i2\pi \bs{y}^TJ\bs{z}},
\end{equation}
where
\begin{equation}
\hat{F}\lr{\bs{y}}=\frac{1}{\det\lr{\CL}} \int_{\mathbb{R}^{2n}/\CL}d\bs{x}\, F\lr{\bs{x}}e^{i2\pi \bs{x}^TJ\bs{y}}.
\end{equation}
Using lemma~\ref{lem:dirac_comb} the correctness of the Fourier expansion can be straightforwardly verified.
Let $\CP = \mathbb{R}^{2n}/\CL$ be a fundamental domain of $\CL$. We compute for $\bs{y} \in \CL^{\perp}$ (such that $e^{i2\pi\bs{x}^TJ\bs{y}}$ is well-defined on $\bs{x} \in \mathbb{R}^{2n}/\CL$)
\begin{align}
\det\lr{\CL}\hat{F}\lr{\bs{y}}
&=\sum_{\bs{\xi}\in \CL} \int_{\bs{x}\in \CP}d\bs{x}\, f\lr{\bs{x}+\bs{\xi}}e^{i2\pi\bs{x}^TJ\bs{y}} \\
&=\sum_{\bs{\xi}\in \CL} \int_{\bs{x}\in \CP-\bs{\xi}}d\bs{x}\, f\lr{\bs{x}}e^{i2\pi\bs{x}^TJ\bs{y}} \\
&=\hat{f}\lr{\bs{y}}, 
\end{align}
such that we have
\begin{equation}
F\lr{\bs{z}}=\frac{1}{\det\lr{\CL}} \sum_{\bs{\xi}^{\perp} \in \CL^{\perp}} \hat{f}\lr{\bs{y}} e^{i2\pi \bs{y}^TJ\bs{z}}.
\end{equation}
Taking $z=0$ completes the proof.
\end{proof}

Using these tools it is now easy game to verify eq.~\eqref{eq:theta_functional}, which we leave as exercise to the reader. Note that in the definition of the lattice theta function the parameter $\tau$ takes the place of a variance-like factor, which is known to become inverted when the corresponding Gaussian function is Fourier transformed.

Since we have seen that $\CL^*$ and $\CL^{\perp}$ only differ by an orthogonal transformation, which does not change the length of lattice vectors, and hence the theta function, we also have 
\begin{equation}
\Theta_{\CL^{\perp}}(\tau)=\Theta_{\CL^*}(\tau),
\end{equation}
such that we can immediately obtain the theta function of  $\CL^{\perp}$ whenever the theta function of the direct lattice $\Theta_{\CL}(\tau)$ is known. 
By definition of the theta function, the distance of a GKP code specified by $\CL$ is given by the smallest non-zero power of $q$ in%
\begin{equation}
Q_{\CL}(\tau):=\Theta_{\CL^{\perp}}(\tau)-\Theta_{\CL}(\tau)=N_{\Delta^2}q^{\Delta^2}+\hdots, \label{eq:distance_theta}
\end{equation}%
and furthermore since we have that $\Theta_{\CL}(\tau)$ is uniquely determined by the distance distribution $(\CD, N_{\delta})$ of $\CL$, which by eq.~\eqref{eq:theta_functional} also fully specifies $\Theta_{\CL^{\perp}}(\tau)$. Note that  $Q_{\CL}(\tau)>0$ since $\CL \subseteq \CL^{\perp}$. 
We arrive at the following insight.
\begin{them}[Code distance is specified by the distance distribution]\label{thm:distance_dist}
The distance of a GKP code specified by $\CL$ is uniquely determined by its distance distribution $(\CD, N_{\delta})$.
\end{them}%
The theta function of a scaled lattice is %
\begin{equation}
\Theta_{\sqrt{d}\CL_0}(\tau)=\Theta_{\CL_0}(d \tau). 
\end{equation}%
Many expressions of theta functions of symplectically self dual lattices, in particular those that also correspond to euclidean self-dual lattices are known in the literature \cite{ConwaySloane} such that their corresponding distances can e.g. be estimated from a logarithmic fit for small $q\ll 1$ of eq.~\eqref{eq:distance_theta}, or by expressing $Q$ in a basis for which the distance distribution is known. 

For concatenated (square) GKP codes, where $\CL=\Lambda\lr{Q}$ is given by a Construction A  lattice, we can express the theta function using the \textit{weight enumerator}. We first introduce the \textit{weight distribution} of a linear code $Q$, which is given by the numbers $\lrc{A_i}_{i=0}^{2n}$ of codewords $q$ in $Q \subset \mathbb{Z}_2^{2n}$ with Hamming weight $\wt(q) = i$. Crucially, the weight distribution $\lrc{A_i}_{i=0}^{2n}$ here refers to the Hamming-weight distribution of the \textit{symplectic representation} of the qubit stabilizers.
The weight enumerator is given by
\begin{equation}
W_Q\lr{x, y}=\sum_{q\in Q} x^{2n-\wt(q)}y^{\wt(q)} = \sum_{i=0}^{2n}A_i x^{2n-i}y^i. 
\end{equation}
Using this definition, we can express the theta function of a construction A lattice \cite{ConwaySloane} by straightforward computation%
\begin{equation}
\Theta_{\Lambda(Q)}\lr{\tau}=W_Q\lr{\theta_3(2\tau), \theta_2(2\tau)},
\end{equation}%
where%
\begin{align}
\theta_3(\tau)&=\Theta_{\mathbb{Z}}(\tau)=\sum_{m\in \mathbb{Z}}q^{m^2}, \\
\theta_2(\tau)&=\Theta_{\mathbb{Z}+\frac{1}{2}}(\tau)=\sum_{m\in \mathbb{Z}}q^{\lr{m+\frac{1}{2}}^2}. 
\end{align}%
Since we see that we can find the theta function of a Construction A lattice corresponding to a concatenated code by means of the (Hamming) weight distribution of its (qubit) stabilizer group, we obtain similar to theorem~\ref{thm:distance_dist} and by eq.~\eqref{eq:conc_dist} the final corollary.

\begin{cor}[Distance of CSS qubit stabilizer code from weight distribution]\label{cor:qubit_distance_weights}
The distance $d$ of a CSS qubit stabilizer code is fully determined by the weight distribution $\lrc{A_i}_{i=0}^{2n}$ of its stabilizers.
\end{cor}

We note that this corollary also follows from a symplectic version of the weight enumerators defined by Shor and Laflamme \cite{ShorLaflamme} and Rains \cite{Rains}, which through their immediate relationship to the quantum error correction conditions \cite{Knill_2000} impose strong restrictions on possible quantum error correcting codes.

\section{GKP codes: A Rosetta stone}\label{sec:rosetta}

In this section we explore the structure of logical Clifford gates for the GKP code, which leads to an algebraic geometric perspective on GKP codes and will allow us to quantify the \textit{moduli space of GKP codes} to relate the fault-tolerance of logical gates for the GKP codes to a concept introduced by Gottesman and Zhang in ref.~\cite{gottesman2017fiber} termed \textit{fiber bundle fault tolerance}. On a high level, GKP Clifford gates are given by lattice automorphisms of the GKP codes. The title of this section stems from the motivation that we connect something discrete and hard-to-tame (fault tolerance) to something geometric and beautiful (compact Riemann surfaces) via an understanding of the related symmetries (lattice automorphisms). This explains the \textit{Rosetta stone} reference in the title of this section. \footnote{This is an ambitious attribution to the famous \textit{Rosetta stone for Mathematics} due to H.~Weil~\cite{Weil2014}, who proposed a bridge between number theory and the geometry of Riemann surfaces.} 

\subsection{Logical Clifford gates for the GKP code}\label{sec:Cliff}
The set of GKP \emph{Clifford gates} form a special and important set of gates that act on the logical space of a GKP code. 
GKP Clifford gates for a GKP code of type $D$ are given by the symplectic automorphism group 
\begin{equation}
{\rm Cliff}\lr{D}\equiv \Aut_{\infty}^{S}\lr{\CL^{\perp}}=\Aut^S\lr{\CL^{\perp}} \ltimes \CL^{\perp},
\end{equation}
such that every logical Clifford gate can be described by the combination of a displacement by a vector in  $\CL^{\perp}$ -- a so-called \textit{trivial} Clifford gate, since it only conjugates Pauli operators (displacements in $\CL^{\perp}$) to an additional phase factor -- and a symplectic automorphism of the lattice $\CL^{\perp}$ that sends logical Pauli operators to logical Pauli operators but preserves the $0$ element. 
Symplectic automorphisms transform vectors constituting the lattice basis $M$ in a way that only implements a symplectic change of basis while leaving the lattice as a geometric object invariant, 
\begin{equation}
\label{eq:AutS}
    \text{Aut}^{S}(\mathcal{L}) \coloneqq \{ g \in \Sp_{2n}\lr{\mathbb{R}} |\, \exists U \in \SL_{2n}( \mathbb{Z}):  UM  = M g^{T}\},
\end{equation}
and we refer to the basis transformation $U\in SL_{2n}( \mathbb{Z})$ as the \textit{integral representation} of the corresponding element.
An important subgroup is the group of symplectic orthogonal automorphisms $\Aut^{SO}\lr{\mathcal{L}}=\Aut^{S}\lr{\mathcal{L}}\cap O_{2n}\lr{\R}$ which is significant due to its interpretation as GKP Clifford operations realizable via passive linear optics without squeezing.

For GKP codes specified by a lattice $\CL \subseteq\CL^{\perp}$, each element of the symplectic automorphism group is uniquely specified by its integral representation given by the group
\begin{equation}
\label{eq:SpD}
    \Sp_{2n}^D\lr{\Z} = \{U\in \GL_{2n}\lr{\Z}:\; UA_DU^T=A_D\}.
\end{equation}
These transformations preserve the symplectic form in its canonical basis $A_D=J_2\otimes D$ \cite{Birkenhake_2004} and yield the so-called integral representation of the symplectic automorphisms.

\begin{lem}\label{lem:integral_rep}
    
    Given a weakly symplectically self dual lattice $\CL \subseteq\CL^{\perp}$ with symplectic Gram matrix (symplectic form) $A=J_2\otimes D$, $D=\mathrm{diag}\lr{d_1,\hdots, d_n},$ we have that
    $\Aut^S\lr{\mathcal{L}}$ is equivalently specified by the integral representation
    \begin{equation}
        \Sp_{2n}^D\lr{\Z}=\{U\in \GL_{2n}\lr{\Z}:\; UAU^T=A\}. 
    \end{equation}

\end{lem}

\proof
Since $M_D=D\oplus I$ is invertible, it holds that the unique $S_U$  for which $UM_D=M_DS_U^T $ is symplectic. Since any basis for a GKP code of type $D$ can be given by $M_D S_0^T$ for some $S_0 \in \Sp_{2n}\lr{\R}$, it follows that the elements of $\Sp_{2n}^D\lr{\Z}$ are integral representations for symplectic automorphisms of $\CL$. Conversely, from $UM=MS^T$ it can also be shown that
$A=MS^TJSM^T=UMJM^TU^T=UAU^T$ every integral representation for a symplectic automorphism needs to admit the defining relation of eq.~\eqref{eq:SpD}.
\endproof

In fact, we also have

\begin{cor}
    \begin{equation}
        \Aut^S\lr{\CL}=\Aut^S\lr{\CL^{\perp}}
    \end{equation}
\end{cor}

\proof
In the canonical basis, we have $M=AM^{\perp}$ where, by lemma~\ref{lem:integral_rep}, a symplectic automorphism $S\in \Aut^S\lr{\CL}$ is specified by a unimodular matrix $U,\, UAU^T=A$. Combining these statements one finds
$M^{\perp}S^T=A^{-1}UAM^{\perp}=U^{-T}M^{\perp}$. Since inverses and transposes preserve the unimodularity of $V=U^{-T}$, we have $S\in \Aut^S\lr{\CL^{\perp}}$ and thus $\Aut^S\lr{\CL}\subseteq \Aut^S\lr{\CL^{\perp}}$. Conversely, we have that  $A^{-1}=-M^{\perp}J\lr{M^{\perp}}^T$, such that for a given $S\in \Aut^S\lr{\CL^{\perp}}$ the relation $VM^{\perp}=M^{\perp}S^T$ for unimodular $V$ implies that the integral representation satisfies $VA^{-1}V^T=A^{-1}$. 
With $M^{\perp}=A^{-1}M$ this yields $MS^T=AVA^{-1}M=V^{-T}M$, such that unimodularity of $V$ implies $S\in \Aut^S\lr{\CL}$ and thus $ \Aut^S\lr{\CL^\perp}\subseteq  \Aut^S\lr{\CL}$.
\endproof

A similar statement can be shown to hold for the orthogonal automorphism group, whose integral representation is 
given by matrices $U\in\GL_{2n}\lr{\Z}$ that preserve the euclidean Gram matrix $G=MM^T$.
In the special case of \textit{scaled GKP codes} (see ref.~\cite{Conrad_2022}), i.e. GKP codes with $D=q I$ we have $\Aut^S\lr{\CL^{\perp}}=\Aut^S\lr{\CL}$, such that every automorphism in $\Aut^S\lr{\CL^{\perp}}$ also descends to an automorphism in $\CL^{\perp}/\CL$ by reducing the outcome modulo $\CL$.

The integral representation of the automorphisms in eq.~\eqref{eq:SpD} can be understood to represent the \textit{logical} action of the (non-trivial) Clifford group on the Heisenberg-Weyl operators in eq.~\eqref{eq:HW_operators}. 
A logical Heisenberg-Weyl operator $O\lr{\bs{l}}=\prod_{i=1}^n X_i^{l_i} Z_i^{l_{i+n}}$ with $X_i, Z_i$ from eq.~\eqref{eq:HW_operators} is specified \footnote{That is, up to phases, as we usually care about the action of these operators on a projective Hilbert space.} by a vector $\bs{l} \in \Z^{2n}_D$, where $\Z^{2n}_D$ denotes $\Z^{2n}$ with the element-wise reduction modulo $I_2\otimes D$, such that vectors are considered equivalent if they differ only by stabilizers. 
The action of a non-trivial Clifford operation is then given by $\bs{l} \mapsto g_I \bs{l}\, \!\mod D\oplus D$ and $g_I$ denotes the integral representation of the corresponding $g\in \Aut^S\lr{\CL}$. 
The reduction modulo $I_2\otimes D$ therefore implements the equivalence relation given by the stabilizers as translations by elements in $\CL$ and can be understood as the map $\Aut^S\lr{\CL^{\perp}} \rightarrow \Aut^S\lr{\CL^{\perp}/\CL}$. 
 
Henceforth, we focus on scaled GKP codes, where $D=dI_n$, such that $\CL=d\CL^{\perp}$ are proportional. 
From the definition of symplectic automorphisms in eq.~\eqref{eq:AutS} observe that $\Aut^S\lr{\CL}$ is one-to-one with its integral representation given by $\Sp_{2n}\lr{\Z}$. 
The relationships between the integral and real representations of symplectic automorphisms for scaled GKP codes and their logical actions are illustrated in fig.~\ref{fig:Cliffordn1}. 
It can also be seen that the action of the symplectic automorphism group on the quotient $\CL^{\perp}/ \CL$ is equivalent to that of $\Sp_{2n}\lr{\Z_d}$, the usual symplectic representation of the non-trivial Clifford group on qudits.

\begin{figure}
\center
\includegraphics[width=.4\textwidth]{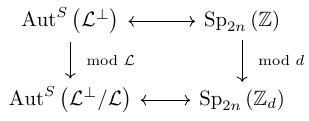}
\caption{Commutative diagram for the structure of nontrivial Cliffords for scaled GKP codes. Note that for $n=1$ we have the equality $\SL_{2}\lr{\mathbb{Z}}=\Sp_{2}\lr{\mathbb{Z}}$.}\label{fig:Cliffordn1}
\end{figure}

A special class of symplectic lattices is given by symplectic lattices that are also euclidean self-dual. Such lattices can be shown to fulfil the property

\begin{lem}[\cite{Sarnak1994}]
A euclidean self-dual lattice $L=L^*\subset \R^{2n}$ is symplectic self-dual if and only if $J$ is a lattice automorphism.
\end{lem}
\proof
A basis $M: \, G=MM^T\in \Z^{2n \times 2n}$ of a euclidean lattice has a canonical dual lattice $L^*$ generated by the matrix $M^{-T}$ similarly to a symplectic self dual lattice, for which any basis has the canonical symplectic dual given by $M^{\perp}=M^{-T}J^T$. From these relations we see that a euclidean self-dual lattice with basis $M=UM^{-T}$ for an unimodular matrix $U$ also satisfies
\begin{equation}
MJ^T=UM^{-T}J^T=UM^{\perp}.
\end{equation}
Hence, if $J$ is an automorphism of the lattice spanned by $M$, it holds that $\CL \sim \CL^{\perp} $ and the lattice is symplectic self-dual.
Conversely, a symplectic self-dual lattice satisfies $M=VM^{-T}J^T$ for some unimodular $U$. If it is also euclidean self-dual, there is some unimodular $W$, such that $WM=M^{-T}$, such that
\begin{equation}
M=VW MJ^T. 
\end{equation}
Since $VW$ is also unimodular this shows that the lattice $\CL=J\CL$ spanned by matrix $M$ has an automorphism given by $J$.

\endproof

This statement can be used to prove symplectic self-duality for the $E_8$ lattice, which was the strategy presented in ref.~\cite{Sarnak1994}.

The symplectic automorphism groups discussed here can be generated by a set of \textit{symplectic transvections}, given by matrices $t_{\bs{\alpha}}, \;\bs{\alpha} \in \CL^{\perp}$, 
\begin{equation}
\label{eq:transvection}
    t_{\bs{\alpha}}=I+\bs{\alpha}\bs{\alpha}^TJ\,,
\end{equation}
which are implemented via the Gaussian unitaries 
\begin{equation}
    U_{\bs{\alpha}}=e^{-\frac{i}{2}\lr{\bs{\alpha}^TJ\bs{\hat{x}}}^2}
\end{equation}
with squeezing value bounded by $\rm sq\lr{t_{\bs{\alpha}}}:=\|t_{\bs{\alpha}}\|_2\leq 1+\|\bs{\alpha}\| $. 

The symplectic transvection adds multiples of $\bs{\alpha}$ to an input vector $\bs{x}$ according to the symplectic inner product $\bs{\alpha}^TJ\bs{x}$, from which it is easy to see that symplectic lattices $\CL_0$ are preserved under transvections by vectors in $\CL_0$.
For elements of a scaled GKP code $\CL^{\perp}=\sqrt{d}^{-1}\CL_0$, a symplectic transvection by one of the canonical basis vectors acts non-trivially on its partner,
\begin{equation}
    t_{\sqrt{d}\bs{e}_i} \bs{f}_i =\bs{f}_i + \bs{e}_i\,,
\end{equation}

and trivially on every other canonical basis vector. 
In particular using $t_{\bs{\alpha}}t_{\bs{\beta}}t^{-1}_{\bs{\alpha}}=t_{t_{\bs{\alpha}} \lr{\bs{\beta}}}$ one observes that for symplectic canonical form basis vectors of the lattice $\bs{\alpha}, \bs{\beta}$, the commutation of the corresponding transvections is determined by whether or not the vectors have a non-trivial symplectic inner product. 

In fact, symplectic transvections are known as representations of \textit{Dehn twists} on compact genus $n$ surfaces $S_n$ \cite{omeara_symplectic_1978}, while the group $\Sp_{2n}\lr{\Z}$ of integral representations of symplectic automorphisms for a symplectic lattice $\CL_0$ forms a representation of their mapping class group ${\rm Mod}\lr{S_n}$ \cite{FarbMargalit+2012} which we now understand to be generated by Dehn-twists.
As homeomorphisms of the surface $S_n$, Dehn-twists preserve the intersection numbers of loops, which is also reflected in the preservation of commutativity of the corresponding symplectic transvections
\begin{equation}
    t_{\bs{\gamma}}\lrq{t_{\bs{\alpha}}, t_{\bs{\beta}}}t^{-1}_{\bs{\gamma}}=\lrq{t_{t_{\bs{\gamma}}\lr{\bs{\alpha}}}, t_{t_{\bs{\gamma}}\lr{\bs{\beta}}}}.
\end{equation}

In fig.~\ref{fig:Dehn} we depict such a generating set (known as the \textit{Lickorish} generators \cite{FarbMargalit+2012}), where each Dehn twist label associates to a corresponding lattice vector given either by a canonical basis element $e_i, f_i$ or a linear combination of such.

\begin{figure}
\center
\includegraphics[width=.7\textwidth]{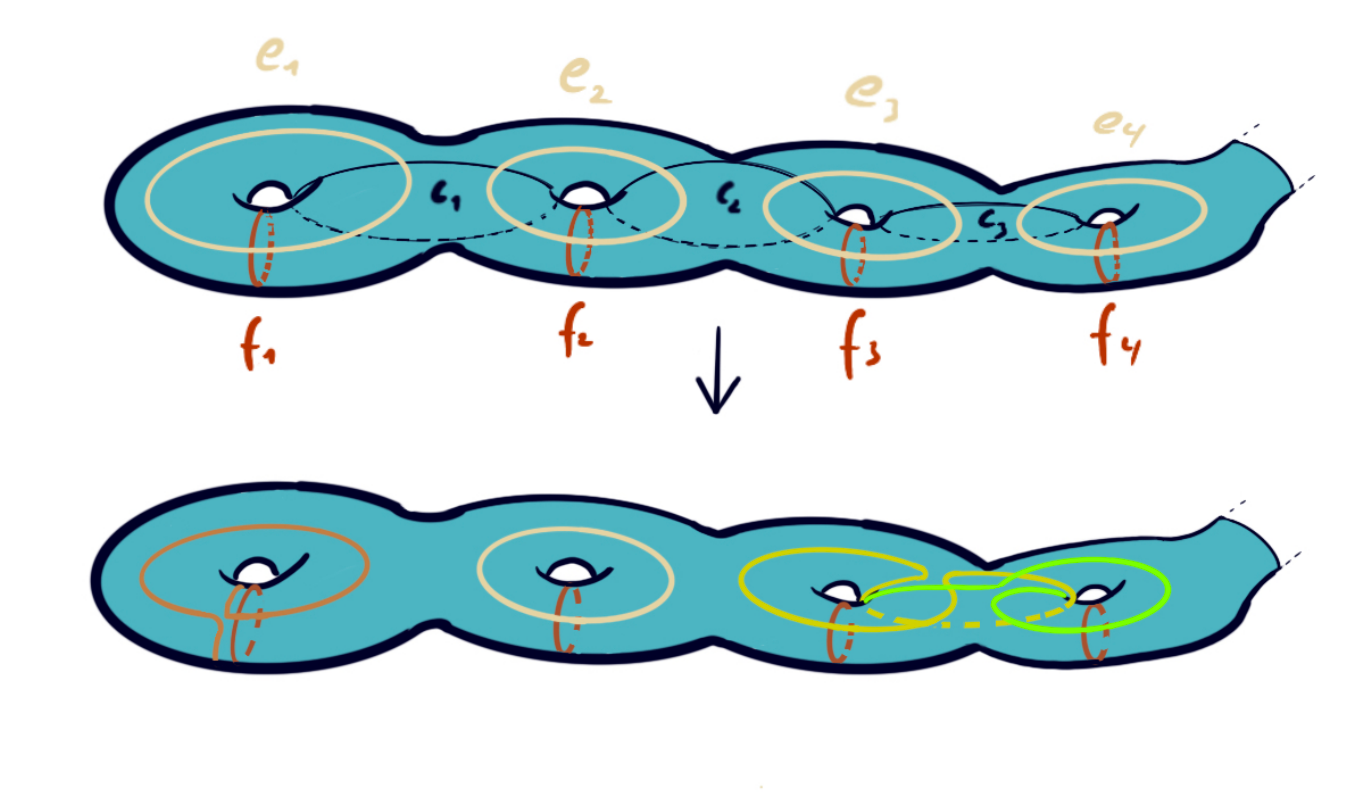}
\caption{A visualization of the surface $S_n$, where logical operators for the GKP code are represented as elements of the first homology group indicated by the elements $(\bs{e}_i, \bs{f}_i)$. 
Logical Clifford transformations are represented by sequences of Dehn-twists of the torus in this representation, and we depict how two such transformations act on these generators going from the top to the bottom figures. 
On the left-most handle, we show how a Dehn-twist about $\bs{f}_1$, a.k.a.\ a symplectic transvection $t_{\bs{f}_1}$, implements a logical phase gate by mapping $\bs{e}_1 \mapsto \bs{e}_1+\bs{f}_1$. 
On the right-most handle, a logical $CZ$ gate is realized via a Dehn twist about the loop with label $\bs{c}_3$, which corresponds to a symplectic transvection $t_{\bs{f}_3+\bs{f}_4}$. 
} \label{fig:Dehn}
\end{figure}

\paragraph*{Example: the square lattice, $\CL=\sqrt{2}\Z^2$.}
For the single-mode square GKP code we can choose bases such that $M=2M^{\perp}=\sqrt{2}I$. 
Relative to this choice the first row of $M^{\perp}$ represents the logical $X$-type Pauli operator 
$\hat{X}=e^{-i\sqrt{\pi}\hat{p}}$ while the second represents $\hat{Z}=e^{i\sqrt{\pi}\hat{q}}$. 

Via eq.~\eqref{eq:AutS} we can identify a symplectic transformation $g$ that implements a non-trivial Clifford gates with its integral representation via  $U=g^T$. 

It is convenient to introduce the $S$ and $T$ matrices,
\begin{equation}
    S=\begin{pmatrix}
        0 & -1 \\ 1 & 0
    \end{pmatrix},\;\;
    T=\begin{pmatrix}
        1 & 1 \\ 0 & 1
    \end{pmatrix},\label{eq:ST}
\end{equation}
which generate $\Sp_2\lr{\Z}=\langle S, T \rangle$.

In the integral representation, the $S$-matrix just introduced can be seen to implement a logical Hadamard gate $U_H=S$, while the logical phase gate $\hat{P}$ can be obtained from $U_P:= T^T$ (the transpose of $T$). 
The $S$-matrix is orthogonal, just as the associated symplectic transformation on the lattice and thus the logical Hadamard can be implemented by a mere passive linear optical element with a representative Gaussian unitary $\hat{U}_H=e^{-i\frac{\pi}{2}\hat{n}}$ corresponding to a $\pi/2$ rotation in phase space. 
The $T$-matrix however is not orthogonal and since the vectors in $\CL^{\perp}$ corresponding to Pauli-$Y$ operators are generically of a different length than corresponding Pauli-$X$ or -$Z$ representatives, the logical phase gate does not admit an orthogonal implementation \cite{Royer_2022}. 

\paragraph*{Example: the hexagonal lattice, $\CL=\sqrt{2}A_2$.}
For the hexagonal GKP code we have $M^{\perp}=M_{A_2}/\sqrt{2}$. 
As a root lattice, orthogonal automorphisms are given by reflections 
\begin{equation}
\label{eq:reflection}
    r_{\bs{\alpha}}=I-2\frac{\bs{\alpha}\bs{\alpha}^T}{\bs{\alpha}^T\bs{\alpha}}
\end{equation}
along the so-called root $\bs{\alpha},\, \bs{\beta}$ contained in the rows of $M_{A_2}$ in eq.~\eqref{eq:GKP_hex}.
Reflections are involutions with a $-1$ determinant, and hence are not symplectic. 
We can thus identify the subset of symplectic orthogonal automorphisms to lie within the even subgroup of the Weyl group $W\lr{A_2}$ which is generated by the product of the two reflections 
\begin{equation}
    R_{\frac{2\pi}{3}}=r_{\bs{\beta}}r_{\bs{\alpha}}=
    \begin{pmatrix}
        \cos \frac{2\pi}{3} & -\sin \frac{2\pi}{3} \\ 
        \sin \frac{2\pi}{3}& \cos \frac{2\pi}{3}
    \end{pmatrix}.
\end{equation}
Solving $UM_{A_2}=M_{A_2}R_{\frac{2\pi}{3}}^T$ yields the integral representation
\begin{equation}
U=\begin{pmatrix}
0 & 1 \\ -1 & 1
\end{pmatrix}.
\end{equation}
By probing its effect on the standard basis, we find that this matrix implements the transformation on logical Pauli operators $X \mapsto Z \mapsto Y \mapsto \hdots$
which realizes a logical $\lr{\hat{P}\hat{H}}^{\dagger}$ gate~\cite{GCB}.

\subsection{Generating symplectic automorphisms}

Due to the one-to-one relationship between symplectic automorphisms stated in lemma~\ref{lem:integral_rep}, a generating set for their integral representation immediately also yields a generating set for the symplectic matrices that need to be implemented to generate and logical Clifford group element. In  lemma~\ref{lem:integral_rep} we have identified the integral representation for symplectic automorphisms with the group $\Sp_{2n}\lr{\Z, D}$ where the group action on vectors representing logical Pauli operators is defined $\mod\, D$. We focus on the case of scaled GKP codes, with $D=dI_n$, such that the group of symplectic automorphisms in their integral representation is given by $\Sp_{2n}\lr{\Z_d}$. In the following we define a generating set for $\Sp_{2n}\lr{\Z_d}$ and proof that every element in $\Sp_{2n}\lr{\Z_d}$ can be generated by at most $O\lr{d n^2}$ elements in this generating set. This group can also be understood as known as the non-trivial Clifford group for a qudit, such that we will borrow from the lingo of Clifford gates on qudits to explain the action of its generators.

Block matrices 
\begin{equation}
S=\begin{pmatrix}
A & B \\ C & D
\end{pmatrix} \in \Z_d^{2n \times 2n}
\end{equation}
are symplectic if $A^TC=C^TA$, $B^TD=D^TB$ as well as$A^TD-C^TB=I$. In particular, we have that for $B=C=0$ the matrix is symplectic if $D=A^{-T}$ such that $S=A\oplus A^{-T}$. If $A=D=I$  and $C=0$ ($B=0$) it becomes necessary that $B=B^T$ is symmetric ($C$ is symmetric). 
Matrices of these constrained types have particularly simple structure and follow simple multiplication rules
\begin{align}
\begin{pmatrix}
A_1 & 0 \\ 0 & A_1^{-T}
\end{pmatrix} \begin{pmatrix}
A_2 & 0 \\ 0 & A_2^{-T}
\end{pmatrix}
&=\begin{pmatrix}
A_1A_2 & 0 \\ 0 & \lr{A_1A_2}^{-T}
\end{pmatrix}, \\
\begin{pmatrix}
I & B_1 \\ 0 & I
\end{pmatrix} \begin{pmatrix}
I & B_2 \\ 0 & I
\end{pmatrix}
&=\begin{pmatrix}
I & B_1+B_2 \\ 0 & I
\end{pmatrix}, \\
\begin{pmatrix}
I & 0 \\ C_1 & I
\end{pmatrix} \begin{pmatrix}
I & 0 \\ C_2 & I
\end{pmatrix}
&=\begin{pmatrix}
I & 0 \\ C_1+C_2 & I
\end{pmatrix} .
\end{align}

In this section we show, building on previous work on qubits \cite{AaronsonGottesman, PatelMarkovHayes}, how for prime dimension $d$, symplectic matrices in $\Sp_{2n}\lr{\Z_d}$ can be synthesized from an elementary gate set $S=\lrc{ J_i ,  P_{i}, C_{i\rightarrow j}}$ of such constrained block matrices consisting of the following matrices in block form,  where $\pi_{i}=\bs{e}_i\bs{e}_i^T$ and $e_{ij}=\bs{e}_i\bs{e}_j^T$:

\begin{itemize}
\item  The quantum Fourier transform on qudit $i$
\begin{equation} J_i=\begin{pmatrix}
I-\pi_{i} & \pi_i  \\ -\pi_i & I-\pi_i
\end{pmatrix} , i\in \lrq{1,n}
\end{equation}
with $J_i^2=-I$ , mapping  $X_i \mapsto Z_i^{-1}, \, Z_i \mapsto X_i,$

\item the phase gate
\begin{equation}
P_i=\begin{pmatrix}
I & 0  \\ \pi_i& I
\end{pmatrix}, i\in \lrq{1,n},
\end{equation}
 mapping $X_i\mapsto X_i Z_i$ and 

\item the  CNOT gate
\begin{equation}
C_{i\rightarrow j}=\begin{pmatrix}
I+e_{ji} &  0 \\ 0& I-e_{ij}
\end{pmatrix}, i \neq j \in \lrq{1,n}
\end{equation}
 that maps $X_i \mapsto X_iX_j$.

\item The CNOT gate is of block diagonal form, and it can be shown by performing the matrix multiplication that the upper triangular elementary block matrix
\begin{equation}
B_{ij}=\begin{pmatrix}
I &  e_{ij}+e_{ji} \\ 0& I
\end{pmatrix} = J_j^{-1} C_{j\rightarrow i} J_j
\end{equation}
 mapping $Z_i\mapsto X_j Z_i$ and $Z_j\mapsto X_i Z_j$ can be obtained by conjugating the CNOT with a Hadamard type gate. This generating set has $|S|=2n+n(n-1)$ elements, where the contribution $n(n-1)$ comes from the fact that we assume all-to-all connectivity for the CNOTs in use. This set can be reduced down to a set of $3n-1$ generators with CNOTs only between a linear number of pairs analogous to the Lickorish generators for the Dehn-twists mentioned before, which however would come at the cost of needing to mediate CNOTs not included in the set via a $O(n)$ number of those that are. 
 \end{itemize}

Denote sequences generated by a finite product from $S$ as 
\begin{equation}
S^k:=\lrc{g_1 g_2\hdots g_k, \, g_i \, \in S}.
\end{equation}
Similar to previous work on generating $\Sp_{2n}\lr{\Z_2}$ we show here that for
\begin{lem}
Let $d$ be prime. For the generating set $S={ J_i ,  P_{i}, C_{i\rightarrow j}}$ defined above, we have
\begin{equation}
\Sp_{2n}\lr{\Z_d} \subseteq S^k
\end{equation}
for $k=O\lr{d n^2}$.
\end{lem}
That is, sequences of $O(dn^2)$ of gates from $S$ suffice to generate all elements in $\Sp_{2n}\lr{\Z_d}$.

\proof It has been shown in ref.~\cite{DopicoJohnson} that every symplectic matrix  $S\in \Sp_{2n}\lr{\Z_d}$, $d$ prime admits a decomposition into symplectic matrices
\begin{equation}
S=Q\begin{pmatrix}
I & 0 \\ C & I
\end{pmatrix}
\begin{pmatrix}
A & 0 \\ 0& A^{-T}
\end{pmatrix}
\begin{pmatrix}
I & B \\ 0 & I
\end{pmatrix},
\end{equation}
where $A\in \mathrm{GL}_n\lr{\Z_d}$ is invertible and $C \in \mathrm{GL}_n\lr{\Z_d}$ and $B \in \mathrm{GL}_n\lr{\Z_d}$ are symmetric and $Q$ is an $O(n)$ length product of the matrices $J_i$ we have defined above. Ref.~\cite{DopicoJohnson} in fact showed this for the field of complex numbers $\C$, but the proof carries over to any number field, such as $\Z_d$ for $d$ prime. 
Using this decomposition, it suffices to check how each individual block matrix can be compiled from the generating set above. Using that 
\begin{equation}
J\begin{pmatrix}
I & B \\ 0 & I
\end{pmatrix} J^T=\begin{pmatrix}
I & 0 \\ -B & I
\end{pmatrix}
\end{equation} together with $J=\prod_{i=1}^n J_i$ we have that the every upper block triangular matrix can be converted to a lower block triangular one with $O(n)$ overhead and that every block upper triangular matrix 
\begin{equation}
\begin{pmatrix}
I & B \\ 0 & I
\end{pmatrix}
\end{equation}
with $B=B^T$ can be obtained from an $O(dn^2)$ fold product of matrices of type $J_i P_i J_i^T$ and $B_{ij}$ for their simple multiplication structure. It remains to bound the complexity of compiling the block diagonal part $A\oplus A^{-T}$. Note that due to the simple multiplication structure of these matrices this problem is equivalent to bounding the complexity of compiling the blocks $A$ as generated by elements $I+e_{ji} $. This is bounded using the same argument as in ref.~\cite{AaronsonGottesman}, which employed a result from Patel et al. \cite{PatelMarkovHayes}, who showed that for underlying field $\Z_2$ an achievable lower bound is given by $O\lr{n^2/\log_2\lr{n}}$. As was already noticed in ref.~\cite{PatelMarkovHayes}, their technique generalizes for any finite field with order $d$, where it yields a bound $O\lr{n^2/\log_d\lr{n}}$. In total, we hence obtain a bound $O(dn^2)$ for the length of the product from $S$ to generate any element in $\Sp_{2n}\lr{\Z_d}$.
\endproof

\begin{mybox}
\subsubsection*{What is ... a Riemann surface?}
A Riemann surface $C$ is a one complex-dimensional complex manifold, i.e. it is a two real-dimensional manifold covered by open sets $\lrc{U_{\alpha}}_{\alpha},\; \cup_{\alpha} U_{\alpha} = C$, such that there is atlas defining coordinate charts $\CA=\lrc{(U_{\alpha}, z_{\alpha})}_{\alpha}$ where 
\begin{equation}
z_{\alpha}:\, U_{\alpha} \to V_{\alpha} \subseteq \C
\end{equation}
is a homeomorphism to a subset of $\C$. The transition functions
\begin{equation}
f_{\alpha, \beta}= z_{\beta} \circ z_{\alpha}^{-1}:\; z_{\beta} \lr{U_{\alpha} \cap U_{\beta}} \rightarrow z_{\alpha} \lr{U_{\alpha} \cap U_{\beta}} 
\end{equation}
are required to be holomorphic and two atlases are equivalent if they only differ locally by holomorphic functions.

Trivially $\C$ is a Riemann surface, same as the Riemann sphere $\hat{\C}=\C\cup \C$.

A \textit{mapping} between Riemann surfaces $f: C \rightarrow D$ is a holomorphic if for every coordinate chart $(U, z)$ on $C$ and every coordinate chart $(V, w)$ on $D$ with $U\cap f^{-1}\lr{V}$ the map
\begin{equation}
w \circ f \circ z^{-1} : z\lr{U\cap f^{-1}\lr{V}} \rightarrow w\lr{V}
\end{equation}
is holomorphic. In particular, a holomorphic mapping into $\C$ is a holomorphic function and a holomorphic mapping into $\hat{\C}$ is a \textit{meromorphic} function. See \cite{Bobenko2011} for a more in-depth discussion on complex differential geometry.
\end{mybox}

\begin{figure}
    \center
    \includegraphics[width=.4\textwidth]{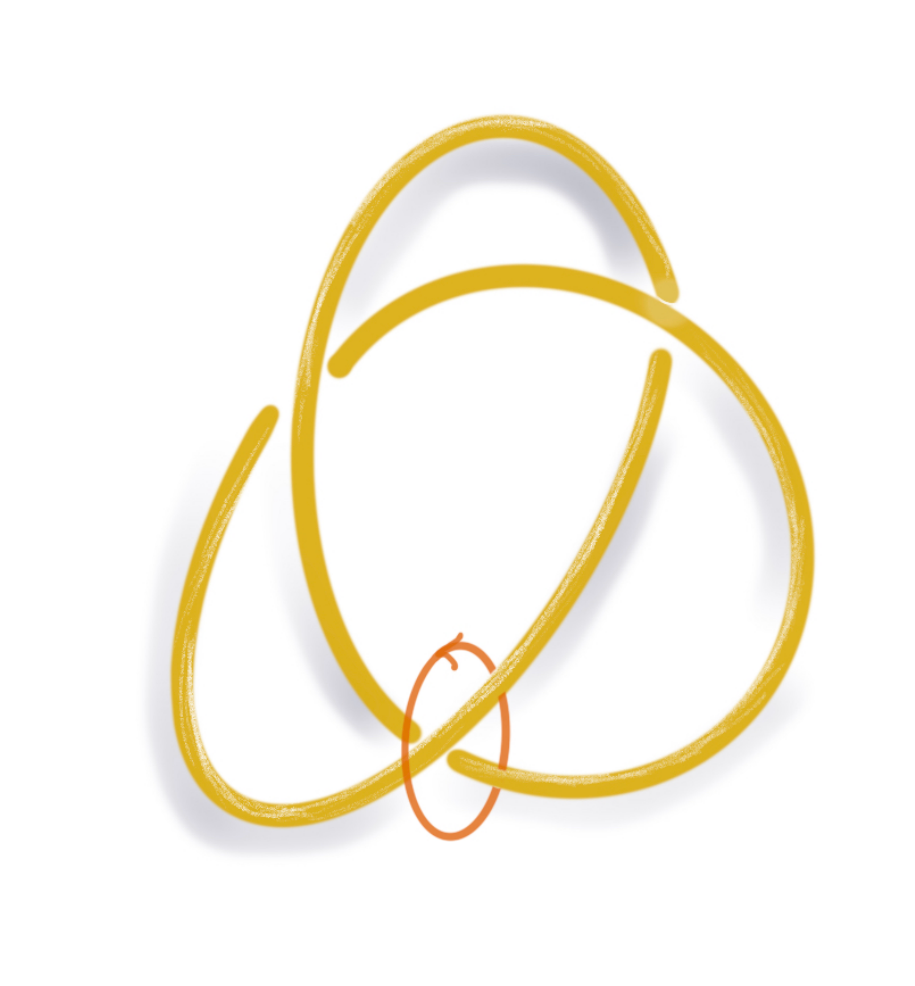}
    \caption{The trefoil knot corresponding to the one dimensional defect of ``distance zero GKP codes'' in $S^3$. A non-trivial link with the trefoil knot given by a $\pi/2$ rotation of the square lattice $\Lambda \mapsto e^{i\phi}\Lambda,\; \phi \in \lrq{0, \pi/2}$ corresponding to a logical Hadamard gates is illustrated. The non-trivial linking is a topological feature of the path that implements the logical Hadamard gate and can be interpreted as the feature that makes the gate implementation fault-tolerant within the fiber bundle framework for fault tolerance \cite{gottesman2017fiber}.}\label{fig:trefoil}
    \end{figure}

\subsection{GKP codes from compact Riemann surfaces}

The connection between the symplectic automorphism group of the $2n$-dimensional symplectic lattice and that of the mapping class group of a compact ($2$-dimensional!) genus $n$ surface $S_n$ is in fact not a coincidence, but hints at a deeper connection between symplectic lattices and the compact surface $S_n$.
As elaborated below, this connection leads to a way of viewing GKP codes as \emph{Jacobians} of algebraic curves. 
We begin by first discussing the relation between logical operators of GKP codes and the homology of compact surfaces. 

The identification of symplectic lattice automorphisms with transformations generated by Dehn twists encountered earlier (or simply elements of $\Sp_{2n}\lr{\Z}$), which are \textit{intersection number preserving} homeomorphisms of genus $g$ surfaces, suggests a more intuitive understanding of the topological nature of scaled GKP codes. 
One way to understand this connection is to realize that the homology groups 
\begin{equation}
    H_1\lr{\R^{2n}/\CL^{\perp}, \Z} \sim H_1\lr{S_n, \Z} \label{eq:Homology}
\end{equation}
are isomorphic, and that the symplectic inner product between $\Z$-valued vectors representing elements in $H_1\lr{\R^{2n}/\CL^{\perp}, \Z}$ is identical with the algebraic intersection number defined for elements in $H_1\lr{S_n, \Z}$. 
Since we have $\CL=d\CL^{\perp}$, the torus $\R^{2n}/\CL$ is a $d^2$-fold cover of $\R^{2n}/\CL^{\perp}$, such that, when regarding elements in $\CL=H_1\lr{\R^{2n}/\CL, \Z}$ as logically trivial elements, the representation of logical operators on $H_1\lr{\R^{2n}/\CL^{\perp}, \Z}$ descends to one on $H_1\lr{\R^{2n}/\CL^{\perp}, \Z_d}$. 
By eq.~\eqref{eq:Homology} we can thus regard elements in $H_1\lr{S_n, \Z_d} $ as representations of logical operators. 
A $d$-fold wind representing a stabilizer group element is corresponds to a trivial logical operator and the intersection number of these loops modulo $d$ determines the commutative phase of the associated displacement operators. 

\subsubsection{Jacobians and compact Riemann surfaces}
The Jacobian of a curve or compact Riemann surface can be thought of as a first-order approximation of a compact Riemann surface, which contains the information about the first homology group of the surface and the intersection between its elements.
We now briefly discuss the essential steps of the construction of the Jacobian and its associated symplectic lattice from a compact Riemann surface; for more detailed treatments see refs.~\cite{Birkenhake_2004, Sarnak1994, Berge}.

Let $C$ be a compact Riemann surface of genus $n=\dim H^0\lr{\omega_C}$, given by the dimension of the space of holomorphic differentials on $C$. 
As a $n$-handled torus, this Riemann surface has a canonical basis that generates its first homology group $\langle \gamma_1\hdots \gamma_{2n} \rangle = H_1\lr{C, \Z}$, where the intersection number between two basis elements $(\gamma_i \cdot \gamma_j)=-J_{ij}$ is  determined by the symplectic form we have encountered earlier. 

Choosing a basis $\omega_1,\hdots, \omega_n$ for $H^0\lr{\omega_C}$ yields a linear map
\begin{equation}
    p: H_1\lr{C, \Z} \rightarrow \C^n :\; \gamma \mapsto \lr{\int_{\gamma} \omega_1, \hdots , \int_{\gamma} \omega_n}^T \label{eq:periods}
\end{equation}
defined from the set of $2n$ generators of $H_1\lr{C, \Z}$ to $2n$ vectors in $\C^n$. 
The $n \times 2n$ matrix 
\begin{equation}
    \Pi=\begin{pmatrix} 
            p(\gamma_1)\; \hdots\; p(\gamma_{2n}) \label{eq:Jacobian}
        \end{pmatrix}
 \end{equation}
is known as the \emph{period matrix}. 

The period matrix admits a standard form. 
In particular, we can always choose a basis and normalization such that $\Pi$ takes the canonical form $\Pi=\lr{ I_n\; \Omega}$ where $\Omega$ is symmetric and $\Im \Omega > 0$ \cite{Birkenhake_2004}. 

Consider the lattice $\Lambda$ spanned by the columns of $\Pi$, $\Lambda=\Pi\,\Z^{2n}$. 
The complex torus $T_{\Lambda} = \C^n/\Lambda$ obtained from the quotient by $\Lambda$ is known as the \emph{Jacobian} variety $J(C)$ of $C$. 
We can map $\Lambda$ into real space lattice $L_{\Lambda}$ by associating with each vector $\Lambda\ni\bs{v} \mapsto (\Re \bs{v}^T, \; \Im\bs{v}^T)^T$; this is known as the real representation. 
Writing $\Omega=X+iY$, $L_{\Lambda}$ is generated by the rows of the matrix $M$ defined by
\begin{equation}
M^T=\begin{pmatrix}
I_n & X \\ 0 & Y
\end{pmatrix},
\end{equation} and satisfies
\begin{equation}
    M \lr{J_2 \otimes Y^{-1}} M^{T}=J.\label{eq:Msymp}
\end{equation}
Since $Y>0$, we can define the rescaled generator matrix
\begin{equation}
    M_C=M(I_2\otimes Y^{-\frac{1}{2}}), \label{eq:MC}
\end{equation}
which, by eq.~\eqref{eq:Msymp}, is symplectic and generates a symplectic lattice, such that it can be scaled to yield a GKP code as discussed earlier.

%
%

Eq.~\eqref{eq:MC} also allows for an interpretation of the lattice generated by $M$, it is simply a stretched version of the symplectic lattice spanned by $M_C$ with ``stretching" $Y^{\frac{1}{2}}\oplus Y^{\frac{1}{2}}$. In the simple case where $Y=dI_n$, this becomes equivalent to the lattice present in a scaled GKP code of type $D=dI_n$.

The lattice spanned by the rows of $\Pi$, can directly be seen to carry the symplectic structure using Riemann's bilinear relations, which tells us that two rows of the period matrix and its conjugate $\bs{a}_i, \bs{b}_j$ given by the period integrals over the forms $\omega_i, \overline{\omega}_j$ have symplectic inner product given by
\begin{equation}
    \int_{C} \omega_i \wedge \overline{\omega}_j = \bs{a}_i^TJ\bs{b}_j,
\end{equation}
which is always real and non-negative for $i=j$. 
In general, we have that
\begin{equation}
    i\int_C \omega_{i} \wedge \overline{\omega}_{j}=i\lr{\Pi J^T \overline{\Pi}}_{ij}
\end{equation}
yields a positive definite matrix; for more details, see ref.~\cite{Birkenhake_2004}. It is this relation that explains the connection between the Jacobian and symplectic lattices.

In general, a complex torus $\C^n/\Lambda$ obtained from a symplectic lattice is also known as a \textit{principally polarized Abelian variety} \cite{Sarnak1994, Birkenhake_2004}, where ``polarized Abelian variety" refers to the fact that there is a Hermitian inner product $H\lr{x,y}=x^{\dagger}Y^{-1}y,\; Y>0$ on this torus with the property that $\Im H\lr{\Lambda, \Lambda}\in \Z $. 
In the real representation, $\Im H\lr{\bs{x}, \bs{y}}=\bs{x}^T(J_2 \otimes Y^{-1})\bs{y}$ is precisely the product appearing in eq.~\eqref{eq:Msymp}.
The adjective ``principal" applies to the special case $Y=I_n$, such that $L_{\Lambda}$ as defined above automatically is a symplectic lattice \cite{Birkenhake_2004, Sarnak1994, Berge} which we have seen to arise above under the appropriate transformation.

Rather than to refer to compact Riemann surfaces, one typically refers to the Jacobian associated to a projective complex algebraic \emph{curve}. 
This underlies a deep connection between algebra and geometry: projective complex algebraic curves can be understood as ``explicit parametrizations" of compact Riemann surfaces. 
This connection is outlined in the box below, and I refer the reader for a more detailed treatment to refs.~\cite{Griffiths1989, Bobenko2011}.

The point is that these connections allow us to construct (scaled) GKP codes from complex curves via their Jacobians. 
The chain of correspondences illustrating the chain of maps that map from curves to GKP codes via the construction of the Jacobian is pictured in figure~\ref{fig:curvesGKP}.

\begin{mybox}
    \subsubsection{Complex algebraic curves and compact Riemann surfaces}
A complex algebraic curve 
\begin{equation}
    C=\lrc{(x,y)\in \C^2,\; P(x,y)=0 }\subset \C^2 
\end{equation}
is the set of roots of a polynomial equation in $2$ variables with maximal degree $d$ and is equivalent to its homogenization 
\begin{equation}
    C_h=\lrc{(x,y,z) \in \C^3, \, P_h(x, y, z)=0 }\subset \CPP^2,
\end{equation}
given by the constant degree polynomial $P_h(x, y, z)=z^d P(x/z, y/z)$ whose set of roots in $C_h$ satisfy the equivalence relation $(\lambda x, \lambda  y, \lambda  z) \sim (x,y,z), \, \lambda \in \C^{\times}$ and $C\subset \CPP^2$ is typically viewed as a projective curve. 
There are finitely many singular points
\begin{equation}S=\{(x_0, y_0)\in C:\; \partial_{x}P(x_0, y_0)=\partial_{y}P(x_0, y_0)=0 \},\end{equation} 
away from which the curve can always be parameterized by points of the form $(x, y(x))\in \C^2$ or $(x(y), y)\in \C^2$ such that either $\partial_x y(x)=-\partial_y P(x,y) / \partial_x P(x,y)$ or $\partial_y x(y)=-\partial_x P(x,y) / \partial_y P(x,y)$ are well-defined and the projection $(x, y) \rightarrow y$ resp. $(x, y) \rightarrow x$ yields local coordinates in $\C$. 
The normalization theorem \cite{Griffiths1989} effectively smoothens out the singular points and guarantees the existence of a compactification of the curve $C^*=C\backslash S$ to obtain a compact Riemann surface $\widehat{C}$ that covers $C$. 
Reversely, Riemann showed that all compact Riemann surfaces can be described as compactifications of algebraic curves. 
While these arguments show the one-to-one correspondence between compact Riemann surfaces and algebraic curves, they are unwieldy in the explicit computation of the period integrals to construct Jacobians associated to curves.
\end{mybox}
\begin{mybox} 
In the special case of hyperelliptic curves given by polynomials of the form
\begin{equation}
P(x, y)=y^2-f(x),\hspace{1.5cm} f(x)=\prod_{i=1}^N (x-\lambda_i),\; \lambda_i\neq \lambda_j\; \forall\, i\neq j
\end{equation}
the identification of the homology basis of the corresponding compact Riemann surface and construction of the homolorphic forms becomes more simple: 
Solutions to the curve are of the form $y=\sqrt{f(x)}$
where a homology basis with $2n$ generators is derived from the branch cuts of the complex square root spanned between the roots $\lambda_i$
(see also ref.~\cite[p. 160]{Silverman2009}). 
A basis of for the holomorphic differentials is then given by
$\omega_i=x^{i-1}\frac{dx}{\sqrt{f(x)}}$ for $i=1\hdots n$ \cite{Birkenhake_2004}.
\end{mybox}

\begin{figure}
\center
\resizebox{\textwidth}{!}{
\begin{tikzpicture}
\node (RC) at (0,0) {\rm Compact Riemann surface $\widehat{C}$};
\node (C) at (10,0) {\rm Curve $C$};
\node (JC) at (0,-2.5) {\rm Jacobian $J(C)$};
\node (L) at (5,-2.5) {\rm symplectic lattice $L_\Lambda$};
\node (G) at (10,-2.5){\rm scaled GKP code $\CL_{\Lambda}$};
\draw[<->] (C) --node[above]{normalization theorem \cite{Griffiths1989}} node[below]{Riemann existence theorem \cite{Griffiths1989}} (RC);
\draw[->] (JC) -- node[above] {eq.~\eqref{eq:MC}} (L);
\draw[->]  (RC) -- node[right] {period mapping eqs. \eqref{eq:periods}, \eqref{eq:Jacobian} } (JC);
\draw[<->] (L) -- node[above]{eq.~\eqref{eq:scaledGKP}}  (G);
\end{tikzpicture}
}
\caption{The chain of correspondences and maps that associate a GKP code to any complex algebraic curve $C\subset \CPP^2$.}\label{fig:curvesGKP}
\end{figure}
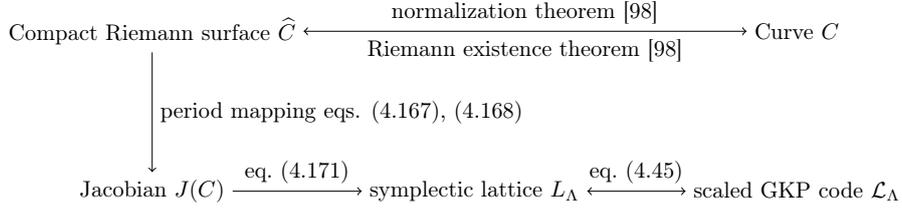
 
The fact that GKP codes may be obtained from compact Riemann surfaces has some interesting implications. 
For example, one may expect that representations of quantum states in code space obtained from pulling back phase-space representations (such as the stellar representation~\cite{Chabaud_2022}) to be constrained by the topology of the Riemann surface. 
However, the more interesting immediate question is whether every GKP code can be understood as a curve. 
Unfortunately, the answer to this is negative: there are symplectic lattices, such as the $E_8$ lattice, that do not arise as the Jacobian of curves \cite{Sarnak1994, Berge}. 
The general question of ``which principally polarized Abelian varieties arise as Jacobians of curves" is a long-standing mathematical quest known as the \emph{Schottky problem} \cite{Sarnak1994, grushevsky2010schottky}.

\subsubsection{Single-mode GKP codes from elliptic curves}

In the previous section we have established a connection between complex curves and GKP codes which we now make more concrete for the case of a single mode $n=1$. 
An \emph{elliptic curve} is a complex torus~\cite{hain2014lectures, Silverman2009}
\begin{equation}
    E=\lr{\C/\Lambda, z },\; z\in\C/\Lambda ,
\end{equation}
where $\Lambda \subset \C$ is a complex, non-degenerate lattice and $z \in\C/\Lambda$ is a point on the torus~\cite{hain2014lectures}. 
The point $z\in \C/\Lambda$ can be thought of as the choice of $0$-point on the torus (since the torus forms an additive group under addition in $\C$ modulo $\Lambda$ we need to fix an identity element). 

This definition of an elliptic curve is one-to-one with an algebraic definition in the following sense. 
The curve 
\begin{equation}
    C_{g_2(\Lambda), g_3(\Lambda)}=\lrc{(x, y) \in \C^2, \; y^2=4x^3-g_2(\Lambda)x-g_3(\Lambda) },
\end{equation}
specified by two complex numbers $g_2, g_3$ that are the image of a lattice under the functions defined below, is parameterized by the Weierstrass $\wp$ function
\begin{equation}
    \wp\lr{z,\, \Lambda}:=\frac{1}{z^2}+\sum_{\omega\in \Lambda\setminus\lrc{0}}\lr{\frac{1}{(z-\omega)^2}-\frac{1}{\omega^2}},
\end{equation}
where the invariance under translations by lattice vectors $\wp(z+\Lambda)=\wp(z)$ shows that this is a well-defined function on the complex torus $\C / \Lambda$ with poles of order $2$ on each lattice point. 
Therefore, distinct lattices $\Lambda$, $\Lambda'$ are distinguished by their $\wp$ functions. 

The Weierstrass function $\wp$ also provides an alternative parametrization of the elliptic curve $C_{g_2(\Lambda),g_3(\Lambda)}$, which can be seen as follows. 
Introduce the (normalized) Eisenstein series of weight $k$
\begin{equation}
    G_k\lr{\Lambda}=\sum_{\omega \in \Lambda\setminus \lrc{0}} \omega^{-k},
\end{equation}
and
\begin{equation}
    g_2(\Lambda)=60G_{4}\lr{\Lambda},\; g_3(\Lambda)=140 G_{6}\lr{\Lambda}\,.
\end{equation}
Then the equation for the elliptic curve is given by 
\begin{equation}
    \wp'^{2}=4\wp^3-g_2\wp -g_3.
    \label{eq:elliptic_curve}
\end{equation}
The elliptic curve is non-singular, i.e.\ it has no cusps or self-intersections, when the discriminant of the right-hand side  
\begin{equation}
    \Delta(\Lambda) = g_2^3-27g_3^2 
\end{equation}
is nonzero, which holds whenever $\Lambda$ is full-rank in $\C$. 

Let $\omega_1, \omega_2$ form a basis for the lattice $\Lambda=\omega_1 \Z \oplus \omega_2\Z$, which is full-rank if $\Im\lr{\omega_2/\omega_1}\neq 0$. 
One can fix an orientation of the basis elements by choosing a basis with $\Im\lr{\omega_2/\omega_1} > 0$, corresponding to a positive intersection of the homology element $\omega_2$ with $\omega_1$ on the torus $\C /\Lambda$ such that, up to an overall factor of rescaling and rotation $\omega_1$, the lattice  $\Lambda_{\tau}= \Z  \oplus \tau \Z$ is parameterized by $\tau \in \hh:=\lrc{ z\in\C,\, \Im(z)>0}$ in the complex upper half plane. 
As a function of $\tau$, the Eisenstein series defined above $g_k(\tau)=g_k(\Lambda_{\tau})$ are modular forms of degree $2k$  \cite{Zagier2008}, implying that they satisfy the transformation rule $f(\gamma.\tau)=(c\tau+d)^k f(\tau) \, \forall \gamma \in \SL_2(\Z)$, where we have introduced the Möbius transformation
\begin{equation}
    {\begin{pmatrix}
        a & b \\ c & d
    \end{pmatrix}
    }.\tau = \frac{a\tau+b}{c\tau+d}\,.
\end{equation}

\begin{mybox}

\subsubsection{M{\"o}bius acrobatics}

To further illustrate the behavior of the Möbius action on the upper half plane we discuss how it can be used to derive the Iwasawa- and Bloch-Messiah decomposition, depicted in fig.~\ref{fig:IwasawaBloch}. Our presentation is guided by the example presented in ref.~\cite{Conrad_lectures}. The main ingredient to this understanding it the transitivity of $\SL_2\lr{\R}$ on the upper half space $\hh = \SL_2\lr{\R}/\SO_2\lr{\R} .i$

To derive the Iwasawa decomposition, recognize that  an arbitrary $z \in \hh$ can be written as
\begin{equation}
z=x+iy=\begin{pmatrix}
1 & x \\  0 & 1
\end{pmatrix}.\begin{pmatrix}
\sqrt{y} & 0 \\  0 & 1/\sqrt{y}
\end{pmatrix} . i,
\end{equation}
where the squeeze ``pushes'' the point $z_0=i$ upwards to $z=iy$ and the final shear moves it horizontally to $z=x+iy$. Since every point $z \in \hh$ can be described by this sequence of Möbius transformations and the upper half plane is one-to-one with elements of $\Sp_2\lr{\R}$ up to a right- rotation, we can deduce that every matrix $S \in \Sp_2\lr{\R}$ can be written as $S=NAK$, where $N$ and $A$ are shears and squeezes as above and $K\in \Sp_2\lr{\R}\cap \SO_2\lr{\R}$ is a rotation.

Similarly, every point $z$ can also be expressed by a squeeze and rotation acting on $z_0=i$, which leads to the Bloch-Messiah decomposition $S=K_1 A K_2$, where $A$ is again a squeeze and $O_1, O_2 \in \Sp_2\lr{\R}\cap \SO_2\lr{\R}$. The steps are geometrically sketched in fig.~\ref{fig:IwasawaBloch}.
\end{mybox}

\begin{figure}
    \center
    \includegraphics[width=.8\textwidth]{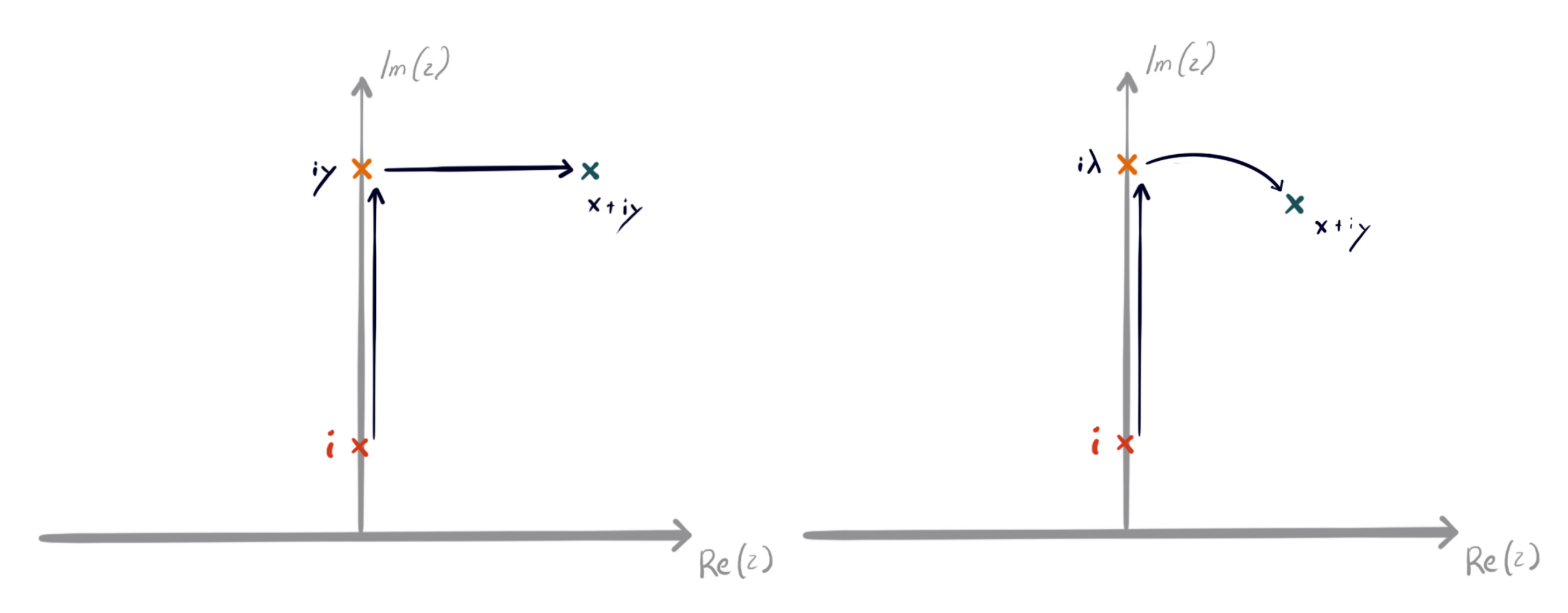}
    \caption{The Iwasawa- (l.) and Bloch-Messiah (r.) decomposition of symplectic matrices understood via M{\"obius} transformations. }
    \label{fig:IwasawaBloch}
\end{figure}

As a function of $\tau$, the discriminant modular form $\Delta(\tau)=\Delta(\Lambda_{\tau})$ is a modular cusp form of weight $12$ that vanishes at $i\infty$. 
Möbius transformations with elements $\gamma=\begin{pmatrix}
a & b \\ c & d
\end{pmatrix}\in \SL_2\lr{\Z}$ can be understood as basis transformation of the corresponding lattice via the map
\begin{equation}
\begin{pmatrix}
1 \\ \tau 
\end{pmatrix}
\mapsto
X\gamma X
\begin{pmatrix}
1 \\ \tau 
\end{pmatrix}
=
(c\tau+d)\begin{pmatrix}
1 \\ \gamma.\tau 
\end{pmatrix},
\end{equation}
where $X=\begin{pmatrix}
0 & 1 \\ 1 & 0
\end{pmatrix}$ and $X\gamma X \in \SL_2\lr{\Z}$. 
For a fixed volume $\det \Lambda_{\tau}$, the lattice $\Lambda$ can always be recovered via appropriate rescaling up to a global rotation.

To associate a single-mode GKP code to an elliptic curve, note that the Möbius transformation defines a transitive action on the upper half plane $\hh$. 
For symplectic orthogonal matrices $K\in \SO_2\lr{\R}=\Sp_2\lr{\R}\cap O_2\lr{\R}$, we have that $i$ is a fixed point, $i=K.i$. 
Therefore, every point in the upper half plane $\tau=S.i\in \hh$ is one-to-one with a symplectic matrix $S\in \Sp_2\lr{\R}/\SO_2\lr{\R}$ up to a rotation. 
We associate with the symplectically self-dual lattice $\Z^2$ the square GKP code encoding a qudit with dimension $d$ by rescaling the lattice $\tau \mapsto \sqrt{d/\det\lr{\Lambda_{\tau}}} \Lambda_{\tau}$. 
Since this rescaling can always be done, it suffices to identify  $\Lambda_{\tau}$, equivalently the torus $\C/\Lambda_{\tau}$, with the corresponding qudit GKP code. 

In fact, this procedure allows to obtain \emph{all} single-mode GKP codes as the orbit $\Sp_2\lr{\R}.i$. 
We can hence identify single-mode GKP codes with elliptic curves $E=\lr{\C/\Lambda_{\tau}, z},\; z\in \C/\Lambda_{\tau}$. 
First we interpret $\Lambda_{\tau}$ as the lattice associated to the stabilizer group of a GKP code. 
Then $z$, which labels a point in $\C$ up to a displacement by a (stabilizer) element in $\Lambda_{\tau}$, is interpreted as the sum of a syndrome $z\mod \frac{1}{d}\Lambda_{\tau}$ and a representative logical displacement label $\overline{z}\in \frac{1}{d}\Lambda_{\tau}$. 
See the discussion in Sec.~\ref{sec:Cliff}.

A level-$d$ structure \cite{hain2014lectures} on an elliptic curve is given by an oriented basis $\lr{\frac{1}{d}, \frac{\tau}{d}}$ of $H_1\lr{E, \Z_d}$ -- the so-called $d$-torsion points on $E$ -- where the intersection number modulo $d$ of the basis elements is $1$ such that the intersection pairing $H_1\lr{E, \Z_d} \times H_1\lr{E, \Z_d} \rightarrow \Z_d$ again defines the desired Heisenberg-Weyl commutation phase of the associated displacement operators (see eq.~\eqref{eq:HW_operators}). 
The level structure defines a finer structure on the elliptic curve. 
Under the above mapping from elliptic curves to GKP codes, it can be understood as the algebra of logical operators (the symplectic dual lattice to $\CL_{\Lambda}$), relative to which $z $ becomes associated with the syndrome of the GKP code. 

\begin{figure}
    \center
    \includegraphics[width=.4\textwidth]{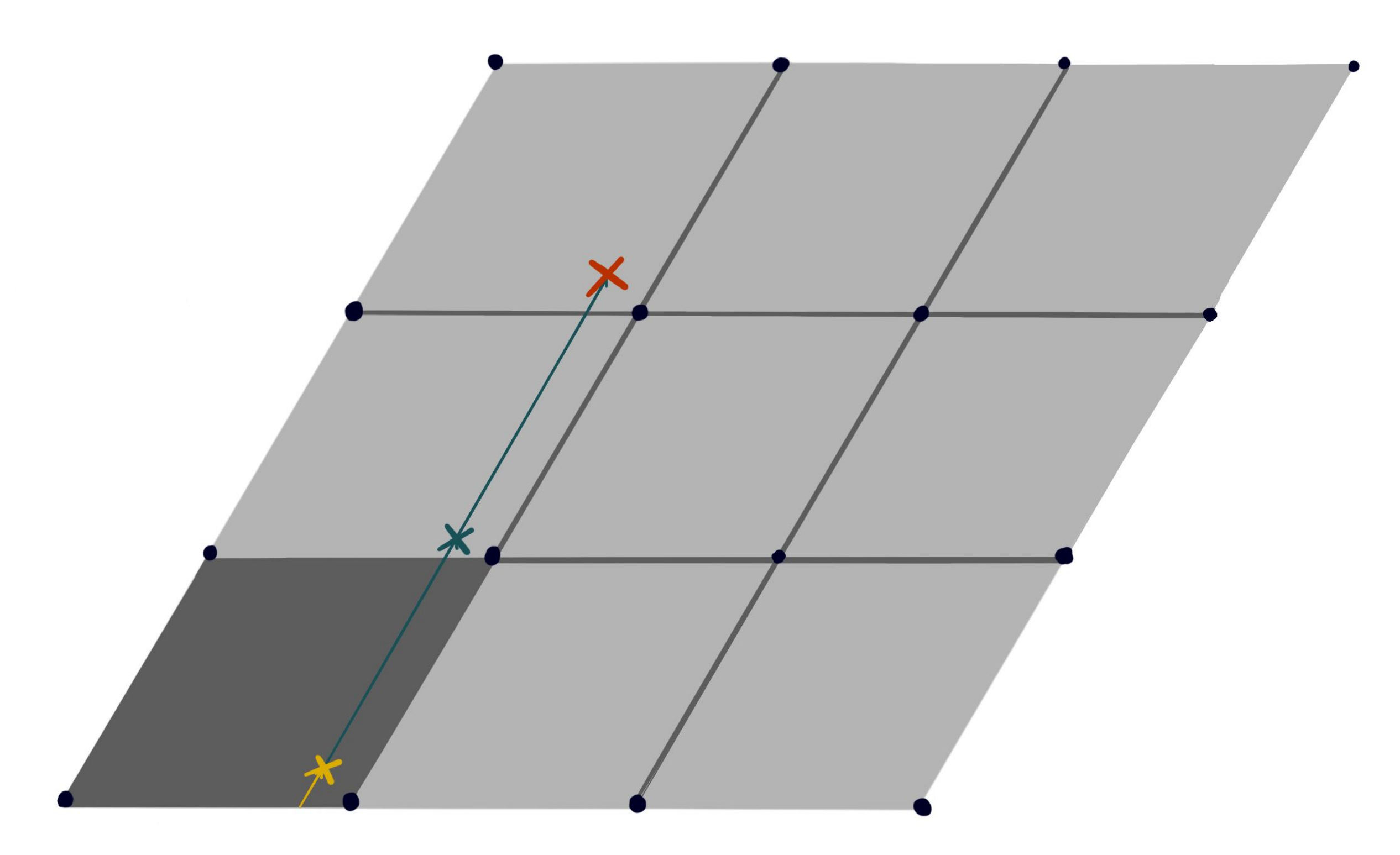}
    \caption{The hexagonal GKP code given by $\Lambda_{\rho},\; \rho=e^{i2\pi/3}$ and relative level structure $d^{-1}\Lambda_{\rho}$ for $d=3$. The points $z\mod d^{-1}\Lambda_{\rho}$ parametrize the syndrome of the GKP code while the $d$-torsion points $d^{-1}\Lambda_{\rho}$ in in $\Lambda_{\rho}$ are interpreted to label logical Pauli operators for the associated GKP code.}
    \label{fig:torsion}
\end{figure}

\subsection{Moduli space of GKP codes and fiber bundle fault tolerance}\label{sec:GKP_FT}

\subsubsection{GKP distance and modular discriminant}

In this section, we discuss a connection between the space of all single-mode GKP codes with a nonzero distance and a complex parameterization in terms of the modular discriminant. 

To understand the space of GKP codes, lets focus on the space of symplectic lattices in $2-$dimensions (equivalently, we focus on the space of elliptic curves ignoring the choice of $z$). 
We have already seen in the previous section that every point $\tau \in \hh$ parametrizes a symplectic lattice up to an overall rotation. Since lattices -- as geometric objects -- are defined independent of the choice of representing basis, the set of symplectic lattices up to basis transformation is given by the left quotient $\Sp_2\lr{\Z}\backslash \Sp_{2}\lr{\R}$ (remember that $\Sp_2\lr{\Z}=\SL_2\lr{\Z}$). $\Sp_{2}\lr{\R}$ has a transitive action on the upper half plane $\hh$, which is trivial for the elements $\lrc{\pm I}$. 
We can hence equally parametrize the space of all symplectic lattices by the quotient
\begin{equation}
M_1 = \PSp_2\lr{\Z} \backslash \hh,
\end{equation}
where $\PSp_2\lr{\Z}=\Sp_2\lr{\Z} / \lrc{\pm I}$ has an \textit{effective} action on $\hh$.

Points of $M_1$ corresponds to isomorphism classes of elliptic curves (GKP codes)\cite{hain2014lectures} and, in fact, $M_1$ again is a Riemann surface, where a holomorphic map to $\C$ is given by the $j$ function
\begin{equation}
j(\tau)=1728\frac{g_2^3(\tau)}{\Delta(\tau)},\label{eq:j_def}
\end{equation} 
which is a modular form of weight $0$. 

We can represent $M_1$ via the fundamental domain
\begin{equation}
\CF=\lrc{\tau \in \hh: \; |\Re\lr{\tau}|\leq \frac{1}{2},\; |\tau|\geq 1},
\end{equation}
shown in fig.~\ref{fig:squeeze_tess}.%
\begin{figure}
\center
\includegraphics[width=.6\columnwidth]{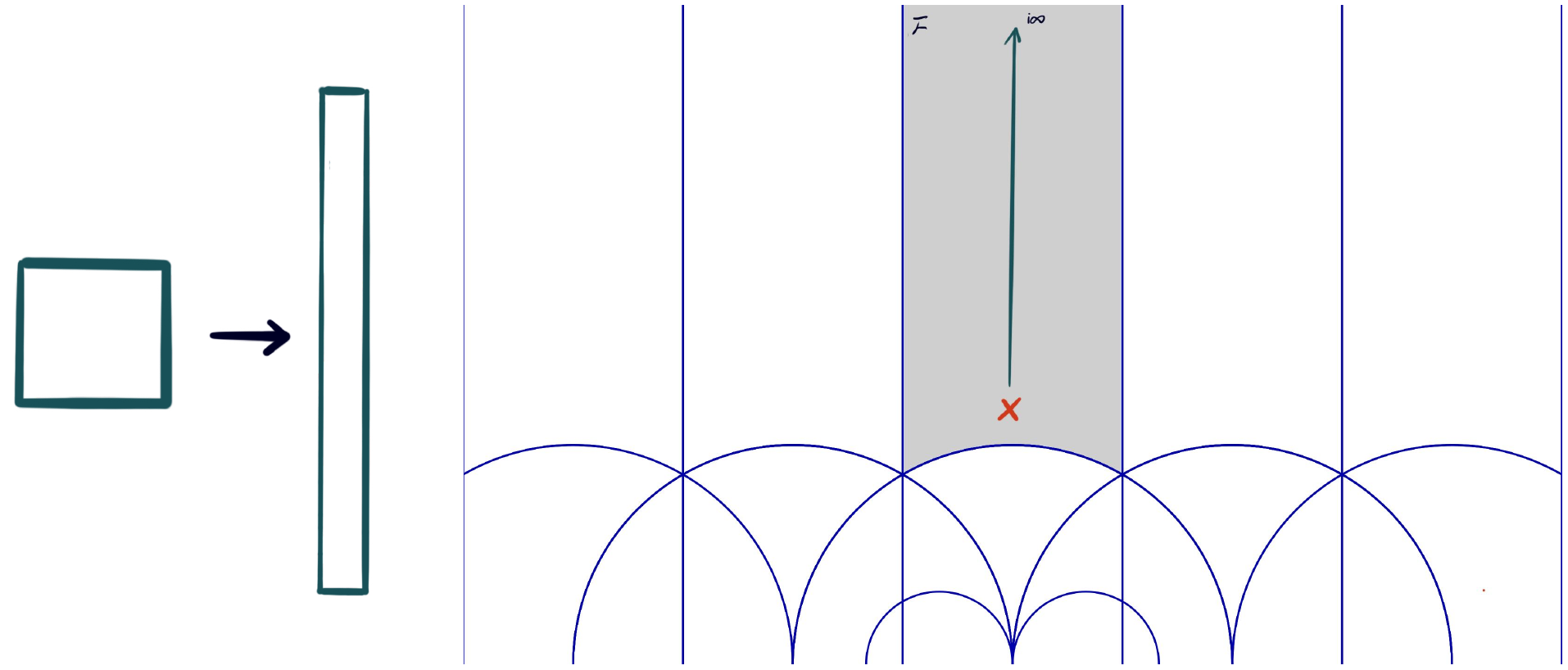}
\caption{The fundamental domain $\CF$ on is marked in grey on the RHS. 
We illustrate the effect of a squeezing operation $\tau \mapsto (\lambda \oplus \lambda^{-1}).\tau=\lambda^2 \tau$ and the corresponding transformation on the lattice $\Lambda_{\tau} \mapsto \Lambda_{\lambda^2\tau}/ \sqrt{\det\lr{\Lambda_{\lambda^2\tau}}}$.}\label{fig:squeeze_tess}
\end{figure}

It is important to note that the space $M_1$ fails to be a quotient manifold in for which every point $\tau \in M_1$ has isomorphic orbits under $\SL_2\lr{\Z}$ action. 
At fault is the existence of fix-points in $\hh$ such that the $\PSp_2\lr{\Z}$ action is not free. 
On the upper half plane, the $S$ and $T$ matrices that generate $\Sp_2\lr{\Z}$ have Möbius actions $S.z=-1/z,\; T.z=z+1$, such that  the points $\tau=i$ and $\tau=\rho:=e^{i2\pi/3}$ are fixed points under $S$ and $ST^{-1}$. 
It is quickly verified that $\tau=i$ corresponds to the GKP code built from the square lattice $\Z^2$, where the $S$ matrix can be understood as the logical Hadamard gate $\hat{H}$ in its integral representation and similarly, $\tau=\rho$ corresponds to the GKP code built from the hexagonal lattice ($A_2$) which has a logical $\hat{H}\hat{P}^{\dagger}$ Hadamard times phase gate corresponding to the $ST^{-1}$ matrix. 
The existence of automorphisms is hence both a blessing and a curse. 
They show the existence and characterize possible logical Clifford gates, but also equip our classifying space $M_1$ with the structure of an orbifold -- meaning that rather than being locally isomorphic to $\C$, it behaves locally like a quotient space of $\C$ modulo a local group action by a group that varies from point to point \cite{caramello2022introduction, hain2014lectures}. 
In fact, as a consequence of our choice of representation, these fixpoints are stabilized by elements in $\SO_2\lr{\R}$, i.e.\ they are associated to GKP codes with logical Clifford gates implementable through passive linear optical elements. 
Later we will construct a moduli space of GKP codes where these points are effectively removed by a choice of additional constraints, such that the moduli space can be fully treated as a complex manifold, and we will leave an investigation of spaces of GKP codes that incorporates the orbifold structure to future work.
We begin by investigating the connection between the topology of $M_1$ and the coding theoretic properties of the associated codes.

The $j$ function diverges in the limit $\tau\rightarrow i\infty$. 
Using $q=e^{i2\pi \tau}$, we can write
\begin{align}
j\lr{\tau}&=q^{-1}+744+19884q + \hdots,\\
\Delta\lr{\tau}&=(2\pi)^{12}q\prod_{k=1}^{\infty}(1-q^k)=\sum_{n=1}^{\infty} \tau(n)q^n,\\
g_2\lr{\tau}&=\frac{4\pi^4}{3}\lr{1+240\sum_{n=1}^{\infty}\sigma_3(n)q^n},
\end{align}
where we have the Ramanujan $\tau(n)$ function and the divisor sum $\sigma_3(n)=\sum_{d|n}d^3$ function, we recognize that the source of this divergence is the simple root of $\Delta(\tau)$ in the limit $\tau\rightarrow i\infty$. 
Comparing to the discussion in the previous chapter, this corresponds to the limit where the lattice $\Lambda_{\tau}$ is not full rank anymore. To understand this point better, write $\tau=x+iy=M.i$, with
\begin{equation}
M=\begin{pmatrix}
\sqrt{y} & x/\sqrt{y} \\ 0 &1/\sqrt{y}
\end{pmatrix}.
\end{equation}
The shortest vector in the lattice $L(M)$ spanned by the rows of $M$ satisfies
\begin{align}
\lambda_1^2\lr{L(M)} &=\min_{(0,0)\neq(n, m) \in \Z^2} \|n^2y+(nx+m)^2/y \|^2 \nonumber\\&\leq \lrc{y+\frac{x^2}{y},\, \frac{1}{y}},
\end{align}
such that in particular we have $\Im\lr{\tau}\leq\lambda_1^{-2}(L(M))$. 
That is, representing the lattice basis in $\hh$, demanding that the lattice (the corresponding GKP code) has finite non-zero distance $\lambda_1\geq const.$, yields an upper bound on the imaginary part of its representation in $\hh$. 
Similarly, one can show that the squeezing value associated to $M$, that is the squeezing necessary to prepare a code state associated to $M$ starting at the canonical square GKP code bounds ${\rm sq}\lr{M}=\|M^T\|_2\geq \Im\lr{\tau} $. 
We illustrate the intuition behind the limit $\tau\rightarrow i\infty$ being associated to a zero distance GKP code in fig.~\ref{fig:squeeze_tess}, where, starting at a square GKP code, a squeezing deformation maps $\tau=i\mapsto \lambda^2i,\,\lambda\in \R$. 
While one of the lattice basis vectors gets increasingly longer, due to the volume-preserving nature of $\Sp_2\lr{\R}$ the other shrinks until it converges to $0$ in the infinite squeezing limit. 

We show that the finiteness of the distance of the GKP code $\lambda_1\geq const.$ also lower bounds the discriminant function $\Delta(\tau)$. With $\Im\lr{\tau}\leq \lambda_1^{-2}$, for large $\Im\lr{\tau}$ we also have
\begin{equation}
|j(\tau)|\leq e^{2\pi/\lambda_1^2}+O(1).
\end{equation}
Using eq.~\eqref{eq:j_def} this bounds
\begin{equation}
|\Delta\lr{\tau}|\geq e^{-2\pi/\lambda_1^2} |g_2\lr{\tau}|^3+O\lr{|g_2\lr{\tau}|^3}.
\end{equation}
Since all the zeros of the Eisenstein series lie on the unit circle $|\tau|=1$ \cite{Rankin1970}, $|\Delta\lr{\tau}|$ will be lower-bounded by $\text{const.} \times e^{-2\pi/\lambda_1^2}$ away from $|\tau|=1$. 
Together with the fact that the discriminant modular form is non-zero for any finite value in $\hh$, in particular on the circle $|\tau|=1$, this shows that any finite distance GKP code with $\lambda_1\propto \Delta_{\rm GKP} > 0$ will also have a non-zero modular discriminant. 

We have arrived at the main result of this subsection: the space of bounded distance $\Delta_{\rm GKP}\lr{\CL}> {\rm const.}$ single-mode GKP codes can be parametrized by $\tau\in \hh$ with bounded $|\Delta(\tau)|>{\rm const.}$. 

\subsubsection{Topological interpretation: the trefoil defect}

We can understand this space topologically via an interpretation presented in refs.~\cite{Ghys,Milnor+1972}. 
As we have argued above, every lattice $\Lambda\subset \C$, through its association to a defining equation for an elliptic curve eq.~\eqref{eq:elliptic_curve}, is equivalently parametrized by the two parameters $(g_2, g_3)\in \C^2$. 
Since for any $c\in \C^{\times}$ we have $g_2\lr{c\Lambda}=c^{-4}g_{2}\lr{\Lambda}, g_3\lr{c\Lambda}=c^{-6}g_{3}\lr{\Lambda}$, one can always rescale the lattice so that $|g_2|^2+|g_3|^2=1$ which is the parametrization of a $3-$sphere $S^3$.
The space of zero-distance GKP codes is given by $0=\Delta=g_2^3-27 g_3^2$. 
In terms of the two complex parameters this equation defines a \textit{trefoil knot} 
\begin{equation}
    K=\lrc{(g_2, g_3)\in \C^2,\; g_2^3-27 g_3^2=0, \, |g_2|^2+|g_3|^2=1 }. 
\end{equation} 
We can therefore understand the space of single-mode GKP codes as the knot complement $S^3-K \sim \Sp_2\lr{\Z}\backslash \Sp_2\lr{\R}$. 
The trefoil knot is illustrated in fig.~\ref{fig:trefoil} and the reader is referred to refs.~\cite{AMSGhys, Ghys} for further reference. 

Any smooth implementation of a Clifford gate on a GKP code naturally traverses a continuous closed loop in the space of lattices $\Sp_2\lr{\Z}\backslash \Sp_2\lr{\R}$ while implementing a basis transformation. 
The topological defect in this space carved out by the trefoil knot illustrates that such loops are in general homotopically non-trivial.
One way to understand this is through the equivalence $\Sp_2\lr{\Z}\backslash \Sp_2\lr{\R}/\SO_2\lr{\R}=\Sp_2\lr{\Z}\backslash \hh=\CF$. 
The space of lattices, up to a rotation, is labeled by an element in the fundamental domain such that each lattice -- including a rotation label --  can be labeled by a point in the fundamental domain $\CF$ \textit{together} with a rotation label in $S^1$ (which my vary across points in $\CF$). 
In order for a smooth transformation on the space of lattices to return to the same point in the fundamental domain with the same rotation label, it must either map to a $\SL_2\lr{\Z}$ equivalent point in $\hh$, or perform a full rotation in $S^1$. 
This decomposition of $\Sp_2\lr{\Z}\backslash \Sp_2\lr{\R}$ is the so-called \textit{Seifert fibration} \cite{Seifert}, which is illustrated in fig.~\ref{fig:seifert}. In fact, the fundamental group of this space $\Sp_2\lr{\Z}\backslash \Sp_2\lr{\R}$ which we now understand as the homotopy group of the know complement $\pi_1\lr{S^3-K}=B_3=SL_2\lr{\Z}$ is the braid group of three strands \cite{Gannon2023}. To see this in generality, lets return to label the lattice $\Lambda=\omega_1\Z +\omega_2\Z$ by the complex basis $\lr{\omega_1, \omega_2}$ for a minute. Since $\C$ is algebraically closed, the defining equation of the elliptic curve takes the form \cite{Silverman2009}
\begin{align}
    \wp'^2 &= \lr{\wp-e_1}\lr{\wp-e_3}\lr{\wp-e_3},\\
    \Delta&=16\lr{e_1-e_2}^2\lr{e_2-e_3}^2\lr{e_1-e_3}^2 \neq 0,
\end{align}
where $e_1=\wp\lr{\omega_1/2}, e_2=\wp\lr{\omega_2/2}$ and $e_3=\wp\lr{\lr{\omega_1+\omega_2}/2}$ with $e_1+e_2+e_3=0$ form the three distinct roots of the equation $\wp'=0$. Since $\wp\lr{z+\Lambda}=\wp\lr{z}$ is defined modulo the lattice and the coefficients $g_2, g_3$ as well as the lattice are uniquely determined by the roots $e_1, e_2, e_3$, any smoothly parametrized basis transformation can also be identified by the evolution $t: [0,1]\rightarrow e_1(t), e_2(t), e_3(t)$, which smoothly implements a permutation of the three roots. Away from the trefoil defect $\Delta=0$, the position of these roots on $\C/\Lambda$ remain distinct along the path, such that every non-trivial basis transformation implemented in this fashion can be identified with a non-trivial element in the braid group of three strands $B_3$ which has a representation in $\SL_2\lr{\Z}=\langle T, T^{-T} \rangle$ \cite{Gannon2023}. 

This shows that every smoothly parametrized logical non-trivial Clifford gate for the single mode GKP code -- given by a closed loop in the knot complement $S^3-K$ -- necessarily implements a homotopically non-trivial element in this space, i.e.\ it implements a nontrivial link with the cut-out trefoil knot as it avoids the ``zero-distance defect" provided by the knot along the path. 
The braids induced by a rotation and a sheer on the square lattice -- corresponding to a logical Hadamard- and phase gate for the square GKP code with $d=2$ -- are pictured in fig.~\ref{fig:braid}.
Note that the reverse is not generally true; there are nontrivial basis transformations of GKP lattices that implement a \textit{trivial} Clifford element, such as the double application of the Hadamard gate for a $d=2$ square GKP code (compare to fig.~\ref{fig:braid}). 

\begin{figure}
    \centering
    \center
    \includegraphics[width=.6\textwidth]{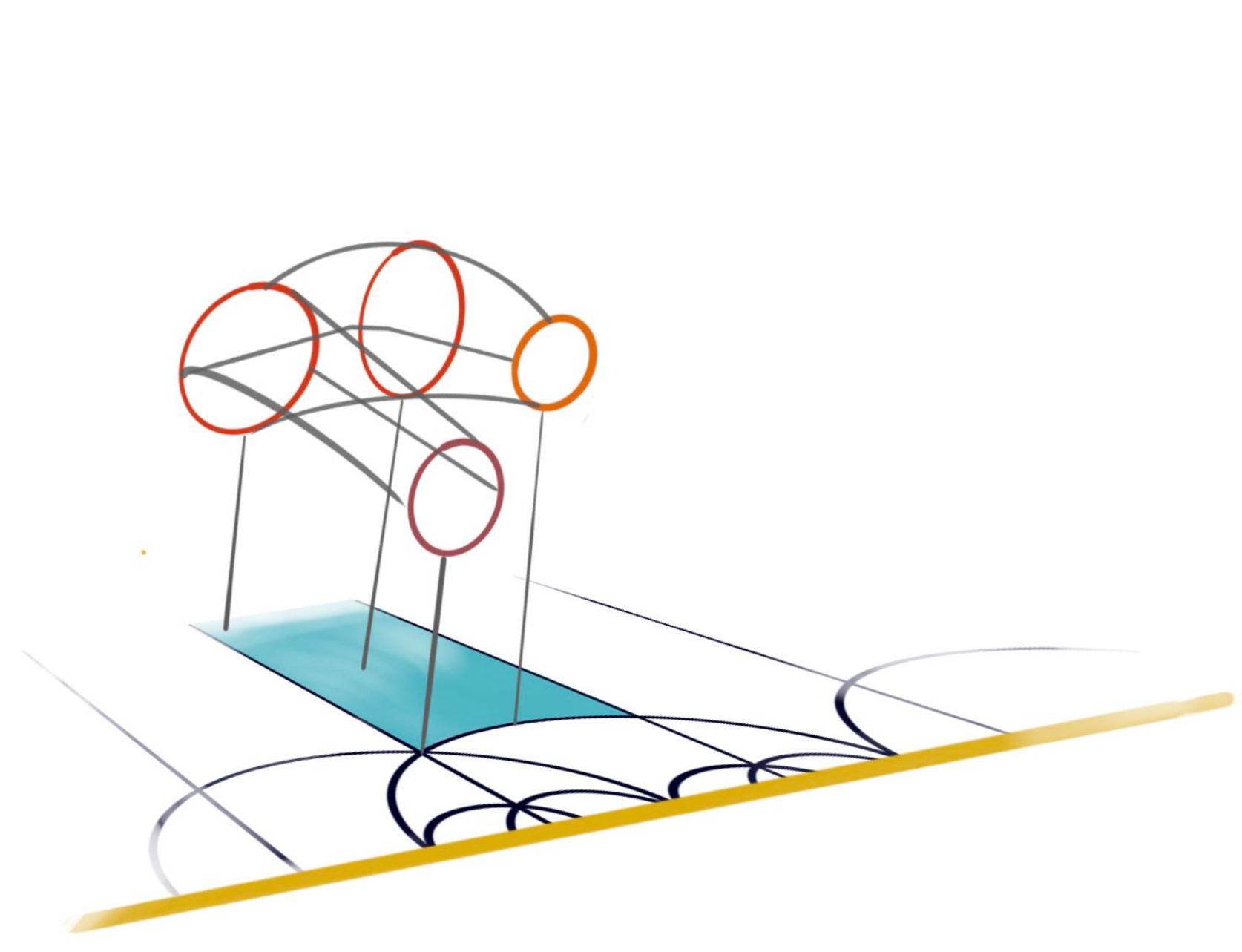}
    \caption{The Seifert fibration describing a decomposition of $\Sp_2\lr{\Z}\backslash \Sp_2\lr{\R}$ into the fundamental domain $\CF$ and a rotation label in $S^1$ for each point in $\CF$. While every lattice has a $\pi$-rotation symmetry, there are special (singular) points $i$ and $\rho=e^{i2\pi/3}$ with additional symmetries under $\pi/2$  and $\pi/3$ rotation. This can be pictured by a smaller circumference rotation index attached to these points in the fibration. In terms of GKP codes, it is these singular points on $\CF$ that correspond to GKP codes with orthogonal symplectic lattice automorphisms.}
    \label{fig:seifert}
\end{figure}

\begin{figure}
    \centering
    \center
    \includegraphics[width=.6\textwidth]{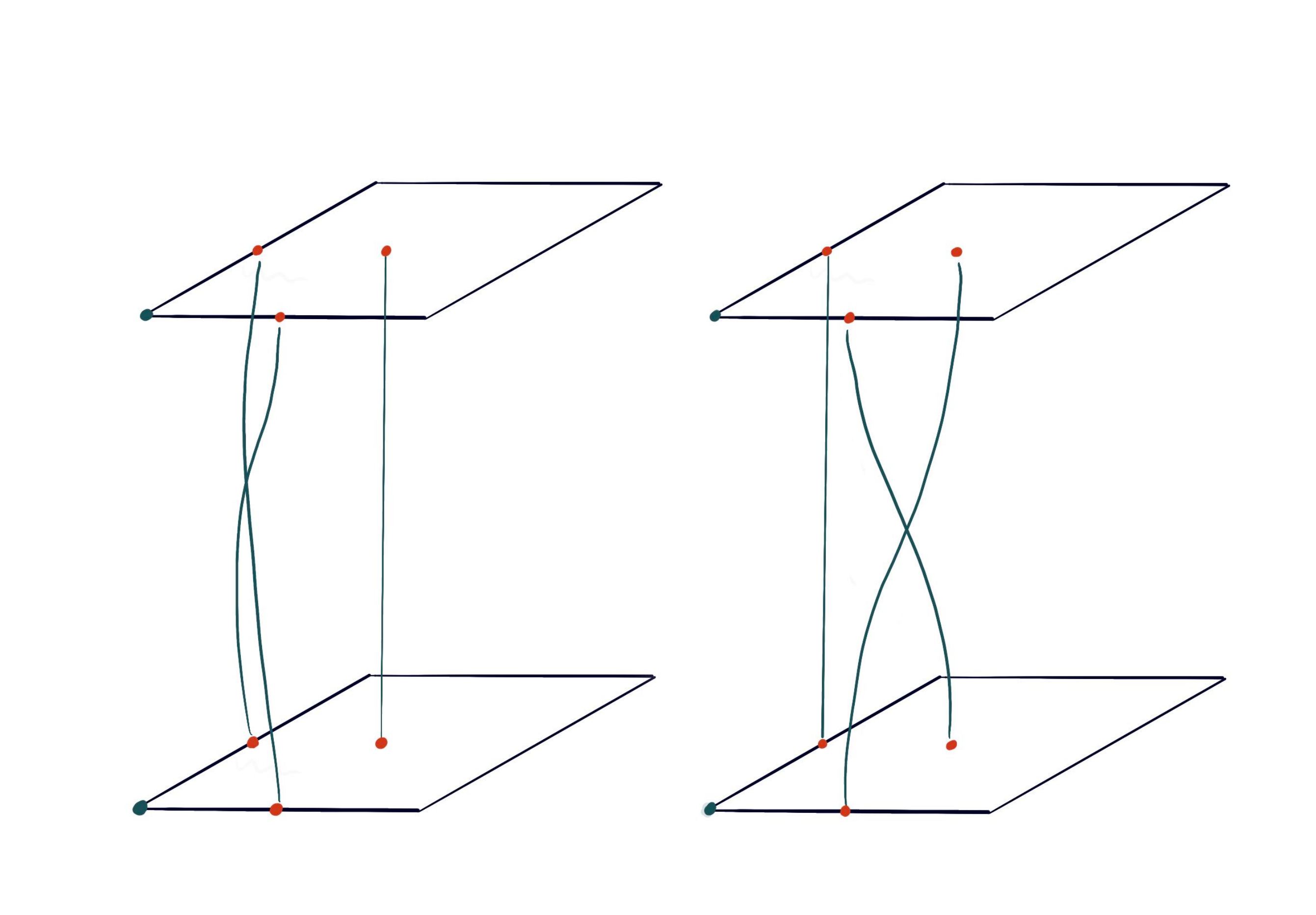}
    \caption{A braid on $e_1=\wp\lr{\omega_1/2}, e_2=\wp\lr{\omega_2/2}, e_3=\wp\lr{\lr{\omega_1+\omega_2}/2}$ implemented through a rotation $(\omega_1, \omega_2)\mapsto (\omega_2, -\omega_1)$ (l.) corresponding to a GKP Hadamard gate and a sheer $(\omega_1, \omega_2)\mapsto (\omega_1+\omega_2, \omega_2)$ (r.)  in the case of $d=2$. Nontrivial braiding of the three roots $e_i$ is induced by smoothly parametrized automorphisms of the underlying lattice.}
    \label{fig:braid}
\end{figure}

In fact, one can define a linking number with the trefoil knot for paths in $\Sp_2\lr{\R}$ which correspond to GKP logical Cliffords. This is done by realizing that one can define a discriminant function $\tilde{\Delta}:\Sp_2\lr{\Z}\backslash \Sp_2{\R} \rightarrow \C^{\times}$ which provides an isomorphism of the homology groups $H_1\lr{\Sp_2\lr{\Z}\backslash \Sp_2{\R}, \, \Z} \sim H_1\lr{\C^{\times}, \Z}$ ~\cite{Duke, Ueki}, such that closed loops in the space of symplectic lattices map to closed loops in $\C^{\times}$. The discriminant function $\tilde{\Delta}$ is  an invariant of the associated lattice, independent of the choice of basis (i.e. it is a weight-$0$ modular form), defined  for $\gamma=\begin{pmatrix}
a & b \\ c & d
\end{pmatrix} \in \Sp_2\lr{\R}$
as
\begin{equation}
\tilde{\Delta}\lr{\gamma}=j_{12}(\gamma,i)\Delta(\gamma.i),
\end{equation}
where we have defined the factor of automorphy
$j_{12}(\gamma,z)=(cz+d)^{-12}$. 
Now let $\gamma_A(t): [0,1]\rightarrow \Sp_2{\R}$ be a continuous curve with $g_A(0)\in \Sp_2\lr{\R}$ and $g_A(1)=Ag_A(0) \, A\in \Sp_2\lr{\Z}$. The linking number is defined by
\begin{equation}
    \mathrm{link} \lr{\gamma_A, K}=\frac{1}{2\pi i}\oint_{\gamma_A} \frac{\mathrm{d}\tilde{\Delta}}{\tilde{\Delta}}=\frac{1}{2\pi i}\oint_{\gamma_A} \frac{\mathrm{d}\Delta}{\Delta}+\frac{1}{2\pi i}\oint_{\gamma_A} \frac{\mathrm{d}j_{12}}{j_{12}}
    \label{eq:link}
\end{equation}
and is  a topological invariant of the path \cite{Duke,Ueki}.
From the modular transformation behavior $\Delta\lr{c\Lambda}=c^{-12}\Delta\lr{\Lambda}$ it can be shown that rotations of the lattices $\Lambda\rightarrow e^{i\phi}\Lambda, \, \phi\in \lrq{0, \frac{\pi}{k}}$ yield linking numbers $\mathrm{link} \lr{\gamma_A, K}=-\frac{6}{k}$, such that the lattice automorphism of the square lattice given by a $\pi/2$ rotation is associated with linking number ${\rm link}\lr{\gamma_S, K}=-3$, while the $\pi/3$ rotation symmetry of the hexagonal lattice associates with a linking number ${\rm link}\lr{\gamma_S, K}=-2$.

In a seminal paper, Ghys~\cite{Ghys} 
showed that for hyperbolic elements $A\in \Sp_2\lr{\Z}$ (i.e.\ those with $\bigl|\Tr\lrq{A}\bigr|>2$) which are implemented via a symplectic squeezing operation
\begin{equation}
    M\in \SL_2\lr{\R} \mapsto M (\lambda \oplus \lambda^{-1}) = AM,\; \lambda>1,
\end{equation}%
the corresponding unique modular geodesic $\gamma_A$ has linking number $  \rm link \lr{\gamma_A, K}=\psi\lr{A}$ with the trefoil knot, where $\psi(A)$ is the well-known \emph{Rademacher function}, which can be computed by compiling $A$ into a product of integer powers of matrices $R=T$, with $T $ as in eq.~\eqref{eq:ST}, $L=T^T$, such that $A=\prod_{i=1}^N R^{r_i}L^{l_i}$. 
Under this expansion, 
\begin{equation}
\psi(A)=\sum_{i=1}^N r_i - l_i  \label{eq:RL}
\end{equation}
is given by the difference of their number of appearances in the product expansion. 
The Rademacher symbol is a class invariant for $\Sp_2\lr{\Z}$, that is for all $g\in \Sp_2\lr{\Z}$ and $A\in \Sp_2\lr{\Z}$ it holds that $\psi\lr{gAg^{-1}}=\psi\lr{A}$. 
In fact, using 
\begin{equation}
R^{r_i}=\begin{pmatrix}
1 & 1 \\ 0 & 1
\end{pmatrix}^{r_i}=\begin{pmatrix}
1 & r_i \\ 0 & 1
\end{pmatrix}, \,
L^{l_i}=\begin{pmatrix}
1 & 0 \\ 1 & 1
\end{pmatrix}^{l_i}=\begin{pmatrix}
1 & 0 \\ l_i & 1
\end{pmatrix}
\end{equation}
the Rademacher symbol also descends to a class invariant on $\SL_2\lr{\Z_d}$ for $A: \,\det\lr{A\!\mod d}=1$ and $\bigl|\Tr\lrq{A} \!\mod d\bigr|>2$, with
\begin{equation}
\psi\lr{A} \!\!\mod d =  \psi\bigl(A \!\!\!\mod d\bigr). 
\label{eq:Rademachermodq}
\end{equation}

The Rademacher function in particular yields a meaningful invariant for symplectic lattice automorphisms provided by a symmetric symplectic matrix. 
In this case, the Bloch-Messiah decomposition  provides a decomposition $S=O^TDO$ of the symplectic matrix into orthogonal symplectic parts $O$ and a squeezing matrix $D$, such that a smoothly parametrized implementation of $S$ can be obtained by concatenating paths in $\Sp_2\lr{\R}$ that implement $O, O^T$ and $D$, respectively. 
The automorphism $S\CL=\CL$ descends to a squeezing automorphism $D$ on the rotated lattice $O\CL$ for which the Rademacher function measures the linking number. 

\subsubsection{The garden of GKP codes}

So far we have identified scaled single-mode GKP codes with elliptic curves with level-$d$ structure and identified the topological defect in the space of all lattices corresponding to such codes with the limit of GKP codes with distance $\Delta_{\rm GKP}=0$ and we have shown how logical Clifford gates quantified by their lattice automorphisms modulo $d$ can be classified according to their linking number with this defect. 
Similar to the case of generic elliptic curves, isomorphism classes of elliptic curves with level structure are classified by the quotient space $\mathcal{M}_1\lrq{d}=\Gamma\lr{d}\backslash \hh$, where we define the congruence subgroup by
\begin{align}
\Gamma(d)&=\lrc{A\in \PSp_2\lr{\Z}, \; A\!\mod d = I } \\
&\subseteq \Gamma(1):=\PSL_2\lr{\Z} . 
\end{align}
Note that here we have defined $\Gamma(d)$ as subgroup of $\PSp_2\lr{\Z}=\SL_2\lr{\Z}/\lrc{\pm I}$ so that $\Gamma(d)$ is torsion free for all $d>1$ and has an effective action on $\hh$.

Define the action of $(m, n ) \in \Z^2$ on $\lr{\tau, z}$ as $\lr{\tau, z+m+n\tau} $, such that the quotient $ \Z^2\backslash\lr{\tau, z}$ is translation symmetric under translations of $z$ by elements in $\Lambda_{\tau}$ and define the $\Gamma\lr{d}$ action as 
\begin{equation}
\gamma:\; (\tau, z) \mapsto \lr{\gamma.\tau, z/(c\tau+d)}
\end{equation}
for $\gamma=\begin{pmatrix}
a & b \\ c & d
\end{pmatrix} \in \Gamma(d)$.
See ref.~\cite{hain2014lectures} for further background.

Action of elements in $\Gamma(d)$ preserve the level-$d$ structure $ d^{-1}\Lambda_{\tau}$ sitting inside of $\C/\Lambda_{\tau}$ and hence represent logically trivial basis transformations of the GKP code.

We assemble the full space of elliptic curves with level-$d$ structure (single-mode GKP codes) as
\begin{equation}
    E\lr{d}=\lr{\Gamma(d) \ltimes \Z^2}\backslash \hh \times \C.
\end{equation}
Understood as GKP codes, this space labels all possible lattices associated with GKP stabilizer groups, i.e.\ that are sublattices to its own symplectic dual (given by its rescaling by $d$), and for each lattice there the element $z$ labels all possible syndromes and logical displacements. 
With our definition, $\Gamma(d)$ is torsion free for $d\geq 2$ such that $\Gamma(d)\backslash \hh$ obtains the structure of a Riemann surface \cite{hain2014lectures}. 

\begin{figure}
\center
\includegraphics[width=.6\columnwidth]{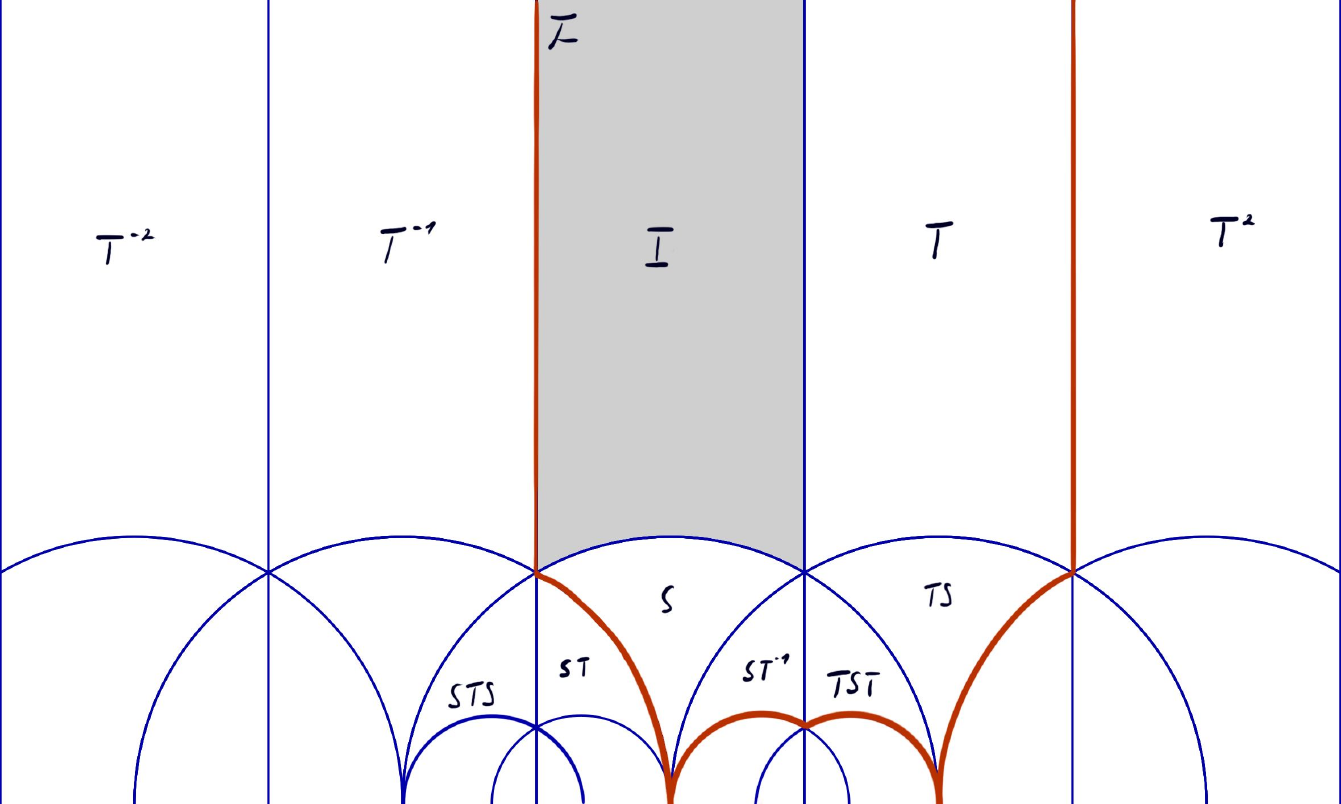}
\caption{The fundamental region $F(2)=\Gamma(2)\backslash\hh = \cup_{\gamma \in \Sp_2\lr{\Z_2}} \gamma F$ is drawn in red and contains the logical Clifford translates of the fundamental region $F=\SL_2{\Z}\backslash \hh$. }\label{eq:F2}
\end{figure}

Finally, we define
\begin{equation}
    E^{\times}\lr{q}=\lr{\Gamma(d) \ltimes \Z^2}\backslash \hh \times_{\tau} \C_d^{\times},\label{eq:Ex}
\end{equation}
where $ \hh \times_{\tau} \C_d^{\times}$ is such that for each point  $\tau \in \hh$, the points $d^{-1}\Lambda_{\tau}$ are removed from the $\C$ factor. 
We define the covering for $d\geq 2$ 
\begin{equation}
    E^{\times}(d) \xrightarrow{\pi} M^{\times}=  \lr{\Sp_2\lr{\Z_d} \ltimes  \lr{\frac{1}{d} \Z_d}^2} \big\backslash E^{\times}(d)\,.
    \label{eq:EMspace}
\end{equation}

The spaces $E^{\times}(d)$ and $M^{\times}$ both have the structure of complex manifolds since the covering group $G=  \lr{\Sp_2\lr{\Z_d} \ltimes  \lr{\frac{1}{d} \Z_d}^2} $ acts freely and properly discontinuously on $ E^{\times}(d)$. 
$E^{\times}(d) \xrightarrow{\pi} M^{\times}$ is a $G$-covering of complex manifolds with the discrete structure group $G$. 
In this construction, we have chosen to exclude the zero section $z=0$ from the space of elliptic curves and its quotients since otherwise 

$M^{\times}$ would not inherit the structure of a complex manifold -- our construction considers only \emph{GKP codes with non-zero syndrome}. 
If we had not excluded these sections, the existence of non-trivial fixed points of the $\Sp_2\lr{\Z}$ action on $\hh$ would prevent the quotients under group action in eqs.~\eqref{eq:Ex}, \eqref{eq:EMspace} to retain manifold structure but allows 
different points on these spaces to retain indeterminacy up to local symmetry groups (in particular $M^{\times}$ would have the structure of an orbifold, which are locally isomorphic to a quotient of a euclidean space with a group which does not have to be constant \cite{caramello2022introduction, hain2014lectures}). 
The family of GKP codes with non-zero syndrome in eq.~\eqref{eq:EMspace} is \emph{universal} \cite{hain2014lectures}, such that every family of GKP codes $E^{\times}\rightarrow M^{\times}$ with non-zero syndrome parametrized over a complex manifold $B$ can be obtained as the pullback of the holomorphic function $\Phi: \, M^{\times} \rightarrow B $ \cite{hain2014lectures} that describes the embedding of $B$ in $M^{\times} $. 
We summarize this property in fig.~\ref{fig:universal_family}.

\begin{figure}
    \centering
    \center
    \includegraphics{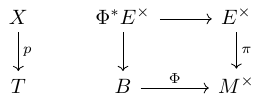}
    \caption{$E^{\times} \rightarrow M^{\times}$ forms a universal family of GKP codes, such that every family of single mode GKP codes with non-zero syndrome can be obtained as pullback of this family. The manifolds $E^{\times}=E^{\times}(d)$, $M^{\times}=M^{\times}(d)$ implicitly depend on the scaling parameter $d$. }
    \label{fig:universal_family}
\end{figure}

\subsection{Towards fiber bundle fault tolerance}

\begin{figure}
\center
\includegraphics[scale=.2]{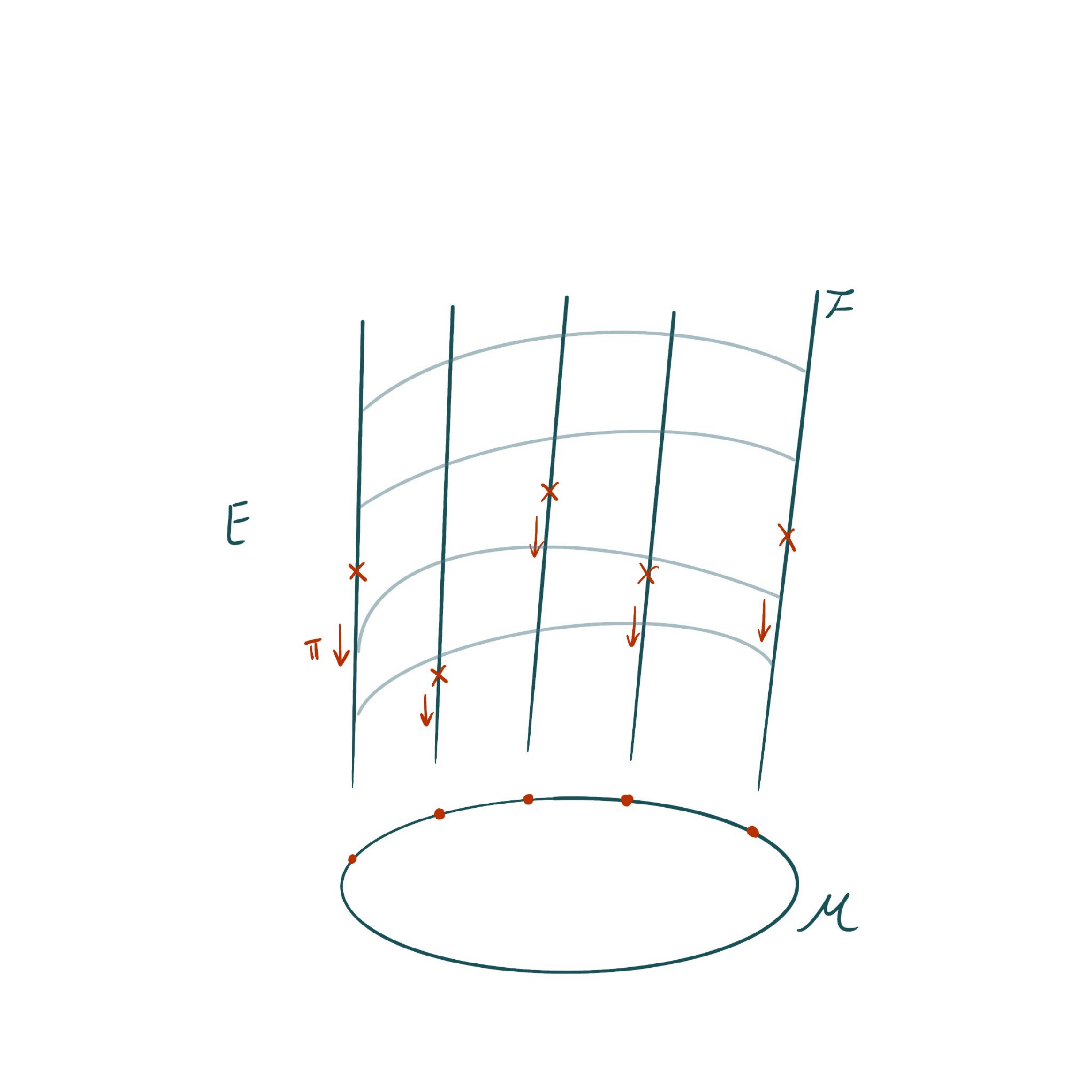}
\caption{A fiber bundle $\pi: E\rightarrow M$. The fibers $F=\pi^{-1}\lr{p}$ are all equivalent to each other. In the background the additional structure is indicated that defines a connection, which is a choice of parallel transport through the tangent of the total space $TE=TH \oplus TV$ and implies a distinguished horizontal- and vertical component of every element in $TE$.}\label{fig:fibrebundle}
\end{figure}

\begin{mybox}
\subsubsection{What is... a fiber bundle?}

A fundamentally important concept in math and physics is that of a fiber bundle, which we briefly sketch while we have already seen many examples throughout this thesis. A fiber bundle is a manifold that locally looks like a product space. Formally, this is provided by a projection from a \textit{total space} down to the \textit{base space} 
\begin{equation}
\pi:\, E\rightarrow M,
\end{equation}
which is such that for every open neighborhood $U\subseteq M$ the preimage  of the projection admits a \textit{local trivialization} 
\begin{equation}
\phi:\, \pi^{-1}\lr{U} \rightarrow U\times F,
\end{equation}
provided by a homeomorphism $\phi$.

One example for a fiber bundle is for instance the \textit{tangent bundle} given by the union of all tangent spaces to a base manifold $M$,
\begin{equation}
TM=\bigcup_{x\in M} T_xM.
\end{equation}

Tangent bundles are ubiquitous in physics: in general relativity, space-time has the structure of a tangent bundle, where (away from extremely dense masses) locally space-time looks like Minkowski space-time. The surface of the sphere, locally looks like its euclidean tangent space and the tangent bundle describes all of these local perspectives. 
Another example it that of the \textit{cotangent bundle}
\begin{equation}
TM^*=\bigcup_{x\in M} T^*_xM,
\end{equation}
which is simply the union of all linear maps from tangent spaces to numbers, e.g. $\R$. While the velocity of points moving along trajectories on a given manifold $M$ is always a tangent vector $\dot{q}\in T_xM$, the canonical momentum $\partial_{\dot{q}}\CL\lr{q, \dot{q}, t}$ derived from a Lagrangian is an element of its cotangent space, such that the classical phase space encountered in the introduction obtains the structure of a cotangent bundle. 
For a more formal treatment see e.g. ref.~\cite{Nakahara}
\end{mybox}

\begin{mybox}
\subsubsection{What is... a connection?}
As a manifold itself, the total space for a given fiber bundle possesses its own tangent bundle $TE$. A \textit{connection} on a fiber bundle is a choice of local decomposition of that tangent space into a \textit{horizontal} and \textit{vertical} component, which is tangent to the fiber $F$. This decomposition allows to compare vectors in $TE=TH \oplus TV$ when pushed around over paths in $M$. For example, in geometric quantization the base manifold is the classical phase space $\R^2$, and we construct a complex \textit{line bundle} with $F\sim \C$ over phase space given by functions representing elements in the quantum Hilbert space where the connection determines how the phase changes when the state is pushed around phase space. As displacement operators do not commute, we know that transporting a state around a loop $D^{\dagger}\lr{\bs{\xi}}D^{\dagger}\lr{\bs{\eta}}D\lr{\bs{\xi}}D\lr{\bs{\eta}}$ yields a phase factor $e^{-i2\pi \bs{\xi}^TJ\bs{\eta}}$ on the fibers. This phase factor arises as \textit{curvature} of the line bundle due to a non-trivial connection. We say a connection is \textit{flat}, if the only way that a non-trivial vertical component is realized by transporting an element around a loop in base space is that this loop is homotopically non-trivial. Again this description has only been very superficial, and I advise the reader to familiarize themselves with more extensive literature, e.g. ref.~\cite{Nakahara}, to obtain a clearer picture.
\end{mybox}

\begin{figure}
    \centering
    \includegraphics[width=\textwidth]{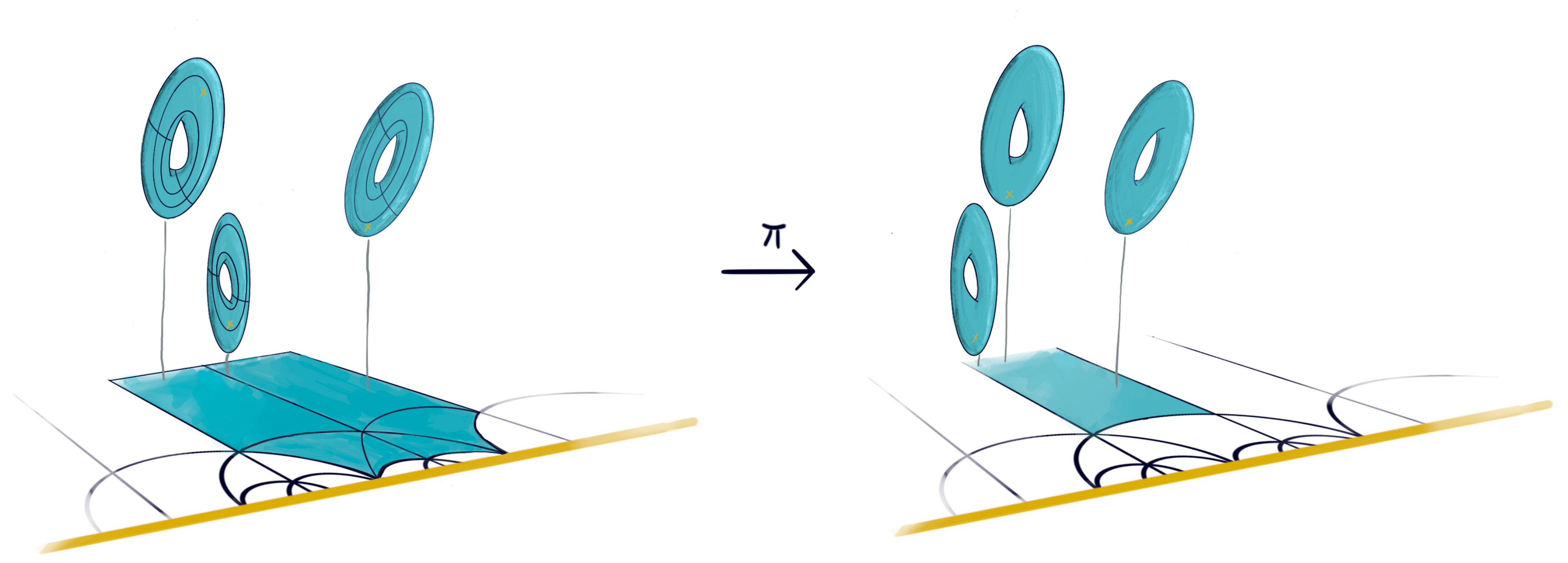}
    \caption{Illustration of the moduli space of GKP codes $M^{\times}$.}
    \label{fig:MX}
\end{figure}
A geometric framework for fault-tolerant gates was proposed by Gottesmann and Zhang, which we discuss very briefly while referring the interested reader to their detailed treatment in ref.~\cite{gottesman2017fiber}. 
This framework considers as fundamental object the Grassmanian $\text{Gr}(K, N)$, the manifold of $K$-dimensional subspaces of an $N$-dimensional Hilbert space. 
Fault-tolerant gates for a code $\CC$ then correspond to homotopically non-trivial loops on a submanifold $\mathcal{M} \subset \text{Gr}(K, N)$ based at $\CC \in \mathcal{M}.$ The manifold $\CM$ is constructed such that every subspace contained within it is an error correction code and thus has some robustness to errors. More concretely one can construct a \textit{vector bundle} over any submanifold of the Grassmanian, whose fiber over a point is the respective codespace. A set of unitary operators is then called fault-tolerant if its left action induces a flat projective connection on this vector bundle.
A fault-tolerant logical gate implemented by a loop based at the code $\CC$ is then determined by the parallel transport of the connection along the path. The connection's flatness implies that non-trivial logical transformations
are necessarily only implemented by loops on $\mathcal{M}$  with non-trivial homotopy. 

The covering space structure of the family of GKP codes discussed above has a similar structure. Taking the place of $\CM$ in Gottesman and Zhang's construction, we consider the set of \textit{subspaces of GKP codes} of a quantum harmonic oscillator and the fibers are given by logical Clifford orbits of the local codes and  choices of syndrome sector. Since these fibers are by construction discrete, paths on the base space have unique lifts to the total space while any smooth path on $E^{\times}(d)$ that implements a non-trivial logical Clifford gate necessarily corresponds to a homotopically non-trivial loop on the base space $M^{\times}$. This base space inherits the topology of the knot complement $S^3-K$ together with that of a torus at each point of the knot-complement. To connect this space $M^{\times}$ to a more elementary decomposition of the space of possible lattices, note that $\CF=\Sp_2\lr{\Z}\backslash \hh$ is one-to-one with $\Sp_2\lr{\Z}\backslash\Sp_2\lr{\R}/\SO_2\lr{\R}$, the space of $2-$dimensional lattices \textit{up to a rotation}, and one can think of the phase of the argument $z\neq 0$ as the label for the corresponding rotation. More concretely, $z\in S_{1,1}$ lives in a \textit{punctured torus}, which has homotopy group $\pi_1\lr{S_{1,1}}=\pi_1\lr{S^1} \times \pi_1\lr{S_{1}}$, equivalent to that of a circle and a torus which captures non-trivial rotations of elements in $M^{\times}$ as non-trivial elements in $\pi_1\lr{S_{1,1}}$. 

The bundle obtained via the forgetful map $M^{\times}\ni (\tau, z)\rightarrow \tau$ can thus be topologically understood as the so-called Seifert fibration \cite{Seifert}, that associates to each element in $\Sp_2\lr{\Z}\backslash\Sp_2\lr{\R}$ an element in $\CF \times S^1$. From the previous discussions we see that the fundamental group of this space has a homomorphism to the single mode GKP Clifford group. It is in this sense, that our construction presented here provides an example of fiber-bundle fault tolerance for the GKP code.

\section{The Dream}

In this chapter we have extensively developed the coding theory of GKP codes, ranging from basic coding theoretic properties of GKP codes and parameter tradeoffs, over proofs of existence of good GKP codes via randomized constructions to the development of an algebraic geometric formulation of the space of (single mode) GKP codes that we have shown to capture fault-tolerance properties in the still underdeveloped, but potentially extremely powerful, fiber bundle framework for fault tolerance by Gottesman and Zhang \cite{gottesman2017fiber}. There are very many dreams that one could formulate building on the work presented in this chapter. The arguable grandest of those would be to extend the analysis of the moduli space of GKP codes to include GKP codes beyond a single mode.  As GKP codes, via concatenation, allow to embed any qubit-based quantum error correcting code into a lattice in continuous space, it would be very interesting to understand the geometry of the moduli space of those lattices and how exactly relevant code properties like the LDPC property -- the characteristic of codes for which there exists a generating set in which stabilizer generators only act on a bounded number of qubits and every qubit is only acted upon by a bounded number of stabilizers -- and the distance of the qubit-based code appear. It is expected that again those spaces can be described by the moduli spaces of complex Abelian varieties \cite{Birkenhake_2004} but now may carry a topological defect of higher dimension. Understanding this structure, whether a non-trivial \textit{systole} exists there and whether such characteristics can be related to fault-tolerance properties and, in general, the development of the theory of fiber bundle fault tolerance for arbitrary qubit-based stabilizer codes \textit{through the lens} of the GKP code is an exciting question to be tackled.
\chapter{Constructing GKP codes}\label{chap:constructions}%
\blfootnote{The content of this chapter is oriented along the publications ref.~\cite{Conrad_2022} and ref.~\cite{conrad2023good}. In particular sec. ~\ref{sec:constructions} is largely adapted from ref.~\cite{Conrad_2022} and the content of sec.~\ref{sec:NTRU} is adapted from ref.~\cite{conrad2023good}.}
Now that the coding theory of GKP codes is understood, we move to discuss constructions of GKP codes by identifying interesting symplectically integral lattices. A list of known symplectically integral lattices is displayed in fig. ~\ref{fig:table_lattices}, which features symplectically self-dual lattices originally identified by Harrington in ref.~\cite{Harrington_Thesis}, GKP codes obtained by concatenation with qubit-based quantum error correcting codes and a new class of GKP codes, dubbed NTRU-GKP codes, that we will discuss in the course of this chapter. In general, the go-to references for lattice theory and general constructions is ref.~\cite{ConwaySloane}, which I refer to for general inspiration and an overview of many important lattice families and properties not explained here. 

\begin{figure}
\center
\resizebox{\textwidth}{!}{
\begin{tabular}{c|c|c| c|c | c|c}
$n$ 	&	$\dim\lr{\CL_0} \,\lr{\CL}$				& $\CL_0$ 									& $\lr{\lambda_1\lr{\CL}}^2$	 		& Symp. ~self-dual										& Eucl.~self-dual 		& Concatenated (trivial sublattice)  \\ \hline
$1$ 	&	$2$									& $\mathbb{Z}^2$						& $1$ 												&\checkmark														& \checkmark						& -- \\
$1$	&	$2$									& $A_2$										&	$\frac{2}{\sqrt{3}}$						&\checkmark														& \checkmark						&	-- \\
$2$ 	& $4$									& $D_4$				& $\sqrt{2}$												&\checkmark \cite{Harrington_Thesis}								& \checkmark			&  $\CL_{\rm triv}\sim \Z^4$ w/ repetition code \cite{Royer_2022}\\ 
$4$	& $8$									& $E_8$										&$2$												&\checkmark														&\checkmark						& $\CL_{\rm triv}\sim 2\Z^8$ w/ Hamming code \cite{ConwaySloane} $\mathcal{H}_8 =[8,4,4]$\\
$6$	& $12$									& $K_{12}$				&$\frac{4}{\sqrt{3}}$	\cite{Harrington_Thesis}					&\checkmark \cite{Harrington_Thesis}	&\checkmark					& $ \CL_{\rm triv}\sim A_2^6$ \cite{Conway_Voroni}	\\
$12$&$24$									&$\Lambda_{24}$	& 4 \cite{ConwaySloane}			&\checkmark \cite{Sarnak1994}				&\checkmark					&$\CL_{\rm triv}\sim 2\Z^8$ w/ Golay code$^*$ \cite{ConwaySloane} $\mathcal{C}_{24}=[24,12,8]$ \\ \hline
$n$ & $2n$ & $\sqrt{\lambda/ q}L_{\rm NTRU}$ & $\Delta \sim O\lr{\sqrt{n/\lambda}} $ & \checkmark  & \checkmark & $\CL_{\rm triv} \sim \sqrt{\lambda q}\Z^{2n}$ \\ \hline
$N$	&$2N$									&$\Lambda_{\square}\lr{\mathcal{Q}}$	& $\Delta\geq  \sqrt{d/2}$ & x												&x	& $\mathcal{Q}=[\![N,k, d]\!]$\\
$N$	&$2N$									&$\Lambda_{\mhexagon}\lr{\mathcal{Q}}$	& $\Delta=\sqrt{d/\sqrt{3}}$ & x				&x	&  $\mathcal{Q}=[\![N,k, d]\!]$

\end{tabular}}
\caption{Some notable symplectically integral lattices that yield GKP codes. The lower block indicates the concatenation of single mode $\CL_{\square}=\sqrt{2}\mathbb{Z}^2$ square GKP and $\CL_{\mhexagon}=\sqrt{2}A_2$ hexagonal GKP codes with qubit quantum error correcting- or detecting codes. Note that concatenation with $\CL_{\mhexagon}$ does not formally produce a Construction A lattice, but is related by a symplectic transformation $S_{\mhexagon}^n=\oplus_i^n S_{\mhexagon}$, $S_{\mhexagon}=M_{A_2}^T$ to the concatenation with the square GKP code generated by $M_{\mathbb{Z}^2}=I_2$, which in fact is Construction A. The symplectically self-dual root lattices listed in this table and their use as GKP codes have previously been identified in
ref.~\cite{Harrington_Thesis}. The rescaled $L_{\rm NTRU}$ lattices that we use here to to construct NTRU-GKP codes are indicated between those and the ``more genuine'' lattices corresponding to concatenated codes. The statements about (symplectic) self-duality are generally up to scaling and rotations. To avoid confusion with the distance $d$ of a qubit QECC here we used $q$ for the scaling parameter for scaled GKP codes.} \label{fig:table_lattices}
\end{figure}

\section{Root systems and root lattices}
Many of the symplectic self-dual lattices listed in fig.~\ref{fig:table_lattices} are so-called \textit{root lattices}, which have an interesting structure and are of fundamental relevance throughout physics and mathematics, which is why they deserve extra attention. Root lattices play an important role in Coexters classification of reflection groups -- which also has application in the construction of quantum error correcting codes based in regular tesselation of hyperbolic surfaces \cite{breuckmann_thesis, Breuckmann_2016, Conrad_2018} -- and by Witt's classification theorem any integral lattice generated by vectors of norm $1$ or $2$ is an orthogonal sum of lattices isometric to $\Z$, elements of the infinite families $ A_n, D_n$ or the exceptional lattices $E_6, E_7, E_8$ \cite{Martinet2003}, which are pictured in fig.~\ref{fig:Dynkin}

A root lattice $L\in \R^n$ has a basis given by a \textit{root system}, which is a set of vectors $\Phi\subset \R^n={\rm span}\lr{\Phi}$  closed under the reflections

\begin{equation}
r_{\bs{\alpha}}\bs{x}=\lr{I-2 \frac{\bs{\alpha}\bs{\alpha}^T}{\bs{\alpha}^T\bs{\alpha}}}\bs{x},
\end{equation}
and for any two roots $\bs{\alpha}, \bs{\beta}\in \Phi$, $2\bs{\alpha}^T\bs{\beta}/\bs{\alpha}^T\bs{\alpha} \in \Z$ is an integer.

\begin{figure}
\center
\includegraphics[width=.8\textwidth]{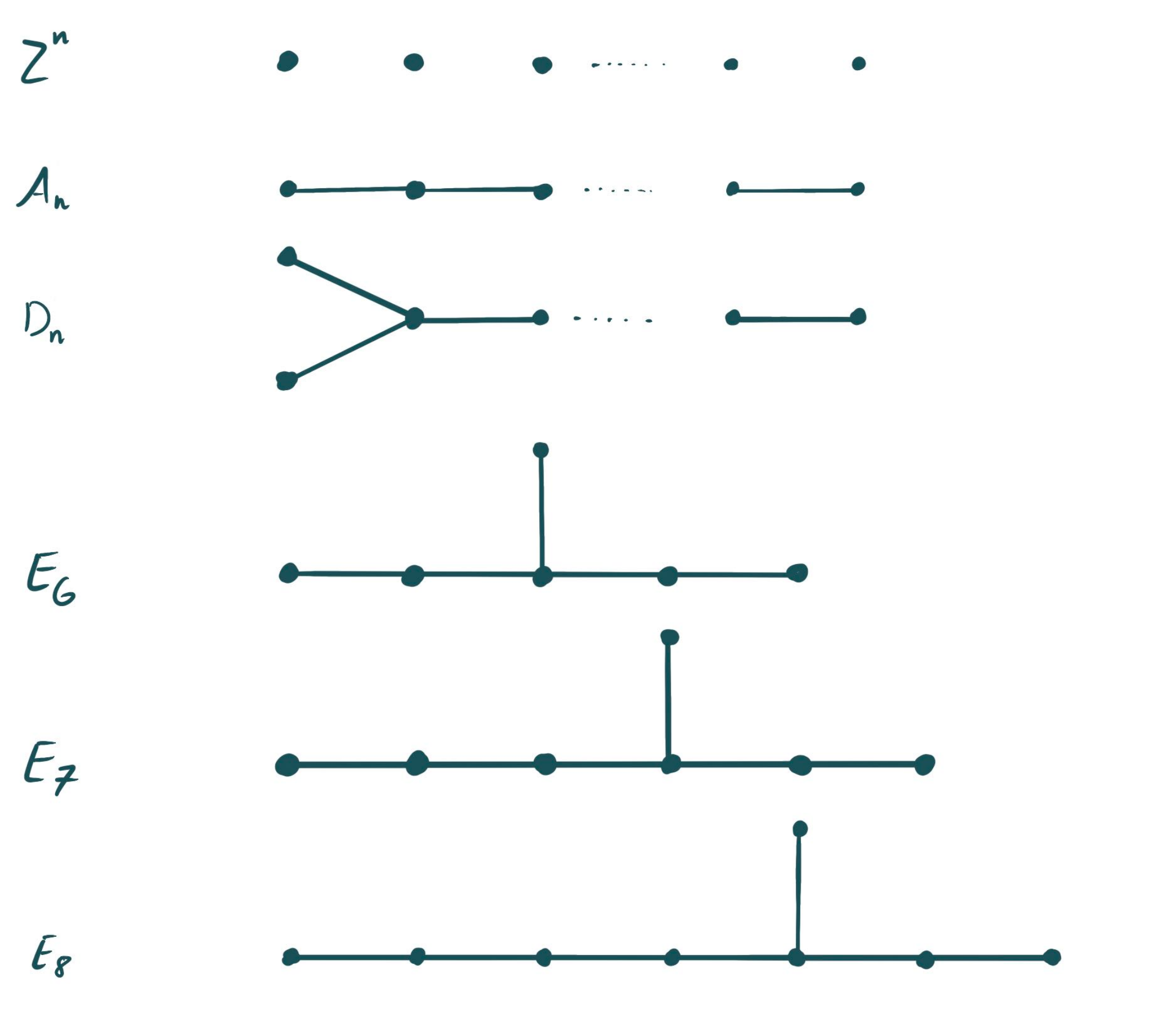}
\caption{Notable root systems represented as Dynkin diagrams. }
\label{fig:Dynkin}
\end{figure}
Root lattices are typically denoted using Dynkin diagrams, which are such that every root is marked by a node and two nodes are connected if they are separated by an angle of $\phi=2\pi/3$. Nodes are not connected if they are orthogonal. In the simple case where we consider roots normalized to length $\|\bs{\alpha}\|^2=2$, this means that two roots $\bs{\alpha}, \bs{\beta}$ share an edge if they have inner product $\bs{\alpha}^T\bs{\beta}=-1$. 

\subsubsection{The $A_n$ lattices}
The family of (alternating) lattices $A_n$ is given by the integer vectors $\bs{x}\in \Z^{n+1}$ that satisfy $\bs{1}^T\bs{x}=\sum_i x_i=0$, i.e. the integers on the plane orthogonal to the all-$1$ vector $\bs{1}\in \Z^n$. For $n=2$ we have already seen the basis for the hexagonal lattice, $A_2$ in eq.~\eqref{eq:GKP_hex}, which is obtained by rotating the integers $\Z^{3}$ into a coordinate system where $\bs{1}=\bs{e}_z$, identifying the points on the $xy$ plane, and rescaling to fix $|\det\lr{M_{A_2}}|=1$. 
Without the rescaling, the $A_n$ lattices have determinant $\det\lr{A_n}=\sqrt{n+1}$ and shortest vector length $\lambda_1^2\lr{A_n}=2$ (for $n\geq 2$).

\subsubsection{The $D_n$ lattices}

The $D_n$ lattices are defined by vectors $\bs{x}\in \Z^n$, such that $\bs{1}^T\bs{x}=0\, \mod 2$, i.e. the coefficients have an even sum. Again (for $n\geq 3$) the shortest vector length is $\lambda_1^2\lr{D_n}=2$ and $\det\lr{D_n}=2$.
A basis for the $D_4$ lattice can be given by
\begin{equation}
M_{D_4}'=\begin{pmatrix} 1 & 1 & 0 & 0  \\ 1 & -1 & 0 & 0  \\ 0 & 1 & -1 & 0  \\ 0 & 0 & 1 & -1  \end{pmatrix}.
\end{equation}
This matrix  however has determinant $|\det\lr{M'_{D_4}}|=2$. By performing a basis transformation using the unimodular matrix (obtained via Gaussian elimination)
\begin{equation}
U=\begin{pmatrix}-1 & 0 & 0 & 0  \\ -1 & 1 & 0 & 0  \\ 0 & 0 & 1 & 0  \\ 1 & 0 & 0 & 1  \end{pmatrix},
\end{equation}
this basis is symplectically diagonalized to yield a symplectic Gram matrix $A=J_2 \otimes {\rm diag}\lr{1, 2}$ and hence naturally encodes a qubit into a collection of $2$ oscillators without rescaling. This observation has also been made in ref.~\cite{Royer_2022}. Refs.~\cite{Harrington_Thesis, Buser_1994} also noted that by performing a suitable rotation, the $D_4$ lattice can also be brought into a symplectically self-dual form. 

\subsubsection{The Gosset $E_8$ lattice}
The $E_8$ lattice is notably one of the most interesting lattices in math and physics. It yields the provably densest lattice packing of spheres in $8$ dimensions \cite{Viazovska_2017} and has many interesting algebraic properties \cite{Garibaldi_2016}. It is also the smallest even self-dual lattice, which can only exist in dimensions $2n= 0 \, \mod 8$ \cite{theta_course}. 
It has shortest vector length $\lambda_1\lr{E_8}=\sqrt{2}$ and $\det\lr{E_8}=1$. The standard basis $M_{E_8}$, which can be read off from the corresponding Dynkin diagram is given below and can be transformed into a symplectic basis with $A=J$ via the unimodular transformation $U$,

\begin{align}
M_{E_8}&=\begin{pmatrix}1 & -1 & 0 & 0 & 0 & 0 & 0 & 0  \\ 0 & 1 & -1 & 0 & 0 & 0 & 0 & 0  \\ 0 & 0 & 1 & -1 & 0 & 0 & 0 & 0  \\ 0 & 0 & 0 & 1 & -1 & 0 & 0 & 0  \\ 0 & 0 & 0 & 0 & 1 & -1 & 0 & 0  \\ 0 & 0 & 0 & 0 & 0 & 1 & -1 & 0  \\ 0 & 0 & 0 & 0 & 0 & 0 & 1 & -1  \\ \frac{1}{2} & \frac{1}{2} & \frac{1}{2} & \frac{1}{2} & \frac{1}{2} & \frac{-1}{2} & \frac{-1}{2} & \frac{-1}{2}  \end{pmatrix}, \\
U&=\begin{pmatrix}0 & 0 & 0 & 1 & 0 & 0 & 0 & 0  \\ 0 & 0 & 0 & 2 & 1 & 0 & 0 & 0  \\ 0 & 0 & 0 & 3 & 2 & 1 & 0 & 0  \\ 1 & 2 & 3 & 4 & 3 & 2 & 1 & 0  \\ 1 & 0 & 0 & 0 & 0 & 0 & 0 & 0  \\ 0 & 1 & 0 & 0 & 0 & 0 & 0 & 0  \\ 0 & 0 & 1 & 0 & 0 & 0 & 0 & 0  \\ 0 & 0 & 0 & 1 & 0 & 0 & 0 & 1  \end{pmatrix}, \nonumber
B_{\mathcal{H}_8}=\begin{pmatrix}0 & 0 & 0 & 0 & 1 & 1 & 1 & 1  \\ 0 & 0 & 1 & 1 & 0 & 0 & 1 & 1  \\ 0 & 1 & 0 & 1 & 0 & 1 & 0 & 1  \\ 1 & 1 & 1 & 1 & 1 & 1 & 1 & 1  \end{pmatrix}.\\
\end{align}

The $E_8$ lattice can also be obtained by applying Construction A to the extended Hamming code which is defined by the basis $B_{\mathcal{H}_8}$ for the corresponding binary code~\cite{ConwaySloane}. 
The group of orthogonal automorphisms of the $E_8$ lattice~\cite{ConwaySloane} is given by the Weyl group $W\lr{E_8}=\langle r_{\bs{\alpha}_1}\hdots r_{\bs{\alpha}_8}\rangle$, generated by all permutations and sign changes of the coordinates together with the (symplectic) matrix $H_4 \oplus H_4$~\cite{ConwaySloane}, where

\begin{equation}
H_4=\frac{1}{2}\begin{pmatrix}
1 & 1 & 1 & 1 \\
1 & -1 & 1 & -1 \\
1 & 1 & -1 & -1 \\
1 & -1 & -1 & 1 
\end{pmatrix}.
\end{equation}
The full Weyl group is of order $|W\lr{E_8}|=4! 6! 8!$ and is claimed to contain a subgroup of symplectic orthogonal automorphisms of order $46080$ in ref.~\cite{Buser_1994}.

\section{Direct sum and product constructions beyond concatenation}\label{sec:constructions}
In this section we discuss two constructions of GKP codes that go beyond the scaled and concatenated GKP codes that are based on lattice glueing and the lattice tensor product. The following constructions are adapted from the discussion of glued (euclidean) lattices and the lattice tensor product in refs.~\cite{ConwaySloane, gannon_thesis}.

\subsection{Glued codes beyond concatenation}
We can understand concatenated GKP codes as GKP codes built on certain types of \textit{glued lattices}, which include those obtained through Construction A. This discussion leans on the description in ref.~\cite{ConwaySloane}. We begin by dissecting a (glued) symplectic lattice with symplectic sublattice in a top-down approach to understand its structure and then move to a bottom-up approach to construct a glued lattice from a base lattice by appending an appropriate glue group. 
Let us assume that we have a $2n-$dimensional symplectic lattice $\mathcal{L}$ that has a (symplectic) sublattice $\CL_0$ with direct sum structure 
\begin{equation}
\CL_0=\bigoplus_{i=1}^k \mathcal{L}_i.
\end{equation} 
Vectors  $\bs{v} \in \mathcal{L}$ can be written as 

\begin{equation}\bs{v} =\sum_i \bs{v} _i,\label{eq:vGlue}
\end{equation} 
where $\bs{v} _i \in  \mathbb{R} \otimes \mathcal{L}_i$. The symplectic inner product of any $\bs{v} _i$ with any vector of $\CL_i$ is integer, such that it can be concluded that $\bs{v} _i\in \mathcal{L}_i^{\perp}$. Moreover we can add to any $\bs{v} _i$ a vector from $\mathcal{L}_i$ without changing the fact that $\bs{v}$  has integer symplectic inner product with any other vector of $\CL$. It hence suffices to demand $\bs{v} _i \in \mathcal{L}_i^{\perp}/\mathcal{L}_i$. Such vectors are called \textit{glue vectors} for $\mathcal{L}_i$, which in the coding language correspond to logical representatives of a local code. $\mathcal{L}_i^{\perp}/\mathcal{L}_i$ is also known as the (symplectic) dual quotient or \textit{glue group} for $\mathcal{L}_i$. We can thus obtain symplectic lattices from a \textit{base lattice} $\CL_0=\bigoplus_{i=1}^k \mathcal{L}_i$ by adding  vectors $\bs{v}$ of the form in eq.~\eqref{eq:vGlue}, where each $\bs{v_i}\in \mathcal{L}_i^{\perp}/\mathcal{L}_i$. We also refer to the set of extra vectors $\bs{v}$ as the \textit{glue group} $G$.
Let $\mathcal{L}= \bigoplus_{i=1}^k \mathcal{L}_i \cup G$ be a glued lattice of this form. it also holds that $\mathcal{L}^{\perp} = \bigoplus_{i=1}^k \mathcal{L}_i^{\perp} \cap G^{\perp}$.

Generally, whenever we have \textit{saturated} sublattices $\mathcal{L}_0\subset \mathcal{L}$, i.e., $\mathbb{R}\otimes\mathcal{L}_0=\mathbb{R}\otimes\mathcal{L}$, we have $|\mathcal{L}_0|/|\mathcal{L}|\in\mathbb{Z}$. For $g\in \mathcal{L}$ we consider the glue classes $[g]=g+\mathcal{L}_0$, which form additive groups that we denote by
\begin{equation}
G=\langle g_1,\dots, g_r\rangle. 
\end{equation} 
The glue group contains the vectors added to $\CL_0$ in order to obtain the glued lattice.
Conversely, in a bottom-up approach, a glued lattice $\mathcal{L}$ can be constructed by considering a general glue group $G\subseteq \CL_0^{\perp}/\CL_0$ and by forming $\mathcal{L}=\mathcal{L}_0[G]=\mathcal{L}_0 \cup G$.  $G$ is cyclic and isomorphic to $\mathbb{Z}_{n_1}\times ..\times \mathbb{Z}_{n_r}$ , where $\mathbb{Z}_n$ is the cyclic group of order $n$, and each $n_i$ is the order of the corresponding generator $g_i$, i.e  the smallest positive integer such that $n_i g_i \in \mathcal{L}_0$ .
The determinant of the glued lattice $\mathcal{L}_0[G]$ can be computed as~\cite{gannon_thesis} %
\begin{equation}
|\mathcal{L}_0[G]|=|\mathcal{L}_0|/|G|=|\mathcal{L}_0|/(\prod_{i=1}^r n_i).\label{det_glued}
\end{equation}%
To construct a \textit{symplectic} glued lattice from a symplectic base lattice $\mathcal{L}_0$, it is important to take care that every $g_i$ has integer symplectic inner product with every other $g_j$ -- i.e., $G$ is itself a finite symplectic group and that each $G$ has integer symplectic inner product with each $x\in\mathcal{L}_0$, i.e., $G\subset \mathcal{L}_0^{\perp}$. It is easy to see that the earlier considerations are reproduced for $\mathcal{L}_0=\bigoplus_{i=1}^k \mathcal{L}_i$. 
Using eq.~\eqref{det_glued} we can obtain the logical  dimension of GKP codes associated to glued lattices $\mathcal{L}_0[G]$. E.g. for a concatenated GKP-qubit code, $G\subset\mathcal{L}_0^{\perp}/\mathcal{L}_0$ is identified with the outer code with, say, $r$ lineary independent generators, each with order $n_i=2$ in $\mathcal{L}_0$. $\mathcal{L}_0=\mathcal{L}(\sqrt{2}I_{2n})$, such that we compute $|\mathcal{L}_0[G]|=2^{n-r}=2^{k}$, consistent with what one would expect.
The GKP code distance corresponding to a code obtained from a glued lattice is however expected to be hard to compute for the same reason that determining the distance of a qubit (qudit) quantum error correcting code is in general hard (see the next chapter for a more in-depth discussion).

\subsection{Tensor product codes}

Aside from the glueing construction, it is also possible to obtain new codes by taking outer products of lattices. The idea behind this construction is akin to product constructions known for qubit quantum error correcting codes, namely the \emph{hypergraph product codes} by Tillich and Zemor~\cite{Tillich_2014}, where the defining structure of the code is a hypergraph, and Homological product codes by Bravyi and Hastings~\cite{HomologicalProduct}, where the code is defined via a cell complex. 
For GKP codes, the defining structure of the codes is given by a lattices, such that the tensor product for lattices serves as an immediate candidate for a similar construction. 

Let $\mathcal{L}_1=\mathcal{L}(M_1)\subset \mathbb{R}^2$ be a symplectic lattice with  symplectic Gram matrix $A_1=M_1J_2M_1^T$ and  $\mathcal{L}_2=\mathcal{L}(M_2)\subset \mathbb{R}^{n}$ an integral lattice with euclidean Gram matrix $G_2=M_2M_2^T$. The tensor product lattice is defined as $\mathcal{L}_{\otimes}=\mathcal{L}(M_1\otimes M_2)=\mathcal{L}_1\otimes \mathcal{L}_2 \subset \mathbb{R}^{2n}$, i.e., a basis for $\mathcal{L}_{\otimes}$ is given by $\lrc{(M_1)_i\otimes  (M_2)_j, \; i=1,2,\; j=1,\hdots,n}$. $\mathcal{L}_{\otimes}$ is a symplectic lattice due to the decomposition $J_{2n}=J_2\otimes I_n$, and its symplectic Gram matrix reads%
\begin{equation}
A_{\otimes}=(M_1\otimes M_2)J_{2n} (M_1\otimes M_2)^T=A_1 \otimes G_2,
\end{equation}%
which is integral by construction. The canonical dual basis is given by
\begin{equation}
M^{\perp}_{\otimes}=(J_{2n}M_{\otimes}^T)^{-1}=M_{\otimes}^{-T}J_{2n}^T=M_1^{-T}J_2^T\otimes M_2^{-T}=M_1^{\perp}\otimes M_2^{*},
\end{equation}
which forms a basis for the symplectically dual lattice $\mathcal{L}_{\otimes}^{\perp}=\mathcal{L}_{1}^{\perp}\otimes \mathcal{L}_{2}^{*}$.
We have 
\begin{align}%
|A_1 \otimes G_2|&=|A_1|^n|G_2|^2,\\
k_{\otimes}=\frac{1}{2}\log_2 |A_1 \otimes G_2| &= \frac{n}{2}\log_2 |A_1| + \log_2 |G_2|.
\end{align}%
\begin{them}[Distance of tensor product codes]
The distance of the tensor product code $$\Delta_{\otimes}=\min_{0\neq x\in \mathcal{L}_{\otimes}^{\perp}/\mathcal{L}_{\otimes}} \|x\|$$ obeys %
\begin{equation}
\max\left\{ \frac{\Delta_1}{\lambda_n(\mathcal{L}_2)}, \frac{\Delta_2}{\lambda_2(\mathcal{L}_1)}\right\} \leq \Delta_{\otimes} \leq \Delta_1\Delta_2,
\end{equation}%
where %
\begin{equation}\Delta_1=\min_{0\neq x\in \mathcal{L}_1^{\perp}/\mathcal{L}_1} \|x\|, \hspace{1cm} \Delta_2=\min_{0\neq x\in \mathcal{L}_2^{*}/\mathcal{L}_2} \|x\|.
\end{equation}%
\end{them} 

\begin{proof}
The proof is analogous to that of Lemma 2 in ref.~\cite{HomologicalProduct}.
To prove the upper bound, let $x\in \mathcal{L}_{1}^{\perp}/\mathcal{L}_{1},\; y\in \mathcal{L}_{2}^{*}/\mathcal{L}_{2}$, be minimal non-trivial  logical representatives of each component codes. It is clear that $x \otimes y \in \mathcal{L}_{\otimes}^{\perp}$. Further we have $x \otimes y \notin \mathcal{L}_{\otimes}$ because we can pick $a \in \mathcal{L}_1^{\perp},\; b \in \mathcal{L}_2^{*}$ such that $a^TJ_2x\in \mathbb{Q}_1,\, b^Ty\in \mathbb{Q}_1$, where $\mathbb{Q}_1=\mathbb{Q}\cap (0,1)$, yielding $(a\otimes b)^T J_{2n}(x\otimes y) \notin \mathbb{Z}$ (note that the symplectic inner product sets the commutation phase for the associated displacement operators). As such, we have  obtained a non-trivial logical operator $x\otimes y \in \mathcal{L}_{\otimes}^{\perp}/\mathcal{L}_{\otimes}$ and $\Delta_{\otimes}\leq \|x\otimes y\|=\Delta_1 \Delta_2$.
Let $0\neq\bs{\psi} \in \mathcal{L}_{\otimes}^{\perp}/\mathcal{L}_{\otimes}$ be a minimal length non-trivial vector. We can always choose $0\neq c\in\mathcal{L}_1^{\perp}/\mathcal{L}_1$ and $d\in \mathcal{L}_2$ such that %
\begin{equation}(c\otimes d)^T (J_2\otimes I_n) \psi\notin\mathbb{Z}.\label{eq:proof_assumption_1}
\end{equation} %
Let  ${\tt vec}^{-1}(\psi)$ be the un-vectorization of $\psi$, i.e., if $\psi=\sum_i x_i \otimes y_i, \, x_i\in \mathcal{L}_1^{\perp},\,y_i \in \mathcal{L}_2^{*}$  we have ${\tt vec}^{-1}(\psi)=\sum_i x_i y_i^T$. It holds that ${\tt vec}^{-1}(\psi)d\in \mathcal{L}_1^{\perp}$ and ${\tt vec}^{-1}(\psi)d \notin \mathcal{L}_1 $ since otherwise $(c\otimes d)^T (J_2\otimes I_n) \psi = c^TJ_2{\tt vec}^{-1}(\psi) d \in\mathbb{Z}$ for all choices of $\bs{c}, \bs{d}$, such that $\psi$ is logically trivial, which is not the case by assumption. Hence, ${\tt vec}^{-1}(\psi)d $ is a non-trivial representative of $\mathcal{L}_1^{\perp}/\mathcal{L}_1$. Using Cauchy-Schwartz we have $\|\psi\|\|d\|=\|{\tt vec}^{-1}(\psi)\|_F\|d\|\geq  \|{\tt vec}^{-1}(\psi)d\|\geq \Delta_1$, where $\|\cdot\|_F$ is the Frobenius norm. 

We can always choose $d\in\mathcal{L}_2$ with length at most the covering radius $\lambda_n(\mathcal{L}_2)$ that satisfies eq.~\eqref{eq:proof_assumption_1}. This is because any $d\in\mathcal{L}_2$  can be written in a basis 
\begin{equation}
d=\sum_{i=1}^n a_i \xi_i, \hspace{.5cm} a_i\in \mathbb{Z},\hspace{.5cm}  \xi_i\in \CL_2 :\; \|\xi_i\|\leq \lambda_n\lr{\CL_2}.
\end{equation}
When $d$ satisfies eq.~\eqref{eq:proof_assumption_1}, we have
\begin{equation}
\mathbb{Z}\not\ni \sum_{i=1}^n (c\otimes d)^T (J_2\otimes I_n) \psi = \sum_{i=1}^n a_i (c\otimes \xi_i)^T (J_2\otimes I_n) \psi.
\end{equation}
There must be at least one summand $i=x$, for which $a_x (c\otimes \xi_x)^T (J_2\otimes I_n) \psi \not\in \mathbb{Z}$. Since $a_x \in \mathbb{Z}$, it must hold that $(c\otimes \xi_x)^T (J_2\otimes I_n) \psi \not\in \mathbb{Z}$. 
Finally, we obtain 
\begin{equation}
\|\psi\|\geq \frac{\Delta_1}{\lambda_n(\mathcal{L}_2)}.
\end{equation}
Following the same procedure, we choose $c\in \mathcal{L}_1,d\in \mathcal{L}_2^{\perp}/\mathcal{L}_2$ such that eq.~\eqref{eq:proof_assumption_1} is satisfied to show that $\|\psi\| \|c\|\geq \Delta_2$.

\end{proof}

\section{GKP codes from the NTRU cryptosystem}\label{sec:NTRU}

 In sec.~\ref{sec:goodrandom} we have seen how \textit{good} GKP codes with average distance scaling $\Delta \propto \sqrt{n}$ and constant encoding rate $\log\lr{\det\lr{\CL}}\propto n$ can be derived from (close-to) uniformly random distributed  full rank lattices $L\in \R^{2n}$.

The construction of random lattices  plays a prominent role in classical- and post-quantum cryptography due to the computational hardness of the associated lattice problems -- even for quantum computers-- in the worst-case, as well as due to the feature of worst-case to average-case reductions for such problems~\cite{Ajtai96}. It is hence natural to consider random lattices that arise in lattice cryptography as candidates for explicit random families of GKP codes.

In this section we introduce the NTRU cryptosystem and show that random \emph{NTRU lattices} obtained from variations of the NTRU cryptosystem are in fact symplectic, such that they allow to construct GKP codes as scaled GKP codes. We discuss scenarios where NTRU lattices are sufficiently random to follow the Gaussian heuristic or, at least, can be shown to admit a lower bound $\lambda_1 \geq O(\sqrt{n})$ with high probability.

The so-derived GKP codes share characteristics of both scaled- and concatenated GKP codes. These NTRU lattices have been originally formulated in the cryptanalysis of attacks on the NTRU cryptosystem~\cite{Hoffstein_Pipher_Silverman_1998, Coppersmith_shamir, May1999, MayCryptanalysis} and their symplecticity has been motivation to further the study of lattice reduction algorithms for symplectic lattices~\cite{Gama_2006}. As GKP codes, these lattices are particularly interesting as they can be understood as certain generalization of \textit{cyclic} quantum error correcting codes such as the well known $XZZX-\, [\![5,1,3]\!]$  quantum error correcting code~\cite{Bennett_1996} or the repetition code..

\begin{mybox}
    \subsubsection{What are polynomial rings, ideals and modules?}\label{whatis:rings}
    A polynomial \textit{ring} is a set of polynomials $R$ in one or more variables $x$, which is closed under addition and multiplication \begin{equation}
    +:\, R\times R \rightarrow R,\hspace{1cm} \cdot:\, R\times R \rightarrow R,
    \end{equation}
    every element $f\in R$ also possesses an additive inverse $-f\in R$ and there exist additive and multiplicative identities $0, 1$. A quick way to say this is that a ring behaves like a additive group and a multiplicative monoid and a polynomial ring is simply a ring where the elements are formal polynomials in some variable $x$. An ideal $I\subset R$ is a subset of a ring that still forms an additive group, but now also closed under multiplications by elements in $R$. Good examples are e.g. the ring $R=\Z$ (not every integer has an integer multiplicative inverse) which has (prime) ideals $p\Z$ for primes $p\in \Z$. Integer multiples of $p$ remain integer multiples of $p$ when multiplied by another integer. 
    \textit{Modules} are to rings what vector spaces are to number fields. These are spaces of ``vectors'' with coefficients in $R$ and carry a linear structure over $R$. In contrast to vector spaces, the elements in a module generally are allowed to lack the ability to divide by elements in $R$.
    
\end{mybox}

\subsection{The NTRU cryptosystem}

In the following we discuss the NTRU cryptosystem\footnote{It is pretty hard to find reference to what these letters ``\textit{NTRU}'' actually stand for. Rumor says that it is meant as an abbreviation for \textit{Number Theorists aRe Us}.} to the degree necessary to understand the structure of the corresponding lattices and GKP codes constructed here. The presentation here is largely derived from the presentations in refs.~\cite{Hoffstein_Pipher_Silverman_1998, May1999, MayCryptanalysis, Bernstein_Buchmann_Dahmen_2009, StehleSteinfeld, Halevi_lecture}. 

The NTRU cryptosystem is formulated using polynomial rings $R=\Z\lrq{x} /\Phi$, where  the quotient $\Phi=x^n+\phi_{n-1}x^{n-1}+\hdots+\phi_0$ will mostly be taken as $\Phi=\Phi_0:=x^n-1$ which is also the setup used in the original description of the NTRU cryptosystem \cite{Hoffstein_Pipher_Silverman_1998}. We will keep $\Phi$ general whenever possible to be able to discuss the provably secure version of the NTRU cryptosystem \cite{StehleSteinfeld} with irreducible $\Phi=x^n+1$ later. We denote $R_q=R/qR$ with a typically large modulus parameter $q$ and $R_p=R/pR$ with a typically small $p$ coprime with $q$. Whenever we take the modulus, $\mod q$ or $\mod p$, we refer to the (coefficient-wise) reduction into the centered fundamental domains $\lrq{-\frac{q}{2}, \frac{q}{2}}$ resp.~$\lrq{-\frac{p}{2}, \frac{p}{2}}$.

Multiplication in $R$  is denoted as $fg \mod \Phi,\, f,g\in R $, where the reduction $\mod \Phi$ ($\mod q/p$) is implicit by specifying the image, and we use a bold $\bs{f}=\mathrm{coeff}(f)$ to refer to the  coefficient vector $\bs{f}=(f_0, f_1, \dots, f_{n-1})$ of $f\in R$ (note that any polynomial in $R$ can be represented with $n-1$ coefficients, for that every power $x^n$ can be replaced by $x^n-\Phi$ when working over $\mod \Phi$. 

Denote the uniform distribution of polynomials $p\in R_q$ with $d_1$ coefficients $+1$, $d_2$ coefficients $-1$ and $n-(d_1+d_2)$ coefficients $0$ as $D\lr{d_1, d_2}$. Further, denote the set of invertible elements in $R_q$, 
i.e., elements $f$ for which $f^{-1}\in R_q$ exists, as $R_q^{\times}$.

The NTRU cryptosystem, specified by parameters $\lr{n, \Phi, d, q, p }$ operates as follows:

\begin{enumerate}
    \item \textbf{Key generation}: Sample $\tilde{f}\hookleftarrow  D(d,d)$ until $f=1+p\tilde{f}\in R_q^{\times}$, sample $\tilde{g} \hookleftarrow  D(d,d)$ to obtain $g=p\tilde{g}$. Return the secret key pair $(f,g)$, and the public key $h=gf^{-1} \in R_q$.
    \item \textbf{Encryption}: Given the public key $h \in R_q$ and a message $m \in R_p$, sample a random polynomial $r\in R_p$ and compute the ciphertext $c=hr+m \in R_q$.
    \item \textbf{Decryption}: Given the ciphertext $c$ and secret key $f$, compute $cf \mod q \mod p= gr+fm \mod q \mod p = m \in R_p$.
\end{enumerate}

The secret key polynomials $(f,g) \in R_q^{\times}\times R_q$ are by construction such that $f \mod p=1$ and $g \mod p=0$. Decryption is guaranteed to be successful whenever all the coefficients involved are sufficiently small, such that $cf  = gr+fm $ holds as equality in $R$, and not just merely in $R_q$ \cite{Silverman_lecture}.

\subsection{Symplectic ideal and NTRU lattices}

The security assumption underlying this cryptosystem as the inretrievability of the secret key is the hardness of the polynomial factorization problem in $R_q$ and secret key retrieval attacks have been formulated already in early analyses of the NTRU cryptosystem \cite{Hoffstein_Pipher_Silverman_1998, Coppersmith_shamir, MayCryptanalysis}.

\begin{Assumption}[Polynomial factorization problem \cite{MayCryptanalysis}]Given a polynomial $h=f^{-1}g \in R_q$ where the coefficients are small compared to q. For suitable parameter settings it is intractable to find small polynomials $f' ,g'\in R_q$ such that $f'  h=g'\in R_q$.\label{ass:PFP}
\end{Assumption}

 Under the premise that the coefficient vectors of the secret key $(f,g)$ are \textit{short}, a typical attack is formulated as the task of finding short polynomials $(f', g') \in R_q^2$ such that $fh=g \in R_q$, where the length of the polynomial pair is defined as the norm of their joined coefficient vectors $\|(f', g')\|_l=\|(\bs{f'}^T , \bs{g'}^T )\|_l$. We will use the $l=2$ norm unless specified otherwise. 
 The attack is carried out by defining the NTRU lattice as an $R$-module $L_R\subseteq R^2$, which admits a basis in its Hermite normal form
\begin{equation}
    H_R=\begin{pmatrix}
        1 & h \\ 0 & q
    \end{pmatrix}.\label{eq:H_R}
\end{equation}
Elements of the $R$-lattice are of the form 
\begin{align}
    (f'\; u)H_R&=(f'\; u)\begin{pmatrix}
        1 & h \\ 0 & q
    \end{pmatrix} \\
    &=(f'\; f'h+uq ) \nonumber\\
    &=(f', f'h+uq) ,\; (f',u) \in R^2,\nonumber 
\end{align}
each of which represent admissible solutions to the equation $f'h=g' \in R_q$, such that short vectors in $L_R$ are expected to correspond to the NTRU secret key pair. In a more general classification, one can view the $R$-lattice $L_R=L_R(h)$ as an rank-$2$ ideal lattice \cite{compactknapsacks}, corresponding to the principal ideal $I=\langle h\rangle \subseteq R$.

$H_R$ is, in fact, also a $q$-symplectic matrix in $R^{2\times 2}$, with respect to the symplectic form
\begin{equation}
J_R=    \begin{pmatrix}
        0 & 1 \\ -1 & 0
    \end{pmatrix}\in R^{2\times 2},
\end{equation}
with
\begin{equation}
    H_RJ_RH_R^T=\begin{pmatrix}
        h^T-h & q \\ -q & 0
    \end{pmatrix}
    =qJ_R \in R^{2\times 2}
\end{equation}
because $h$ is a scalar in $R$. 

We can define a homomorphism that maps the rank-$2$ $R$-lattice $L_R$ to a rank-$2n$ $\Z$-lattice $L\subseteq \Z^{2n\times 2n}$ by defining a map onto a $\Z^{n \times n}$ matrix
\begin{align}
    C_{\Phi}:\; 
    R&\rightarrow \Z^{n \times n} \\
    f &\mapsto C_{\Phi}(f), \\ 
    \lr{C_{\phi}\lr{f}}_{i,j}&=(T_{\Phi}^i\bs{f})_j,\, i,j=0,\dots ,n-1,
\end{align}
where the rows are given by the vectors $T_{\Phi}\bs{f}$ and
    \begin{equation}
        T_{\Phi}=
        \begin{pmatrix}
            \bs{0}^T & -\phi_0 \\
            I_{n-1} & -\bs{\phi}_{1:n-1}
        \end{pmatrix}
    \end{equation}
implements the map $f\mapsto xf \mod \Phi \in R$ on the coefficient vector $\bs{f}$ of $f$ by left multiplication.

$C_{\Phi}$ is linear over $\Z$, such that we can express the homomorphism on every polynomial $f\in R$ as
\begin{align}
    C_{\Phi}\lr{f}
    &=\sum_{i=0}^{n-1}f_i C_{\Phi}\lr{x^i}\\
    &=\sum_{i=0}^{n-1}f_i C_{\Phi}\lr{1}\lr{T_{\Phi}^T}^i,
    \nonumber 
    \\
    &=\sum_{i=0}^{n-1}f_i\lr{T_{\Phi}^T}^i,\nonumber
\end{align}
where we have used that $C_{\Phi}\lr{1}=I_n\; \forall \Phi$.
In this representation, it is evident that $C_{\Phi}\lr{f}$ acts via right action
\begin{equation}
    \mathrm{coeff}(fg \mod \Phi)=\bs{f}^TC_{\Phi}\lr{g}=\bs{g}^TC_{\Phi}\lr{f}
\end{equation}
and that 
\begin{equation}
    C_{\Phi}\lr{f g \mod \Phi}= C_{\Phi}\lr{f} C_{\Phi}\lr{g}
\end{equation}
indeed represents a homomorphism. When $\Phi=\Phi_0=x^n-1$, $C_{\Phi}\lr{f}$ is simply the (row) circulant matrix of the coefficient vector $\bs{f}$. Circulant matrices are not symmetric, but have a mirror symmetry along the anti-diagonal, $R_n C_{\Phi_0}^T\lr{f} R_n=C_{\Phi_0}\lr{f},$ where $R_n$ is the anti-diagonal matrix $(R_n)_{i,j}=\delta_{i, n-1-j}, \, i,j=0,\dots, n-1$.
We also define a related map
\begin{align}
    A_{\Phi}:\; 
    R&\rightarrow \Z^{n \times n}, \\
    f &\mapsto A_{\Phi}(f), \\ 
    \lr{A_{\phi}\lr{f}}_{i,j}&=(T_{\Phi}^{-i}\bs{f})_j,\, i,j=0,\dots ,n-1,
\end{align}
where $T_{\Phi}^{-i}=\lr{T_{\Phi}^{-1}}^i$ and for $\Phi=\Phi_0=x^n-1$ this is the symmetric anti-circulant matrix of the coefficient vector $\bs{f}$, $A_{\Phi_0}^T(f)=A_{\Phi_0}(f)$. Since $A_{\Phi}$ is also $\Z$-linear, here we have
\begin{align}
    A_{\Phi}\lr{f}
    &=\sum_{i=0}^{n-1}f_i A_{\Phi}\lr{x^i}\\
    &=A_{\Phi}\lr{1}C_{\Phi}\lr{f},\nonumber
\end{align}
where, for $\Phi=\Phi_0=x^n-1$, 
we have that $ A_{\Phi}\lr{1}=1\oplus \overline{I}_{n-1}=:\sigma$ is the orthogonal coefficient mirror $\sigma=\sigma^T$ that maps the coefficient vector $f(x) \in R$ to that of $f\lr{x^{-1}}=f\lr{x^{n-1}}\in R$ \cite{MayCryptanalysis}.

The so-defined maps allow us to map the earlier defined $R$-lattice $L_R$ onto a lattice $L=L(h)\subseteq \Z^{2n}$ by applying the corresponding homomorphism on the entries of the basis

\begin{equation}
    H_R=\begin{pmatrix}
        1 & h \\ 0 & q
    \end{pmatrix}
    \mapsto 
    H_{\Z}=\begin{pmatrix}
        I_n & C_{\Phi}\lr{h} \\ 0 & q I_n
    \end{pmatrix}.
\end{equation}
It can be checked that the lattice spanned by the basis contains all secret key pairs $(\bs{f'}^T\; \bs{g'}^T)$ corresponding to solutions $fh=g \in R_q$. It is however not symplectic yet. To obtain a symplectic matrix, notice that
for $\Phi=\Phi_0$,
\begin{equation}
    H^{cs}=\begin{pmatrix}
        I_n & A_{\Phi_0}\lr{h} \\ 0 & q I_n
    \end{pmatrix}=\begin{pmatrix}
        I_n & \sigma C_{\Phi_0}\lr{h} \\ 0 & q I_n
    \end{pmatrix}\label{eq:HCS}
\end{equation}
is indeed symplectic and corresponds to a rotation of the lattice $L$,
\begin{equation}
    (\sigma \oplus I_n)H^{cs} = H_{\Z}(\sigma^T \oplus I_n),
\end{equation}
since $(\sigma \oplus I_n)$ is unimodular. This is the basis used by Coppersmith and Shamir in their attack on the NTRU cryptosystem \cite{Coppersmith_shamir, MayCryptanalysis}. We generalize this observation to the following statement.

\begin{lem}
    An NTRU lattice $L\subseteq \Z^{2n} \subset \R^{2n}$ given by generator
    \begin{equation}
         H_{\Z}=\begin{pmatrix}
        I_n & C_{\Phi}\lr{h} \\ 0 & q I_n
    \end{pmatrix}
    \end{equation}
    is equivalent to a $q$-symplectic lattice $L'\subset \R^{2n}$ for all $h$ if
    there exists a signed permutation matrix $\sigma_{\Phi} \in \Z^{n \times n} \cap O\lr{n}$ such that 
    \begin{equation}
        \lr{\sigma_{\Phi}C_{\Phi}\lr{h}}^T=\sigma_{\Phi}C_{\Phi}\lr{h}
    \end{equation}
    is symmetric.
\end{lem}
\proof 
A lattice generated by $M$ is equivalent to a lattice generated by $M'$, such that $\det M=\det M'$ if and only if there exists an unimodular matrix $U$  and an orthogonal matrix $O$ such that $M'=UMO$ \cite{ConwaySloane}. Take $O=\lr{\sigma_{\Phi}^T\oplus I_n}$ and $U=\lr{\sigma_{\Phi} \oplus I_n}$.
\endproof
\begin{cor}
    NTRU lattices over $\Phi=\Phi_0$ and $\Phi=x^n+1$ are equivalent to $q$-symplectic lattices
\end{cor}
\proof
For $\Phi=\Phi_0$ we already saw earlier that $\sigma_{\Phi_0}=\sigma$ provides a symmetric matrix $\sigma C_{\Phi}\lr{h}$ for all $h$. For $\Phi=x^n+1$ this is also the case, with
   \begin{equation}
        \sigma_{\Phi}=
        \begin{pmatrix}
            1 & \bs{0}^T \\
            \bs{0} & -\overline{I}_{n-1}
        \end{pmatrix}=A_{\Phi}\lr{1}
    \end{equation}
and $A_{\Phi}\lr{1}C_{\Phi}\lr{h}=A_{\Phi}\lr{h}$ is such that the first row is $\bs{h}^T$ and every other row is generated by permuting the first element around the ``periodic boundary'' on the right to the left while adding a $-1$ factor. This matrix is clearly symmetric and $\sigma, \sigma_{\Phi}$ are signed permutations.
\endproof

Finally, the fact that these NTRU lattices $L$ corresponds to ideals $I=\langle h\rangle \subseteq R$ equips them with the symmetry $L=(T_{\Phi}\oplus T_{\Phi})L$. When $\Phi=\Phi_0$ we have that the symmetry is $n$-fold, $T_{\Phi_0}^n=I$ and similarly for $\Phi=x^n+1$ we have $T_{\Phi}^n=-I$. 

Henceforth, we will default to $\Phi=\Phi_0$ unless specified otherwise and omit the corresponding $\Phi_0$ index from $C_{\Phi_0}$ and $A_{\Phi_0}$. The anti-circulant matrix $A\lr{h}=\sigma C\lr{h}$ implements a homomorphism from $R$ with respect to a modified matrix multiplication
\begin{equation}
    A\lr{f}\sigma A\lr{h}=\sigma C\lr{f}C\lr{f}=A\lr{fg}.
\end{equation}
We denote $A\lr{f}\sigma=:A^{\sigma}\lr{f}$, such that $A^{\sigma}\lr{f}A\lr{g}=A\lr{f}$.

On $\Z^n$, ciphertexts produced by the NTRU encryption with secret key pair $(f, g)$ and public key $h$ take the form
\begin{align}
    \bs{c}^T
    &=\bs{m}^T+\bs{r}^T C(h)\, \mod q\\
    &=\bs{m}^T+\lr{\sigma \bs{r}}^T  \sigma  C(h)\, \mod q
    \nonumber
\end{align}
and decryption is carried out by left-multiplying with $A^{\sigma}\lr{f}$ and reducing $\mod q$ and $\mod p$.

The corresponding $q$-symplectic generator of the underlying lattice is given by
\begin{equation}
    H=\begin{pmatrix}
        I_n & A\lr{h} \\ 0 & qI_n
    \end{pmatrix}, 
\end{equation}
which is already a $q$-symplectic basis for the symplectically integral $L$. 

We use this lattice as starting point to define a scaled GKP-code by taking $\CL=\sqrt{({d}/{q})}L$ with generator 
$M=\sqrt{({d}/{q})}H$. Notice that these lattice generators form a subclass of  those considered in theorem\ref{them:randsymH} for which we have already shown the expected average $\lambda_1 \sim \sqrt{n}$ scaling. 

Similar to the discussion earlier, the GKP code built this way will encode $D=d^n$ logical dimensions with symplectic dual 
\begin{equation}\CL^{\perp}=L/\sqrt{d q}
\end{equation} and  distance 
\begin{equation}\Delta=\lambda_1\lr{L}/\sqrt{d q}.
\end{equation}
For randomly chosen $f, g \in \CR$, the Gaussian Heuristic~\ref{GaussianHeuristic} and them.~\ref{them:randsymH} hence suggest that a \textit{good} parameter scaling

\begin{align}
D &= d^n,\\
\Delta &\geq  \sqrt{\frac{n }{d \pi e}},
\end{align}
is possible.

However, the Gaussian heuristic does not always hold for NTRU lattices with arbitrary parameters. Due to the sub-lattice structure $q\Z^{2n}\subset L$ there always exist trivial vectors $q \bs{e}_i,\, i \in \lrq{1,2n}$ of length $q$ in  $ L$ which yield logically non-trival vectors of length $\sqrt{q/d}$. A shortest vector length $\lambda_1\lr{L}$ growing with $\sqrt{n}$ can however be maintained by choosing suitably large $q$ scaling with $n$. Furthermore, NTRU lattices (with $\Phi_0$) are constrained by 1. being \textit{cyclic} lattices and 2. having an existing inverse of $f \in R_q$ and 3. having a fixed number $2d$ of non-zero coefficients in the vector corresponding to the secret key $(\sigma\lr{\bs{f}}^T,\bs{g}^T)^T \in L$, which on the one hand make it not immediately clear if they would be sufficiently random for the Gaussian heuristic to hold, and on the other hand already present short vectors of length $\leq  O(\sqrt{d})$. These points have been addressed in refs.~\cite{Qi_earchive, Jingguo_Qi}, where the authors show the following statement.

\begin{cor}[{\cite[{Corollary 3}]{Qi_earchive, Jingguo_Qi}}]\label{cor:Qi}
If $d = \lfloor n/3 \rfloor $, then with probability greater than $1-2^{-0.1n}$ the shortest vector in a random NTRU lattice has length greater than
$\sqrt{0.28 n}$.
\end{cor}

This statement gives us confidence to claim that random NTRU lattice based GKP codes as constructed above can be expected to be \textit{good} when the parameters are chosen properly, as summarized by the following.
\begin{them}[Good codes from NTRU lattices]
\label{prop_1}
    A GKP code with $\CL=\sqrt{({2}/{q})}L$, where $L$ is the NTRU lattice over $\Phi_0$ specified in the basis eq.~\eqref{eq:HCS} with $d = \lfloor n/3 \rfloor $ encodes
    \begin{equation}
        k=n
    \end{equation}
    qubits and has with probability greater than $1-2^{-0.1n}$ a distance given by
    \begin{equation}
        \Delta= \min \lrc{\sqrt{\frac{0.14 n}{q}}, \sqrt{\frac{q}{2}}}.
    \end{equation}
    For sufficiently large constant $q$ and $n\leq q^2/0.28$ this defines a randomized family of \textit{good} GKP codes.
\end{them}
\proof
Follows immediately from corollary~\ref{cor:Qi} and the GKP-code construction laid out in the main text.
\endproof

\subsection{Numerical results}
In fig.~\ref{fig:NTRU_sample} are plotted the shortest vector lengths for $N_{\rm sample}=100$ randomly sampled NTRU lattices for varying $q$ and $n$ with $p=3$. In the figures, we compare samples over NTRU-like random cyclic lattices, where $h$ is sampled randomly from $R_q$ in row $a)$ with NTRU lattices over $\Phi=x^n-1$ with $f$ invertible in $R_q$ and bounded non-zero entries $d=\lfloor n/3 \rfloor$ (in row $b)$). We also compare the average length of shortest vectors for even more constrained NTRU lattices where the public key $h$ is also required to be invertible in $R_q$ in row $c$. In this case we obtain $g$ from the amended distribution $g \sim p D(d+1,d)$ since otherwise $g$ -- and thus $h$ -- would have a trivial root $g(1)=0$ rendering the polynomial non-invertible.
Finally, in row $d)$, the experiment is carried out using the setup of ref.~\cite{StehleSteinfeld}, where the quotient $\Phi=x^n+1$ is chosen to be irreducible and $n$ is a power of $2$. 

In these statistics we observe that random cyclic lattices (row $a)$) appear to agree well with the Gaussian heuristic, while the growth of the shortest vector length of the NTRU lattices in row $b)$ and $c)$ degrades with increasing $q$, consistent with the  bound given in corollary~\ref{cor:Qi}. The numerical results suggest that simply picking a random polynomial $h\in R_q$  is very likely to yield the $\lambda_1\sim \sqrt{n}$ scaling. This is summarized as the following conjecture.

\begin{conjecture}[Good GKP codes]\label{conj:random_h}
    A GKP code with $\CL=\sqrt{{d}/{q}}L$, where $L$ is specified by the basis in \eqref{eq:HCS} and $h$ is selected at random from $R_q=\Z_q\lrq{x}/\langle x^n-1\rangle$, is likely a good code with $k=n$ and 
    \begin{equation}
        \Delta \geq \min\lrc{ \sqrt{\frac{n }{d \pi e}}, \sqrt{\frac{q}{d}}}.
    \end{equation}
\end{conjecture}

Note that the present conjecture is a somewhat stronger statement than what is implied by theorem~\ref{them:randsymH}. Theorem~\ref{them:randsymH} requires uniform randomization over symmetric upper right blocks of the corresponding generator matrices $X=X^T$, which still maintains $n(n-1)/2$ free parameters. Here, the (quasi-) circular structure of the blocks already reduces the number of independent parameters to $n$, which is only a small subset of the space of all lattices generated with upper right block $X=X^T$. The numerical findings and the conjecture predict that this small subset is large enough and sufficiently well-distributed within the space of all lattices to maintain the $\lambda_1$ scaling.

Finally, in row $d)$, one observes a good agreement of the shortest vector lengths with the scaling proposed by the Gaussian heuristic. In ref.~\cite{StehleSteinfeld} a probabilistic lower bound on the smallest infinity norm $\lambda_1^{\infty}\lr{L}$ has been proven, which is included in the figure. As we'll discuss later, GKP codes derived from this particular NTRU- setup are of cryptographic relevance and based on these numerical observations we can also conjecture that such GKP codes are likely good.

\begin{conjecture}[Good GKP codes]\label{conj:SS}
    A GKP code with $\CL=\sqrt{{d}/{q}}L$, where $L$ with $\det L = q^n$ is equivalent to  NTRU lattice specified by the basis in \eqref{eq:HCS} and $h=g/f\leftarrow f, g$ are sampled  at random from a Gaussian distribution with variance $\sigma^2=q$ in  $R_q=\Z_q\lrq{x}/\langle x^n+1\rangle$, $q\geq \mathtt{poly}(n)$ and $n\geq 8$ a power of $2$ is likely a good code with $k=n$ and 
    \begin{equation}
        \Delta \geq  \sqrt{\frac{n }{d\pi e}}.
    \end{equation}
\end{conjecture}

In contrast to the previous statement in proposition~\ref{prop_1}, these distance bounds do not suffer from choosing larger modulus $q$, but we can pick $q$ arbitrarily large to obtain high distances. 

The trivial sub-lattice $L_q=q\Z^{2n}\subseteq L$ which enforces the $q$ modularity in the cryptographic setup is analogue to the structure of concatenated (hypercubic) GKP codes $\CL_{\rm triv}=\sqrt{d q}\Z^{2n} \subseteq \CL$, such that the lattices $\CL$ defined above may be interpreted as a concatenated (qudit) GKP code where $\CL_{\rm triv}$ defines the underlying single mode qudit code with $D=d q$. It is interesting that this class of NTRU-GKP codes thus shares characteristics of both \textit{scaled-} as well as \textit{concatenated} GKP codes. 

 \begin{figure*}
 \center
 \includegraphics[width=\textwidth]{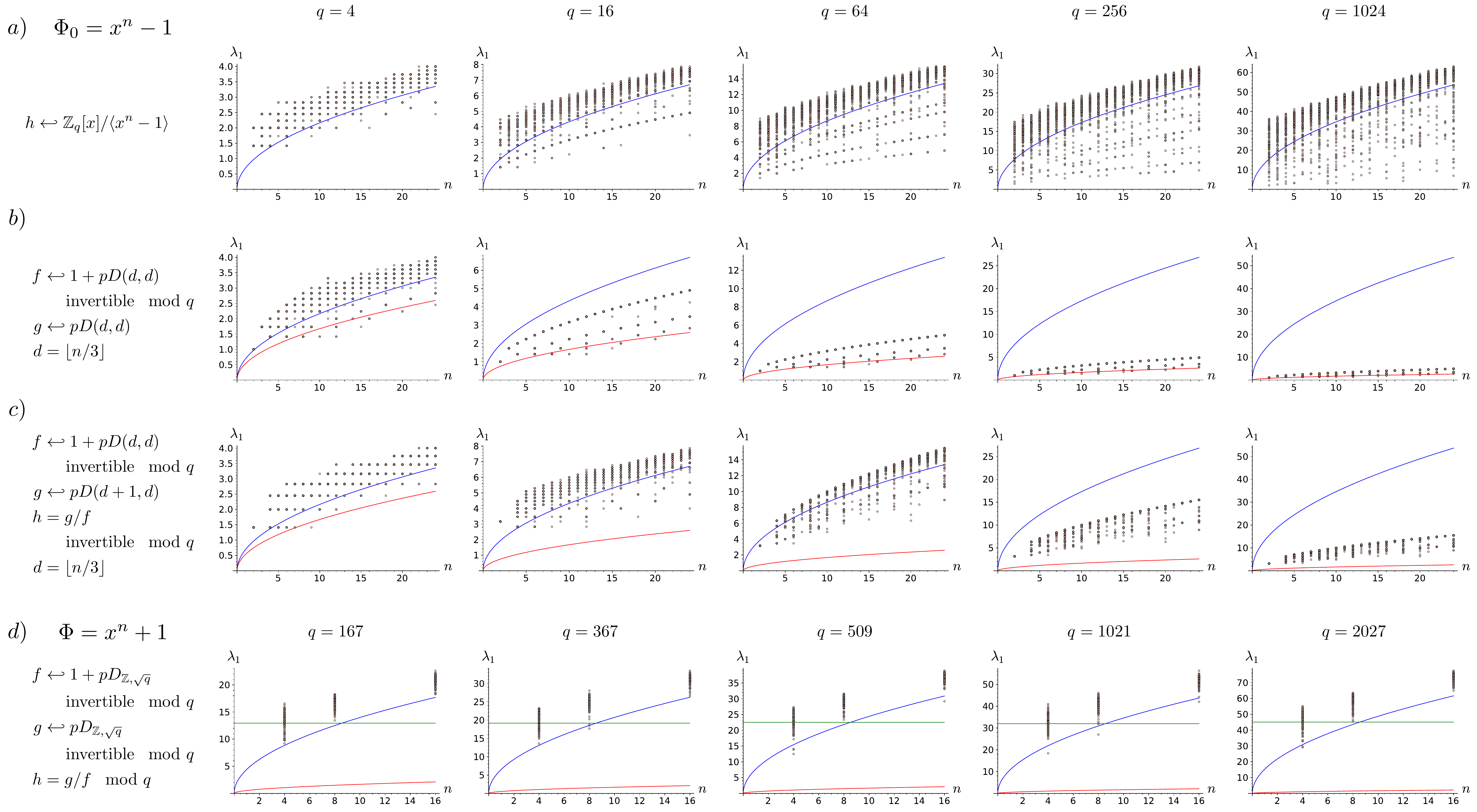}
 \caption{ Shortest vector lengths computed via full \texttt{HKZ} reduction of $a)$ random cyclic ($\Phi_0=x^n-1$) lattices as generated by the hard lattice generator in \texttt{sagemath}, $b)$ random NTRU lattices with $p=3$ and $d=\lfloor n/3 \rfloor$ and $c)$ random NTRU lattices where $h$ is invertible in $R_q$  for varying $q=2, \dots , 2048$. In $d)$ we sample NTRU lattices generated with the irreducible quotient $\Phi=x^n+1$, where $n$ is a power of 2. 
 For each $n \in [2,24]$ we sample $100$ NTRU lattices and compute the shortest vector by computing the HKZ reduced lattice basis. For reference, we plot the expected shortest vector length from the Gaussian heuristic $\lambda\lr{n}=\sqrt{nq/\pi e}$ in blue and the expected lower bound $\lambda_0\lr{n}=\sqrt{0.28 n}$ in red. In panel $d)$, we have also included a green line at $\sqrt{q}$, which is the standard deviation of the discrete Gaussian distribution $f, g$ are sampled from and is related to a probabilistic lower bound for $n\geq 8$ a power of $2$ on the shortest infinity norm $\lambda_1^{\infty}\lr{L}$ derived in ref.~\cite{StehleSteinfeld}.
 The \texttt{sagemath} \cite{sagemath} code as well as all numerical data presented here is available under ref.~\cite{GitLink}. The \texttt{sagemath} functionalities to construct NTRU lattices are partially adapted from ref.~\cite{LatticeHacks}.} \label{fig:NTRU_sample}
 \end{figure*}

\section{The Dream}

In this chapter we have discussed ways to implement GKP codes by means of scaling known symplectically self-dual lattices derived from root systems and introduced a novel class of GKP codes constructed upon instances of the NTRU cryptosystem, which is also shown to yield a family of random \textit{good} GKP codes. On both these fronts there is more to explore. 
Root lattices carry the defining feature of having (orthogonal) automorphisms generated by phase-space reflections through hypersurfaces normal to the roots. This is a natural starting point to search for symplectic orthogonal automorphisms -- i.e. logical Clifford gates -- within those groups. For instance, as has been pointed out in ref.~\cite{Buser_1994}, the $E_8$ lattice possesses $46080$ different symplectic automorphisms, which naturally translates into a meaningful fraction of the logical Clifford group to be implementable through relatively simple physical operations provided by \textit{passive} linear optical elements.  Furthermore, the existence of such automorphisms also implies the ability for GKP codes built on those lattices to distill magic states from the Gaussian vacuum state. We discuss this relation in appendix~\ref{app:Magic} more in depth. A concrete challenge here would be to systematically identify the subgroup of symplectic automorphism group within the reflection (Weyl) groups of symplectic root lattices.

A more pertinent question is to explore how the weight of an optimal generating set for GKP lattices influences possible distances in the following sense. In chapter~\ref{chap:Theory}, relationships between the euclidean $\|\cdot\|_2$ norm of lattice basis vectors and the distance $\Delta$ of GKP codes have already been discussed through transcendence theorems and theta functions. A physically more meaningful setting would be to constrain the $\|\cdot\|_0$ norm of the rows and columns of the lattice generator and ask what euclidean distances are possible, as this constraint quantifies the physical connectivity between different quadratures and modes necessary to measure the associated stabilizers. For qubit-based quantum error correcting codes, it was recently shown that families of so-called good \textit{Low-Density-Parity-Check} (LDPC) code exist \cite{panteleev2022asymptotically, Breuckmann_2021_balanced, Breuckmann_2021}.  When concatenated with single mode GKP codes such codes imply the desired $\Delta \propto \sqrt{n}$ scaling while retaining short basis vectors for the lattice in $\|\cdot\|_0$ norm. The fact that the bounds in chap. ~\ref{chap:Theory} are derived using the euclidean $\|\cdot\|_2$ norm, however, suggests that there is room to adjust the scaling factor for the distance by allowing to vary the euclidean length of the shortest $\|\cdot\|_0$ norm lattice vectors. Conversely, given a lattice and searching for optimal (short) bases relative to the $\|\cdot\|_0$ norm is a relevant problem to examine.

\chapter{Decoding GKP codes}\label{chap:complexity}
\section{Decoding is hard}
\begin{figure}[H]
\center
\hspace{-2.5cm}
\includegraphics[width=1.2\textwidth]{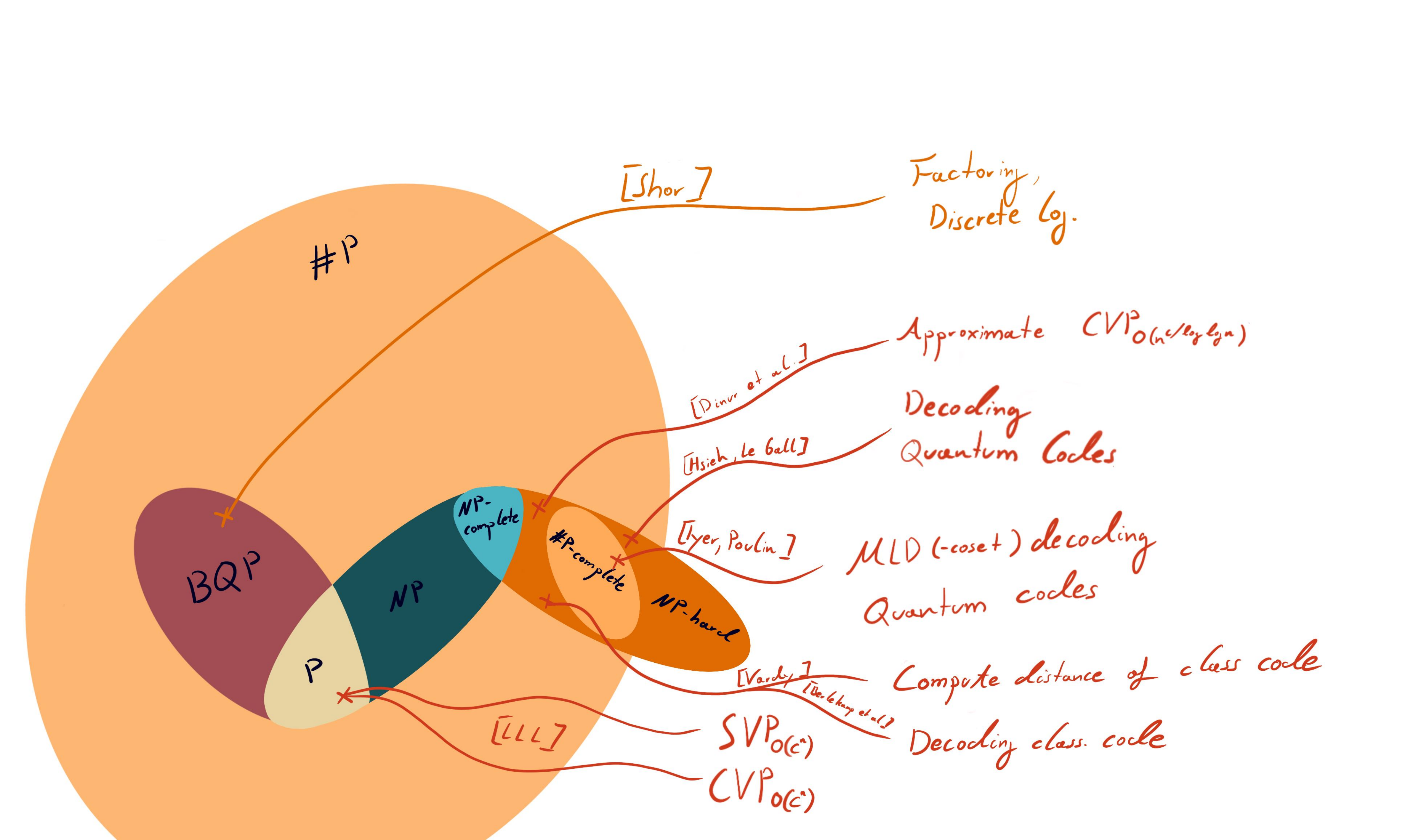}
\caption{Sketch of a Venn diagram illustrating the relationships between relevant complexity classes and the categorization of decoding- and lattice problems. Note that to be very precise, one should substitute all problems and problem classes with their decision version variants to avoid comparing apples to oranges: Instead of $\# P$ it would be more appropriate to consider the class $P^{\# P}$ of polynomial time decidable problems with access to a $\# P$ oracle and similarly, substitute $\tt SVP$ and $\tt CVP$ by its appropriate decision variants. For clarity of presentation we sweep these details under the carpet.}\label{fig:Complexities}
\end{figure}

\begin{mybox_blue}
\subsubsection{What is... complexity?}
Computational complexity theory is the categorization of computational problems into their solvability relative to the size of the input of an individual instance. Most commonly, one focuses on \textit{decision problems}, which are questions formulated on length $n$ bitstrings $\bs{b}$ that are answered by a simple \texttt{Yes} or \texttt{No}, that is they evaluate a function $p:\, \Z_2^n\rightarrow \Z_2$. 

\begin{itemize}
\item The class $\tt P$ is the class of all problems answerable within $O\lr{{\tt poly}\lr{n}}$ elementary computational steps, one also says \textit{polynomial time}. 

\item The class $\tt NP$ is the class of all \textit{non-deterministic polynomial} time problems. I.e. the questions doesn't necessarily be answerable in polynomial time, but there exists an $O\lr{{\tt poly}\lr{n}}$ time algorithm that decides whether a given answer was correct or not.

\item The class $\tt BQP$, termed \textit{bounded-error quantum polynomial}, captures the class of problems decidable by executing a quantum computation within a polynomial number of steps: preparing initial states, applying gates and measurements with high probability. As quantum measurements are non-deterministic in general, it is necessary to relax definition to only require to obtain the right answer with probability bounded away from $p\geq 0.5+c, c>0$. The success probability can then be arbitrarily amplified by repeating the experiment and taking the majority.

\item The class $\tt \#  P$ is not a decision problem but counts the number of \texttt{Yes} answers  to a problem in $\tt NP$.  Given access to a $\tt \# P$ solver,  one can also find out whether the answer to a problem is \texttt{Yes} by simply querying whether there is a non-zero number of \texttt{Yes} answers. One specifies this by defining the class $\tt P^{\# P}$, the class of polynomial time algorithms that have black-box access to a $\tt \# P$ solver.

\end{itemize}
We say that a problem $f$ can be \textit{reduced} to a problem $g$, or $f\leq g$, if there exists a polynomial time algorithm that can solve $f$ given access to an oracle that solves $g$. A problem is called $\tt C-hard$ if all problems in $\tt C$ can be reduced to it, i.e. it is at least as hard as all problems in $\tt C$;  if it is contained in $\tt C$ itself too, we call it $\tt C-complete$.

See fig.~\ref{fig:Complexities} for an overview of the relationship between these classes and how some relevant problems fall into them. Further reference is found in ref.~\cite{ComplexityZoo}.

\end{mybox_blue}

\subsubsection{Decoding classical error correction}
The decoding problem is one of the cornerstones of computational complexity theory.  In classical error correction, where the code space is a space of bitstrings $C\in \Z^n_2$, we have access to the \textit{parity check matrix} $H\in \Z_2^{r\times n}$, which is such that ${\rm ker}\lr{H}=C$, and one can compute the \textit{syndrome} of codewords $\bs{c}\in C$ that are perturbed by an error $\bs{e}\in \Z_2^n$,

\begin{equation}
\bs{s}=H\lr{\bs{c}+\bs{e}}\, \mod 2=H\bs{e} \mod 2.
\end{equation}

In practice, one assumes small errors to happen more likely than large errors, such that the decoding problem becomes the task of solving the optimization problem

\begin{equation}
{\tt Dec}\lr{\bs{s}; C}=\argmin_{\substack{\bs{e}\in \Z_2^n,\\  H\bs{e}=\bs{s}\mod 2}} \| \bs{e}\|,
\end{equation}
where $\|\cdot\|$ is typically taken as the Hamming norm $\|\cdot\|_0$ but since we are working over binary strings other norms $\|\cdot\|_0=\|\cdot\|_1=\|\cdot\|^{1/2}_2$ only differ by a power.

One can turn this into a decision problem
\begin{equation}
{\tt Dec}^w\lr{\bs{s}; C}=\stackrel{?}{\exists}\bs{e}\in \Z_2^n: \,  H\bs{e}=\bs{s}\mod 2 , \;   \| \bs{e}\|\leq w,
\end{equation}
such that the corresponding optimization problem can be solved by varying $w\in \lrc{1, \hdots,  \lfloor d_C/2 \rfloor }$ up to the half distance of the code, which is defined as the norm of the smallest codeword,

\begin{equation}
d_C=\min_{\substack{\bs{c}\in C}{\bs{c}\neq 0}} \| \bs{c}\|_0.
\end{equation}

This is an example of a reduction of an optimization problem to a corresponding decision problem. The classical decoding problem was proven to be $\tt NP-hard$ in ref.~\cite{Berlekamp} -- which makes the decision version $\tt NP-complete$ --  and, similarly, ref.~\cite{Vardy} proved that the problem of computing the distance is $\tt NP-hard$ with accompanying $\tt NP-complete$ decision version
\begin{equation}
{\tt Dist}^w\lr{C}=\stackrel{?}{\exists}\bs{c} \in C:\, \|\bs{c}\|_0 \leq w.
\end{equation}

The hardness of decoding classical error correcting codes has been an important ingredient in the design of cryptographic protocols \cite{McEliece1978APK, Bernstein_Buchmann_Dahmen_2009}. The basic idea in such protocols is to devise a way to draw error correcting codes at random in a way that allows the user of the protocol to keep a \textit{secret key}, which allows to convert the error correcting code together with its syndromes back into a form that is easy to decode. This idea is vastly powerful and lies at the basis of many modern cryptographic protocols: see e.g. ref.~\cite{Bernstein_Buchmann_Dahmen_2009} and references therein.

\subsubsection{Decoding quantum error correction}

We have already looked at the decoding problem of quantum error correcting codes in chapter~\ref{chap:QC}, which we briefly review. For simplicity, we consider only qubit-based quantum error correcting codes with Pauli-type stabilizers. We label each Pauli operator by a binary string $\bs{l} \in \Z_2^{2n}$
\begin{equation}
P\lr{\bs{l}}=X_1^{l_1}\otimes\hdots \otimes X_n^{l_n}\otimes Z_1^{l_{n+1}}\otimes\hdots \otimes Z_n^{l_{2n}},
\end{equation}
such that two Pauli operators commute as
\begin{equation}
P\lr{\bs{l}}P\lr{\bs{l'}}= (-1)^{\bs{l}^T J \bs{l'} }P\lr{\bs{l'}}P\lr{\bs{l}},
\end{equation}
where the commutation phase is determined by the symplectic inner product $\bs{l}^T J \bs{l'} \mod 2$. The generating set of a quantum error correcting code $\CG=\lrc{P\lr{\bs{g}_1},\hdots, P\lr{\bs{g}_r}}\subset \CP$ with $r=n-k$ independent generators can hence be summarized by the parity check matrix
$H=\lr{\bs{g}_1, \,\hdots , \bs{g}_r}^T$ and the stabilizer group is given by its row-span $S={\rm span}_{\Z_2}\lr{H}$. Logical Pauli operators live in the commutant of the stabilizer group with labels $L=S^{\perp}$ relative to the symplectic inner product $\mod 2$, such that the quotient group  $L/S$ yields a set of representatives of the $4^k$ logical Pauli operators.

The syndrome measured upon applying a Pauli error $P\lr{\bs{e}}$ to a code state  is hence 
\begin{equation}
\bs{s}=H\bs{e}\mod 2,
\end{equation}
which is analogous to the situation in classical error correction, except for the additional requirement that the stabilizer generators labeled in the rows of $H$ commute, i.e. $HH^T = 0\mod 2$. 

Let's assume that we are given a probability distribution over Pauli errors $p(\bs{e})$ and measure a syndrome $\bs{s}=H\bs{e}\mod 2$. The first step in decoding is to find a generic error $\bs{d}\in \Z_2^{2n}$ that yields the same syndrome. Applying this correction returns the state to code space and what remains is to find a logical post-correction by computing the probability that this correction has returned us to the wrong element in code space. This computational  task, known as \textit{maximum likelihood decoding}, is to evaluate

\begin{equation}
{\rm MLD}\lr{\bs{s}}=\argmax_{\bs{l}\in L/S}  \sum_{s\in S} p\lr{\bs{d}+\bs{l}+\bs{s}},\label{eq:MLD_q}
\end{equation}
which finds the most likely logical error incurred by applying correction $\bs{d}$ \textit{up to stabilizer equivalences}. The generic difficulty in this problem stems from the fact that there may be many error configurations that are of low probability by themselves, but add up to a high probability configuration due to a combinatorial factor when added up over stabilizer-equivalent configurations. This difficulty was made concrete in ref.~\cite{Iyer}, who showed that MLD decoding quantum error correcting codes is generally $\tt \# P-complete$. 

When the probability for errors are low and sufficiently well-behaved, such that the most likely error coset $\bs{d}+\bs{l}$ is expected to also be given by the most likely individual configuration, the quantum decoding problem simply reduces to the problem 
\begin{equation}
{\tt Dec}\lr{\bs{s}; Q}=\argmin_{\substack{\bs{e}\in \Z_2^n,\\  H\bs{e}=\bs{s}\mod 2}} \| \bs{e}\|, \label{eq:MED_q}
\end{equation}
where $Q={\rm ker}\lr{H}$. For quantum error correcting codes this problem has been shown to be $\tt NP-hard$ in ref.~\cite{Hsieh_2011}. This mode of decoding has also been dubbed \textit{minimum energy decoding} (MED) in ref.~\cite{toricGKP}, due to the interpretation of the failure probability of the MLD decoder in eq.~\eqref{eq:MLD_q} as the free energy of a certain statistical mechanical model \cite{toricGKP, Dennis_2002}, whose energy  -- i.e. without the combinatorial entropic contribution -- is minimized by solving this problem.

\subsubsection{Lattice problems}

Another presumably hard class of problems is formulated on lattices. Given a lattice $L\subset \R^n$ and an arbitrary vector $\bs{t}\in \R^n$ the analogues of the distance-computation and decoding problems are the \textit{shortest-} and \textit{closest vector problem}.

\begin{align*}
{\tt SVP}\lr{L}&=\argmin_{0\neq \bs{x}\in L} \|\bs{x}\|_2,\\
{\tt CVP}\lr{\bs{t}; L}&=\argmin_{ \bs{x}\in L} \|\bs{t}-\bs{x}\|_2.\\
\end{align*}
 
 These problems are in fact so hard, that even approximations are hard to obtain. To quantify this, define the approximate problems with approximation parameter $\gamma$ and let $\lambda_1\lr{L}=\|{\tt SVP}\lr{L}\|_2$, ${\rm dist}\lr{\bs{t}, L}=\|{\tt CVP}\lr{\bs{t}; L}\|$ denote the length of the shortest vector, resp. the minimal distance between $\bs{t}$ and the lattice. Approximate versions of the shortest-  and closest vector problem can then be defined as
 \begin{align*}
{\tt SVP}_{\gamma}\lr{L}&={\tt return} \;  \bs{x} \in L-\lrc{0}: \; \|\bs{x}\|\leq \gamma\, \lambda_1\lr{L}, \\
{\tt CVP}_{\gamma}\lr{\bs{t}; L}&={\tt return} \;  \bs{x} \in L: \; \|\bs{t}-\bs{x}\|\leq \gamma \,{\rm dist}\lr{\bs{t}, L}.
\end{align*}

It was proven in ref.~\cite{Dinur} that these problems are hard even for approximation factors $\gamma=n^{c/\log\log\lr{n}}$. On the contrary, if the approximation factor is allowed to be exponentially large, $\gamma=2^{n\lr{\log\log n}^2/\log n}$, the Lenstra-Lenstra-Lov{\'a}sz algorithm solves them efficiently \cite{LLL, Regev2010, Schnorr1987}.

It is somewhat unintuitive to see the hardness of the above problems in the simple case of two-dimensional lattices. The generic hardness of these problems stems from the fact that the computationally efficient way to represent a lattice is through its generator matrix $M\in \R^{n\times n}$, which may contain arbitrarily long non-orthogonal vectors, and there are $|\GL_n\lr{\Z}|=\infty$ possible bases to pick from. If one is lucky to possess a ``good" basis for a lattice, which e.g. contains the shortest lattice vector, or even better: which is such that the basis vectors represent the $n$ successive minima, it obviously helps to solve the $\tt SVP$ problem. By presenting a very fine-grained resolution of the lattice, a good basis of similar type then also helps in solving the $\tt CVP$ problem as it allows to represent any lattice vector in a relatively minimal linear combination of its basis vectors \cite{Schnorr1987}.
 
\section{Maximum likelihood decoding GKP codes}\label{sec:MLD}

To derive the decoding problem for GKP codes, assume a stochastic Gaussian displacement noise channel as specified in eq.~\eqref{eq:noise_channel} with variance $ \tilde{\sigma}^2$. For comparison with the literature, when the displacement operators are defined by a more ``standard'' convention without the overall constant $\sqrt{2\pi}$, this corresponds to a physical variance of $\sigma^2=2\pi  \tilde{\sigma}^2$. %
Upon sampling an error $\bs{e}$ and measuring the stabilizers, a syndrome vector of the form%
\begin{equation}
\bs{s}(\bs{e})=MJ\bs{e} \mod 1
\end{equation}%
is obtained as the phases of the eigenvalues of the stabilizer generators $\lrc{D\lr{\bs{\xi}_i}}_{i=1}^{2n}$ when acting on a code state vector $\ket{\psi}$ displaced by an error vector $\bs{e}$,
\begin{equation}
D\lr{\bs{\xi}_i}D\lr{\bs{e}}\ket{\psi}=e^{i2\pi \bs{\xi}_i J\bs{e}} D\lr{\bs{e}}\ket{\psi},
\end{equation}
where $\bs{\xi}_i^T=M_i$ is the i'th row of $M$.

Since we are dealing with full rank lattices, given the syndrome, we can assign a pure error %
\begin{equation}
\bs{\eta}(\bs{s})=(MJ)^{-1}\bs{s} \label{eq:eta}
\end{equation} %
that has the same syndrome $\bs{s}$ as $\bs{e}$ as initial guess for the correction. 

To find the appropriate logical post-correction, for every $\bs{\xi}^\perp \in \CL^\perp/\CL$ we evaluate the probabilities that, given syndrome $\bs{s}$, the actual error is stabilizer equivalent to $\bs{\eta}(\bs{s})+\bs{\xi}^{\perp}$, which is given by
\begin{equation}
P([\bs{\eta}(\bs{s})+\bs{\xi}^{\perp}] | \bs{s}) = P^{-1}(\bs{s}) \sum_{\bs{\xi} \in \mathcal{L}} P_{\tilde{\sigma}}(\bs{\eta}(\bs{s})+\bs{\xi}^{\perp}+\bs{\xi}), \label{MLD_1}
\end{equation} 
where $\lrq{\bs{x}} = \left\{ \bs{x} + \bs{\xi}, \bs{\xi}\in \CL\right\} $ and  $P_{\tilde{\sigma}}$ is as specified in eq.~\eqref{eq:noise_channel}. %
This can be rewritten as%
\begin{align}
P([\bs{\eta}(\bs{s})+\bs{\xi}^{\perp}] | \bs{s}) &= P^{-1}(\bs{s}) \sum_{\bs{\xi} \in \mathcal{L}+\bs{\eta}(\bs{s})+\bs{\xi}^{\perp}} P_{\tilde{\sigma}}(\bs{\xi}) \\
&=\sqrt{2\pi\tilde{\sigma}^{2n}}^{-1}  P^{-1}(\bs{s})  \Theta_{\mathcal{L}+\bs{\eta}(\bs{s})+\bs{\xi}^{\perp}}\left(\frac{i}{2\pi \tilde{\sigma}^2}\right), \label{MLD_2}
\end{align}%
proportional to the theta series of the \textit{packing} $\CP=\CL+\bs{\eta}(\bs{s})+\bs{\xi}^{\perp}$ evaluated in $z = \frac{i}{2\pi \tilde{\sigma}^2}$ \footnote{ $\CP$ as the translate of a lattice $\CL$ is formally not a lattice, in particular $\CP$ may not contain the origin.} .
Let $\lr{\tilde{\CD},\, N_{\tilde{\delta}}}$ denote the distance distribution of $\CP$, that is, 
\begin{equation}
\tilde{\CD}=\{\| \bs{\xi} + \bs{\eta}(\bs{s})+\bs{\xi}^{\perp}\|^2,\; \bs{\xi} \in\CL \}
\end{equation}
is the set of possible lengths in the shifted lattice and $N_{\tilde{\delta}}$ counts the multiplicity of these lengths in the shifted lattice.
We can thus write the coset probabilities above in a \textit{small error} or \textit{``low temperature expansion''},
\begin{equation}
 \Theta_{\mathcal{L}+\bs{\eta}(\bs{s})+\xi^{\perp}}\left(\frac{i}{2\pi \tilde{\sigma}^2}\right) = \sum_{\tilde{\delta} \in \tilde{\mathcal{D}}} N_{\tilde{\delta}} q^{\tilde{\delta}}. \label{MLD_3}
\end{equation}%
evaluated at $q = \exp\lr{ -1/2\tilde{\sigma}^2 }$.  MLD decoding is implemented by applying the total correction %
\begin{equation}
\overline{\bs{\eta}}=\bs{\eta}(\bs{s}) + \argmax_{\bs{\xi}^{\perp} \in \mathcal{L}^{\perp}/\mathcal{L}} P([\bs{\eta}(\bs{s})+\bs{\xi}^{\perp}] | \bs{s}).
\end{equation} 
\subsection{Minimum energy decoding}
In the limit $\tilde{\sigma}, q \rightarrow 0$, the sum \eqref{MLD_3} becomes sharply distributed around solutions with minimal $\tilde{\delta}$. That is, the bulk of the sum \eqref{MLD_3}  is determined by
\begin{equation}
\argmin_{\bs{\xi} \in \mathcal{L}} \|\bs{\xi} + \bs{\eta}(\bs{s})+\bs{\xi}^{\perp}\|,
\end{equation}
 such that the logical post-correction becomes
\begin{align}
\overline{\bs{\eta}}
&=\bs{\eta}(\bs{s}) - \argmin_{\bs{\xi}^{\perp} \in \mathcal{L}^{\perp}}   \|\bs{\eta}(\bs{s})-\bs{\xi}^{\perp}\| \\
&=\bs{\eta}(\bs{s}) - {\tt CVP}\lr{\bs{\eta}(\bs{s}), \CL^{\perp}}.\label{ME_1}
\end{align}
For small error rates $\overline{\sigma}\rightarrow 0$, the most likely coset as computed in \texttt{MLD} is given by the most likely individual error consistent with the syndrome. In this limit \texttt{MLD} reduces to \texttt{CVP}.

As already noted, this is a classical computationally hard problem. In the following we show that 1. for GKP codes, MLD decoding is at least as hard as MED decoding and 2. MED decoding a concatenated (qubit-) GKP code implies a decoder for the corresponding qubit-code. 


\begin{lem}{($\mathtt{eMLD}\geq \mathtt{MED}$)}
    Given an oracle that evaluates 
    $$\mathtt{eMLD}\lr{\bs{x}, \bs{\xi}^{\perp}, \CL, \overline{\sigma}}=\Theta_{\CL+\bs{\xi}^{\perp}+\bs{x}}\lr{\frac{i}{2\pi \overline{\sigma}^2}},$$ $\mathtt{CVP}\lr{\bs{x}, \CL^{\perp}}$ can be solved efficiently.
\end{lem}
\proof
Denote by $\mathtt{DecCVP}\lr{\bs{x}, \CL, r}$ the decisional CVP problem that outputs $\mathtt{True}$ if $\text{dist}\lr{\bs{x},\CL}\leq r.$  This is polynomially equivalent to the optimization- and search variants of $\mathtt{CVP}$ \cite{Regev_lecture}. 
First notice that we generally have
\begin{align} \Theta_{\CL^{\perp}+\bs{x}}\lr{\frac{i}{2\pi \overline{\sigma}^2}}
&=\sum_{\bs{\xi}^{\perp} \in \CL^{\perp}/\CL} \Theta_{\CL+\bs{\xi}^{\perp}+\bs{x}}\lr{\frac{i}{2\pi \overline{\sigma}^2}}
\nonumber
\\
&\geq e^{-\frac{1}{2\overline{\sigma}^2} \text{dist}\lr{\bs{x}, \CL^{\perp}}^2}.
\end{align} 
If $\mathtt{DecCVP}\lr{\bs{x}, \CL, r}$ is true, then we further have 
\begin{equation}
e^{-\frac{1}{2\overline{\sigma}^2} \text{dist}\lr{\bs{x}, \CL^{\perp}}^2} \geq e^{-\frac{r^2}{2\overline{\sigma}^2} }
\end{equation}
for all $\overline{\sigma} \in \R$, and hence we can solve $\mathtt{DecCVP}\lr{\bs{x}, \CL^{\perp}, r}$ by checking if above condition is true for sufficiently small $\overline{\sigma} < r$.
Alternatively, w.l.o.g. assume that $\CL \subset \Z^n$ and $\bs{x} \in \Z$. Given access to $$\Theta_{\CL^{\perp}+\bs{x}}\lr{z}=\sum_{m\in \mathbb{N}} a_m e^{i\pi z m},$$ we can compute 
\begin{equation}
2 a_m=e^{m \pi \tau  }\int_{-1}^1 dt\, e^{-it\pi m}\Theta_{\CL^{\perp}+\bs{x}}\lr{t + i\tau}
\end{equation}
to evaluate $\{a_m\}$ for 
$m=1, \dots, M$, where $M$ can be bounded by Mikowski's convex body theorem, to find the smallest non-zero coefficient $a_m$. This solves optimization-\texttt{CVP} which is polynomially equivalent to 
its search version.\endproof
Note that here we did not show that the full $\mathtt{MLD}$ problem
\begin{equation}
\mathtt{MLD}\lr{\bs{x}, \CL, \overline{\sigma}}=\argmax_{\bs{\xi}^{\perp} \in \CL^{\perp} / \CL}\Theta_{\CL+\bs{\xi}^{\perp}+\bs{x}}\lr{\frac{i}{2\pi \overline{\sigma}^2}}    
\end{equation}
is hard. 

An important class of GKP codes are concatenated codes, which we have learned to correspond to construction A lattices in chapter~\ref{chap:Theory}. Lattices corresponding to concatenated GKP codes have the special structure of containing a trivial sublattice $\sqrt{2}\CL_{n\square}=2\Z^{2n}\subseteq \sqrt{2}\CL$, such that $\CL^{\perp}\subseteq \CL_{N\square}^{\perp}$ is a sublattice of the dual-trivial lattice. One can hence build a decoder by first applying a correction that takes an error $\bs{e}\in \R^{2n}$ back onto the dual sublattice $ \CL_{N\square}^{\perp}$, which is always exact and efficient due to its orthogonal structure, and then use the applied shift to inform a secondary correction, that takes the error back to $ \CL^{\perp}$. Pictorially, we implement a sequence

\begin{equation} 
 \mathbb{R}^{2n} \rightarrow  \CL_{N\square}^{\perp}  \xrightarrow{\mathtt{CVP}\lr{ \mu}} \CL^{\perp}.\label{eq:dec_conc}
\end{equation}

In  fact, in ref.~\cite[p. 450]{ConwaySloane}, it has (constructively) been shown that given a soft decoder for a binary code $C$, this procedure is always exact.

\begin{lem}[\cite{ConwaySloane},
p.~450]
\begin{equation}
\mathtt{CVP}\lr{\cdot,\,\Lambda \lr{C}}=\mathtt{Decode\lr{C}}.
\end{equation}
\end{lem}

\proof
$C$ is embedded in $\Z^n$ by identifying the (scaled and shifted) Construction A lattice $\Lambda\lr{C}=1-2C + 4 \Z^n $, where every bit string $\bs{b}\in C$ is mapped to $1-2\bs{b} \in \{-1,1\}^n$. In this representation we consecutively solve $\mathtt{CVP}\lr{\cdot, 4\Z^n}$ and then apply the soft decoder for $C$, which finds the closest transformed code word $\bs{c} \in 1-2C \in \{-1,1\}^n$ to input $\bs{x}' \in \mathbb{R}^n$. As both decoders are exact, with a little care (see ref.~\cite[p.~450]{ConwaySloane}), this solves CVP exactly. Note that the reverse direction is trivially true via the embedding of $C$ into $\mathbb{R}^n$ provided by Construction A and taking modulo $4\Z^n$. A hard decoder, that solves
\begin{equation}
\argmin_{\bs{c}\in C} d_H\lr{\bs{c}_b, \bs{x}_b} \label{eq:dec_classical}
\end{equation}
on binary input $\bs{x}\in \{-1,1\}^n$  is also derived from a soft decoder by noticing that $\|\bs{c}-\bs{x}\|_2^2=4 d_H\lr{\bs{c}_b, \bs{x}_b}$, where $\bs{x}_b$ represents the binary $\{0,1\}$ representation of $\bs{x}$ and $d_H$ is the Hamming distance.
\endproof


\subsection{Decoding NTRU-GKP codes}

We review the decoding problem for the NTRU-GKP code discussed in section~\ref{sec:NTRU}. Remember that the NTRU-GKP code had a natural concatenated structure, that is there is a trivial sublattice structure $\CL_{\rm triv}=\sqrt{dq}\Z^{2n} \subset \CL$ associated to the lattices describing NTRU-GKP codes, such that it is natural to split the decoding into two steps:  1. the correction of the error back onto one living on $\CL_{\rm triv}^{\perp}$ and, 2. correct back from $\CL_{\rm triv}^{\perp}$ to $\CL^{\perp}$.

A code state that (either through a natural error process or by deliberate modification) undergoes a displacement by
\begin{equation}
\bs{e}=\begin{pmatrix}
\bs{x} \\ \bs{y}
\end{pmatrix}
\end{equation}
gives rise to trivial syndrome

\begin{align}
\bs{s}_{\rm triv}&=\sqrt{d q}\bs{e} \mod 1.
\end{align} Due to the simple orthogonal structure of $\CL_{\rm triv}$ a first step of the correction is easily carried out by applying the correction $\bs{\eta}=-\bs{s}_{\rm triv}/\sqrt{d q} $.
After correcting for the trivial syndrome (associated to the underlying hypercubic GKP code) the remaining error is the unknown, but likely short, vector 
\begin{equation}
\bs{e}'=
\frac{1}{\sqrt{d q}}
\begin{pmatrix}
\bs{u} \\ \bs{v} 
\end{pmatrix} \in \CL_{\rm triv}^{\perp},\; \bs{u}, \bs{v} \in \Z^n.
\end{equation}
The residual error can be considered as living on the scaled $q$-ary "lattice''\footnote{strictly speaking, this is not a lattice but a finite subgroup of one when lattices are considered as infinite Abelian groups.} 
\begin{equation}
\CL_q=\frac{1}{\sqrt{d q}}\mathbb{Z}^{2n}_q
\end{equation}
dual to the trivial stabilizer lattice  and has a probability distribution induced by the trivial syndrome and correction 
\begin{equation}
P\lr{\bs{e}'}\propto \sum_{\bs{t} \in \CL_{\rm triv}} e^{-\frac{\lr{\bs{e}'+\bs{s}_{\rm triv}+\bs{t}}^2}{2\overline{\sigma}^2}}. \label{eq:res_prob}
\end{equation}
The remaining syndrome is
\begin{align}
\bs{s}&=MJ\bs{e}' \mod 1 \nonumber \\
&= \frac{1}{q}\begin{pmatrix}
\bs{v}-A\lr{h}\bs{u} \mod q \\ 0 \mod 1 
\end{pmatrix} .
\end{align}
We recognize that the first block of the syndrome $q\bs{s}_1=\bs{v}-A\lr{h}\bs{u} \mod q$  syndrome takes the same form as the ciphertext of the NTRU cryptosystem (compare to section~\ref{sec:NTRU}).  The position of the message is now taken by $\bs{m}=\bs{v}$ and the random vector is replaced by $\bs{r}=-\sigma(\bs{u})$. Following the standard NTRU decryption process now allows to obtain $\bs{v} \mod q \mod p$ as well as 
\begin{equation}
\bs{u}=qA^{\sigma}\lr{h^{-1}}(\bs{v}-q\bs{s}_1) \mod q
\end{equation} 

We can also decompose the remaining syndrome as

\begin{equation}
q\bs{s}=\begin{pmatrix}\bs{v} \\ -\bs{u} \end{pmatrix} + \underbrace{\begin{pmatrix}-A\lr{h}\bs{u} \\ \bs{u} \end{pmatrix}}_{\in  \CL_{\rm cs}^J}, \label{eq:decomp}
\end{equation}
where the vector on the RHS is element of the flipped NTRU lattice generated by the public basis
\begin{equation}
H^J=\begin{pmatrix}
q I & 0 \\ -A\lr{h} & I
\end{pmatrix}.
\end{equation}

Equation \eqref{eq:decomp} shows that a likely, i.e., small, error vector $\begin{pmatrix}\bs{v} \\ -\bs{u} \end{pmatrix}$ can indeed be obtained by solving $\mathtt{CVP}\lr{\CL^J, q\bs{s}}$, which can be expected to be at least as hard as finding the shortest lattice vectors in $\CL_{\rm cs}$ if not given the secret key to the corresponding instance of the NTRU cryptosystem.



\section{Quantum public key communication from NTRU-GKP codes}

The fact that \textit{decoding} the NTRU-GKP code essentially is equivalent to \textit{decrypting} the corresponding instance of the NTRU cryptosystem creates an interesting situation. Given access to the secret key of the NTRU instance, one can devise decoders built on the NTRU decryption mechanism. Different strategies to this end were  numerically investigated in ref.~\cite{conrad2023good}, to which the interested reader is referred. 

More interesting is the fact that \textit{without} access to the secret key, we can also expect that decoding the GKP code becomes as hard as breaking the corresponding instance of the cryptosystem. This suggests that the NTRU-GKP codes presented here may be used for both, quantum error correction and a new kind of quantum public key communication scheme at the same time. One may interpret NTRU-GKP codes as \textit{trapdoor decodable quantum error correcting codes}. That is, while stabilizer measurements can be performed and code states prepared using only access to the public key $h$, knowledge of the corresponding secret keys $\lr{f, g}$ of the NTRU cryptosystem is necessary for reliable and efficient decoding.

This observation naturally leads to the idea of trying to build a  \emph{private quantum channel} \cite{PrivateQuantumChannel} using the NTRU-GKP code.

The setup is that two parties, Alice and Bob (see figure~\ref{fig:PQC}), would like to communicate a quantum state over a public channel, where a potential evesdropper, Eve, could evesdrop on their message. 
If Alice and Bob were able to also exchange classical information over a classical secret channel inaccessible to Eve, e.g. if  they met up at some point very far in the past and interchanged this information if ever needed,  there is a simple strategy that allows them to also setup a private quantum channel using the \textit{(quantum-) one-time-pad} \cite{mosca2000private}. For every message -- an $n$-qubit quantum state $\ket{\psi_m}$ -- that Bob wants to send to Alice, Alice simply draws a $2n$-bit random bitstring $\bs{r}\in \Z_2^{2n}$ and secretly communicates it to Bob, who transmits $P\lr{\bs{r}}\ket{\psi_m}$ to Alice. That is, Bob perturbs the quantum state by a Pauli operator corresponding to the bitstring before sending it over to Alice. Knowing what she told Bob, Alice can then simply undo the Pauli operator by applying $P^{\dagger}\lr{\bs{r}}$ to her inbox to decrypt the quantum message.

This strategy is secure for the following reason. The eavesdropper Eve does not know $\bs{r}$. Hence, to Eve, the transmitted state looks like the message state with a \textit{random} Pauli operator applied to it. It is easy to show (see also the ``what is...'' box~\ref{whatis:groupprojector}, where this becomes an example of a \textit{state-twirl}) that this completely scrambles the quantum message and all that Eve is able to see is random gibberish. This strategy is known as the \textit{quantum one-time-pad} (quantum OTP) \cite{mosca2000private} and is amongst the most fundamental ideas in quantum cryptography.

The protocol proposed to set up a private quantum channel is similar to the quantum OTP and uses that the syndrome of the random displacement error encodes a ciphertext of the NTRU scheme. The idea is that Alice draws a random instance of the NTRU cryptosystem by sampling a secret key pair $(f,g)$ and tells Bob via a public classical channel the public key $h$, which we have seen earlier to fully specify the corresponding GKP code (they fix all other parameters $n, d, q, \Phi$ beforehand). When Bob now perturbs the state by a small error, Alice can measure stabilizers and decode with the help of the secret key. An evesdropper without the knowledge of the secret key, however, cannot.

The public key protocol, also described in fig.~\ref{fig:PQC} is sketched as follows:
\begin{enumerate}
\item Alice samples a secret key pair $(f,g)$ and computes the public key $h$, which is communicated to Bob.
\item Bob produces a code state described by the GKP code using the basis $\sqrt{d/{q}} H(h)$ and samples an error corresponding to a random message $\bs{e}_0=(-\bs{r}, \bs{m})/\sqrt{\lambda q}$, according to the specifications of the NTRU cryptosystem, by which he displaces the state. He transmits the state to Alice.

\item Alice measures the stabilizers and decodes the state, e.g., via the NTRU decryption routine or by employing Babai's algorithm as outlined before using the secret key pair $(f,g)$. She has hence received the to her unknown state from Bob through the error corrected private quantum channel. 

\end{enumerate}

 \begin{figure*}
     \centering
\includegraphics[width=\textwidth]{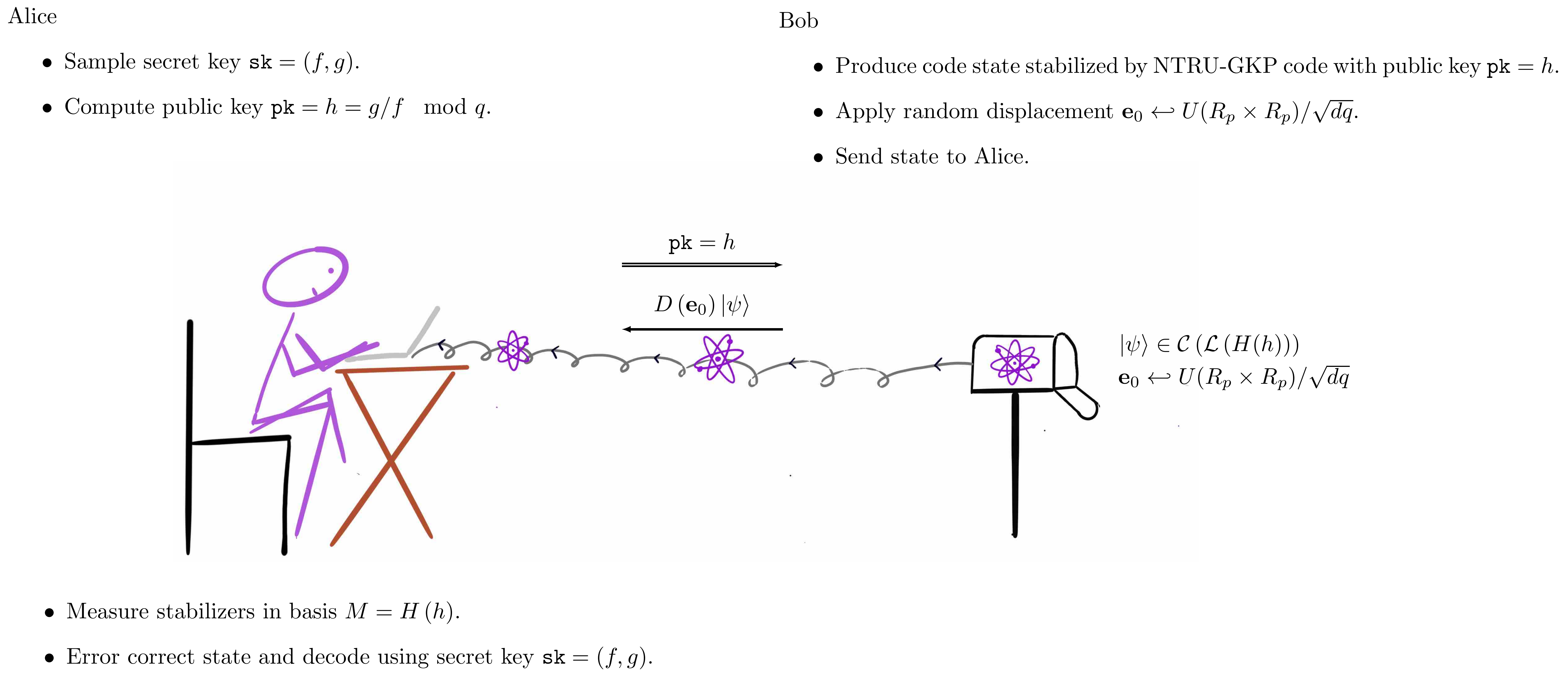}
     \caption{Outline of the private quantum channel established using the NTRU-GKP code as described in the main text. 
     }
     \label{fig:PQC}
 \end{figure*}

The security of this scheme under the assumption that classically decoding a quantum error correcting code -- i.e., finding small errors that are consistent with the syndrome --  is \textit{necessary} to retrieve its logical content is then immediately inherited from the corresponding classical NTRU cryptosystem. While we are not presenting a rigorous proof of the security of this scheme, supporting points are as follows.

\subsubsection{Necessity to decode.} In order to unambiguously obtain the logical code state, it is necessary to find a correction $\bs{e'}$ consistent with the syndrome such that $ \|\bs{e}_0+\bs{e'}\| \leq \Delta /2$. Since $\|\bs{e}_0\|_{\infty} \leq 1/\sqrt{d q}$ and the smallest element in $\CL^{\perp}$ is of length $\Delta$, this amounts to decrypting the NTRU ciphertext in the syndrome to identify $\bs{e}_0$. 
A first cryptanalysis goes as follows.
Let $\ket{\overline{\psi}}$ be a logical code state vector specified by a GKP-NTRU code with lattice $\CL$. We examine the eigenvalue of logical Pauli observables
obtained when the initial code state is encrypted by applying the random displacement $D\lr{\bs{e}_0}$, a syndrome $\bs{s}\lr{e_0}=MJ\bs{e}_0 \mod 1$ is obtained and a generic correction via $\bs{\eta}=(MJ)^{-1}\bs{s}\lr{\bs{e}_0}$ is applied. 
With 
\begin{align}
M^{-1}=\frac{1}{\sqrt{dq}}
\begin{pmatrix}
    qI & -A_{\Phi}\lr{h} \\ 0 & I
\end{pmatrix},
\end{align}
this yields a generic correction
\begin{align}
    \bs{\eta}=\frac{1}{\sqrt{d q}} 
    \begin{pmatrix}
        0 \\ \bs{c},
    \end{pmatrix}
\end{align}
where $\bs{c}=\bs{m}+A_{\Phi}(h)\bs{r} \mod q$ is the associated NTRU ciphertext. 
The total remaining error after correction thus is 
\begin{align}
\bs{e}_0-\bs{\eta}=\frac{-1}{\sqrt{d q}}
    \begin{pmatrix}
         \bs{r} \\ \bs{c}-\bs{m}
    \end{pmatrix}=    \begin{pmatrix}
         \bs{r} \\ A_{\Phi}(h)\bs{r} \mod q
    \end{pmatrix}.
\end{align}
We compute
\begin{equation}
M^{\perp}J\lr{\bs{e}_0-\bs{\eta}}=\frac{1}{d}\begin{pmatrix}
0 \\ \bs{r}  
\end{pmatrix} \mod q/d,
\end{equation}
which shows that for an input code state vector $\ket{\overline{\psi}}$, after encoding and generic correction, the eigenvalues of logical Pauli operators corresponding to rows $i=n+1,\dots, 2n$ in $M^{\perp} $ obtain a random phase $e^{i\frac{2\pi}{d}r_{i-n}}$. This observation suggests that, for $d=2$, without access to the random string $\bs{r}$ embedded in the NTRU ciphertext in every instance, the quantum state is effectively projected onto a state that is diagonal in the logical Pauli-$Z$ basis and quantum superpositions are washed out. This situation is similar to that of half a quantum OTP, where only one type (either $X$ or $Z$) of Pauli operators is used in the encryption. 

\subsubsection{Orthogonality.} For a fixed quantum state vector $\ket{\overline{\psi}}$, different error realizations $D\lr{\bs{e}_0}$ where $\|\bs{e}_0\|< \Delta /2$ map the state to mutually orthogonal states (sectors of the QECC). This is guaranteed by the quantum error correction conditions. Without applying suitable corrections, separate encodings of the same logical quantum state vector $D\lr{\bs{e}_i}\ket{\overline{\psi}}$ are expected to appear uncorrelated.

\subsection{Quantum cryptography with computational security}
The design of the NTRU-GKP codes in sec.~\ref{sec:NTRU} is flexible enough to allow the use of versions of the NTRU-cryptosystems that are secure from quantum attacks under the computational assumption that $\tt SVP$ is hard on a quantum computer \cite{StehleSteinfeld, Regev_2005}, such that the protocol proposed here exemplifies the idea of designing quantum cryptographic protocols using \textit{computational} assumptions. This idea stands in contrast to usual designs of quantum cryptographic protocols, which are typically designed to be \textit{information theoretically secure}, such as the quantum OTP. The upshot of this approach is that computational security may suffice for many tasks considered in practice and escapes known no-go theorems for information theoretically secure protocols. Typical quantum cryptographic communication protocols, such as quantum key distribution (QKD) \cite{Ekert1991}, require Alice and Bob to share some a priori entangled state which also needs to be distributed securely somehow. The protocol presented here is a promising approach towards resolving this requirement.

\subsubsection{Are quantum decoders more powerful than classical decoders?}

The core of proving security for the private quantum channel outlined above is summarized by the question of whether the classical decoding problem fully reduces to the quantum decoding problem. If that was the case, then solving the quantum decoding problem would always allow to break the corresponding instance of the NTRU cryptosystem. From the arguments presented so far this seems very likely the case, but it will be necessary to treat this question more rigorously nevertheless to make strong security claims.
It makes sense to define these decoding problems as follows

\begin{mydef}[Quantum decoding problem]\label{def:quantum_dec}
Let $\CC\subset \CH$ be the code space associated to a stabilizer group $\CS=\langle \CG\rangle$ which is finitely generated by the set $\CG$  and let
\begin{equation}
\CN\lr{\cdot}=\sum_{E\in \CE} p(E) E \cdot E^{\dagger}
\end{equation}
be a noise channel with error operators $E\in \CE$ and probability distribution $p:\; \CE \rightarrow [0,1]$. Let $\ket{\overline{\psi}}\in \CC$ be a code state and $\overline{\rho}=\ketbra{\overline{\psi}}$. The quantum decoding problem is solved by black box with the following in- and outputs.

\begin{itemize}
\item \tt Input: $\CN\lr{\overline{\rho}}$, $\CE$, $p$.
\item \tt Ouput: $\argmax_{\ket{\psi} \in \CC} \CF\lr{\ketbra{\psi}, \overline{\rho}}$,
\end{itemize}
where $\CF(X, Y)$ is the fidelity in $\CC$.
 
\end{mydef}

\begin{mydef}[Classical decoding problem]\label{def:class_dec}

Let $\CC\subset \CH$ be the code space associated to a stabilizer group $\CS=\langle \CG\rangle$ which is finitely generated by the set $\CG$ and centralizer $\CC\lr{\CS}$.
Let
\begin{equation}
\CN\lr{\cdot}=\sum_{E\in \CE} p(E) E \cdot E^{\dagger}
\end{equation}
be a noise channel with error operators $E\in \CE$ and probability distribution $p:\; \CE \rightarrow [0,1]$. Further assume that every element $E\in \CE$ is such that 
\begin{equation}
g_k^{\dagger}E^{\dagger} g_k E = e^{is_k(E)}I \; \forall k =1\hdots |\CG|\label{eq:destabilizer2}
\end{equation}
for a function $\bs{s}:\, \CE \rightarrow \R^n$. Let $\ket{\overline{\psi}}\in \CC$ be a code state and let $\eta:\, \R^n \rightarrow \CE $ be an arbitrary but fixed inverse to $\bs{s}$. The classical decoding problem is solved by black box with the following in- and outputs.

\begin{itemize}
\item \tt Input: $\bs{s}$, $\CE$, $p$.
\item \tt Ouput: $\argmax_{L \in \CC\lr{\CS}/\CS}  \sum_{S\in \CS} p(\eta \lr{\bs{s}}+S+L)$.
\end{itemize}
\end{mydef}

By design, we already know the opposite inclusion. 
\begin{lem}
Under the noise assumptions of the classical decoding problem, the   quantum decoding problem reduces to the classical decoding problem.
\end{lem}
\proof
Measure the stabilizer generators in $\CG$. This collapses the state $\CN\lr{\overline{\rho}}$ to a mixture of terms with a fixed syndrome $\bs{s}$ corresponding to the measurement outcome. Applying the classical MLD decoder in def.~\ref{def:class_dec} yields the optimal fidelity by definition.
\endproof

In fact, the above statement can be made even stronger. When the actual noise model does not adhere to eq.~\eqref{eq:destabilizer2} but contains coherent noise processes that are linear combinations of elements in an error basis that does satisfy eq.~\eqref{eq:destabilizer2}, the measurement collapse will ``diagonalize'' the noise channel into one that can be handled by the classical decoder, albeit with a potential adaption on the prior distribution $p$, which can be computed.
The different versions of the decoding problem (specialized to the case of GKP codes) are sketched in fig. ~\ref{fig:decoding_problems}.

\begin{figure}
\center
\includegraphics[width=\textwidth]{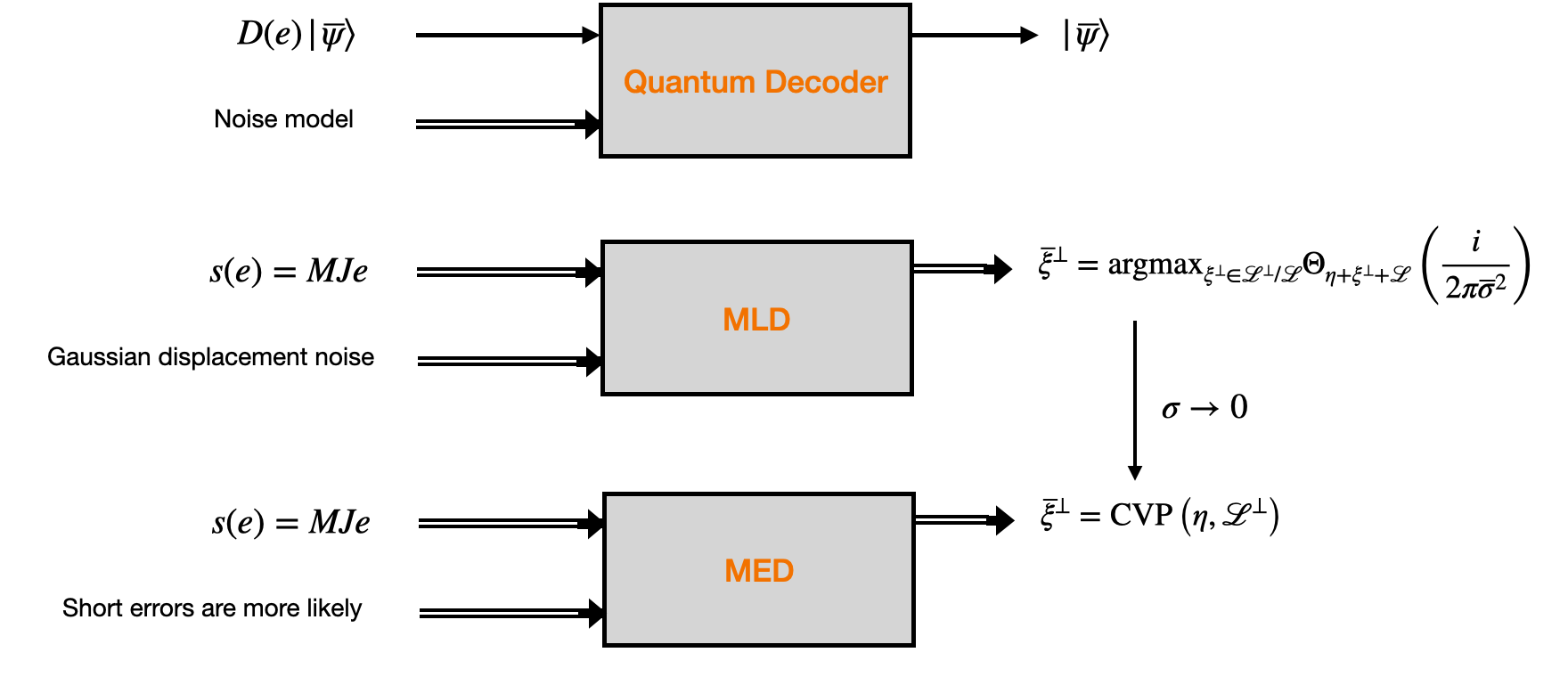}
\caption{A (top) quantum decoder, a (middle) classical MLD decoder and (bottom) a MED decoder for the GKP code. In order to devise a quantum decoder, one usually imposes additional structure by forcing the quantum decoder to perform stabilizer measurements and thus reduces the quantum decoding task to the MLD decoding task in the middle.}\label{fig:decoding_problems}
\end{figure}

Interestingly, the converse statement is not so straightforward and remains an open problem. The difficulty in proving the converse stems from the fact that the quantum decoder in the way defined above does not output any information about \textit{which} correction was applied. To infer this, a smarter strategy needs to be employed, such as  quantum process tomography \cite{Chuang_1997, Kliesch_2019}. We do not attempt such an analysis here but leave this as an interesting open question.

\section{The Dream}

In this chapter we have extensively discussed the decoding problem for quantum error correcting codes and found that the decoding problem for NTRU-GKP codes proposed in section~\ref{sec:NTRU} possesses an intimate link to the decryption of corresponding instances of the NTRU cryptosystem. We discussed a potential use of this link to build a quantum cryptographic protocol, a private quantum channel, which bases its security on computational assumptions where it is presumed to be hard for an adversary without a secret key to perform quantum error correction. This idea is in principle very similar to the seminal idea of McElies \cite{McEliece1978APK}, where a classical error correcting code is drawn by applying a permutation to Goppa codes such that inverting the permutation becomes necessary for decrypting (i.e. decoding an ``error''), but becomes more profound when applied to quantum error correction.

My dream along this line of research would be to devise a so-called \textit{blind delegated quantum computation} protocol, where a (classical) client attempts to run universal quantum computations on a quantum server but can only communicate with the server through a classical channel. It is known that this task is very unlikely to be possible \cite{aaronson2019complexitytheoretic} when information theoretic security is required. Whether a softer demand of security under computational assumptions is possible, however, remains unclear to the best of my knowledge. The broad idea is, similar to our setup of the private quantum channel, to let the client sample a random quantum error correcting code which the server is instructed to use, but let the client retain a secret key that renders the computationally difficult decoding problem easy. Upon termination of the algorithm run by the server, the server measures the syndromes as well as the computational output, which are both communicated to the client. Under possession of the secret key, the client now can decode and obtain the error corrected outcome to her computation.
There are many challenges that come along with the design of such a protocol, a suitable class of codes needs to be identified and computational steps that the server is supposed to run need to be communicated to the server in an encoded fashion such that the server cannot tell what computation exactly is executed. And finally, security needs to be proven rigorously. 

This is an exciting challenge, and worthwhile to dream about.

\chapter{Implementation of GKP error correction}\label{chap:implementation}
\blfootnote{Sections~\ref{sec:FSE} and~\ref{sec:Steane} are extensions of work presented in ref.~\cite{Terhal_2020}. Sec.~\ref{sec:floquetGKP} presents content published in and is adapted from ref.~\cite{Conrad_2021} while the general perspective on twirling presented there is part of work in preparation in ref.~\cite{Conrad_2024_shadow}.}

In this chapter we are going to discuss implementation of GKP codes. In contrast to the previous chapters, this is going to place a stronger emphasis on the physics of the relevant quantum systems and how it interplays with the structure of the GKP code. To implement the GKP code, the requirements to a physical system are, roughly, 

\begin{enumerate}
\item  It has a configuration space isomorphic to $\R^n$, and an associated quantum phase space described by $\R^{2n}$. 
\item There is a mechanism to measure stabilizers, given by displacement operators $D\lr{\bs{\xi}}$ ,
\item There is a mechanism to measure logical operators $D\lr{\bs{\xi}^{\perp}}$,
\item there is a mechanism to implement Gaussian unitary operations.
\end{enumerate}

This is of course only a very crude sketch of sensible desiderata. Depending on the concrete ambition -- whether one wants to implement a long-lived quantum memory, perform some quick computations or aims at using the GKP code for communication \cite{Noh_Capacity, Noh_2020} or metrology \cite{Duivenvoorden_Sensor} these should be adapted. Notably, this list does not include preparation of GKP states, i.e. code states of the GKP code. The reason is that stabilizer measurements and GKP state preparation enjoy a circular relationship. Provided the ability to perform stabilizer measurements, one can prepare code states of the GKP code by simply measuring stabilizers and applying corrective shifts. Similarly, as we will see in sec.~\ref{sec:Steane}, the ability to prepare code states allows to implement GKP stabilizer measurements when combined with certain Gaussian unitary operations and homodyne measurements, i.e. measurements of $\bs{\hat{q}}$ and $\bs{\hat{p}}$. We begin this chapter by reviewing some basics of photonic and superconducting systems to identify the physics behind the phase-space variables and discuss some basic tools needed to implement the GKP code in these systems.

\section{Physical systems for GKP codes}

\begin{figure}
\center
\includegraphics[width=.8\textwidth]{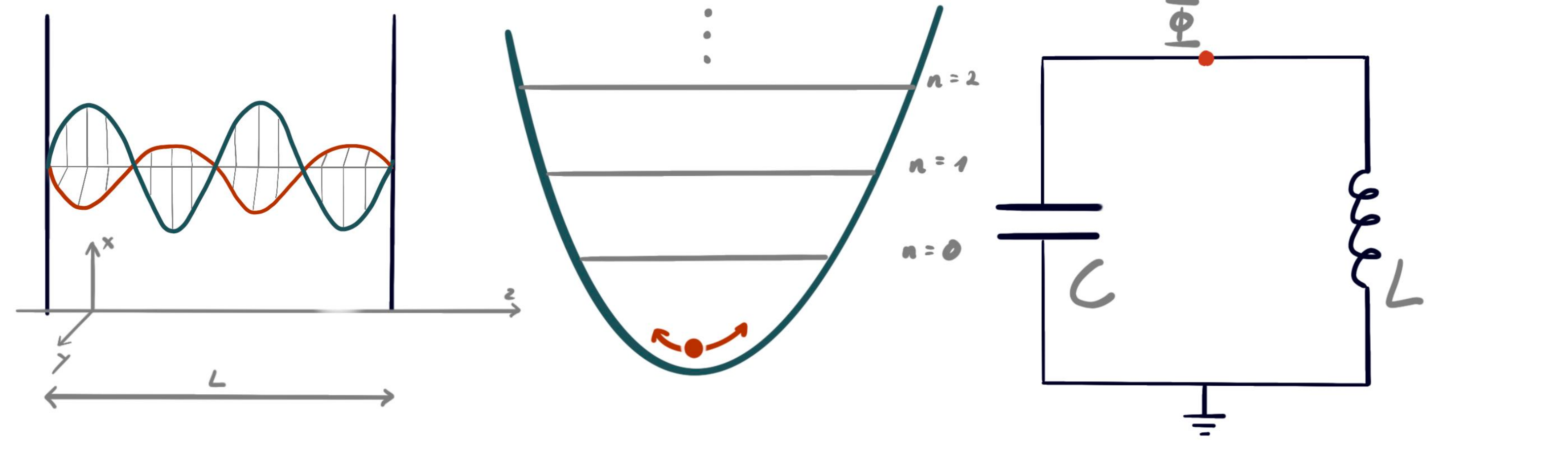}
\caption{Illustration of a photonic mode in a cavity of length $L$ (left), the energy landscape of a quantum harmonic oscillator (middle) and a superconducting LC-circuit (right).}\label{fig:imp}
\end{figure}

The most popular systems considered for the implementation of the GKP code are photonic systems, associated to optical photons propagating in either optical fibers or free space \cite{Konno_2024}, or superconducting circuits \cite{Campagne_Ibarcq_2020}. These two platforms are distinctive: photonic systems can in principle operate at room temperature, only requiring more intricate cooled down components for state generation, and are equipped with a natural means to perform homodyne detection -- that is, direct measurement of the quadratures $\bs{\hat{q}}$ and $\bs{\hat{p}}$ of the encoded modes. In contrast to what is possible in superconducting architectures, however, optical photon-photon interactions are difficult to engineer. Non-linear unitary evolutions of the quadratures are very difficult to generate, and it becomes necessary to compartmentalize any circuit to shift as many non-Gaussian resources as possible into the  state-preparation part, where in particular non-Gaussian measurements help in the generation of the states. Superconducting circuits, through the existence of the Josephson junction, naturally possess the means to implement strongly non-linear Hamiltonian evolutions, but require a high level of cooling and homodyne measurement of the quadratures are slow and costly. We briefly discuss the physics and identification of quadratures for these systems, where we follow the discussion in ref.~\cite{Gerry_Knight_2004} for basic quantum optics and ref.~\cite{Ciani_2024, Girvin2014} for our discussion on superconducting circuits. See fig. ~\ref{fig:imp} for an illustration of these platforms.

\subsection{Photonics}
The dynamics of a single mode of an electromagnetic field confined in a one dimensional cavity of length $L$ (see fig. ~\ref{fig:imp}) is classically described by Maxwells equations \cite{Gerry_Knight_2004}
\begin{align}
\nabla \times \bs{E} &= -\frac{\partial \bs{B}}{\partial t},\hspace{1cm} &&\nabla \times \bs{B} =\frac{1}{c^2}\frac{\partial \bs{E}}{\partial t},\\
\nabla \cdot \bs{E}&=0, \hspace{1cm}  &&\nabla\cdot \bs{B} =0.
\end{align}
Assuming a polarization $\bs{E} || \bs{\hat{e}}_x$, solutions to these equations are given by the standing waves
\begin{align}
E_x(z,t)&=\sqrt{\frac{2\omega^2}{V\epsilon_0}}q(t) \sin\lr{kz},\\
B_y(z,t)&=\frac{1}{c^2k}\sqrt{\frac{2\omega^2}{V\epsilon_0}}\dot{q}(t) \cos\lr{kz},
\end{align}
with frequency $\omega=kc$ and the classical Hamiltonian of the system can be derived as \cite{Gerry_Knight_2004}
\begin{equation}
H(q,p)=\frac{p^2+\omega^2 q^2}{2},\; \omega\in \frac{c\pi}{L}\Z.
\end{equation}
In the process of this derivation, we can determine the canonical position operator with $q(t)$, the time-dependent component of the electric field and the canonical momentum $p(t)=\dot{q}(t)$ is identified with the time dependent component of the magnetic field. Canonical quantization lifts these quadrature variables to operators on an infinite dimensional Hilbert space with commutation $\lrq{\hat{q}, \hat{p}}=i\hbar$. We hence find that the quantum quadratures of a photonic mode with fixed $\omega$ are simply provided by its electric-  and magnetic field components at a fixed point in space.  

The implementation of the GKP in photonic systems typically proceeds by implementing an intricate system of single-photon sources and photon counters to produce so-called cat states, which are non-Gaussian states and then processed to GKP states through a breeding protocol \cite{Weigand_2018} that successively builds up GKP states by interlacing them with cat states using beam splitters and performs homodyne measurements on one of the output legs. Due to the lack of strong non-linear Hamiltonian elements to process existing photonic states, efforts to implement quantum computing in photonic systems are typically concentrated around the engineering of powerful GKP resource states such that the actual steps of the computation are carried out using sequences of (classically) adaptive measurements on such states. This is a computational model called \textit{measurement based quantum computation} (MBQC) \cite{Bourassa_Blue_2021}.

\subsection{Superconducting oscillators}

An electric LC system (on the right in fig.~\ref{fig:imp}) is a classical example of a harmonic oscillator system. Electrons loaded on the capacitor $C$ establish a voltage drop across the system and incur a current flowing through the inductor $L$, which in turn generates an inductive voltage pointing in the reverse direction. The dynamics of this system is naturally described by a harmonic oscillating behavior of the flux $\Phi$ across the branches of the system (set the flux of the ground to $\Phi_{\rm ground}=0$). Kirchoff's laws imply that the sum of branch fluxes around a closed loop equals the total magnetic flux piercing the loop, which we for now assume to be zero. The canonically conjugate variable to $\Phi$ is the charge $Q$ that loads the capacitor and the dynamics of the system is described by the classical Hamiltonian
\begin{equation}
H=\frac{Q^2}{2C}+\frac{\Phi^2}{2L}, \label{eq: LC_class}
\end{equation}
which oscillates with natural frequency $\omega=\sqrt{LC}^{-1}$

Modern experimental techniques allow to engineer superconducting LC oscillators on microscopic scales $~O(100\mu m)$ with capacitances $~O(1 pF)$ and inductances $~O(10 nH)$, large enough that the wavelength associated with the frequency exceeds the dimensions of the circuit  \cite{Devoret1997}. When the system is of high quality and sufficiently cooled $~ O(20mK)$, thermal fluctuations become smaller than the energy gaps $k_B T\ll \hbar \omega$, warranting a quantum mechanical treatment of these systems. Canonical quantization lifts the variable $\Phi, Q$ to the corresponding operators with canonical commutation relation $\lrq{\hat{\Phi}, \hat{Q}}=i\hbar$ and the Hamiltonian in eq.~\eqref{eq: LC_class} ascends to that of a quantum harmonic oscillator
\begin{equation}
H=\frac{\hat{Q}^2}{2C}+\frac{\hat{\Phi}^2}{2L}=\hbar\omega\lr{\hat{n}+\frac{1}{2}}\label{eq: LC_quant}
\end{equation}
with a suitably chosen annihilation operator $\hat{a}=\frac{1}{\sqrt{2L\hbar\omega}}\hat{\Phi}+\frac{i}{\sqrt{2C\hbar\omega}}\hat{Q}$ and Fock number operator $\hat{n}=\hat{a}^{\dagger}\hat{a}$.

A special element that can be engineered in superconducting systems is the so-called Josephson Junction (JJ), which comprises of two superconductors separated by a thin oxide layer. Although the superconducting materials are isolated from each other, Cooper pairs may still tunnel through the barrier to give rise to a current

\begin{equation}
I(t)=I_0 \sin\lr{\frac{2e}{\hbar} \Phi(t)},
\end{equation}
where $\Phi$ is the flux across the junction, which can be described by a (quantum) Hamiltonian contribution
\begin{equation}
H_{JJ}=-E_J\cos\lr{\frac{2\pi}{\Phi_0}\hat{\Phi}},
\end{equation}
with $\Phi_0=\frac{h}{2e}=2eR_Q$ the flux quantum and $R_Q$ the resistance quantum. The Josephson junction may be thought of as a non-linear inductance: performing a Taylor expansion allows to write $H_{JJ}=\frac{\hat{\Phi}^2}{2L'}+\epsilon \hat{\Phi}^4 +O(\hat{\Phi}^6)$.  When the inductance in the harmonic oscillator Hamiltonian in eq.~\eqref{eq: LC_quant} is replaced by such a junction, we obtain a system that for large capacitances $C$ behaves very similar to a quantum harmonic oscillator, but with a non-uniform energy separation $\Delta E_{12}> \Delta E_{01}$. This is the typical mechanism used to engineer a superconducting qubit \cite{Koch}:  by creating a large energy gap between levels $n=1$ and $n=2$, one effectively isolates the two level system $\lrc{\ket{n=0}, \ket{n=1}}$, which is then treated as the logical qubit in the ``trivial encoding''.
This encoding is somewhat wasteful:  we have started with an infinite dimensional Hilbert space and used the perturbative anharmonicity of the Joshephson junction to ``throw away'' infinite but two levels of the system to derive a qubit. 

Implementations of the GKP code in superconducting architectures typically take advantage of the strong non-linearity differently. A typical strategy is to use the trivially encoded superconducting qubits together with controlled displacement operation $cD\lr{2^k\bs{ \xi}}=\ketbra{0}_q\otimes  I+ \ketbra{1}_q\otimes  D\lr{2^k\bs{ \xi}}, k=1\hdots K$ between the qubit and the oscillator to perform quantum phase-estimation on the displacement operators \cite{Terhal_2016, Terhal_2020} so to measure the corresponding GKP stabilizer operators $D\lr{\bs{\xi}}$. Since the eigenvalues of GKP stabilizers $D\lr{\bs{\xi}}$ are continuous, in this approach one would need to implement phase estimation to infinite order $K\rightarrow \infty$ to obtain a strong measurement of the stabilizer operator and collapse the input state to a (shifted) GKP state. This is of course not possible and existing experiments \cite{Campagne_Ibarcq_2020, Sivak_2023} focus on implementing only very low levels (i.e. $K=1$) of phase estimation for the preparation and stabilization of GKP states, which has already demonstrated impressive performances with logical qubit lifetimes of $\approx 1-2 ms$ \cite{Sivak_2023}.

In section ~\ref{sec:floquetGKP} we will see how one uses the fact that the cosine term 
\begin{equation}
 \cos\lr{\sqrt{2\pi}\bs{\xi}^TJ\bs{\hat{x}}}=\frac{1}{2}D\lr{\bs{\xi}} +h.c.,
\end{equation}
 the Hermitian part of GKP stabilizers, emerges naturally from the structure of the Josephson junction and how this structure can be used to engineer a GKP-encoded superconducting qubit that takes advantage of the infinitude of the Hilbert space more efficiently. 

 Interesting passive implementations of the GKP code on superconducting systems not further discussed here can also be derived via the use of a non-reciprocal Gyrator element that allows to engineer a system mimiking physics present in the quantum Hall effect \cite{Rymarz2021} or by advantage of reservoir engineering \cite{lachancequirion2023autonomous}. We refer to the cited references for more insight on these approaches.

\subsection{Other notable systems}

Since quantum harmonic oscillators are somewhat omnipresent in physics, it is no surprise that implementations of the GKP code are not solely limited to photonics and superconducting systems. The first experimental demonstration of preparation of GKP code states was in fact carried out on a mechanical oscillator \cite{Fluehmann_2019}, where the physical vibration of a magnetically trapped ion represents the motion of the quantum harmonic oscillator and, under application of a strong laser, internal degrees of freedom of the ion couple to its motion and allow the implementation of a displacement operator conditioned on its internal state. On a high level, this again allows the implementation of a mechanism similar to the controlled-displacement based phase estimation process to project onto motional states onto eigenstates of the displacement operator.

In a similar spirit, it was also proposed to implement the GKP code on a nanomechanical oscillators \cite{Weigand_2020}. Here, a nanomechanical oscillator couples to the electromagnetic (EM) field of a photon via photon pressure, resulting in an interaction Hamiltonian of the form $H_I=\gamma \hat{n}_{\rm EM}\hat{q}_{\rm osc}$. Under unitary evolution, the phase of the EM field is rotated by an angle $\phi \mapsto \phi + \gamma t q_{\rm osc}$ given by the position of the oscillator. By implementing a measurement of the phase accumulation on the oscillator relative to an initial auxiliary state with a fixed phase, one is then able to measure $\gamma t q_{\rm osc} \mod 2\pi$, which is a modular measurement of the quadrature of the oscillator. 

This is only a short excerpt of a very long list of interesting implementations of the GKP code, which is notably extended by the identification of GKP code states with states present in the quantum Hall effect \cite{Rymarz2021, Ganeshan_2016, GKP}. In the realm of solid state physics it is no surprise that the translation invariant GKP states are realized in the physics of solid state systems, which typically are assumed with natural lattice-like translational invariances. It becomes an interesting quest to relate properties of the GKP code with more traditional physical quantities in the hope of either finding new implementations of the GKP code or with hope to use the GKP code to simulate those intricate physical systems.

\section{GKP stabilizer readout and bottlenecks}

The task of GKP stabilizer measurements is that of implementing either direct projective of the displacement operators $D\lr{\bs{\xi}}=e^{-i\sqrt{2\pi}\bs{\xi}^TJ\bs{\hat{x}}}$ or equivalently, any measurement of the physical quadratures $\bs{\hat{x}} \mod \sqrt{2\pi} \CL^{\perp}$ that do not reveal information about the absolute quadratures or that distinguish between shifts via vectors in the dual lattice $\CL^{\perp}$. The first way, via phase estimation, has been outlined above and we have identified the infeasibility to implement full projective stabilizer measurements from the simple inability to measure the continuous eigenvalue of the displacement operator to perfect precision. The second way, that was also historically first analysed in refs.~\cite{GKP, Glancy_2006}, ``copies'' the information of a selected quadrature onto an auxiliary GKP state, which already possesses a translation invariance along the selected quadrature to obfuscate the copied down absolute quadrature information such that neither the absolute quadrature value nor potential logical information can be measured upon homodyne measurement of the auxiliary mode. This trick carries the name \textit{Steane error correction}, following ref.~\cite{Steane_1997}, and is a generally applicable method to measure stabilizers for quantum error correcting codes by implementing a logical CNOT gate coupling the storage system to an auxiliary code block initialized in a state that hides the copied down logical information from being read out. 

Steane error correction for the GKP code is closely related to a strategy dubbed \textit{Knill error correction} after refs.~\cite{Knill2005, knill2004faulttolerant}, which implements a quantum teleportation circuit using a logically encoded Bell state and a decoded logical bell-measurement, such that errors are essentially ``filtered out" in the teleportation process. It turns out that the implementations of Steane- and Knill error correction for the GKP code are effectively equivalent on the full CV level once the choice of correction in Steane error correction is fixed suitably;  this is shown in sec.~\ref{sec:Knill}.

\subsection{Finite squeezing error}\label{sec:FSE}
The bottleneck in Steane- and Knill error correction is presented by the need of high-quality auxiliary GKP states to facilitate the measurement. This is a suitable point to address the elephant in the room:  GKP code states, i.e. quantum states that are translation invariant under a full rank generating set of displacement operators, do not actually exists. Throughout this work we have acted as if they existed, ignoring the fact that translation invariant states would also necessarily come with infinitely extended support, occupy infinite photon number states and would fail to be normalizable within a $L^2\lr{\R^n}$ Hilbert space.\footnote{The better way to treat them would be to use the Segal-Bargmann representation, which has a different normalization condition and doesn't interpret the representing functions as immediate physical quantities.} In practice, one is ever only able to produce approximations to GKP code states relative to a regularizing parameter such as the order of phase estimation used to distil the state from a reference state like the vacuum, or by regularizing via a photon-number cut-off that can be chosen as smooth or hard. The analytical strategy is then to treat this regularization as an ``error'' applied to an exact code state \cite{GKP, Tzitrin_2020, Terhal_2020}.\footnote{Which is not actually a \textit{state} in terms of a typical use of quantum mechanics. This ``error" also violates the quantum error correction conditions as it renders previously orthogonal logical states non-orthogonal. The breakdown of the quantum error correction conditions is however manageably small in the regularization parameter, such that one can regard the realistic version of the GKP code as an approximate quantum error correcting code \cite{Tzitrin_2020, Schumacher2002}.}

A good example for this approach is given for $n=1$ modes by the unique stabilizer state of $\CL=\mathbb{Z}^2$, \footnote{Also regarded as the sensor- \cite{Duivenvoorden_Sensor} or ``qunaught'' state \cite{Mensen_2021}.} 
\begin{equation}
\ket{\varnothing}=\sum_{n\in \Z} \ket{\sqrt{2\pi} n}_q.
\end{equation}
This state is an infinite sum of improper states, $\ket{\sqrt{2\pi} n}_q$ that evaluate to Dirac deltas $\delta \lr{q-\sqrt{2\pi} n}$ in their position representations. As it fails to be normalizable and has a physically impossible photon occupation $\braket{\varnothing|\hat{n}|\varnothing}=\infty$, this is hardly a physical state. To obtain a physical state, a sensible approximation is provided by 
\begin{equation}
\ket{ \tilde{\varnothing}}=N_{\beta} e^{-\beta \hat{n}} \ket{\varnothing},
\end{equation}
where $N_{\beta}$ is a normalization parameter. The regularization operator $e^{-\beta \hat{n}}$ implements an exponential damping of the state along the Fock basis, which by writing $\hat{n}=\frac{\hat{p}^2+\hat{q}^2-1}{2}$, can also be realized as a Gaussian envelope applied to the state in phase space. Using a result of ref.~\cite{CahillGlauber} and adapting to our conventions, it can be shown that
\begin{equation}
\Tr\lrq{D^{\dagger}\lr{\bs{x}}e^{-\beta \hat{n}}}=\frac{1}{1-e^{-\beta}}e^{-\frac{\pi}{2\tanh\lr{\beta / 2}}\|\bs{x}\|^2 }, \label{eq:FSE_disp}
\end{equation}
such that the regularization operator can be written as coherent superposition over Gaussian distributed displacements
\begin{equation}
e^{-\beta \hat{n}}=\frac{1}{1-e^{-\beta}}\int d\bs{x}\, e^{-\frac{\pi\|\bs{x}\|^2}{\Delta^2}} D\lr{\bs{x}},
\end{equation}
where $\Delta^2=2\tanh\lr{\beta /2}$ is the variance of the coherent ``displacement error'' incurred by this process \cite{GKP}. 

The wave function of the resulting state is then

\begin{align}
\beta_{\Delta}(q)&=\braket{q|\varnothing_{\Delta}} \nonumber \\
&=N_{\Delta} e^{-\frac{\Delta^2}{2}q^2}\sum_{n\in \Z}e^{-\frac{1}{2\Delta^2}(q-\sqrt{2\pi})^2} \nonumber\\
&=N'_{\Delta} e^{-\frac{\Delta^2}{2}q^2} \int_{-\infty}^{\infty} dx\, g_{\Delta}\lr{x}\Sha_{\sqrt{2\pi}}\lr{q-x}, \nonumber \\ 
&=N''_{\Delta} g_{\Delta^{-1}}(q) (g_{\Delta} *  \Sha_{\sqrt{2\pi}})(q),\label{eq:beta}
\end{align}
where $\Sha_{c}\lr{x}=\sum_{n\in \Z}\delta\lr{x-nc}$ is the Dirac comb, $g_{\Delta}\lr{x}=(2\pi\Delta^2)^{-1/2}e^{-x^2/2\Delta^2}$ the Gaussian distribution with variance $\Delta^2$ and $N^{(\cdot)}_{\Delta}$ are normalization parameters. The $\Z^2$ lattice is $\pi/2$ rotation symmetric, which implies a rotation symmetry by $e^{-i\frac{\pi}{2}\hat{n}}$ of the states $\ket{\varnothing}$ and $\ket{ \tilde{\varnothing}}$. Since this $\pi/2$ rotation is equally understood as a change of basis between the position- and momentum representation of the state (i.e. a Fourier transform), it follows that $\braket{p|\varnothing_{\Delta}}=\beta_{\Delta}(p)$, resolved along the momentum axis, has exactly the same wave function.

The presentation above shows that the approximate GKP wave function is composed of two contributions: the $\sqrt{2\pi}$ translation invariant part, which is rendered ``fuzzy" by the convolution with a Gaussian distribution of variance $\Delta^2$ and a Gaussian part with variance $\Delta^{-2}$, that weakly biases the weight towards the center $q=0$. In the limit $\Delta\rightarrow 0$ this state expresses the unique joint eigenstate of the commuting displacements $S_p=e^{-i\sqrt{2\pi}\hat{p}},\, S_q=e^{i\sqrt{2\pi} \hat{q}}$ and the dependency on small deviations of $\Delta$ away from this limit has an interesting interpretation. Along the $p-$quadrature, the variance $\Delta^{-2}$ of the envelope contribution is Fourier dual to the variance of the individual peaks of the fuzzy comb along the $q-$quadrature and has the interpretation that the weak localization of the global observable $\hat{p}$ is in correspondence with the disturbance of the localization of the modular observable  $\hat{q}\mod \sqrt{2\pi}$ via the Heisenberg uncertainty relation.

\subsection{Steane error correction circuits}\label{sec:Steane}

A quadrature-controlled displacement between two GKP codes on a set of \text{data modes} $(d)$ and a single \textit{auxiliary mode} $(a)$  can facilitate a stabilizer measurement via a process called \textit{Steane error correction}. GKP stabilizers are measured by preparing an auxiliary mode $(a)$ in a logical stabilizer state, entangle the data modes with the auxiliary mode and perform homodyne measurement on the auxiliary mode. One can for instance always measure a GKP stabilizer $D\lr{\bs{\xi}}$ on a collection of data modes by coupling via the generalized sum gate

\begin{equation}
\mathsf{gSUM}_{ad}\lr{ \bs{e}_1,   \bs{\xi}; 1} =e^{-i \hat{p}_a \otimes \bs{\xi}^\mathrm{T} J \hat{\bs{x}}_d } \label{eq:gSUM}
\end{equation} to an auxiliary mode $\ket{\varnothing}$ and performing homodyne measurements on the $\hat{q}-$ quadrature of the auxiliary state. 
We sketch this protocol in fig.~\ref{fig:stabilizer_meas}, where, the syndrome $s$ of a stabilizer $D\lr{\bs{\xi}}$ is obtained from the measurement outcome $\tilde{q}$ as

\begin{equation}
\tilde{q}/\sqrt{2\pi} \mod 1 = s \mod 1.
\end{equation}

\begin{figure}
\center
\includegraphics[width=.4\textwidth]{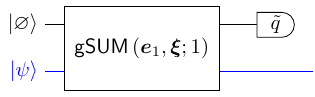}
\caption{Generalized stabilizer measurement protocol for a stabilizer given by $D\lr{\bs{\xi}}$. $\ket{\varnothing}$ is the unique code states of a GKP code given by $\CL=\mathbb{Z}^{2}$, also known as the sensor state \cite{Duivenvoorden_Sensor}.}\label{fig:stabilizer_meas}
\end{figure}

We can evaluate the action of the circuit in fig.~\ref{fig:stabilizer_meas} on the target mode. To this end write

\begin{equation}
\ket{\varnothing_{\Delta}}=\int_{-\infty}^{\infty} dq\, \beta_{\Delta}(\hat{q}) \ket{q}=  \beta_{\Delta}(\hat{q}) \ket{0}_p,
\end{equation}%
with the infinitely squeezed state $\ket{0}_p$
such that 
\begin{align}
\bra{\tilde{q}}_a\mathsf{gSUM}_{ad}\lr{ \bs{e}_1,   \bs{\xi}; 1} \ket{\varnothing_{\Delta}}_a \ket{\psi}_d
&=\bra{\tilde{q}}_a \beta_{\Delta}(\hat{q}_a - \bs{\xi}^TJ\bs{\hat{x}}_d) \mathsf{gSUM}_{at}\lr{ \bs{e}_1,   \bs{\xi}; 1} \ket{\tau_{\infty}}_a\ket{\psi}_d \nonumber\\
&=\beta_{\Delta}(\tilde{q} - \bs{\xi}^TJ\bs{\hat{x}}_d) \bra{\tilde{q}}_a \mathsf{gSUM}_{at}\lr{ \bs{e}_1,   \bs{\xi}; 1} \ket{\tau_{\infty}}_a\ket{\psi}_d \nonumber \\
&= \beta_{\Delta}(\tilde{q} - \bs{\xi}^TJ\bs{\hat{x}}_d) \ket{\psi}_d. 
\end{align}

Given a state $\ket{\psi}_d$, the probability to measure outcome $\tilde{q}$ is given by $P\lr{\tilde{q}}=\braket{\psi | \beta^2_{\Delta}(\tilde{q} - \bs{\xi}^TJ\bs{\hat{x}}_d)|\psi }$, which is also the factor that needs to be used to normalize the state above. Due to the structure of the state we see that this distribution is again governed by two contributions. For small $\Delta$, the translation invariant part of the wave function $\beta_{\Delta}$ (see eq.~\eqref{eq:beta}) dominates and the distribution $P\lr{\tilde{q}}$ is centered around the values of $\bs{\xi}^TJ\bs{\hat{x}}_d \mod \sqrt{2\pi}$ with variance $\Delta^2$. Up to this uncertainty,  this is precisely the information content of the modular observable $D\lr{\bs{\xi}}$ that the circuit is supposed to measure (compare to the limit $\Delta\rightarrow 0$). The envelope contribution in $\beta_{\Delta}$ weakly biases the distribution $P(\tilde{q})$ towards the non-modular value of $\bs{\xi}^TJ\bs{\hat{x}}_d$, such that the localization of the modular observable dual to $\bs{\xi}^TJ\bs{\hat{x}}_d$ will also be disturbed by $\Delta^2$. 
After extracting a a syndrom $\bs{s}$, a correction can e.g. be implemented by displacing the state back by the amplitude $\bs{\eta}=(MJ)^{-1}\bs{s}$, which takes into account the weakly measured full quadrature value and recenters the state around $0$. This strategy was investigated in ref.~\cite{Terhal_2020} for a single mode of the GKP code with $\CL=\sqrt{2}\Z^2$, where this strategy was shown to stabilize the average photon number and a logical post-correction was derived to minimize the logical error probability.

\begin{figure}
\center
\includegraphics[width=\textwidth]{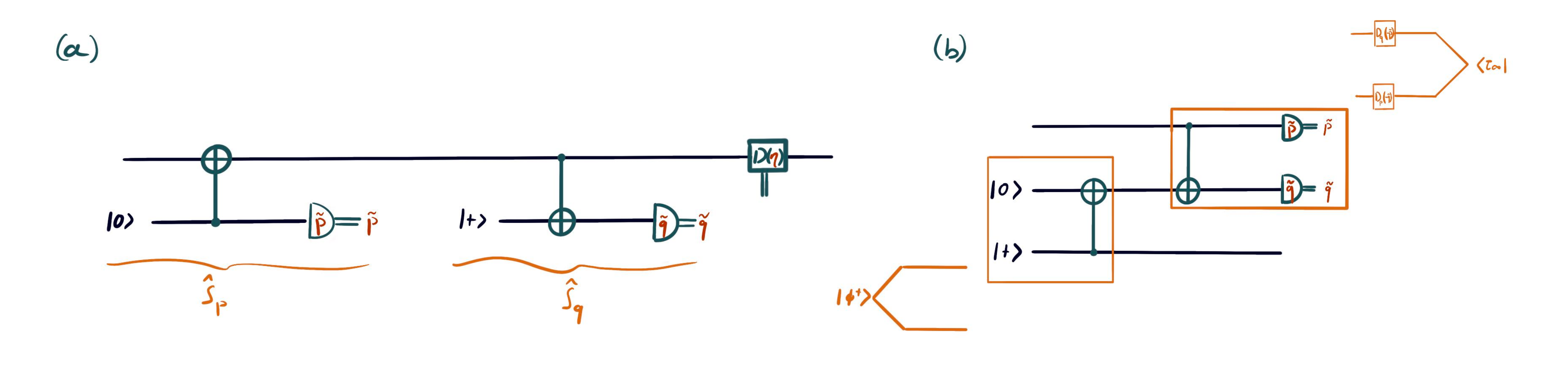}
\caption{$(a)$ Steane- and $(b)$ Knill stabilizer measurement circuits for a single-mode GKP code with $\CL=\sqrt{2}\Z^2$. In $(a)$ individual segments realize measurements of the $S_p=e^{-i2\sqrt{\pi}\hat{p}}$ and $S_q=e^{i2\sqrt{\pi}\hat{q}}$ stabilizer sequentially. 
The Knill circuit $(b)$ can be understood as a logical teleportation circuit, where on the bottom left a logical $\ket{\Phi^+}$ Bell state is prepared and on the top right an EPR measurement is performed. The EPR measurement with a fixed outcome $(\tilde{q}, \tilde{p})$ can also be understood as a displacement followed by a projection onto the EPR state $\ket{\tau_{\infty}}$. The circuits potentially require to be followed up by an additional logical displacement correction.
}\label{fig:Steane_Knill}
\end{figure}

Instead of using the sensor state together with the $\mathsf{gSUM}$ gate, the more commonly discussed strategy is to employ logical $\ket{0}$ and $\ket{+}$ states together with the two-mode $\sf SUM_{ct} = e^{-i\hat{q}_c \hat{p}_t}$ gate that displaces a target mode $(t)$ by the position value of the control mode $(c)$. This gate also serves as the logical CNOT in the single-mode square GKP code that encodes a qubit $\CL=\sqrt{2}\Z^2$. The perfect logical $\ket{0}$ state enjoys a $\sqrt{\pi}$ translation symmetry along the momentum axis as it is stabilized by $\CS_0=\langle Z, X^2 \rangle$ with the $\sqrt{\pi}$-momentum displacement $Z=e^{i\sqrt{\pi}\hat{q}}$ and  $\sqrt{\pi}$-position displacement $X=e^{-i\sqrt{\pi}\hat{p}}$. This state is Fourier-dual to the $\ket{+}$ state that is the unique fixpoint of $\CS_+=\langle X, Z^2 \rangle$ and has a $\sqrt{\pi}$ translation symmetry along the position axis. The circuit to implement the Steane-style stabilizer measurements is depicted in fig.~\ref{fig:Steane_Knill}. In the real world, the auxiliary states $\ket{0}, \ket{+}$ need to be replaced by their finitely squeezed variants, which can be achieved by applying the regularization operator $e^{-\beta \hat{n}}$ as discussed above to implement a smooth cutoff.

\subsection{Knill error correction circuits}\label{sec:Knill}
The Knill error correction circuit -- also sometimes referred to as \textit{teleportation based error correction} or \textit{tele-correction}-- is depicted in panel $(b)$ in figure~\ref{fig:Steane_Knill}. The conceptual idea is to implement quantum teleportation protocol for the logically encoded information in the error correcting code, which can be conducted without teleporting potential errors. This protocol is enabled by the entangling $\sf SUM$ gate, which is given by
\begin{equation}
S_{12}=e^{-i\hat{q}_1\hat{p_2}}:\; \begin{pmatrix}
\hat{q}_1 \\ \hat{q}_2 \\ \hat{p}_1 \\ \hat{p}_2
\end{pmatrix}
\mapsto 
\begin{pmatrix}
\hat{q}_1 \\ \hat{q}_2+\hat{q}_1 \\ \hat{p}_1-\hat{p}_2 \\ \hat{p}_2
\end{pmatrix}.
\end{equation}
This gates serves as the entangling gate that prepares a GKP Bell state\footnote{We shall keep in mind that these are not actually \textit{states} in the rigorous sense, but we shall look at these equations nevertheless to develop some intuition about their structure.}  

\begin{equation}
\ket{\Phi^+}=e^{-i\hat{q}_1\hat{p_2}} \ket{\overline{+}}\ket{\overline{0}}=\sum_{m,n \in \Z} \ket{\sqrt{\pi}n}_q\ket{2\sqrt{\pi}m+\sqrt{\pi}n}_q
\end{equation}
and also defines the infinitely squeezed EPR state \cite{EPR},

\begin{equation}
\ket{\tau_{\infty}}=e^{-i\hat{q}_1\hat{p_2}} \ket{0}_p\ket{0}_q=\sum_{n=0}^{\infty}\ket{n,n}=\int_{\R} dx\, \ket{x}_q\ket{x}_q. 
\end{equation}

In fact, it can be shown that, conditioned on the measurement outcomes $\lr{\tilde{q}, \tilde{p}}$, this circuit implements the same physical action as the Steane circuit together with a corrective shift via
\begin{equation}
e^{-i\tilde{q}\hat{p}}e^{i\tilde{p}\hat{q}} = e^{-i\tilde{q}\tilde{p}/2}D\lr{\lr{-\tilde{q}, \;  \tilde{p}}^T
/\sqrt{2\pi}}\label{eq:Knillcorr}
\end{equation}
and a $\pi$-rotation.
This equivalence only assumes a $\pi$-rotation symmetry on the auxiliary $\ket{0}, \ket{+}$ states and holds when the auxiliary GKP states are finitely squeezed. The key trick in deriving this equivalence is the following circuit identity, that compiles a CV $\sf SWAP$ gate into a sequence of $\sf SUM$ gates and a $\pi$-rotation, which can be easily checked by writing out the symplectic matrices that implement these Gaussian unitary gates. Similar to ordinary quantum teleportation, this circuit typically needs to be followed up by a logical Pauli correction. 

\begin{figure}[H]
\center
\includegraphics[width=.5\textwidth]{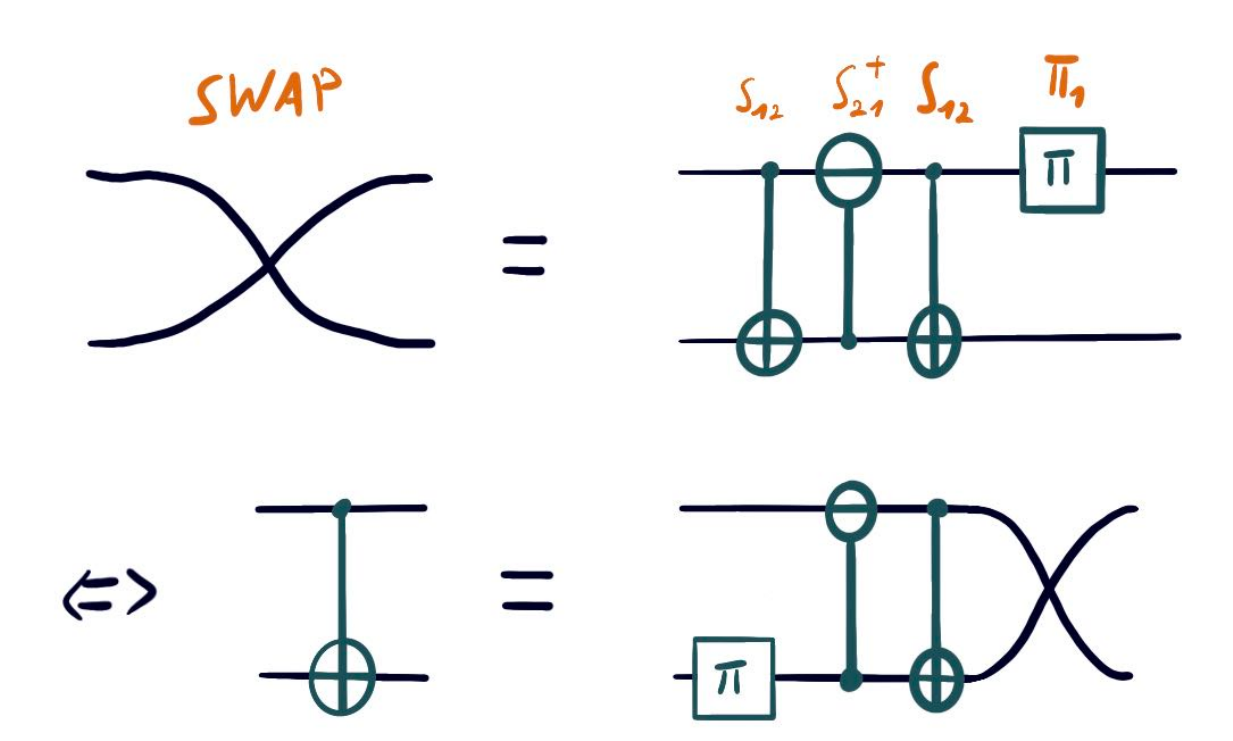}
\caption{CV $\sf SWAP$ gate compiled into a sequence of $\sf SUM$ gates and a $\pi$-rotation.}\label{fig:SWAPtrick}
\end{figure}

Using this result, one can derive the equivalence between the Knill- and the Steane error correction circuit following the steps displayed in fig.~\ref{fig:KnillEquiv}, which shows a proof based on mostly graphical calculus.

\begin{figure}[H]
\hspace{-2cm}
\includegraphics[width=1.3\textwidth]{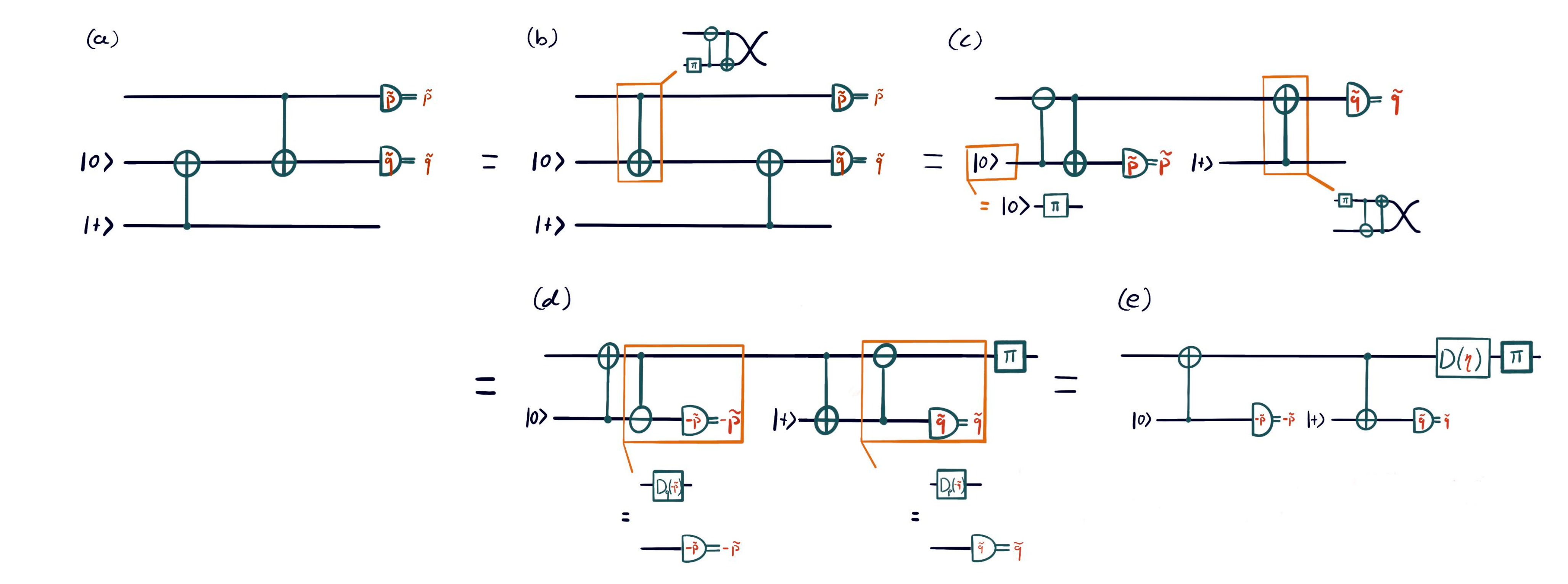}
\caption{The derivation of the equivalence between the Knill- and the Steane stabilizer measurement circuit. We find that the Knill error correction circuit is equivalent to the Steane circuit paired with a final corrective shift by $\bs{\eta}=\lr{-\tilde{q}, \;  \tilde{p}}^T
/\sqrt{2\pi}$.}\label{fig:KnillEquiv}
\end{figure}

The teleportation-based Knill error correction circuit is particularly relevant in photonic quantum computing architectures, where it is difficult to engineer strong non-linear Hamiltonian interaction between photonic modes and non-linear interactions dominantly need to be inherited from quantum measurement processes. The teleportation-based nature has the upshot that in every step the old quantum state is replaced by a ``fresh'' one. In the above derivation we see, by comparing the correction implemented by the Knill circuit in its Steane equivalent, that the Knill circuit always returns a state with a well-centered envelope as the displacement imposed on the Steane side is precisely that performs this job.
 A  relevant point is that, as in ordinary quantum teleportation, a logical post correction is generally necessary when one attempts to perform error correction using the Knill circuit. In the equivalence we prove in fig.~\ref{fig:KnillEquiv} this necessity stems from the fact that (in the limit $\Delta \rightarrow 0$) the measurement outcomes $(\tilde{q}, \tilde{p})$ of the Knill circuit, which feed back as displacements on the equivalent Steane circuit, should contain no logical information and hence implement a random logical Pauli term which needs to be corrected. Strategies to derive a logical post-correction based on the finite squeezing error for the Steane circuit are presented in ref.~\cite{Terhal_2020} and can straightforwardly be adapted to the Knill circuit using their equivalence.

\section{Passive error correction}\label{sec:passiveGKP}

The oldest dream in quantum error correction is to build a quantum memory -- given by a Hamiltonian built from stabilizers -- that does not require active monitoring and intervention but stabilizes itself via a natural physical mechanism. The hope is to establish a quantum memory that functions similar to a ferromagnetic hard-drive where the energy to create small errors in the system outweights the entropic fluctuations such that there is a temperature below which it enters a stable phase. This dream is usually termed \textit{passive} quantum error correction, which is to be distinguished from \textit{dissipative} quantum error correction, where one also allows to constrain the system to interact with its environment in a tailored way. In this section we will briefly review the idea of passive quantum error correction tailored to the GKP code and, how the physics of this system behaves like that of a charged particle moving on a torus under a magnetic field or in a 2d crystal. These connections had been pointed out in refs.~\cite{GKP, Rymarz2021} and we briefly review them here since they point to a fascinating connection between GKP error correction and more conventional physics.
Finally, we discuss how a GKP-Hamiltonian suitable for this task can be obtained through Floquet-engineering, i.e. by tailoring a time-dependent periodic driving sequence that effectively gives rise to the target Hamiltonian. 

\subsection{Stabilizer Hamiltonians}
The stabilizer Hamiltonian associated with a stabilizer group for qubit systems, generated by a set of generators $\CG=\lrc{g_1,\hdots, g_r}\subseteq \CP$, is typically defined as
\begin{equation}
H_{\CG}=-U\sum_{i=1}^r g_i , 
\end{equation}
where $U$ is some constant. Since all the generators $g_i$ commute, the ground space of this Hamiltonian will be the shared $+1$ eigenvalued eigenspace of all the stabilizer generators $g_i$, which is simply the code space $\CC$ of the associated stabilizer code, and it has a degeneracy given by the dimension of the associated logical Hilbert space. Although it is not strictly necessary, we assume that the generating set used above is minimal, and all the $g_i$'s are (linearly) independent. This can be relaxed to other generating sets as long as the number of terms needed in the Hamiltonian does not grow exponentially in the system size.
For an elementary set of errors $\lrc{E_i}\in \CP$ that satisfy the quantum error correction conditions for this code, excited states are simply provided by elements in the spaces $E_i \CC$, which have the same degeneracy of the code space. 
Let $E$ be an error that yields a syndrome vector 
\begin{equation}
\bs{s}\in \Z_2^r  : E^{\dagger}g_iE = (-1)^{s_i(E)}.
\end{equation}
The excited state given by the error sector $E \CC$ then has the energy
\begin{equation}
H(E)= -Ur+2U\sum_{i=1}^r s_i(E),
\end{equation}
that is the energy required to create an error $E$ grows with the total syndrome $\bs{1}^T\bs{s}$. The challenge in designing good Hamiltonians for passive quantum error correction is thus mapped to the problem of finding families of quantum error correcting codes where the total syndrome generally grows with the size of any error so that the total energy needed for errors to accumulate to grow to logical errors is extensive in the system size \cite{Dennis_2002, Bravyi_2009}. Even below the thermodynamic limit, stabilizer Hamiltonians defined via stabilizer codes with a finite number of modes $n$ are expected to be useful in that the energy penalty associated with error events dampens their likelihood to grow to logical errors and the system is left with an enhanced lifetime.

Given a generating set $M=\lr{\bs{\xi}_1,\hdots, \bs{\xi}_{2n}}^T$ for a GKP code, we define the stabilizer Hamiltonian for the $n$-mode system as
\begin{equation}
H_{\rm GKP}\lr{M}= -\frac{U}{2}\sum_{i=1}^{2n}D\lr{\bs{\xi}_i} + {\rm h.c.}= -U\sum_{i=1}^{2n} \cos\lr{ \sqrt{2\pi} \bs{\xi}^T_i J \bs{\hat{x}}}.
\end{equation}
The ground space of this Hamiltonian consists of the associated GKP code space with states translation symmetric under displacements by vectors in $\CL$, the lattice spanned by $M$. It is interesting to note that this Hamiltonian is \textit{gapless}:  By continuity of the cosines, small displacements of the Hamiltonian can change its eigenvalues by arbitrarily small amounts
\begin{equation}
D^{\dagger}\lr{\bs{e}}H_{\rm GKP}\lr{M}D\lr{\bs{e}}=-U\sum_{i=1}^{2n} \cos\lr{ \sqrt{2\pi} \bs{\xi}^T_i J \bs{\hat{x}}+2\pi\bs{\xi}^T_i J \bs{e}  },
\end{equation}
such that displaced ground states of the Hamiltonian $D\lr{\bs{e}} \ket{\overline{\psi}}$ acquire an energy 
\begin{equation}
-U\sum_{i=1}^{2n} \cos\lr{2\pi\bs{\xi}^T_i J \bs{e} }
\end{equation}
relative to the ground state.
To construct a slightly more realistic scenario, it is reasonable to add a small perturbation $H_{0}\lr{\epsilon}=\epsilon \sum_{i=1}^n \hat{n}_i\propto \epsilon \bs{\hat{x}}^{\dagger}\bs{\hat{x}}$ to the stabilizer Hamiltonian to obtain the total Hamiltonian (up to an irrelevant constant)
\begin{equation}
H=H_{0}\lr{\epsilon}+H_{\rm GKP}\lr{M}= \sum_i \epsilon \frac{\hat{p}_i^2+\hat{q}_i^2}{2}-U\sum_{i=1}^{2n} \cos\lr{ \sqrt{2\pi} \bs{\xi}^T_i J \bs{\hat{x}}}.
\end{equation}

In general, this perturbation breaks the translational symmetry of the GKP stabilizer Hamiltonian, and the resulting eigenstates are not to be expected to be exact code states anymore (which were not really physical in the first place anyhow). Importantly, the perturbation may imply a gap in the energy landscape. If the GKP code had some rotational symmetries, those would remain preserved so that one can find a basis for the perturbed ground space in terms of rotation symmetric states.

We can analyse this perturbed Hamiltonian with some crude physical intuition. For small $\epsilon$, the terms $\frac{ \hat{p}_i^2}{2\epsilon^{-1}}+\frac{\epsilon}{2}\hat{q}_i^2$ in the the perturbation describe a particle in a very flat harmonic potential with a very large mass, where the flat potential approximatively preserves the translation symmetry of the GKP Hamiltonian and the large mass motivates that the inert particle, sitting in the minima of the cosine landscape, behaves as if each well of the cosine landscape was a quantum Harmonic oscillator. This motivates a perturbative treatment of the GKP Hamiltonian by expanding the cosine terms and manually adding the translational invariant structure implicit in the cosines via

\begin{align}
H_{\rm GKP}\lr{M}\approx H_{\rm LE} =-\pi U \sum_{\bs{\xi}^{\perp}\in \CL^{\perp}} D^{\dagger}\lr{\bs{\xi}^{\perp}} \lrq{\bs{\hat{x}}^{\dagger} V_M\bs{\hat{x}} }D\lr{\bs{\xi}^{\perp}},
\end{align}
where 
\begin{equation}
V_M=J^T\lrq{\sum_{i=1}^{2n} \bs{\xi}_i\bs{\xi}_i^T } J.
\end{equation}
Using the fact that every basis $M$ for a GKP code with type $D$ can be understood as a symplectic transform of a canonical basis  (see chap.~\ref{chap:Theory}, an equivalent statement can also be made using Williamsons theorem \cite{Weedbrook_2012})
\begin{equation}
M=\lr{\sqrt{D}\oplus \sqrt{D} } S^T,
\end{equation}
we find that
\begin{align}
\bs{\hat{x}}^{\dagger} V_M\bs{\hat{x}} &= \lr{S^{-1}\bs{\hat{x}}}^{\dagger} \lr{D\oplus D}\lr{S^{-1}\bs{\hat{x}}}\\
&= U_{S^{-1}}^{\dagger} \bs{\hat{x}}^{\dagger}  \lr{D\oplus D}\bs{\hat{x}}  U_{S^{-1}}.
\end{align}
The terms $\bs{\hat{x}}^{\dagger}  \lr{D\oplus D}\bs{\hat{x}}$ simply describe the direct sum of $n$ quantum harmonic oscillators with energy gaps $d_i$. The total Hamiltonian can be approximatively viewed as a particle moving in the wells of a an infinite sum of squeezed quantum Harmonic oscillators. 

For large $U\rightarrow \infty $ in the absence of a perturbation, the  distance $\Delta$ of the GKP code is understood as the distance a particle needs to tunnel through the bulk of the cosines to realize a logical error while the presence of a perturbation $H_{0}\lr{\epsilon}$ realizes a coupling between the different degenerate ground states. To analyze this situation, we define a projector onto the ground space

\begin{equation}
\Pi^{\infty}_{\CL}=\sum_{s\in \CS}s =\sum_{\bs{\xi}\in \CL}e^{i\phi_M \lr{\bs{\xi}}} D\lr{\bs{\xi}}.
\end{equation}

This projector is such that it takes a unit value on the code space, i.e. the ground space of $H_{\rm GKP}\lr{M}$ and vanishes elsewhere. Its Wigner-function can be computed as
\begin{equation}
W_{\Pi^{\infty}_{\CL}}\lr{\bs{x}}=\sum_{\bs{\xi}\in \CL}e^{i\phi_M \lr{\bs{\xi}}} e^{i2\pi \bs{x}^TJ\bs{\xi}},
\end{equation}
which evaluates for symplectically even lattices with $A\in 2\Z^{2n \times 2n}$ to
 \begin{equation}
W_{\Pi^{\infty}_{\CL}}\lr{\bs{x}}=\frac{1}{\det\lr{\CL}}\sum_{\bs{\xi}\in \CL^{\perp}}\delta\lr{\bs{x}-\bs{\xi}^{\perp}}.
\end{equation}
per lemma~\ref{lem:dirac_comb} (see also ref.~\cite{Conrad_2021}).

Note that due to the infinite size of the GKP stabilizer group, this projector is ill-defined\footnote{Which is nothing that scares a physicist.}, in particular we have $\lr{\Pi^{\infty}_{\CL}}^2=(\infty)  \Pi^{\infty}_{\CL}$ and a more rigorous treatment would consider this as the limit of a regularized approximate projector

\begin{equation}
\Pi^{\infty}_{\CL}=\lim_{n\rightarrow \infty} \frac{1}{|\CG^n|} \sum_{s \in \CG^n} s, 
\end{equation}
where $\CG^n$ is the approximate group built from an $n-$step random walk over a generating set $\CG$ for the group $\CS=\langle \CG\rangle$. 

The coupling between the ground space degeneracy sectors realized by the perturbation is, with $\hat{N}=\sum_{i=1}^n \hat{n}_i$,
\begin{align}
\Pi^{\infty}_{\CL}\epsilon \hat{N}\lr{\epsilon}\Pi^{\infty}_{\CL}
&\approx \Pi^{\infty}_{\CL}\lr{1-e^{-\epsilon \hat{N}}}\Pi^{\infty}_{\CL} \\
&=\Pi^{\infty}_{\CL}- \frac{1}{\lr{1-e^{-\epsilon}}^n} \int_{\R^{2n}} d\bs{x}\,    e^{-\frac{\pi \|\bs{x}\|^2}{\sigma^2}} \Pi^{\infty}_{\CL} D\lr{\bs{x}}\Pi^{\infty}_{\CL}
\end{align}
where we used the expansion from eq.~\eqref{eq:FSE_disp} with $\sigma^2=2\tanh\lr{\epsilon /2}$. We have that $\Pi^{\infty}_{\CL} D\lr{\bs{x}}\Pi^{\infty}_{\CL} \neq 0$ only if  $\bs{x}\in \CL^{\perp}$, such that we obtain the effective ground space coupling
\begin{equation}
 \int_{\R^{2n}} d\bs{x}\,    e^{-\frac{\pi \|\bs{x}\|^2}{\sigma^2}} \Pi^{\infty}_{\CL} D\lr{\bs{x}}\Pi^{\infty}_{\CL}= \sum_{\bs{x}\in \CL^{\perp}/\CL} \Theta_{\bs{x}+\CL}\lr{\frac{i}{2\sigma^2}} D\lr{\bs{x}},
\end{equation}
with the $\Theta$ functions defined in chap.~\ref{chap:Theory}. In this model the coupling between logical states in the ground space is exponentially suppressed with the GKP code distance $\Delta$ and it is interesting to observe that the coupling constants for projections onto different sectors $D\lr{\bs{\eta}}\Pi^{\infty}_{\CL}D\lr{\bs{\eta}}^{\dagger}$ reproduce the MLD probabilities that we have derived for decoding the GKP code in sec.~\ref{sec:MLD}.

\subsection{GKP Hamiltonians and the Hall effect}

\subsubsection{Particles on the torus}
There is an alternative representation of the approximate GKP Hamiltonian using the so-called Zak basis \cite{Zak1967}, which allows for a more illusive interpretation. The following derivation presents a simplified version of a discussion in ref.~\cite{Ganeshan2016}. 

Consider the approximate GKP Hamiltonian 
\begin{equation}
H=\epsilon \frac{\bs{\hat{x}}^{\dagger}\bs{\hat{x}}}{2} - U\sum_{k=1}^{2n} \cos \lr{\sqrt{2\pi} \bs{\xi}_k^T J \bs{\hat{x}}}.
\end{equation}
We use the trick that every basis for a GKP code can be written as $M=\lr{D\oplus I_n} S_1^T$ with a symplectic matrix $S_1$, and apply a Gaussian unitary with symplectic transformation $S=S_1^TS_0$ where $S_0=\lr{\sqrt{2\pi}^{-1}I_n\oplus \sqrt{2\pi}I_n}$ and $D=\diag\lr{d_1,\hdots, d_n}$, we can compute

\begin{equation}
U^{\dagger}_{S} H U_{S} = \epsilon \frac{\bs{\hat{x}}^{\dagger}S^TS \bs{\hat{x}}}{2} -U\sum_{k=1}^{n} \cos\lr{d_k\hat{p}_k}+\cos\lr{2\pi \hat{q}_k}\label{eq:H_traf}
\end{equation}

The periodic part of this Hamiltonian naturally commutes with displacement operators in $\Z^n \oplus 2\pi D^{-1}\Z^n$, and it becomes convenient to use of version of the co-called Zak-basis \cite{Zak1967} with $\bs{\theta}, \bs{\phi} \in \left[0, 2\pi\right)^n$

\begin{equation}
\ket{\bs{\theta}, \bs{\phi}}= \frac{1}{\lr{2\pi}^{n/2}}\sum_{\bs{j}\in \Z^{n}} e^{i\bs{j}^T\bs{\theta}} \ket{\bs{j}+\bs{\phi}/2\pi}_{\bs{q}},
\end{equation}
where $\ket{\bs{x}}_{\bs{q}}=\bigotimes_{i=1}^n \ket{q_i}_q$ denotes the $n$-dimensional position eigenstate.

These states are eigenstates of the commuting operators $e^{i\hat{p}_k}, e^{i2\pi \hat{q}_l},\; \forall k,l = 1\hdots n$ with
\begin{align}
e^{i\hat{p}_k}\ket{\bs{\theta}, \bs{\phi}}&= e^{i\theta_k} \ket{\bs{\theta}, \bs{\phi}},\\
e^{i2\pi\hat{q}_l}\ket{\bs{\theta}, \bs{\phi}}&= e^{i\phi_l} \ket{\bs{\theta}, \bs{\phi}},
\end{align}
and simple displacements act as
\begin{align}
e^{i\bs{a}^T\bs{\hat{q}}} \ket{\bs{\theta}, \bs{\phi}} =e^{i\bs{a}^T\bs{\phi}/2\pi} \ket{\bs{\theta}+\bs{a}, \bs{\phi}}\\
e^{i\bs{b}^T\bs{\hat{p}}} \ket{\bs{\theta}, \bs{\phi}} = \ket{\bs{\theta}, \bs{\phi}-2\pi \bs{b}}.\label{eq:disp_zak}
\end{align}

The states $\ket{\bs{\theta}, \bs{\phi}}$ form a complete basis, as can be checked by computing

\begin{align}
\int_{\left[0, 2\pi\right)^{2n}} d\bs{\theta}d\bs{\phi}\, \ket{\bs{\theta}, \bs{\varphi}}\bra{\bs{\theta}, \bs{\varphi}} 
&=\sum_{\bs{j},\bs{j'}\in \Z^{n}}  \int_{\left[0, 2\pi\right)^n}d\bs{\varphi}\, \lrq{\frac{1}{\lr{2\pi}^{n}}\int_{\left[0, 2\pi\right)^n} d\bs{\theta}e^{i\lr{\bs{j}-\bs{j'}}^T\bs{\theta}} }
\nonumber \\ &\hspace{2cm} \ket{\bs{j}+\bs{\varphi}/2\pi}_{\bs{q}}\bra{\bs{j'}+\bs{\varphi}/2\pi}_{\bs{q}} \nonumber \\
&= \sum_{\bs{j}\in \Z^{n}}\int_{\left[0, 2\pi\right)^n} \ket{\bs{j}+\bs{\varphi}/2\pi}_{\bs{q}}\bra{\bs{j}+\bs{\varphi}/2\pi}_{\bs{q}} \nonumber \\
&= \hat{I},
\end{align}
where the expression in the square brackets evaluates to $\delta_{\bs{j}, \bs{j'}}$, and in this basis the quadrature variables become represented by
\begin{equation}
\bs{\hat{x}}=\begin{pmatrix}
-\bs{\hat{p}_{\theta}} + \bs{\hat{\varphi}}/2\pi \\ 2\pi \bs{\hat{p}_{\varphi}}
\end{pmatrix}=\bs{\hat{p}_{\theta, \phi}} + \bs{A}\lr{\bs{\hat{\varphi}}},
\end{equation}
where the momenta of the compact variables are represented by $\hat{p}_{\theta_k}\equiv -i\partial_{\theta_k}$ and $\hat{p}_{\varphi_k}\equiv -i\partial_{\varphi_k}$ such that (when exponentiated) they behave as if $\lrq{\hat{\theta}_k, \hat{p}_{\theta_l}}=\lrq{\hat{\varphi}_k, \hat{p}_{\varphi_l}}=i\delta_{kl}$.
In the Zak representation, it holds that 
\begin{equation}
e^{i2\pi \hat{p}_{\theta_k}}=\hat{I}, \hspace{1cm}  e^{i2\pi \hat{p}_{\varphi_k}}=e^{i\hat{\theta}_k} ,
\end{equation}
which is verified using eq.~\eqref{eq:disp_zak}, such that also
\begin{equation}
e^{i2\pi \hat{q}_k}=e^{i\hat{\phi}_k},\hspace{1cm} e^{id_k\hat{p}_k}=e^{id_k\hat{\theta}_k}.
\end{equation}

Combining all of these points, the Hamiltonian in eq.~\eqref{eq:H_traf} becomes
\begin{equation}
H'=\frac{\epsilon}{2} \lr{\bs{\hat{p}_{\theta, \phi}} + \bs{A}\lr{\bs{\hat{\varphi}}}}^{\dagger}S^TS\lr{\bs{\hat{p}_{\theta, \varphi}} + \bs{A}\lr{\bs{\hat{\varphi}}}} -U\sum_{k=1}^{n}  \cos\lr{d_k\hat{\theta}_k}+\cos\lr{ \hat{\varphi}_k}.
\end{equation}
This Hamiltonian describes the motion of $n$ particles with generalized total quadrature $\bs{\hat{\Theta}}=\bs{\hat{\theta}}\oplus\bs{\hat{\varphi}}$ in the cosine potential on the $n$-torus $[0,2\pi)^{2n}$. In the simple case of diagonal $C=S^TS$, we recognize that $\bs{A}\lr{\bs{\hat{\varphi}}}$ takes the role of a vector potential which can be interpreted to describe the presence of a magnetic field driving these particles, while for non-diagonal matrix $C$ the dynamical momenta $\bs{\hat{\pi}}=\bs{\hat{p}_{\theta, \phi}} + \bs{A}\lr{\bs{\hat{\varphi}}}$ couple non-trivially. 

\subsubsection{An electron in a crystal}

It has already been recognized early in ref.~\cite{GKP} that the system of a single electron moving on a torus with a normal magnetic field can give rise to states resembling GKP states. Unfortunately, the geometry of a torus is relatively rare in nature and hard to build in the lab. An electron moving in the plane under a crystal potential, however, is a more realistic scenario and one would expect it to behave as if it was on a torus when the electron is of low energy and only ``sees'' one potential valley at a time.

 An electron moving in the plane with a perpendicular magnetic field is forced by the Lorenz force into a deflected trajectory, which leads to the physics of the Quantum Hall effect. This is a physical setting with a long history, and the connection to the GKP code has been explained in refs.~\cite{Rymarz2021, Ganeshan2016}. We review this connection following the presentation in ref.~\cite{Rymarz2021} and ref.~\cite{TongQHE}, as this is a likely gateway to a deeper exploration of interesting physics related to the GKP code. This discussion focuses only on the dynamics of a single electron on the plane in a first-quantized language to avoid dealing with statistical properties of the system.  A discussion focused on incorporating statistical properties and the inclusion of an edge\footnote{See also refs.~\cite{Ganeshan2016,Ganeshan2022}.} are likely to imply further interesting insights.

Let's consider an electron moving in a crystal potential in the plane. The crystal is parametrized by a $2d$-lattice $\CL_{crys}=a\CL_0$ with basis $M^T=a\lr{\bs{v_1}\; \bs{v_2}} = a S$ such that the unit cell has area $A=a^2$ and can be obtained by a symplectic transformation  $S\in \Sp_2\lr{\R}$ from the square lattice with $M=a I_2$.

The associated crystal potential with $\bs{\hat{q}}=(\hat{q}_1, \hat{q}_2)^T$ is
\begin{align}
V_{\rm crystal}\lr{\bs{\hat{q}}}
&=-V\lrc{\cos\lr{\frac{2\pi}{a}\bs{v}_1^TJ\bs{\hat{q}} }+\cos\lr{\frac{2\pi}{a}\bs{v}_2^TJ\bs{\hat{q}} }}\\
&=-\frac{V}{2}\lrc{e^{i\frac{2\pi}{a}\lr{\bs{v}_1^TJ\bs{\hat{q}} }}+ e^{i\frac{2\pi}{a}\lr{\bs{v}_2^TJ\bs{\hat{q}} }} + {\rm h.c. }},
\end{align}
such that the Hamiltonian of the electron with mass $m$ and charge $-e$ moving in the magnetic field becomes
\begin{equation}
H=\frac{\bs{\hat{\pi}}^2}{2m} +V_{\rm crystal}\lr{\bs{\hat{q}}}, \label{eq:Ham_crys}
\end{equation}
with the dynamical momenta $\bs{\hat{\pi}}=\bs{\hat{p}}+ e\bs{A}\lr{\bs{\hat{q}}}$. The vector potential $\bs{A}(\bs{q}) $ is such that it gives rise to the magnetic field pointing perpendicular to the plane, $\nabla \times \bs{A}(\bs{q})  = B \bs{e}_3$, which could be realized by different choices of gauge, such as the Landau gauge with $\bs{A}\lr{\bs{q}}=q_1 \bs{\hat{e}}_2$ which we describe here. The dynamical momenta $\bs{\hat{\pi}}$ are gauge invariant. 
Physically, an electron moving in the magnetic field is pushed into a direction perpendicular to its direction of movement by the Lorenz force $\bs{F}=-e/m\bs{p}\times \bs{B}$ and its trajectories become circular in the $(q_1, q_2)$ plane -- similar to that of a harmonic oscillator in the $(q,p)$ plane. The trajectories rotate around the guiding center variables 
\begin{equation}
\bs{\hat{R}}=\bs{\hat{q}}-\frac{1}{m\omega_c}J\bs{\hat{\pi}}\label{eq:GCV}
\end{equation}
with cyclotron frequency $\omega_c=eB/m$, where $\bs{\hat{\pi}}$ also describe the relative coordinates of the electron. These operators fulfil the commutation relations
\begin{equation}
\lrq{\hat{R}_i, \hat{R}_j}=il_B^2 J_{ij} ,\hspace{.4cm} \lrq{\hat{\pi}_i, \hat{\pi}_j}=-i\frac{\hbar^2}{l_B^2} J_{ij},\hspace{.4cm} \lrq{\hat{\pi}_i, \hat{R}_j}=0\label{eq:comm_Jpi}
\end{equation}
with magnetic length $l_B^2=\hbar/eB$.

By introducing the annihilation operator $\hat{a}=\frac{l_B}{\sqrt{2}\hbar}\lr{\hat{\pi}_2+i\hat{\pi}_1},\; \lrq{\hat{a}, \hat{a}^{\dagger}}=1$, the kinetic part of the Hamiltonian can be rewritten as
\begin{equation}
\frac{\bs{\hat{\pi}}^2}{2m}= \hbar\omega_c \lr{\hat{a}^{\dagger}\hat{a}+\frac{1}{2}},
\end{equation}
which is the familiar Hamiltonian of the quantum harmonic oscillator. The (Fock) eigenspaces $\ket{n}$ of this Hamiltonian are called \textit{Landau levels}, and quantize the rotational movement of the electron around its guiding center.

We can also introduce the displacement operators 
\begin{align}
T\lr{\bs{v}}&=e^{-i\bs{v}^TJ\bs{\hat{R}}/l_B^2}:\; T^{\dagger}\lr{\bs{v}}\bs{\hat{R}} T\lr{\bs{v}}=\bs{\hat{R}} +\bs{v}, \nonumber \\
D_{\pi}\lr{\bs{v}}&=e^{-i\frac{l_B^2}{\hbar^2}\bs{v}^TJ\bs{\hat{\pi}} }:\; D_{\pi}^{\dagger}\lr{\bs{v}} \bs{\hat{\pi}}D_{\pi}\lr{\bs{v}}=\bs{\hat{\pi}}+\bs{v}.
\end{align}
With these definitions, and inserting eq.~\eqref{eq:GCV}, the Hamiltonian in eq.~\eqref{eq:Ham_crys} becomes
\begin{align}
H&= \hbar\omega_c \lr{\hat{a}^{\dagger}\hat{a}+\frac{1}{2}} \nonumber\\
&\hspace{.3cm} -\frac{V}{2}\lrc{  T\lr{\frac{\Phi_0 a}{\Phi}\bs{v}_1} D_{\pi}\lr{\frac{h}{a} \bs{v}_1} +T\lr{\frac{\Phi_0 a}{\Phi}\bs{v}_2}D_{\pi}\lr{\frac{h}{a} \bs{v}_2}   + {\rm h.c.}},
\end{align}
where we also inserted the flux quantum $\Phi_0=h/e$ and Flux $\Phi=a^2 B$ that pierces the individual unit cell.

This Hamiltonian has a nice interpretation. The terms $D_{\pi}\lr{\cdot}$ couple between the different Landau levels, which are separated by an energy gap of $\Delta E =\hbar \omega_c \propto B$. When the magnetic field is very strong, the energy gap becomes large, and we can restrict the system to an effective theory on the lowest landau level (LLL) at $\ket{n=0}$. Projecting onto the LLL, the effective theory is now given by
\begin{equation}
H_{LLL}= -\frac{V'}{2}\lrc{  T\lr{\frac{\Phi_0 a}{\Phi}\bs{v}_1}+ T\lr{\frac{\Phi_0 a}{\Phi}\bs{v}_2}   + {\rm h.c.}},\label{eq:HLLL}
\end{equation}
with $V'\propto V$ \cite{Rymarz2021}, which is close to the kind of GKP Hamiltonian we were after. 
The operators $T\lr{\cdot}$ commute as 
\begin{equation}
T\lr{\bs{x}}T\lr{\bs{y}}=e^{i2\pi \frac{ \Phi\lr{\bs{x}, \bs{y}}}{\Phi_0}} T\lr{\bs{y}}T\lr{\bs{x}}, 
\end{equation}
where $ \Phi\lr{\bs{x}, \bs{y}}=B\bs{x}^TJ\bs{y}$ is the magnetic flux enclosed in the parallelogram spanned by the vectors $\bs{x}, \bs{y}$. The displacements appearing in the Hamiltonian thus commute to 
\begin{equation}
T\lr{\frac{\Phi_0 a}{\Phi}\bs{v}_1} T\lr{\frac{\Phi_0 a}{\Phi}\bs{v}_2}=e^{i2\pi\Phi_0 /\Phi } T\lr{\frac{\Phi_0 a}{\Phi}\bs{v}_2} T\lr{\frac{\Phi_0 a}{\Phi}\bs{v}_1}.
\end{equation}
Choosing $\Phi=\frac{p}{q}\Phi_0$, the operators hence commute if $q/p \in \Z$ and the Hamiltonian in eq.~\eqref{eq:HLLL} becomes a GKP Hamiltonian. We define the variables $\bs{\hat{x}}=\bs{\hat{R}}/l_B$, which are dimension less and commute like the usual position- and momentum operators $[\hat{x}_i, \hat{x}_j]=iJ_{ij}$. With $2\pi l_B/a=\sqrt{2\pi \Phi_0/\Phi}$ the LLL Hamiltonian becomes

\begin{equation}
H_{LLL}=-V' \lrc{\cos\lr{\sqrt{2\pi\frac{\Phi_0}{\Phi}}\bs{v}_1^TJ\bs{\hat{x}}}+ \cos\lr{\sqrt{2\pi\frac{\Phi_0}{\Phi}}\bs{v}_2^TJ\bs{\hat{x}}} }.
\end{equation}

This is the stabilizer Hamiltonian of a scaled GKP code with stabilizer lattice $\CL=\sqrt{q/p}\CL_0$ which encodes $d=\frac{\Phi_0}{\Phi}=q/p$ logical dimensions. From this correspondence again we see that we need $q/p\in \Z$ for the terms in the Hamiltonian (the stabilizers) to commute. When $q$ and $p$ are chosen to be coprime integers, as is usually assumed in quantum Hall physics, it becomes necessary that $p=1$ and the Hamiltonian describes the simple scaled GKP code with $\CL=\sqrt{q}\CL_0,\; \CL_0^{\perp}=\CL_0$. From our earlier discussions we know that this Hamiltonian encodes the $q$-dimensional code space as its degenerate ground space.

The presence of the magnetic field that controls how the variables $\hat{R}_{1,2}$ commute, such that now it is not merely the unit cell volume $a^2$ of the underlying lattice $\CL_{crys}=a \CL_0$ that defines the dimension of the code space (the ground space degeneracy of this Hamiltonian) but the strength of the magnetic field provides an extra handle on the commutation phase between operators. Compared to the phase-space picture in the $q, p$ plane, this system -- in the $R_1, R_2$ plane -- behaves as if one could scale the value of the canonical commutation phase $\hbar \mapsto q\hbar$, which manifests equivalent to the mechanism underlying the class of scaled GKP codes.

\subsection{Floquet engineering via dynamical decoupling}\label{sec:floquetGKP}

Now to something more practical. Consider the stabilizer Hamiltonian for a single mode square GKP code, scaled to encode a single qubit. The Hamiltonian is given by (set $U=1$ for simplicity)

\begin{equation}
H_{\rm GKP}= -\cos\lr{2\sqrt{\pi}\hat{p}}-\cos\lr{2\sqrt{\pi}\hat{q}}.\label{eq:GKPHsq}
\end{equation}
How does one engineer such a Hamiltonian?  
From the previous discussion on superconducting circuits -- identifying $\hat{q}\propto\hat{\Phi}$ and $\hat{p}\propto\hat{Q}$ we see that one of the cosine terms can be produced by a Josephson junction but none of the physical mechanisms discussed above seem to allow the engineering of a complementary $-\cos\lr{2\sqrt{\pi}\hat{Q}}$ term. In the typical process of circuit quantization all the variables that enter the potential terms need to commute, such that obtaining the complementary cosine is very difficult \cite{Terhal_2020}. One proposal that states such a device uses a Josephson junction together with the exotic phase-slip junction, which simply have the correct Hamiltonian terms \cite{Le_2019}. The understanding of such devices is however still subject of ongoing research (see ref.~\cite{Koliofoti2023}) and beyond the scope of our discussion. Ref.~\cite{Rymarz2021} showed that one can engineer the GKP Hamiltonian effectively using a non-reciprocal Gyrator element to mimic the GKP Hamiltonian, which is an approach motivated by the previous discussion on the quantum Hall effect realization of the GKP Hamiltonian, where the magnetic field non-trivially couples the guiding center variables. In this proposal the gyrator allows to implement a comparable coupling. The gyrator element, however, is also a non-standard element in the toolbox of superconducting circuits and awaits realizations in the necessary parameter regimes.

\vspace{.5cm}

\begin{mybox}

\subsubsection{What is... a group projector? }\label{whatis:groupprojector}
\small
Given a finite group $G$ with linear action on a vector space $V$ $(g, v)\mapsto g.v \in V$  a group projector is defined as
\begin{equation}
\Pi_{G}\lr{x}=\frac{1}{|G|}\sum_{g\in G} g.x.
\end{equation}
Using closure of the group multiplication and the definition of the action $g.(h.x)=(gh).x$ it quickly follows that $\Pi_{G}\circ\Pi_{G}\lr{x}=\Pi_{G}\lr{x}$ and that it
 projects the input onto the $g$-symmetric elements $g.x=x\, \forall g\in G$ in $V$.

We have already seen a group projector before, given by the code space projector associated with a stabilizer code $\CS\subset \CU$
\begin{equation}
\Pi_S=\frac{1}{|\CS|}\sum_{s \in \CS} s,
\end{equation}
that, via the usual action on states in the physical Hilbert space takes any input state to a code state. 

A group projector can also be defined to act on operators on a Hilbert spaces under adjoint action. We stick to groups of unitary operators $G\subset \CU$, the group projector is defined as 
\begin{equation}
\Pi_{G}\lr{X}=\frac{1}{|G|}\sum_{g\in G} gXg^{\dagger}
\end{equation}
and now projects onto the \textit{commutant} of $G$, i.e. operators that commute with all elements in $G$. This idea is useful in state preparation relative to a quantum error correcting code: let $\CS$ be a stabilizer group, the projector $\Pi_{\CS}\lr{\rho}$ projects any input state onto a state supported only on operators in the centralizer $\CC \lr{\CS}$ of the stabilizer group: it projects onto a \textit{logical} Bloch vector.

One can go even a step further, the set of channels acting on a quantum system also constitute a vector space and one can define for a group of unitary operators $G\in \CU$ the action

\begin{equation}
(g, \CN\lr{\cdot}) \mapsto g^{\dagger}\CN\lr{g \cdot g^{\dagger} } g,
\end{equation}

which can also be understood as an adjoint action of the group $\lrc{g\otimes \overline{g},\; g\in G}$ on the natural representation of the channel (i.e. the representation where states are written as vectors). The use of group projectors in quantum information is colloquially referred to  as \textit{twirling} and some applications are summarized in fig.~\ref{fig:Twirling}.
\end{mybox}

\begin{figure}
\center
\includegraphics[width=\textwidth]{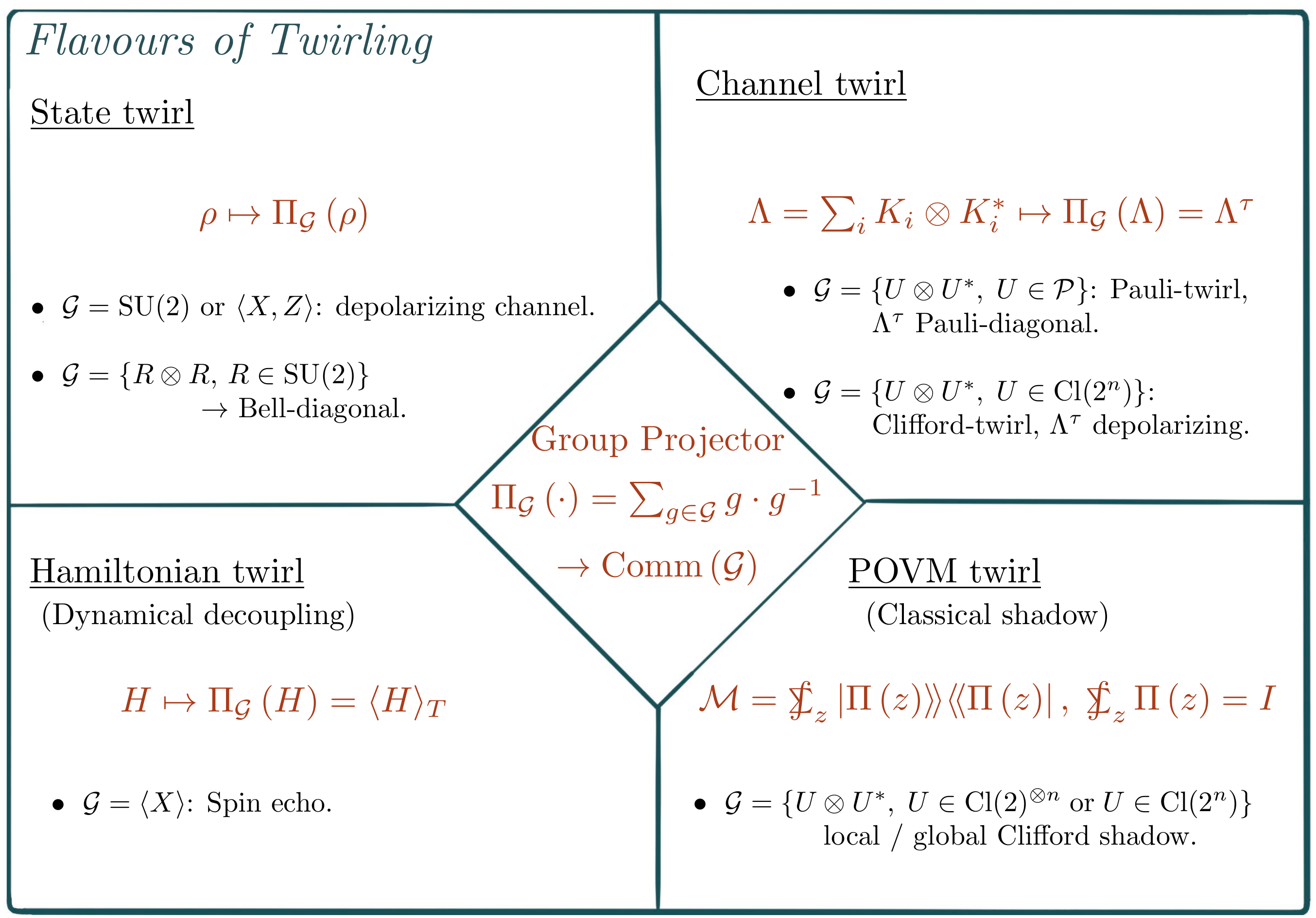}
\caption{Different applications of the group projector in quantum information, colloquially referred to as \textit{twirling}.}\label{fig:Twirling}
\end{figure}

\begin{figure}
\center
\centering
\includegraphics[width=.4\textwidth]{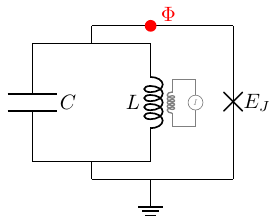}
\caption{Quantum harmonic oscillator comprising a cavity and a (super-) inductance coupled to a Josephson Junction. In gray a circuit element is indicated to implement displacements on the oscillator via inductive coupling.}\label{fig:circuit}
\end{figure}

To propose an implementation of a GKP Hamiltonian using superconducting circuits, we follow a different approach using only textbook circuit elements but time-dependent periodic driving of the Hamiltonian; the circuit is shown in fig.~\ref{fig:circuit}. 

The broad idea of this approach is to take advantage of a technique called \textit{dynamical decoupling}, where we periodically interlace the natural time evolution of a Hamiltonian $H_0$ with instantaneous unitary pulses $P_k$ after time periods $\Delta t_k=\tau_k T$, where $T_C$ is the periodicity of the periodic control sequence $U_C(t+nT_C)=U_C(t)\, \forall n \in \N$. The  pulses are set to satisfy $P_0=I$ and $\prod_{k=1}^M P_k = I$ so that they have a net-zero effect when the natural Hamiltonian evolution is trivial. 
When the periodic window is chosen sufficiently short, the time evolution of the system over each window is given by
\begin{equation}
U(T_C)=    P_M e^{-i\tau_{M-1}T_C H_0 }P_{M-1}\hdots P_1 e^{-i\tau_1T_C H_0 }P_0
\end{equation}
and can be described by the average Hamiltonian
\begin{equation}
     \overline{H}^{(0)}=\sum_{k=1}^{M} \tau_k Q_k^{\dagger} H_0 Q_k, \label{eq:dd2}
\end{equation}
with the accumulated control pulse $Q_k=P_{k-1}P_{k-2}...P_1$, with $P_k=Q_{k+1}Q_k^{\dagger}$.

Equation~\eqref{eq:dd2} reflects the structure of a group projector. Assume we have a group $G$ that contains  the cumulative pulses $Q_k$ of the driving sequence, then the average Hamiltonian

\begin{equation}
 \overline{H}^{(0)}= |G|^{-1} \sum_{g\in G} g H_0 g^{\dagger} =  \Pi_G\lr{H_0}
\end{equation}

is given by the projection of the Hamiltonian $H_0$ onto the commutant of the unitary group $G$. A common use of such driving sequence is given by $\tau_1=1/2,\; P_1=P_2=X$ so to get rid of an unwanted Hamiltonian evolution $H_0=\epsilon Z$. This trick is called \textit{spin echo} \cite{LEVITT1979473} and realizes the periodic evolution via

\begin{equation}
U(T_C)= X e^{-i \frac{T_C}{2} \epsilon Z } X  e^{-i \frac{T_C}{2} \epsilon Z } I= e^{-iT_C\frac{\epsilon}{2} \lr{XZX+Z} } = I,\label{eq:spin_echo}
 \end{equation}
which becomes trivial due to the fact that the $Z$ operator is not in the commutant of the group $\langle X\rangle $, i.e. $\Pi_{\langle X\rangle} \lr{Z}$. In the exponent of eq.~\eqref{eq:spin_echo} we can also recognize that the average Hamiltonian that describes the time evolution over the interval $T_C$ is indeed given by the group-projected Hamiltonian $\Pi_{\langle X\rangle} \lr{H_0}$. 

Dynamical decoupling is a powerful tool when combined with quantum error correction. When one designs the pulse sequence such that the cumulative pulses in eq.~\eqref{eq:dd2} reflect the stabilizers of a stabilizer group, the effective Hamiltonian obtained will commute with the stabilizer group which allows to implement logical unitary evolutions where the dynamical decoupling sequence filters out unwanted errors from spurious couplings. 
From a Hamiltonian engineering perspective, this trick allows the engineering of target Hamiltonians that are potentially hard to realize via a direct physical implementation by building a \textit{substrate} Hamiltonian that is perhaps easier to build but has potential support on many terms that are unwanted in the target and then filter out these unwanted terms by an appropriate dynamical decoupling sequence. When targeting a stabilizer Hamiltonian, the pulses will constitute precisely of elements in the centralizer of the stabilizer group.

This is the strategy we follow to find an implementation for the GKP stabilizer Hamiltonian. First we show how one can find a suitable substrate Hamiltonians by means of a Josephson junction coupled to a high-frequency $LC$ oscillator, and then we show how a dynamical decoupling sequence with displacement pulses in $\CL^{\perp}$ can be designed that filters out the unwanted terms in that substrate Hamiltonian. To analyse Hamiltonians, it is convenient to decompose them into displacement operators and study the structure of their characteristic function

\begin{equation}
H=\int_{\R^{2n}} d\bs{x}\, h\lr{\bs{x}} D\lr{\bs{x}},
\end{equation}
where $ h\lr{\bs{x}}=\Tr\lrq{D^{\dagger}\lr{\bs{x}} H}$. 

The targeted stabilizer Hamiltonian in eq.~\eqref{eq:GKPHsq} simply has a characteristic function proportional to

\begin{equation}
h_{\rm GKP}\lr{\bs{x}}= -\delta\lr{\bs{x}- \bs{\xi}_1}-\delta\lr{\bs{x}+ \bs{\xi}_1}-\delta\lr{\bs{x}- \bs{\xi}_2}-\delta\lr{\bs{x}+ \bs{\xi}_2},\label{eq:GKPchar}
\end{equation}
with $\bs{\xi_i}=\sqrt{2}\bs{e}_i$. That is, it is simply given by $4$ Dirac-delta peaks distributed on the phase space points corresponding to the generating set and its inverse, which are depicted in fig. ~\ref{fig:hchar}. 

\begin{figure}
\center
\includegraphics[width=.5\textwidth]{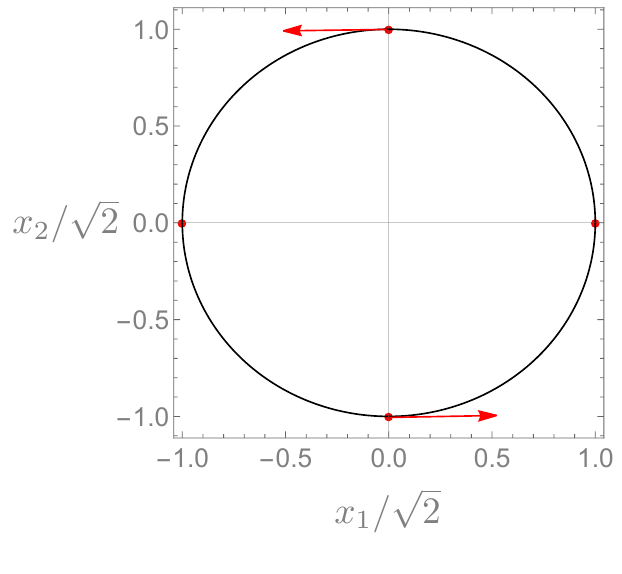}
\caption{The characteristic function $h_{\rm JJ}(\bs{x}; t)$  with $\varphi=\sqrt{2\pi}$ traverses the indicated circle with time. For $\omega t\geq 2\pi$ the rotating points are smeared out over the circle which represents the first order RWA $\overline{h}^{(1)}_{JJ}(\bs{x}; \,t)$. The delta peaks of the characteristic function $h_{GKP}(\bs{x})$ are indicated in red.} \label{fig:hchar}
\end{figure}

\subsubsection{The substrate Hamiltonian}
To obtain a suitable substrate Hamiltonian, consider a Josephson junction coupled to an LC oscillator as in fig.~\ref{fig:circuit}.
Using standard circuit quantization \cite{Girvin_LectureNotes,Ciani:765439} the Hamiltonian can be expressed in terms of flux $\Phi$ and charge $Q$ with $[\Phi,Q]=i\hbar$ as

\begin{equation}
H=\frac{Q^2}{2C}+\frac{\Phi^2}{2L}-E_J\cos\left(\frac{2\pi}{\Phi_0} \Phi\right),
\end{equation}
where $\Phi_0=\frac{h}{2e}=2e R_Q$ is the flux quantum.
Expressing the Hamiltonian in terms of the cavity frequency $\omega=\sqrt{LC}^{-1}$ and creation and annihilation operators
\begin{align}
a &=\frac{1}{\sqrt{2L\hbar\omega}}\Phi+\frac{i}{\sqrt{2C\hbar\omega}}Q,\\
a^{\dagger} &=\frac{1}{\sqrt{2L\hbar\omega}}\Phi-\frac{i}{\sqrt{2C\hbar\omega}}Q,
\end{align}
which obey the commutation relation $[a,a^{\dagger}]=1$, the flux and charge operators can be expressed as
\begin{align}
\Phi&=\sqrt{\frac{\hbar Z}{2}} (a+a^{\dagger}) \\
Q &=-i\sqrt{\frac{\hbar}{2Z}} (a-a^{\dagger}),
\end{align}
where $Z=\sqrt{\frac{L}{C}}$ is the impedance of the cavity mode.
In this representation the Hamiltonian becomes
\begin{equation}
H= \underbrace{\hbar \omega\left(a^{\dagger}a+\frac{1}{2}\right) }_{H_0}-E_J \cos\left( \underbrace{\frac{2\pi}{\Phi_0} \sqrt{\frac{\hbar Z}{2}}}_{\varphi} (a+a^{\dagger}) \right).
\end{equation}
The constant factor inside the $cos(\cdot)$ term can be simplified to
\begin{equation}
\varphi=\sqrt{\frac{(2\pi)^2}{\frac{h^2}{(2e)^2}} \frac{\hbar Z}{2}}=\sqrt{\frac{\pi Z}{R_Q}}.
\end{equation}
In the frame rotating with $H_0$ the Hamiltonian reads
\begin{align}
H &=-E_J \cos\left( \varphi( e^{-i\omega t}a+e^{i\omega t}a^{\dagger}) \right)\\
&=-\frac{E_J}{2} \left\{  \exp( i\varphi e^{i\omega t}a^{\dagger}+i\varphi e^{-i\omega t}a  )+\exp( -i\varphi e^{i\omega t}a^{\dagger}-i\varphi e^{-i\omega t}a  ) \right\} \\
&=-\frac{E_J}{2} \left\{ D\lr{\frac{\varphi}{\sqrt{\pi}}R^{\omega t}\bs{e}_2  } + D\lr{-\frac{\varphi}{\sqrt{\pi}}R^{\omega t}\bs{e}_2  } \right\},
\end{align}
where we have expressed $\hat{a}=\lr{\hat{q}+i\hat{p}}/\sqrt{2}$ and defined the rotation
\begin{equation}
R^{\omega t} = 
\begin{pmatrix}
\cos \omega t & -\sin \omega t \\ \cos\omega t & \sin \omega t
\end{pmatrix}
\end{equation}
such that the Hamiltonian characteristic function in the rotating frame becomes
\begin{equation}
h_{\rm JJ}(\bs{x}; t)=-\frac{E_J}{2} \lrc{\delta\lr{\bs{x}-\frac{\varphi}{\sqrt{\pi}}R^{\omega t}\bs{e}_2 }+\delta\lr{\bs{x}+\frac{\varphi}{\sqrt{\pi}}R^{\omega t}\bs{e}_2 } }. \label{eq:hJJ}
\end{equation}

For the characteristic function to define support on the targeted stabilizer and overlap with the characteristic function in eq.~\eqref{eq:GKPchar}, we hence see that one needs to tune the impedance such that $\varphi=\sqrt{2\pi}$, which corresponds to a value $Z=2R_Q$.

Under time evolution, in the frame co-rotating with $H_0$, the Hamiltonian of the Josephson junction traces out a path that reproduces the GKP stabilizer Hamiltonian at every quarter cycle $T_{\rm \GKP}=\frac{\pi}{2\omega}$. If one assumes the possibility to selectively turn the coupling between the Josephson junction on and off via some time-dependent function $E_J\mapsto f(t)E_J$ that effectively modulates the Josephson energy, one could simply reproduce the GKP Hamiltonian by ``blinking'' the JJ up at those selected points. This is in fact a simplified version of the strategy to design a GKP Hamiltonian proposed in refs.~\cite{Kolesnikow_2024, nathan2024selfcorrecting}, where ref.~\cite{Kolesnikow_2024} proposes the use of either a flux modulated squid loop and to facilitate this drive and ref.~\cite{nathan2024selfcorrecting} assumes the existence of a fast  switching mechanism.

With $\frac{d}{dt}U_0(t)=-iH_0 U_0(t)$, the full rotating frame Hamiltonian is given by $H_{rot}(t)=U_0^{\dagger}HU_0+i\frac{d U_0^{\dagger}}{dt}U_0$. 
The unitary evolution of the system in this frame is given by the Magnus expansion
\begin{equation}
U(t)=\mathcal{T}\exp{\left(-i\int_0^{t} H_{rot}(t') dt\right)}= \exp{\left(-i \overline{H}(t)\right)}, \label{eq:Magnus}
\end{equation}
for which the first order characteristic functions  of $\overline{H}(t)=\sum_k \overline{H}^{(k)}(t)$ can be evaluated to
\begin{align}
\overline{h}^{(1)}(\bs{x}; \,t)&=\int_0^t h_{\rm JJ}(\bs{x}; t) dt', \label{eq:h_first}
\end{align}
and higher orders vanish as $\hbar\omega \ll E_J$ \cite{Blanes_2010, Terhal_2020}. Hence, for sufficiently large frequency $\omega$, the path traced out by the rotating frame evolution of eq.~\eqref{eq:hJJ} \textit{smears out} to become a solid circle which has the support we were looking for. This process of replacing the effective rotating frame Hamiltonian with its time-average is also called the \textit{rotating wave approximation}.
We can also consider the effective Hamiltonian in the Fock basis, by computing
\begin{equation}
\int_0^t dt'\,\braket{n| D\lr{R^{\omega t'}\bs{x}} | m } \xrightarrow{\omega t\geq 2\pi} \delta_{m,n} L_n\lr{\pi \|\bs{x}\|^2} e^{-\frac{\pi \|\bs{x}\|^2}{2}}t,
\end{equation}
where $L_n\lr{\cdot}$ are the Laguerre polynomials. As the resulting Hamiltonian now is diagonal in the Fock basis, it has also become rotation invariant.

\subsubsection{The decoupling sequence}

To distil a GKP Hamiltonian from the substrate we need to derive a time- and pulse sequence such that the operation applied to $H_0$ in eq.~\eqref{eq:dd2} looks like a group projector onto the centralizer of the GKP stabilizer group, that is, we are trying to construct the projector

\begin{equation}
\Pi_{\CL^{\perp}} \lr{\cdot}=\sum_{\bs{\xi}^{\perp}\in \CL^{\perp}} D\lr{\bs{\xi}^{\perp}} \cdot D^{\dagger}\lr{\bs{\xi}^{\perp}}.
\end{equation}
As already discussed above this projector is unphysical, so that the best we can do is to approximate it via a distribution $\mu_{M^{\perp}} \lr{\bs{x}}$ that converges against an invariant measure on $\CL^{\perp}$ in some limit.

Equipped with such a formulation, an approximate Hamiltonian twirl is implemented via

\begin{align}
\widetilde{\Pi}_{\CL^{\perp}}\lr{H_0} \nonumber
&= \int_{\R^{2n}} d\bs{x}\, h\lr{\bs{x}} \int_{\R^{2n}} d\mu\lr{\bs{\gamma}}\, D\lr{\bs{\gamma}}D\lr{\bs{x}}D^{\dagger}\lr{\bs{\gamma}} \nonumber\\
&=  \int_{\R^{2n}} d\bs{x}\, h\lr{\bs{x}} \underbrace{\lrq{ \int_{\R^{2n}} d\mu\lr{\bs{\gamma}}\, e^{-i2\pi \bs{\gamma}^TJ \bs{x} }  }}_{\nu\lr{x} } D\lr{\bs{x}}, \label{eq:PiH0}
\end{align}
that is, it implements a simple multiplication of the characteristic function of the Hamiltonian with a kernel $\nu\lr{\bs{x}}$, which is given by the symplectic Fourier transform of the measure. As a Fourier transform, it is easy to convince ourselves that convolution measures $\mu_1 * \mu_1$ yield the product of the kernels $\nu_1\nu_2$, such that 
any non-trivial initial measure $\mu^1$ with $\nu^1$ is approximatively compactly supported on $\CL^{\perp}$ can be amplified to sharpen its support by considering the $N$-fold convolution
$\mu^{* N}$ with Fourier transform $\nu^N$.

Specifically, we define a random walk from the joint distribution of $N'=2N$ half-steps $\pm \bs{\xi}_i / 2$, each of which are selected with $1/2$ probability at each step. Define for $i=1\hdots 2n$ the associated (discrete) measure 
\begin{equation}
\mu'_{i}\lr{\bs{x}} =\frac{1}{2}\delta \lr{\bs{x}-\bs{\xi}^{\perp}_i/2}+\frac{1}{2}\delta \lr{\bs{x}+\bs{\xi}^{\perp}_i/2},
\end{equation}
so that we obtain the measure corresponding to $N$ steps of the random walk as $\mu_i^{(*N)} :=\mu_i'^{(*2N)} $.
The corresponding kernel function is for a single step is
\begin{align}
\nu_i\lr{\bs{x}}&= \int_{\R^2} d\mu_i \lr{\bs{\gamma}} e^{-i2\pi \bs{\gamma}^TJ\bs{x}} \nonumber \\
&=\cos^2\lr{\pi \lr{\bs{\xi}_i^{\perp}}^TJ\bs{x}},
\end{align}
such that, in the limit $N\rightarrow \infty $, $\nu_i^N\lr{\bs{x}}$  suppresses all contributions $\bs{\alpha}$ except for those in the symplectic dual of $\bs{\xi}_i^{\perp}$. We define the joint measure over all generators in $M^{\perp}$ to be the joint random walk given by the $2n$-fold convolution
\begin{equation}
\mu_{M^{\perp}} = \mu_1 * \mu_2 *\hdots *\mu_{2n}
\end{equation}
which has Fourier transform
\begin{equation}
\nu_{M^{\perp}}\lr{\bs{x}}=\prod_{i=1}^{2n} \cos^2\lr{\pi \lr{\bs{\xi}_i^{\perp}}^TJ\bs{x}}.
\end{equation}

For the square GKP code with $M^{\perp}=\frac{1}{\sqrt{2}} I_2$,  one can rewrite
\begin{equation}
\mu_{M^{\perp}}^N \lr{\bs{\gamma}}=\mu_{M^{\perp}}^{* N} = \sum_{n,m=-N}^{N} P^N\lr{n,m} \delta\lr{n \bs{\xi}^{\perp}_1+m \bs{\xi}^{\perp}_2-\bs{\gamma}}
\end{equation}
with 
\begin{equation}
P^N\lr{n,m}= 2^{-4 N}\binom{2N}{n+N}\binom{2N}{m+N}.
\end{equation}

The kernel functions for the square and hexagonal GKP codes are depicted in fig.~\ref{fig:nuMap}.
\begin{figure}
\center
\begin{tikzpicture}[scale=1.1]
\node at (-3.5,0){\includegraphics[width=3.7cm]{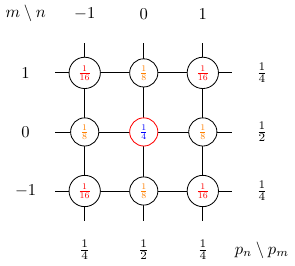}};
\node at (0,0){\includegraphics[width=3.5cm]{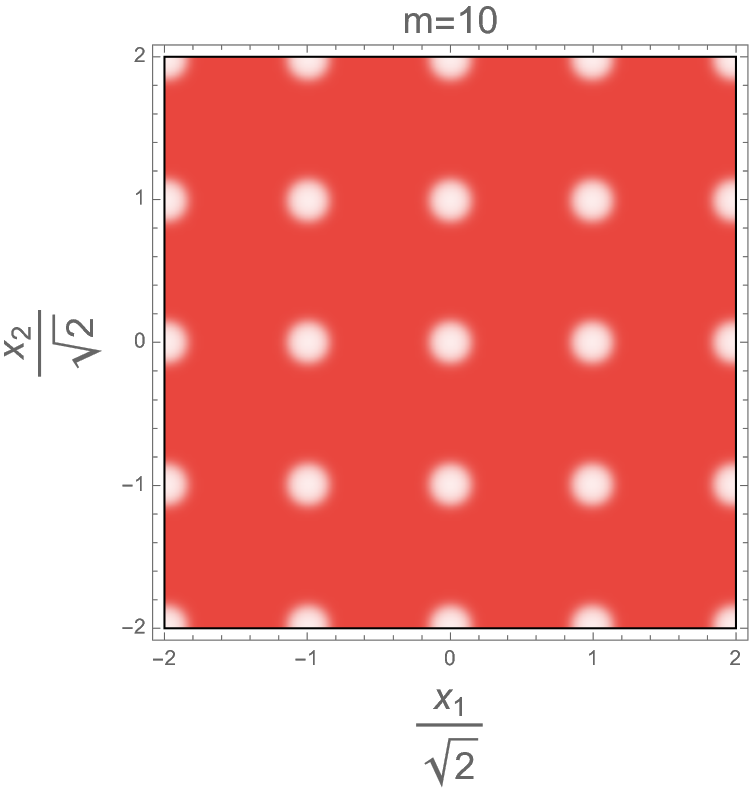}};
\node at (3.5,0){\includegraphics[width=3.5cm]{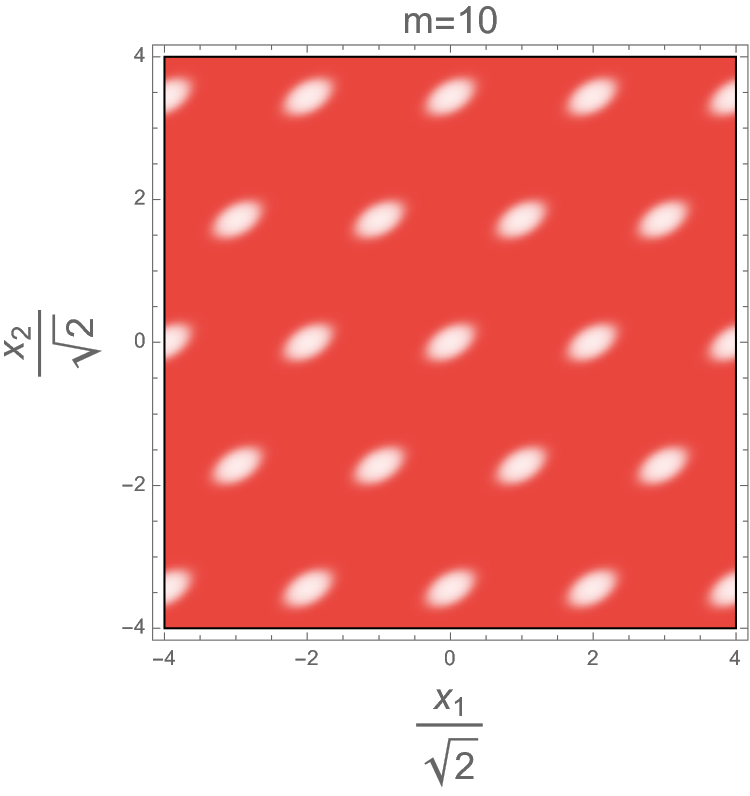}};
\node at (5.7,.2){\includegraphics[width=.8cm]{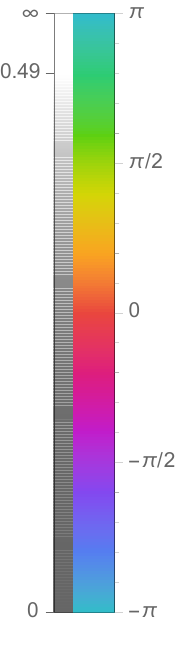}};
\end{tikzpicture}

\caption{(left) probability distribution for $n=1$ step of the random walk and (right) the kernel functions $\nu_{M^{\perp}}$ for square and hexagonal GKP codes. } \label{fig:nuMap}
\end{figure}
Using this probability measure, the expression for the approximately twirled Hamiltonian becomes

\begin{align}
\widetilde{\Pi}^N_{\CL^{\perp}}\lr{H_0} 
&= \sum_{n,m=-N}^N  P^N\lr{n,m} D\lr{n \bs{\xi}^{\perp}_1+m \bs{\xi}^{\perp}_2} H_0 D^{\dagger}\lr{n \bs{\xi}^{\perp}_1+m \bs{\xi}^{\perp}_2} \label{eq:twirlH1}\\
&=\int_{\R^{2n}} d\bs{x}\, h\lr{\bs{x}} \nu_{M^{\perp}}^N\lr{\bs{x}}\, D\lr{\bs{x}}.\label{eq:twirlH2}
\end{align}

While the last line, eq.~\eqref{eq:twirlH2}, shows exactly the approximate projection of the Hamiltonian onto one whose characteristic function shares support with $\nu_{M^{\perp}}^N\lr{\bs{x}} $, which is confined to $\bs{x}\in \CL$ for large $N$, the expression in eq.~\eqref{eq:twirlH1} looks oddly like the expression for the effective Hamiltonian under dynamical decoupling from eq.~\eqref{eq:dd2}! These expressions become exactly the same when we equate 
\begin{align}
\tau_{k}&=\tau_{k(n,m)}=P^N(n,m), \nonumber\\
Q_k&=Q_{k(n,m)}=D\lr{n \bs{\xi}^{\perp}_1+m \bs{\xi}^{\perp}_2}.
\end{align}
What is left is hence to find the function $k(n,m)$ that assigns a time step to the corresponding displacement label.

\subsubsection{Control path ordering} 
To minimize the experimental effort of implementing the control pulses $P_k=Q_{k+1}Q_k^{\dagger}$ and to maintain $\prod_k P_k=I$, we construct a \textit{control graph}:  the vertices  $(n,m)$ of the control graph label the (accumulated) displacement amplitudes $\{Q_k\}$ and edges represent the allowed transitions, that is choices of $\{P_k=Q_{k+1}Q_k^{\dagger}\}$ to map between different accumulated control pulses. To minimize the necessary displacement amplitude at each instance we choose the edge connectivity as in a  \textit{kings graph} $C=(V=\{(n,m)\},E_{king})$, which is known to have Hamiltonian cycles for each $N$. The ordering $k(\cdot,\cdot)$ is then given by a  Hamiltonian cycle on the vertices of $C$ starting at $k(0,0)=1$. This construction ensures that each instantaneous control pulse displacement amplitude is bounded by a constant $\|\bs{\xi}\| \leq 1$. For an example of a possible control sequence for $N=1$ see fig.~\ref{fig:Hamiltonian_cycle}. 
The $N-step$ Hamiltonian twirl is mapped to an open-loop control sequence consisting of $M=(2N+1)^2$ displacement pulses.

\begin{figure}
\center
\includegraphics[width=.4\textwidth]{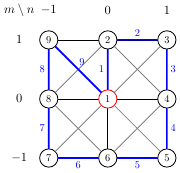}
\caption{One possible ordering of the control path as given by a Hamiltonian cycle on the control graph. Each vertex is associated with the accumulated control pulse $Q_k$ at time-index $k=k(n,m)$ and each edge with the instantaneous control pulses $P_k$ applied at the respective time step $k$ as indicated by the time labels on the edges.} \label{fig:Hamiltonian_cycle}
\end{figure}

\subsubsection{Numerical results}

Tuning $Z=2R_Q$, the substrate Hamiltonian is given by
\begin{equation}
H_{sub}/E_J=-e^{-\pi}\sum_n L_n(2\pi)\ket{n}\bra{n}, \label{eq:H_sub}
\end{equation}
to which the decoupling sequence is applied. The resulting average time Hamiltonian can then be numerically diagonalized; the Wigner functions for the lowest two eigenstates of the effective average Hamiltonian under $N-level$ logical Twirl $H^{(0)}_{av}$  for $N=1, 5, 10, 15$ and $E_J=1$  as well as the spectrum are shown in fig.~\ref{fig:hJJ_approx}.

The lowest eigenstates are found to approximate the GKP magic states

\begin{align}
\ket{H^+_{\Delta}}&=\cos{\left(\frac{\pi}{8}\right)}\ket{0_{\Delta}}+\sin{\left(\frac{\pi}{8}\right)}\ket{1_{\Delta}},\\
\ket{H^-_{\Delta}}&=-\sin{\left(\frac{\pi}{8}\right)}\ket{0_{\Delta}}+\cos{\left(\frac{\pi}{8}\right)}\ket{1_{\Delta}},
\end{align}
with approximation parameter $\Delta_{q/p} \propto N^{-0.185}$, which can be understood from the fact that the substrate Hamiltonian had a rotational symmetry, of which the $\pi/2$ rotation symmetry survives the twirl as this is the symmetry shared with the lattice $\CL^{\perp}$. 

In the above description, the approximate GKP computational basis states are parametrized as

\begin{align} 
\ket{0_{\Delta}} &=\int_{\mathbb{R}} dq\;  \sum_{n \in \mathbb{Z}} e^{-2\Delta^2 \pi n^2 } e^{-\frac{1}{2 \Delta^2} (q-2n\sqrt{\pi})^2}\ket{q}, \label{eq:codestate_0}\\
\ket{1_{\Delta}}  &=\int_{\mathbb{R}} dq\;  \sum_{n \in \mathbb{Z}} e^{-2 \Delta^2 \pi n^2 }e^{-\frac{1}{2 \Delta^2} (q-(2n+1)\sqrt{\pi})^2}\ket{q},\label{eq:codestate_1}
\end{align}
while the finite squeezing parameter is measured via \cite{Weigand_2020}
\begin{equation}
\Delta_q=\sqrt{\frac{-1}{\pi}  \mathrm{ln}\left( {|\Tr\lr{\rho D\lr{\bs{\xi}_1} }|}\right)} ,\hspace{1cm} \Delta_p=\sqrt{\frac{-1}{\pi}  \mathrm{ln}\left( {|\Tr\lr{\rho D\lr{\bs{\xi}_2} }|}\right)}.
\end{equation}

\begin{figure}
\includegraphics[width=\textwidth]{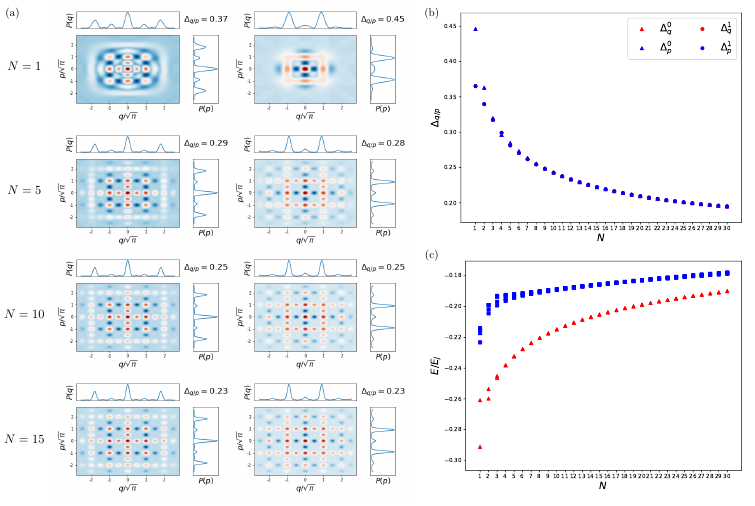}
\caption{$(a)$ Wigner functions of the two lowest eigenstates of $H^{(0)}_{av}$ for twirling level $N=1,..,15$ together with their effective squeezing parameter are shown. $(b)$  Finite squeezing parameters for $N=1..30$  and $(c)$  the ten lowest eigenenergies of $H^{(0)}_{av}$ are plotted. The two lowest (degenerate for large $N\geq 4$) eigenenergies approximating the GKP $\ket{\overline{H_{\pm}}}$ states are colored red and are separated by a gap that shrinks with $N$ from the higher levels.
The effective squeezing of the two lowest eigenstates becomes symmetric in $q,\,p$ for $N\geq 8$ and approximately scales with the twirling-level $\Delta_{q/p}\propto N^{-0.185}$.\\
Code that was used to produce this figure can be found in ref.~\cite{git_GKPdd}.
} \label{fig:hJJ_approx}
\end{figure}

We find that already for a low level of the twirl, the low energy spectrum of the average time Hamiltonian reproduces the desired behavior of an approximate GKP stabilizer Hamiltonian. 

\subsubsection{Parameter regimes \& related work}

The RWA for $h_{JJ}^{(1)}$ is generally valid for sufficiently large cavity frequencies $\hbar \omega \gg E_J$ which could be considered smeared out in the limit $\omega t_{min} \geq 2\pi$, where $t_{min}=2^{-4N}T_C$ corresponds to the smallest timescale where this approximation needs to be valid. The latter inequality sets the lower bound for the period length of the stroboscopic evolution  as  $T_C \geq \frac{2\pi}{{\omega}} 2^{4 N}$. For the average Hamiltonian to remain valid, it would be desirable to have a cavity with large $\omega$ to minimize this bound for some finite $N$.

The final limitation of this scheme is given by the speed limit for displacement operations. For the assumption of using instantaneous pulses (which is also sometimes called \textit{bang-bang dynamical decoupling}) to hold, it is necessary that displacements of amplitude $\sqrt{2\pi}\|\bs{\xi}^{\perp}_{1/2}\|=\sqrt{\pi}$ can be realized in a time $t_{disp}\ll 2^{-4N}T_C$.
This necessitates that $\sqrt{2} T_X \ll 2^{-4N}T_C$ where $T_X$ is the minimal time it takes to realize a (elementary) displacement by $\sqrt{2\pi}\|\bs{\xi}^{\perp}_{1/2}\|=\pi$. Altogether, this imposes a bound of $T_X \ll\frac{1}{\sqrt{2}} \frac{2\pi}{\omega}$ on the speed to implement displacements.

As a generic example, assuming a cavity frequency of $\frac{\omega}{2\pi}=5.26 GHz$  \cite{Campagne_Ibarcq_2020}, elementary displacements must be realizable in a time $T_X \ll 0.13 ns$, which is a demanding assumption but appears within experimental reach. 

While we have been motivating this discussion from the point of view of dynamical decoupling, this pathway to engineer GKP Hamiltonians can be more broadly understood as a version of \textit{Floquet engineering} \cite{Rubio_Abadal_2020, Mori_2018} and a rigorous connection can be established by comparing \textit{Floquet-Magnus} expansion to the expression for the average Hamiltonian given earlier. Extending this work, there have been proposals by Kolesnikow and Grimsmo \cite{Kolesnikow_2024}, Sellem et al. \cite{sellem2023gkp}, as well as Nathan et al.\cite{nathan2024selfcorrecting}, who proposed protocols where the effective Josephson energy is modulated by inclusion of a rapid switch. Ref.~ \cite{Kolesnikow_2024} proposed the use of such protocol to prepare GKP states by an adiabatic ramp on the effective Josephson energy and  ref.~\cite{nathan2024selfcorrecting} investigated its self-correcting behavior in an open quantum system, with coherence time estimates growing exponentially with the loss rate up to $O(1ms-1s)$ in their estimates. While this seems highly promising, the results rely on the existence of the rapid switching process, which as yet appears difficult to engineer \cite{nathan2024selfcorrecting}. 

Finally, note that the scheme discussed here is not restricted to the engineering of GKP stabilizer Hamiltonians. By simply adapting the pulses to stabilizer shifts (corresponding to an impedance $Z=R_Q$ in the rotating JJ), the decoupling sequence preserves Hamiltonian elements corresponding to \textit{logical} GKP displacement. This strategy can be used to remove spurious couplings in the implementation of logical gates for the GKP code, which is closer in spirit to the original intention of dynamical decoupling.

\section{The Dream}

In this chapter we have developed a high-level understanding of some experimental approaches to realize the GKP code. As the GKP code bears large promise in providing a quantum memory with good protection against displacements or photonlosses (see refs.~\cite{Terhal_2020, Noh_Capacity, Albert_2018}), many research groups from academia and industry are developing implementation strategies and contribute to the development of this kind of ``quantum engineering theory''~\cite{Bourassa_Blue_2021, Royer_2020, lachancequirion2023autonomous, rojkov2023twoqubit}. I refer to ref.~\cite{Brady_2024} for a more extensive overview over experimental approaches. This chapter comes with two final dreams.
\\

The first dream, of course, is that it works. 
\\

Successful implementation of the GKP code so to facilitate fault-tolerant quantum computing with the demonstration of advantages in using quantum algorithms for real-world problems is a desirable technological milestone. The cumulative developments in the field at the time of writing show hope that we are moving closer to this goal. But this would not be a thesis on the theory of GKP quantum error correction if that was all we dared to dream about. 

This chapter has also shown that the GKP code connects to many interesting topics in physics. Notably, the discussion on the GKP Hamiltonians presented here is vastly extended by work presented in refs.~\cite{Ganeshan2016, Ganeshan2022}, where a connection between Hamiltonians of the GKP code and Abelian topological phases has been pointed out. The understanding of the GKP code Hamiltonian in relation with the physics of the quantum Hall effect provides an interesting gateway to future investigations and it woudl be desirable to understand how the GKP code appears in a many-particle treatment of the quantum Hall physics discussed in this chapter to understand how topological properties, such as the existence of chiral edge modes and Chern numbers manifest in relation to coding theoretic properties. Furthermore, the formulation of this connection provided here left the choice of lattice as a variable parameter. It would be interesting to investigate whether phenomena such as Hall viscosity \cite{Read_2009, Fremling_Thesis, haldane2009hallviscosityintrinsicmetric} manifest when the crystal lattice undergoes a lattice automorphism through an adiabatic deformation process and how the related Berry phase manifests in comparison to the understanding developed in chapter~\ref{chap:Theory}.
Further beyond, as the quantum Hall effect has been subject to field-theoretic treatments in the literature, this connection is likely to provide a gateway to develop a field-theoretic generalization of the GKP code, which I expect to necessitate the development of a coding theoretic machinery in the language of vertex operator algebras \cite{Hermanns_2008} and may point to yet another bridge  between interesting topics from mathematics and physics. 
The final dream of this thesis is thus to pursue a deeper understanding how exactly the GKP code and existing\footnote{Either on paper or in the lab.} physical systems with topological properties relate, and how exactly the coding theoretic properties manifest in physics and vice versa.

\chapter{An open door}
In this thesis,  we have developed a coding theory of GKP codes,  vastly extending previous work by Gottesman et al.  in ref.~\cite{GKP} and Harrington and Preskill in refs.~\cite{HarringtonPreskill,  Harrington_Thesis} and built a lattice theoretic-,  as well as algebraic geometric understanding of the structure of GKP codes where the latter lead to a classification of fault-tolerance for GKP Clifford gates within the fiber bundle fault-tolerance framework proposed in ref.~\cite{gottesman2017fiber}.  Equipped with the lattice theoretic understanding,  we have shown how a lattice-based cryptographic scheme can be used to derive families of GKP codes with good parameters and,  by quantifying the hardness to decode these codes,  proposed a novel quantum cryptosystem that builds entirely on decoding hardness for generic GKP codes.  Finally,  we discussed possible routes to implement the GKP code,  with the key result proposal of a Floquet-implementation of the GKP stabilizer Hamiltonian.  

The focus of this thesis has generally not been on deriving results,  but on the development of theoretical methods and tools.  In this quest,  we have seen how GKP error correction connects to a vast number of fields of active mathematical and physical research and pointed towards many possible routes for future research,  of which selected highlights were phrased as \textit{Dreams} at the end of each chapter. 

My hope for this work is that it conveys the exciting richness of the theory of quantum error correction with the GKP code and motivates the reader to simply\\

\center dream on.

\blfootnote{For the curious: the title of this thesis is a hommage to the $2001$ movie ``\textit{The fabulous destiny of Am{\'e}lie Poulain}'', which carries the german title, verbatim translated, ``The fabulous world of Amelie''. Am{\'e}lie's character and life story I found to nicely mirror that of the GKP code, which coincidentally was also published in the same year.}
\leavevmode\thispagestyle{empty}\newpage

\newpage
\setcounter{section}{0}
\renewcommand\thesection{\Alph{section}}
\leavevmode\thispagestyle{empty}\newpage
\raggedright
\chapter*{Appendix}
\addcontentsline{toc}{chapter}{Appendix}  
\section{Complex theta functions}\label{app:theta}
We have seen in chap.~\ref{chap:Theory} that every GKP code can be described by a type $D$ and a complex matrix $\Omega \in \HH_n$, summarized in the period matrix $\Pi=\begin{pmatrix}
D & \Omega
\end{pmatrix}$.
Similar to the Jacobi theta function defined in the introduction, there is a quasiperiodic holomorphic multivalued function $ \C^n\rightarrow \C$. This is the the theta function well defined for $\Omega \in \HH_n$~\cite{Tata_1}

\begin{equation}
\vartheta\lr{\bs{z}, \Omega}=\sum_{\bs{n}\in \Z^n} e^{i\pi \bs{n}^T\Omega \bs{n} + i2\pi \bs{n}^T\bs{z}}.
\end{equation}

Similar to the $n=1$ case, this theta function is quasi-periodic with respect to the complex lattice $\Lambda_{\Omega}=\Z^n + \Omega \Z^n$, such that

\begin{align}
\vartheta\lr{\bs{z}+\bs{m}, \Omega}&=\vartheta\lr{\bs{z},  \Omega},\; \bs{m}\in \Z^n,\\
\vartheta\lr{\bs{z}+\Omega\bs{m}, \Omega}&=e^{-i\pi \bs{m}^T\Omega \bs{m}-i2\pi\bs{m}^T\bs{z}} \vartheta\lr{\bs{z},  \Omega},\; \bs{m}\in \Z^n,
\end{align}
and we can also define the modification for type $D=\diag\lr{d_1, \hdots ,d_n} $
\begin{align}
\vartheta_D\lr{\bs{z}, \Omega}&:=\vartheta\lr{D^{-1}\bs{z}, D^{-1}\Omega D^{-1}}\\
&=\sum_{\bs{n}\in \Z^n} e^{i\pi \bs{n}^TD^{-1}\Omega D^{-1} \bs{n} + i2\pi \bs{n}^T D^{-1}\bs{z}},
\end{align}
which is quasi-periodic over the lattice $\Lambda_{D, \Omega}=D\Z^n + \Omega \Z^n$
with
\begin{align}
\vartheta_D\lr{\bs{z}+D\bs{m}, \Omega}&=\vartheta_D\lr{\bs{z},  \Omega},\; \bs{m}\in \Z^n, \nonumber\\
\vartheta_D\lr{\bs{z}+\Omega\bs{m}, \Omega}&=e^{-i\pi \bs{m}^T\Omega \bs{m}-i2\pi\bs{m}^T\bs{z}} \vartheta\lr{\bs{z},  \Omega},\; \bs{m}\in \Z^n, \label{eq:quasi_Omega}
\end{align}
where eq.~\eqref{eq:quasi_Omega} is easily checked by completing the square. 
Consequentially, the roots of the theta function 
\begin{equation}
\Theta_{D, \Omega}=\lrc{\bs{z}\in \C^n:\; \vartheta_D\lr{\bs{z}, \Omega}=0} \subseteq \C^n /\Lambda_{D, \Omega}
\end{equation}
are a translation symmetric space under the lattice $\Lambda_{D, \Omega}$ and behave as if they lived on a complex torus.

\section{Magic states}\label{app:Magic}
\blfootnote{This section is also found in the appendix ref.~\cite{Conrad_2024}, from where it was taken. I include it here to provide a more comprehensive reference on all things GKP.}
The symplectic orthogonal automorphism group $\Aut^{SO}\lr{\CL^{\perp}}$ of GKP codes have a special application in that they give rise to magic states. Let $\ket{0}^{\otimes n}$ be the $n$-mode vacuum state. The vacuum state is rotation symmetric and arguably the simplest state to prepare. Further let 
\begin{equation}
\Pi_{M}=\sum_{\bs{\xi}\in \CL\lr{M}} e^{i\phi_M\lr{\bs{\xi}}} D\lr{\bs{\xi}} 
\end{equation}
be the code space projector of a GKP code with generator $M$. In the case of a scaled GKP code where the symplectic Gram matrix $A$ has only even entries we further have that the phases appearing in the group elements are trivial $\phi_M\lr{\bs{\xi}} =0 \mod 2\pi$, such that we simply write $\Pi_{\CL}$ and for $\hat{U}_S$ the Gaussian unitary associated to a symplectic automorphism $S\in\Aut^{SO}\lr{\CL^{\perp}}=\Aut^{SO}\lr{\CL}$, we have
\begin{equation}
\lrq{\hat{U}_S, \Pi_{\CL}}=0.
\end{equation} 
This implies that 
\begin{equation}
\ket{M}=\Pi_{\CL}\ket{0}^{\otimes n} 
\end{equation}
is a $+1$ eigenvalued eigenstate of $\hat{U}_S$. $\ket{M}$ lives in the codespace of the GKP code and is the $+1$ eigenstate of the logical Clifford gate associated to $S$. Using a logical CNOT gate and the ability to perform computational basis measurements states of this type can be consumed to implement non-Clifford gates to lift the previously discussed Clifford gates to a universal gate set \cite{Bravyi_2005}. 
Ref.~\cite{Bravyi_2005} distinguished between T- and H-types of magic states given by the single qubit Clifford orbit of the states \cite{Bravyi_2005}

\begin{align}
\ketbra{H}&=\frac{I}{2}+\frac{1}{\sqrt{2}}\lr{\hat{X}+\hat{Z}},\\
\ketbra{T}&=\frac{I}{2}I+\frac{1}{\sqrt{3}}\lr{\hat{X}+\hat{Y}+\hat{Z}},
\end{align} 

\paragraph*{Example: $\CL=\sqrt{2}\Z^2$}
For the square GKP code we have already identified the logical Hadamard gate realized by $e^{-i\pi/2 \hat{n}}$ as the only Clifford gate realizable using passive Gaussian unitary. Furthermore (in codespace) the $+1$ Eigenstate of the Hadamard gate is unique such that we obtain $\ket{M}=\ket{H+}$ the $+1$ eigenvalued eigenstate of the logical Hadamard gate. This fact was observed in ref.~\cite{AllGaussian}, where it was also shown that performing quantum error correction allows for the production of those magic states.
\paragraph*{Example: $\CL=\sqrt{2}A_2$}
It was realized in ref.~\cite{GCB} that the hexagonal GKP code has a symplectic orthogonal automorphism that realizes the $\hat{H}\hat{P}^{\dagger}$ gate given by $\hat{U}_{HP^{\dagger}}=e^{-i\frac{2\pi}{3}\hat{n}}$. The logical $HP^{\dagger}$ gate is a symmetry of the $\ket{T}$-type magic state defined in ref.~\cite{Bravyi_2005}, such that the state $\ket{M}$ obtained by projecting the vacuum onto code space again yields a magic state.

For one mode the lattices $\CL$ denoted above can be uniquely described by a single parameter $\tau$ that transforms via $\tau \mapsto S^{-1}.\tau$ for  $S\in \SL_2\lr{\R}$ when the associated code space projector transforms with $\Pi_\CL\mapsto U_S\Pi_{\CL} U_S^{\dagger}$. Similarly, every Gaussian state can also be labeled by an element $z\in \hh$ by considering the unique state annihilated by $\hat{a}_z=\hat{p}-z\hat{q}$. This labeling is such that for a Gaussian unitary $U_S$ $\hat{a}_{z'} = U_S \hat{a}_{z}U_S^{\dagger}$ satisfies $z'^{-1}=S. z^{-1}$. This allows us to compactly describe the evolution of a state of type $\ket{M} $ under Gaussian unitary evolution

\begin{align}
\ket{M}&\mapsto U_S\ket{M}\\ \lr{\tau,\; z^{-1}} & \mapsto \lr{S^{-1}.\tau ,\, S. z^{-1}}.
\end{align}

In ref.~\cite{Royer_2022} some non-Clifford logical gates implementable via non-Gaussian unitary gates are identified such as $\sqrt{\hat{H}}$ and a version of a controlled Hadamard gate. It would be interesting to extend the geometric classification discussed in the main text to such gates, which is left for future work. 
\leavevmode\thispagestyle{empty}\newpage

\thispagestyle{empty}

\bibliographystyle{plain}
{\footnotesize
\addcontentsline{toc}{chapter}{References}  
\bibliography{thesis_bib}
}
\leavevmode\thispagestyle{empty}\newpage
\leavevmode\thispagestyle{empty}\newpage
\subsubsection*{Declaration of authorship}
\vspace{2cm}
\raggedright
\begin{itemize}
\item Name: Conrad
\item First name: Jonathan
\end{itemize}

I declare to the Freie Universit{\"a}t Berlin that I have completed the submitted dissertation independently and without the use of sources and aids other than those indicated.  The present thesis is free of plagiarism.  I have marked as such all statements that are taken literally or in content from other writings.  This dissertation has not been submitted in the same or similar form in any previous doctoral procedure. 

I agree to have my thesis examined by a plagiarism examination software. \vspace{1cm}\\

Date:\hrulefill\hspace{3cm}  Signature:\hrulefill

\end{document}